\documentclass[pdflatex,sn-apa]{Ostwald-et-al-R1-sn-jnl}
\usepackage{graphicx}
\usepackage{multirow}
\usepackage{amsmath}
\allowdisplaybreaks
\usepackage{amssymb}
\usepackage{amsfonts}
\usepackage{amsthm}
\usepackage{mathrsfs}
\usepackage[title]{appendix}
\usepackage{xcolor}
\usepackage{textcomp}
\usepackage{tabularx}
\usepackage{subcaption}
\usepackage{manyfoot}
\usepackage{booktabs}
\usepackage{placeins}
\usepackage{algorithm}
\usepackage{algorithmicx}
\usepackage{algpseudocode}
\usepackage{listings}
\usepackage{soul}
\usepackage{mathtools}
\usepackage{stmaryrd}
\DeclarePairedDelimiter{\ib}{\llbracket}{\rrbracket}
\DeclareMathOperator*{\argmax}{argmax}
\newtheorem{definition}{Definition}
\newtheorem{theorem}{Theorem} 
\usepackage{titlesec}
\renewcommand{\figurename}{Figure}
\renewcommand{\tablename}{Table}
\usepackage[font=small,labelfont=bf,labelsep=period,justification=justified,singlelinecheck=false]{caption}
\usepackage{enumitem}
\usepackage{kantlipsum}
\usepackage{lineno}

\newcommand*{\fullnameref}[1]{\hyperref[#1]{\ref*{#1}.\ \nameref*{#1}}}
\newcommand*{\fullref}[1]{\hyperref[{#1}]{\autoref*{#1}}}
\newcommand*{\defref}[1]{\hyperref[{#1}]{Definition \autoref*{#1}}}
\newcommand{\eqnref}[1]{eq.~\eqref{#1}}

\newcommand{\fitproofline}[1]{\begingroup\setbox0=\hbox{$\displaystyle #1$}\ifdim\wd0>0.92\linewidth\resizebox{0.92\linewidth}{!}{$\displaystyle #1$}\else\box0\fi\endgroup}
\begin{document}
\title[On Agentic Behavioral Modeling]{On Agentic Behavioral Modeling}
\author[1]{\fnm{Dirk} \sur{Ostwald}}
\author[2,3]{\fnm{Rasmus} \sur{Bruckner}}
\author[4]{\fnm{Franziska} \sur{Usée}}
\author[1]{\fnm{Belinda} \sur{Fleischmann}}
\author[1]{\fnm{Joram} \sur{Soch}}
\author[1]{\fnm{Sean} \sur{Mulready}}
\affil[1]{\orgdiv{Institute of Psychology}, \orgname{Otto von Guericke University Magdeburg}}
\affil[2]{\orgdiv{Institute of Psychology}, \orgname{University of Hamburg}}
\affil[3]{\orgdiv{Department of Education and Psychology}, \orgname{Freie Universität Berlin}}
\affil[4]{\orgdiv{Institute of Psychology}, \orgname{Philipps-Universität Marburg}}

\abstract{Integrating theoretical neuroscience, decision theory, and
probabilistic inference offers a promising route to understanding human
cognition, yet concrete methodological bridges between agentic artificial intelligence (AI) models and
behavioral data analysis remain formally underdeveloped. We advance this
synthesis within the framework of \textit{agentic behavioral modeling} (ABM),
which treats artificial agents as latent, generative hypotheses about cognitive
mechanisms and evaluates them by their statistical adequacy in explaining human
behavior. After outlining its conceptual foundations, we apply the framework to
two minimal laboratory paradigms: a binary perceptual contrast-discrimination
task and a symmetric two-armed bandit learning task. We formalize each
task–agent–data system as a joint probability model, derive explicit conditional
log-likelihoods for behavioral inference, validate different model variants
using model and parameter recovery simulations, and evaluate them in light of
empirical data. Using these minimal examples, we provide a novel agent-centric
interpretation of the psychometric function, derive optimal policies for both
tasks under explicit assumptions on sensory bias and post-decision noise,
and show that the Bayesian posterior expectation and induced decisions admit an
equivalent adaptive-rate Rescorla–Wagner representation in the symmetric bandit
task. More broadly, this work may serve as a
conceptual and practical foundation for applying ABM to cognitive behavioral
science.}

\maketitle

\vspace{2mm}

\small 
Corresponding author: Dirk Ostwald (dirk.ostwald@ovgu.de)

\normalsize 
\clearpage
\section{Introduction}\label{sec:introduction}
% \linenumbers
The formulation and empirical evaluation of neurocomputational theories of human
behavior are central to contemporary cognitive behavioral science. It is a
widely held belief that integrating perspectives from probabilistic inference,
decision theory, and neuroscience in this research area can contribute to a more
principled understanding of human cognition
\citep[e.g.,][]{lewandowsky2011, hsiao2024, naselaris2018, kriegeskorte2018}.
Such research may, in turn, inform approaches to longstanding problems in
artificial intelligence (AI), support the development of individualized
interventions in computational psychiatry, and improve the general transparency
and interpretability of cognitive behavioral science
\citep[e.g.,][]{guest2021, hassabis2017, montague2012, series2020}.

Both AI research and cognitive behavioral science have developed the concept of
\textit{artificial agents} as a central theoretical and methodological construct.
An agent is a computational system that perceives its environment through
inputs, selects actions according to some internal policy, and acts on that
environment to achieve specified objectives
\citep[e.g.,][]{russell2010, sutton2018, botvinick2020, friston2025}. Although
AI research and cognitive behavioral science share a deeply intertwined history
and intellectual tradition, an important distinction between these fields lies
in their fundamental goals. In AI research, the primary objective is to develop
agents that outperform prior approaches and, in contemporary commercial
contexts, yield deployable products. In contrast, cognitive behavioral science
treats artificial agents as scientific hypotheses about the computational
mechanisms of cognition. This approach implies that, from the perspective of
behavioral researchers, artificial agents are latent constructs that can only be
observed indirectly through participants’ behavior. Moreover, the central metric
for evaluating agents in cognitive behavioral science is their relative
statistical merit in explaining observed human behavior, rather than their task
performance. Finally, while AI research typically prioritizes scalability,
efficiency, and robustness, cognitive behavioral science places greater emphasis
on psychological plausibility and, most importantly, interpretability.

Several recent perspectives have articulated the potential benefits of more
fully integrating contemporary agentic AI with cognitive behavioral research
\citep[e.g.,][]{lieder2020, chandramouli2024, hess2025, lake2017}. Collectively,
these accounts emphasize that combining probabilistic inference and optimization
within an agent-centric framework — most commonly formalized in terms of
reinforcement learning (RL) and partially observable Markov decision processes
(POMDPs) — offers a promising avenue for advancing cognitive behavioral science.
Although such high-level reviews are valuable for mapping the conceptual
landscape and motivating this synthesis, they typically fall short of providing
concrete methodological guidance or in-depth worked examples of applications in
basic scientific settings. In the present work, we aim to address this gap by
making this program concrete under the moniker of \textit{agentic behavioral
modeling} (ABM) \citep{horvath2021, usee2025}.

Specifically, we first provide a conceptual overview of the ABM framework
(\autoref{sec:conceptual-framework}). In our formulation, ABM entails an
explicit decomposition of the model space into a \textit{task model}, an
\textit{agent model}, and a \textit{data model}. The task model furnishes a
formalized specification of the experimental paradigm. The agent model provides
a representation of the hypothesized cognitive mechanisms underlying task
execution and specifies the process by which the agent’s internal beliefs are
transformed into goal-directed decisions. Finally, the data model supplies the
statistical framework that integrates observed behavioral data with the task and
agent models, enabling parameter estimation and model comparison. In essence,
then, the ABM framework constitutes a language that enables the explicit formal
specification and systematic differentiation of the assumptions associated with
each individual model component, as well as the rigorous formalization of their
interactions.

Having provided a conceptual overview of the ABM framework, we examine its
application to two deliberately simple laboratory tasks drawn from decision
neuroscience and behavioral economics. The first is a binary perceptual
decision-making task in which participants discriminate between the contrasts of
two simultaneously presented Gabor stimuli
(\autoref{sec:gabor-contrast-discrimination}). The second is a symmetric
two-armed bandit learning task in which participants seek to maximize reward by
identifying the more profitable of two actions
(\autoref{sec:symmetric-bandit-learning}). Using these minimal and well-studied
experimental scenarios, we demonstrate how such tasks can be embedded within the
ABM framework and clarify the analytical consequences of doing so. To this end,
we place particular emphasis on treating the interaction between task, agent,
and data as a joint probability distribution in which core agent variables
remain latent from the experimenter's perspective. Echoing this perspective, we
derive explicit forms of the conditional log-likelihood functions that serve as
the central link between agent-based cognitive models and behavioral data
analysis. From a global point of view, the two case studies discussed herein can
be regarded as two complementary edge cases of an RL-POMDP landscape: The first
case study concerns agent-internal (perceptual) uncertainty but omits
trial-by-trial state dependencies, while the second case study concerns
agent-internal trial-by-trial dependencies (learning) but omits additive
internal uncertainty \citep{bach2012, summerfield2012, bruckner2025}. While our
case studies and behavioral findings align with established results in the
literature, their in-depth study in the ABM framework also allows us to make
several specific novel contributions: first, we provide a new, agent-centric
interpretation of the well-known psychometric function in binary perceptual
choice; second, we provide explicit derivations of optimal behavioral policies
for both experimental paradigms under explicit assumptions on sensory bias and
post-decision noise. Third, we explicate an adaptive-rate Rescorla–Wagner
representation of the Bayesian posterior expectation and induced decisions in
the symmetric bandit task.

Notational conventions are summarized at the end of
\autoref{sec:conceptual-framework}, and experimental details are provided in
\fullnameref{sec:participants-and-procedure}. The key mathematical results are
presented in the main text, while their formal derivations are deferred to
\fullnameref{sec:gabor-contrast-discrimination-supplement} and
\fullnameref{sec:symmetric-bandit-learning-supplement}. Throughout, we make
direct reference to the implementation level by indexing the corresponding
Python functions in the study's associated data repository.

\clearpage
\section{Conceptual framework}\label{sec:conceptual-framework}

In ABM, we are concerned with the fundamental setting of individual participants
performing a cognitive task. Our aim is to devise and compare different
algorithmic-level theories that participants may entertain to perform well on
the task, as inferred from their observable behavior. To achieve this, we
partition the fundamental setting into a \textit{task model}, an \textit{agent
model}, and a \textit{data model} (\fullref{fig:1}A). By the task model, we
understand a probabilistic formulation of the choice environment; by the agent
model, we understand an algorithmic formulation of a task solution; and by the
data model, we understand the statistically embedded task and agent models. In
the following, we briefly indicate what these three models typically entail,
without claim to completeness or exclusivity.

A \textit{task model} commonly comprises structural components, such as sets of
\textit{states}, \textit{actions}, and \textit{rewards}, as well as dynamical
components, such as \textit{state-state transitions} and \textit{action- and
state-dependent reward probabilities}. From the perspective of an agent
interacting with the task, the task state may comprise both indirectly
observable (latent) and observable components. Often, the observable state
components have distributions that are conditional on the latent state
component. In the ABM framework, the task model makes the cognitive task
explicit and amenable to algorithmic solution. From a practical perspective, the
task model is best conceived as an abstract formulation of the program
implementing the task, such as a PsychoPy \citep{peirce2007, peirce2019},
Psychotoolbox \citep{brainard1997}, or Unity \citep{brookes2020} script.
Ideally, such a program is fully integrated into the ABM framework and
constitutes a special agent that interacts with the human participant performing
the task, minimizing programming code variance between the actual experiment and
its computational model.

An \textit{agent model} is a computational system that interacts with the task
in order to achieve a certain goal and serves as a hypothesis about the
participants' cognitive processes when performing the task. An agent model
typically embodies some knowledge of core task aspects. This is usually
represented by including some task model features as properties of the agent
model object. From an experimental perspective, this corresponds to carefully
informing the participant about the task aspects that the agent model uses to
solve the task. For example, if an action yields a reward with a given but not
directly observable probability, and this probability is learned by an agent
model, it is sensible to inform the participant about the probabilistic nature
of the action-conditional reward (but not its precise value). On a
trial-by-trial basis, agents generally have some knowledge about the current
state of the environment, which is captured by their \textit{belief state}
\citep[e.g.,][]{astrom1965, smallwood1973, georgee.monahan1982, puterman1994, bertsekas1995, sutton1998, krishnamurthy2016, kochenderfer2022}.
Belief states are probability distributions on latent task states that can, for
example, be used to infer the most likely current task state value
\citep[e.g.,][]{kalman1960, briers2010, sarkka2013}. To entertain belief states,
agents require a \textit{belief state update function} that determines how
recent observations are integrated with the agent's previous knowledge. Based on
their beliefs about the task environment, agents typically assign values to
potential avenues of action
\citep[e.g.,][]{pratt1995, russell2010, glimcher2014}. We describe this value
allocation by means of \textit{decision value functions} that ascribe numerical
values to the decisions available to the agent. Because agents are assumed to
make a choice on every trial of the task, a core principle of the framework is
that they choose the option that maximizes their decision value function, a
process we take to be implemented by a \textit{decision function}
\citep[e.g.,][]{wald1950, raiffa1967, berger1985}.

Formally, the composition of a belief state function, a decision value function,
and a decision function constitutes a function that maps the agent's current
knowledge state to a decision. Functions of this form are commonly referred to
as \textit{policies} in the operations research and RL literature
\citep[e.g.,][]{puterman1994, krishnamurthy2016}, a convention we follow here.
However, in contrast to work in these fields, the focus in behavioral modeling
is not necessarily on \textit{optimal policies}. Although any policy can be
made optimal under a suitably chosen, sufficiently unconstrained cost function,
such a construction need not provide a meaningful explanation of behavior.
Oftentimes, optimality considerations play no role in behavioral modeling.
An example of this is the common reliance
on Rescorla–Wagner-based learning algorithms with free learning rates that may
or may not be good estimators of a decision value in a given scenario
\citep[e.g.,][]{siegel1996, wilson2019, yau2023, soto2023}. In the current work,
we explicitly consider the optimality of some agent policies but do not mean to
imply that optimality considerations are a necessary prerequisite for ABM.

A further important distinction for ABM is between \textit{stochastic} and
\textit{deterministic} agents. Like stochastic policies, stochastic agents
integrate some degree of internal randomness into their decision architecture,
which, in the case of ABM, is not directly observable from the experimenter's
perspective. In essence, given identical task histories, the decision of a
stochastic agent is non-deterministic, whereas the decision of a deterministic
agent is deterministic. From a behavioral modeling perspective, stochastic
agents endow an ABM with an additional level of non-observable uncertainty that
needs to be integrated out. In the current work, our first case study features a
stochastic agent model, whereas the agent models in the second case study are
fully deterministic.

The interaction between agent models and task models described thus far is a
typical feature of contemporary agentic AI research
\citep[e.g.,][]{brockman2016, mnih2015, bertsekas2019, sutton2018}. When
applying this paradigm to behavioral research, however, we must carefully
distinguish between the observable and non-observable aspects of an agent's
behavior. While the aim in agentic AI research is to devise agents that, in some
sense, perform optimally on a given task, in behavioral modeling, the aim is
ultimately to evaluate different algorithmic theories in light of human
behavioral data. In line with standard approaches in statistical inference, it
is thus beneficial to view an algorithmic agent as a (deterministic or
stochastic) theory of the participant's cognitive process for solving the task
at hand while simultaneously exercising some modesty with respect to its
universal validity, i.e., to account for our uncertainty as experimentalists.

\FloatBarrier
\clearpage
\vspace*{\fill}
\begin{center}
\includegraphics[width=\linewidth]{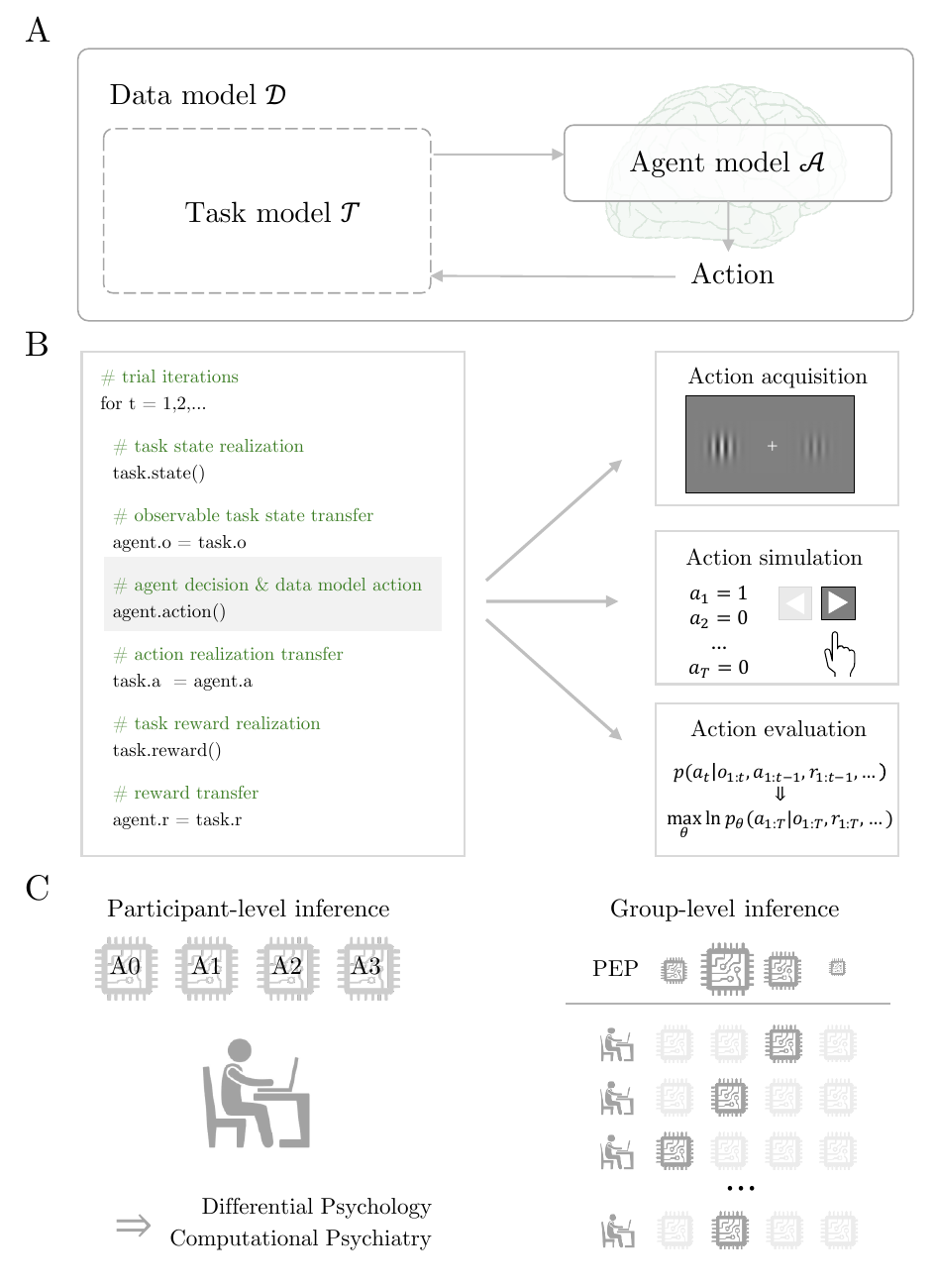}
\end{center}
\vspace*{\fill}
\clearpage

\begingroup
\noindent\rule{\linewidth}{0.4pt}\par
\captionof{figure}{\textit{(Figure on previous page.)} Agentic behavioral
modeling (ABM). (A) Basic model partitioning of the ABM scenario. A \textit{task
model} describes the agent choice environment, an \textit{agent model} describes
an algorithmic trial-by-trial solution of the task, commonly aiming to capture
human decision processes, and a \textit{data model} accounts for the latent
nature of the agent's internal dynamics and their indirect observation under
measurement noise. (B) Core implementation of the ABM framework. At the heart of
the implementation of a given ABM scenario is a routine that realizes the
trial-by-trial dynamics of the agent interacting with the task while being
observed by the experimenter. In the code accompanying the case studies of
\autoref{sec:gabor-contrast-discrimination} and
\autoref{sec:symmetric-bandit-learning}, this routine is implemented in
\textit{abm\_abm.py}. Commonly, such implementations make use of standard
object-oriented programming concepts that allow both data and
functionality to be bound in single entities, such as `task' and `agent'. Crucially, such a
core routine usually has to subserve at least three functions, as shown on the
right: first, it can be used for experimental presentation to human
participants; second, it can simulate artificial behavioral datasets; finally,
it can evaluate conditional action probabilities which, in logarithmic form,
define the ABM conditional log-likelihood. (C) Levels of inference in ABM. Agent
model comparison is facilitated by ABM both at the level of individual
participants, serving inferential objectives in differential psychology and
computational psychiatry, and, at the group level, providing nomothetic
insights. For the latter, the ABM approach is augmented with a protected
exceedance probability (PEP) Bayesian model comparison approach that allows data
to be integrated across participants by evaluating null-hypothesis-adjusted
group-level posterior model probabilities.}
\label{fig:1}
\endgroup

The common inferential approach to this end is to model both a theory-based
mechanistic component of the phenomenon at hand and, at the same time, account
for the residual data variation that may not be explained by the mechanistic
model \citep[e.g.,][]{gauss1809, laplace1810}. In the statistical literature,
this second component is variably referred to as \textit{observation error},
\textit{residual variance}, or \textit{noise}
\citep[e.g.,][]{stigler1986, hald2007, searle1971}. In the ABM framework, we
make this distinction explicit by distinguishing \textit{agent decisions} from
\textit{agent actions} \citep[cf.][]{krpan2025}. Here, a decision is an internal,
unobserved selection from the decision space, which has the same elements as
the action space. Specifically, we assume that
the agent's decisions are the direct consequence of the agent's belief state,
decision value allocation, and decision value maximization process. Agent
actions, on the other hand, are conceived as the post-decision, observable
readout of agent decisions that, on a trial-by-trial basis, may or may not
correspond to the agent's decisions. As is common, for example, in the
application of drift-diffusion models to reaction time and accuracy data, we
thus assume a post-decision error process, which we describe probabilistically
\citep[cf.][]{ratcliff1978, ratcliff2008}. In ABM, decisions and actions are thus
conceived as distinct random variables.

To this end, the \textit{data model} corresponds to a decision-conditional
action distribution that is typically parameterized to span a continuum of
action distributions. For example, a deterministic, value-maximizing agent
would select a given decision/action option on every trial if that option
consistently had the highest value. The decision-conditional action distribution
allows observed actions to deviate from this internal decision. In our case
studies, this is implemented by selecting a Bernoulli distribution over binary
actions according to the agent's internal decision, with a parameter controlling
the probability of a post-decision deviation. This post-decision variability is
distinct from random exploration, which in ABM would correspond to a stochastic
agent's internal decision process. Such a process would be augmented by the
possibility of deviations from it through a decision-conditional action
distribution, requiring inference to account for the unobserved internal
dynamics, for example through Bayesian filtering or other latent-state
estimation approaches.

Together, the architecture of task, agent, and data models corresponds to a
parameterized joint probability distribution of task, agent, and action random
variables across all trials of a given experiment. This joint distribution fully
describes the postulated interaction between the agent and the task based on its
distributional form and encoded conditional independence properties, and it
accounts for the fact that internal agent variables are not directly observable
from the experimenter's perspective. It is this joint probability distribution,
then, that we refer to as an ABM, and we explicate ABMs as joint probability
distributions in the two case studies discussed below.

Typically, task random variables are fully observable from the perspective of
the experimenter but only partially observable from the perspective of the
agent. The agent variables, in turn, are fully observable from the perspective
of the agent but not from the perspective of the experimenter, whereas the
action random variables are observable from the perspectives of the task, the
agent, and the experimenter. The fundamental goal of behavioral modeling, then,
is to evaluate the conditional probability of a participant's actions given
the task inputs and interaction history available before each action,
marginalizing non-observable internal agent variables. The sum of the logarithms
of these trialwise conditional probabilities is referred to as the
\textit{conditional log-likelihood}. In the examples below, this defines either
a conditional likelihood given exogenous task inputs or a joint action--reward
likelihood up to a parameter-independent factor. It is usually a function of model
parameters, for example, internal agent parameters such as sensory sensitivities
or learning rates, as well as the post-decision noise parameter of the data
model. A standard approach to estimating these parameters, based on a single
participant's dataset, is to determine them using the maximum likelihood
principle and to utilize the maximized conditional log-likelihood to evaluate
the plausibility of different model architectures in light of the observed data
\citep{fisher1922, fisher1925b, wilks1938, cox1974, aldrich1997, wilson2019}.

There are two principal ways to evaluate the conditional log-likelihood function
of an ABM. First, from a computational perspective, given an implementation of
the ABM architecture, this architecture can be subjected to the experimental
task variables and, instead of simulating agent actions, can be used to evaluate
the probability of an observed participant action based on the task-agent
interaction history (\fullref{fig:1}B). While this approach necessarily leaves
the mathematical form of the conditional log-likelihood function implicit, it
can nevertheless be validated using model and parameter recovery simulations and
can form a legitimate basis for data analytical procedures
\citep{horvath2021, usee2025}. We refer to this mathematically implicit form of
the conditional log-likelihood function as an ABM's \textit{agentic conditional
log-likelihood function}. Given the complexity of current ABM applications,
using agentic conditional log-likelihood functions seems to be the norm;
however, this is often not explicitly discussed in scientific communications.
Second, from an analytical perspective, given a formulation of the ABM joint
distribution, the trialwise conditional action probabilities described above
can be derived (or at least approximated) mathematically by marginalizing
non-observable internal agent variables. The sum of their logarithms can then
be supplemented with the
experimental task variables and observed participant actions and considered a
function of the models' parameters. We refer to this mathematically explicit
form of the conditional log-likelihood function as an ABM's \textit{analytical
conditional log-likelihood function}. Below, we exemplify the derivation and
functional form of the analytical conditional log-likelihood function in both
case studies. From a theoretical perspective, this serves to explicate this
central function in the ABM approach by example and to validate the pragmatic
computational approach that dominates the literature. Finally, it may also
provide a basis for a more analytical approach to understanding the expressivity
and identifiability of ABM frameworks in the future.

As previously mentioned, evaluating the conditional log-likelihood function
serves two fundamental purposes in model-based scientific inference: parameter
estimation and model comparison
\citep[e.g.,][]{shao2003, claeskens2008, aitkin2010, berger2024}. In ABM, model
comparisons typically concern both individual participant-level and group-level
inference \citep[e.g.,][\fullref{fig:1}C]{cronbach1957, molenaar2004}. First,
one is often interested in determining the most plausible agent model for an
individual participant and, if model and parameter identifiability permit,
participant-specific model parameter estimates that can serve as measures of
individual particularities. This level is of main interest in the study of
psychological inter-individual differences and computational psychiatry
\citep[e.g.,][]{curran2011, schwartenbeck2016, karvelis2023}. On the other hand,
there remains broad interest in how far individual cognitive strategies may
generalize at the group level
\citep[e.g.,][]{shadish2001, ince2021, yarkoni2022, lakens2022, meteyard2020, schnack2016, piray2025}.

A versatile method for simultaneously addressing both goals is the
\textit{protected exceedance probabilities} (PEP) framework introduced by
\citet{stephan2009} and \citet{rigoux2014}. In brief, the PEP framework rests on
a summary-statistics approach to hierarchical inference that uses
participant-level approximations of the marginal log-likelihood as the basis for
group-level Bayesian model comparison. Accordingly, it necessitates the separate
likelihood-based assessment of ABMs for each participant, in line with an
inter-individual differences perspective. On the other hand, it integrates
participant-level information to evaluate the posterior probability that a given
model has the highest true but unknown generative probability under the
alternative hypothesis that not all model probabilities are equal. The protected
exceedance probability combines this probability with a uniform contribution
under the null hypothesis, weighted by the posterior probabilities of the two
hypotheses. The original
(unprotected) exceedance probability framework has been shown to be less
outlier-prone than fixed-effects Bayesian group model comparisons based on log
model evidence sums \citep{stephan2009}. In addition, the PEP framework
eliminates the overconfidence bias of unprotected exceedance probabilities that
arises from discounting the null hypothesis that all model probabilities are
equal as a potential explanation for apparent differences in model
probabilities. Since the PEP framework for group-level model inference is
sufficiently general to accommodate any particular ABM scenario, we only mention
it in passing in the case studies below and fully review and derive it in
\fullnameref{sec:protected-exceedance-probabilities}. The corresponding
implementation is available in \textit{abm\_bmc.py}.

Having outlined the ABM conceptual framework from the task model to group-level
inference, we now make it concrete with two case studies.
To this end, we use the following notational conventions.
We use indexed lowercase roman letters such as $x_t$ to denote random variables.
As a shorthand notation for sets of random variables, we use the colon notation
such that, e.g., $x_{1:n} := \{x_1,...,x_n\}$. We use $p$ to denote the
probability mass and/or density functions of joint, conditional, and marginal
random variable distributions, such as $p(x_t,y_t), p(x_t|y_t)$, and $p(x_t)$.
For the functional forms of the distributions of interest herein, we use the
shorthand notations listed in \fullref{tab:probability-distributions}. We use
the Iverson bracket notation \citep{iverson1962}
\begin{equation}
\ib{P} := 
\begin{cases} 
1 & \mbox{ if $P$ is true } \\
0 & \mbox{ if $P$ is false }  
\end{cases}
\end{equation}
to denote the indicator of a proposition $P$ such
as $x_t = y_t, x_t < 0$ or $x_t \le 0$.

\begin{table}[!t]
\caption{Parametric distributions of interest. $\Gamma$ denotes the gamma function.}
\label{tab:probability-distributions}
\centering
\small
\renewcommand{\arraystretch}{1.8}
\begin{tabularx}{\linewidth}{@{}l X@{}}
\toprule
Bernoulli PMF
& $\mathrm{Bern}(x; \mu) := \mu^{x}(1-\mu)^{1-x}$ for $x\in \{0,1\}$, $\mu \in [0,1]$ 
\\ \midrule
Uniform PDF
& $\mathrm{Unif}(x;I) := \frac{1}{b-a}\ib{x\in I}$ for $x\in\mathbb{R}$ and
$I$ an interval with endpoints $a<b$
\\ \midrule
Normal PDF 
& $\mathrm{Norm}(x; \mu, \sigma^2) :=
\frac{1}{\sqrt{2\pi\sigma^2}} 
\exp\left( -\frac{1}{2\sigma^2}(x-\mu)^2 \right)$ for $x,\mu \in \mathbb{R}, \sigma^2 > 0$
\\ \midrule
Normal CDF & 
$\Phi(x; \mu, \sigma^2) :=
\frac{1}{\sqrt{2\pi\sigma^2}} 
\int_{-\infty}^{x} 
\exp\left( -\frac{1}{2\sigma^2}(s-\mu)^2 \right) ds$ for $x,\mu \in \mathbb{R}, \sigma^2 > 0$
\\ \midrule
Beta PDF
& $\mbox{Beta}(x; (\alpha_1,\alpha_2)) :=
\frac{\Gamma(\alpha_1 + \alpha_2)}
{\Gamma(\alpha_1)\Gamma(\alpha_2)}
x^{\alpha_1 - 1}(1-x)^{\alpha_2 - 1}$ for $x \in ]0,1[, \alpha_1,\alpha_2 > 0$
\\ \bottomrule
\end{tabularx}

\end{table}

\clearpage

\section{Gabor contrast discrimination}\label{sec:gabor-contrast-discrimination}
\subsection{Experimental paradigm and descriptive statistics}

As a first case study, we consider a perceptual choice task (\fullref{fig:2}).
In each task trial,
participants were presented with two Gabor patches on the left and right of a
central fixation cross (\fullref{fig:2}A). One of the Gabor patches' contrasts
was larger than that of the other. The participants' task was to identify the
laterality of the Gabor patch with the higher contrast by pressing one of two
buttons (left cursor or right cursor). The correct response yielded a reward of
$+1$, while the incorrect response yielded a reward of $+0$. The laterality of
the Gabor patch with the higher contrast and the absolute size of the contrast
difference between the Gabor patches were chosen at random in each trial. Each
of $n = 59$ participants completed 300 trials in 10 blocks of 30 trials each. On
average, the participants correctly identified the laterality of the Gabor patch
with the higher contrast in 82.7\% of the trials, with a standard deviation of
5.9\%, corresponding to an average reward rate of 0.827 (SD $0.059$)
(\fullref{fig:2}B). Plotting the group average proportion of ``right'' responses as
a function of the normalized contrast difference between the Gabor patches, with
positive values of the contrast differences indicating a higher contrast of the
Gabor patch on the right side, yielded the familiar sigmoidal psychometric
function
\citep[e.g.,][]{green1966, barlow1971, falmagne1980, treutwein1999, wichmann2001, gold2007, hanks2017, carandini2024}.
Individual participants showed similar behavior, with some variability
concerning the proportion of ``right'' responses toward the boundaries of the
normalized contrast difference space (``stimulus-independent errors'' or
``lapses'') and the location of the 50\% ``right'' response rate with respect to a
normalized contrast difference of 0 (``bias'') (\fullref{fig:S1}). For more
in-depth coverage of the experimental procedures and the descriptive analyses
reported, please refer to \fullnameref{sec:participants-and-procedure},
\fullnameref{sec:gabor-contrast-discrimination-supplement}, and
\textit{abm\_describe.py}.

\begin{figure}[!htbp]
\includegraphics[width=\linewidth]{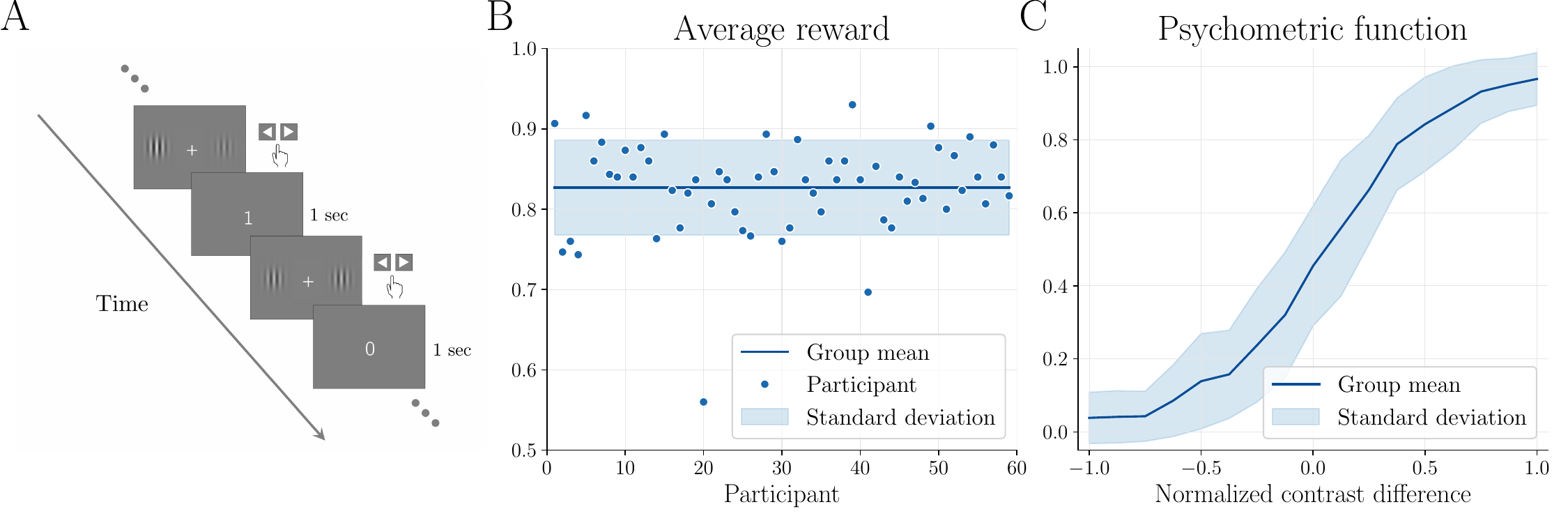}
\caption{Gabor contrast discrimination task and descriptive statistics. (A)
Gabor patch contrast discrimination trial sequence. Two Gabor patches differing
in contrast were presented to participants, who were required to indicate
whether they perceived the contrast of the left or the right Gabor patch as
higher. The correct response yielded a reward of $+1$, the incorrect response a
reward of $+0$. Rewards were presented immediately after the response for 1
second. (B) Average reward obtained by the group of $n = 59$ participants. On
average, participants responded correctly and thus obtained a reward of $+1$ in
82.7\% of the trials with a group standard deviation of 5.9\%. (C) Group
psychometric function. The group mean and standard deviation correspond to the
proportion of ``right'' responses as a function of the normalized contrast
difference. Negative normalized contrast differences correspond to a physically
higher contrast of the left Gabor patch, whereas positive normalized contrast differences
correspond to a physically higher contrast of the right Gabor
patch.}\label{fig:2}
\end{figure}

\subsection{Agentic behavioral model}

In the following, we formulate an ABM for the Gabor contrast discrimination task. 

\paragraph{Task model} We describe the experimental task using the tuple
\begin{equation}\label{eq:T-1}
\mathcal{T} := (S,C,A,R,p(s_t), p(c_t|s_t), p(r_t|s_t,a_t)),
\end{equation}
where
\begin{itemize}
\item $S := \{0,1\}$ is the state set, modeling the laterality of the higher
contrast Gabor patch, with $0$ corresponding to the left and $1$ corresponding
to the right,
\item $C := [-\kappa,\kappa]$ with $\kappa > 0$ is the contrast difference set,
modeling the contrast difference in a given trial, where $\kappa$ specifies
the maximum absolute contrast difference,
\item $A := \{0,1\}$ is the action set, modeling button presses ($0$ = left, $1$
= right),
\item $R := \{0,1\}$ is the reward set,
\item the state distribution is the distribution of the state random variable
$s_t$, taking values in $S$, and is given by
\begin{equation}
p(s_t) :=\mathrm{Bern}\left(s_t;0.5\right),
\end{equation}
\item the state-conditional contrast difference distribution is the conditional
distribution of the contrast difference random variable $c_t$, taking values
in $C$, given $s_t$, and is given by  
\begin{equation}
p(c_t|s_t) := \mathrm{Unif}\left(c_t;[-\kappa,0]\right)^{1-s_t} \mathrm{Unif}\left(c_t; ]0,\kappa]\right)^{s_t},
\end{equation}
\item the state- and action-conditional reward distribution is the conditional
distribution of the reward random variable $r_t$, taking values in $R$, given
$s_t$ and the action random variable $a_t$, taking values in $A$, and is given
by
\begin{equation}
p(r_t|s_t,a_t) :=\mathrm{Bern}(r_t; \ib{a_t = s_t}).
\end{equation}
\end{itemize}
From a neurocognitive perspective, the interaction between task and agent in
each trial $t = 1,...,T$ thus evolves as follows (\fullref{fig:3}, task model).
First, the laterality of the higher contrast Gabor patch $s_t$ is realized from
$p(s_t)$ based on a Bernoulli distribution with parameter $0.5$. Then, according
to $p(c_t|s_t)$, for a higher contrast on the left, the contrast difference
$c_t$ is realized based on a uniform distribution on the interval $[-\kappa,0]$,
or, for a higher contrast on the right, the contrast difference $c_t$ is
realized based on a uniform distribution on the interval $]0,\kappa]$. For a
given state, the contrast difference may thus assume a more or less pronounced
level. Finally, if the agent's action $a_t$ matches the laterality of the higher
contrast Gabor patch, a reward of $+1$ is realized from $p(r_t|s_t,a_t)$;
otherwise, a reward of $+0$ is realized from $p(r_t|s_t,a_t)$. Note that the
state evolution is independent of the agent's action, and the agent's decision
is independent of the realized rewards. For implementational details of the task
model, please refer to \textit{abm\_task.py}.

\begin{figure}[!t]
\includegraphics[width=\linewidth]{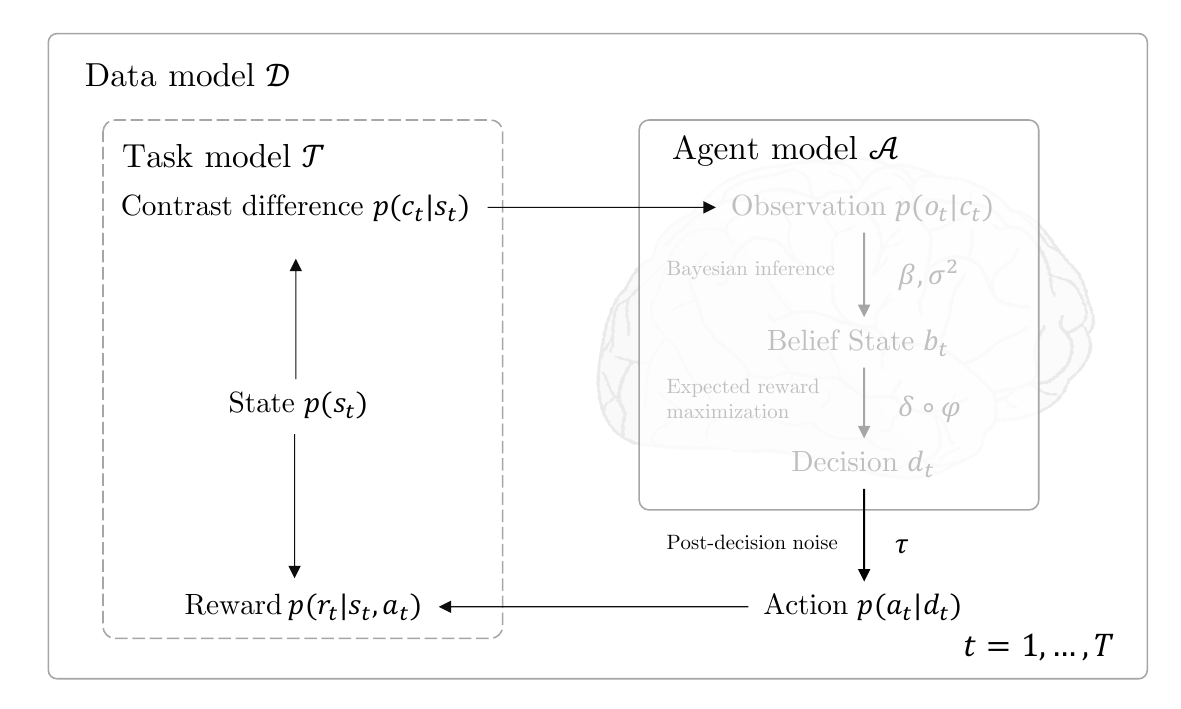}
\caption{Neurocognitive perspective of the Gabor contrast discrimination
ABM.}\label{fig:3}
\end{figure}

\paragraph{Agent model}
We consider the agent model
\begin{equation}\label{eq:A-1}
\mathcal{A} := \left(\mathcal{T}, O, B, V, D, p(o_t|c_t), \beta, \phi,\delta \right),
\end{equation}
where  

\begin{itemize}
\item $\mathcal{T}$ is the task model and represents the agent's knowledge of
the task properties,
\item $O := \mathbb{R}$ is the set of observations,
\item $B$ is the set of belief states
\begin{equation}
B := \left\lbrace b | b : S \to \mathbb{R}_{\ge 0} \mbox{ with } \sum_{s \in S} b(s) = 1 \right\rbrace, 
\end{equation}
i.e., the set of all probability mass functions on $S$,
\item $V := \{0,1\}^{2}$ is the decision value set, 
\item $D := \{0,1\}$ is the decision set, modeling the choice to execute a
specific button press (0 = left, 1 = right),
\item $p(o_t|c_t)$ is the contrast difference-conditional observation
distribution of the observation random variable $o_t$, taking values in $O$,
and is given by 
\begin{equation}\label{eq:p_o_c}
p(o_t|c_t) := \mathrm{Norm}(o_t;c_t+\eta,\sigma^2),
\end{equation}
for a sensory bias $\eta \in \mathbb{R}$ and a sensory standard deviation
$\sigma > 0$,  
\item $\beta$ denotes the belief state function
\begin{equation}\label{eq:beta-1}
\beta : O \to B, o_t \mapsto \beta(o_t) := 
\begin{pmatrix}
p(s_t = 0|o_t) \\
p(s_t = 1|o_t) \\
\end{pmatrix}
=:
\begin{pmatrix}
b^0_t \\
b^1_t \\
\end{pmatrix}
=: b_t,
\end{equation}
such that the belief state $b_t$ takes values in $B$. As shown in
\autoref{thm:belief-state-1}, in the absence of sensory bias, $\eta=0$,
\begin{equation}\label{eq:belief-state-1}
p(s_t|o_t) = 
\frac{\left(\Phi(0;o_t,\sigma^2) - \Phi(-\kappa; o_t, \sigma^2)\right)^{1-s_t}
      \left(\Phi(\kappa;o_t,\sigma^2) - \Phi(0; o_t, \sigma^2)\right)^{s_t}}
     {\Phi(\kappa;o_t,\sigma^2) - \Phi(-\kappa;o_t,\sigma^2)},
\end{equation}
\item $\phi$ denotes the decision value function
\begin{equation}
\phi : B \to V, b_t \mapsto \phi(b_t) := \left(v^d_t\right)_{d\in D} =: v_t
\mbox{ with }
v^d_t := \ib{b^d_t \ge 0.5},
\label{eq:phi-1}
\end{equation}
such that the decision value vector $v_t$ takes values in $V$.
\item $\delta$ denotes the agent's decision function
\begin{equation}
\delta : V \to D, v_t \mapsto \delta(v_t) := \min\left(\argmax_{d \in D} v_t^d\right) =: d_t,
\label{eq:delta-1}
\end{equation}
such that the decision random variable $d_t$ takes values in $D$. Ties between
decision values are resolved in favor of decision $0$.
\end{itemize}
The thus formalized agent decision architecture is motivated by both
neurocognitive and rational analysis perspectives. First, from a neurocognitive
perspective, the agent's decision-making process on a given task trial unfolds
as follows (\fullref{fig:3}, agent model). Based on an external contrast
difference $c_t$, the agent makes an observation $o_t$ in its sensorium that
corresponds to the external contrast difference under the additive contributions
of a deterministic, trial-independent sensory bias $\eta$ and a random,
trial-dependent, normally distributed, zero-centered error term of variance
$\sigma^2$. From a neurobiological perspective, $\eta$ can be interpreted as a
slow time-scale shift in the agent's sensory apparatus that induces the
observation of a contrast difference when there is, in fact, none in the
external world \citep[e.g.,][]{dosher1998, greenlee1988, bao2012, webster2015}.
The random, trial-specific error term, on the other hand, can be interpreted as
a fast time-scale perceptual bias, such as transient changes in visual attention
or contrast gain control
\citep[e.g.,][]{georgeson1985, bex2007, carandini2012, kohn2007}. Note that the
assumption of such additive noise processes is not unique to the framework
considered here but is implicit in many psychometric models, such as the
classical signal detection framework by \citet{green1966}. Crucially, while
$c_t$ is observable from the perspective of the experimenter, it is only
indirectly observable from the perspective of the agent. In contrast, $o_t$ is
observable from the perspective of the agent and only indirectly observable from
the perspective of the experimenter.

Upon observation of $o_t$, the agent is conceived to form a belief state $b_t$
with respect to the decision-relevant state $s_t$ by means of the belief-state
function $\beta$. For each $o_t$, $\beta$ returns a discrete probability
distribution over the two possible states of $s_t$ in the form of the
two-dimensional stochastic vector $b_t$. Its elements are evaluated under the
agent's unbiased observation model after integrating the intervening variable
$c_t$, as detailed in \autoref{thm:belief-state-1} and documented in
\eqnref{eq:belief-state-1}. Based on $b_t$, the agent's decision-value function
$\phi$ assigns decision value $1$ to each decision whose corresponding belief
is at least $0.5$, and decision value $0$ otherwise. When both beliefs equal
$0.5$, both decision values equal $1$. Finally, using the decision function
$\delta$, the agent selects a decision with maximal decision value, resolving
ties in favor of decision $0$. Notably, like $o_t$, the random vectors $b_t$,
$v_t$, as well as the agent's decision $d_t$ are conceived as latent from the
perspective of the experimenter. However, conditional on the agent's observation
$o_t$, it is assumed that these quantities follow deterministically through the
composition of the agent's belief-state, decision-value, and decision functions.

From a rational analysis perspective, the deterministic Markovian policy
implemented by the agent's decision architecture maximizes the agent's expected
cumulative reward across all trials of the task under the assumptions of zero
sensory bias and zero post-decision noise, $\eta=0$ and $\tau=0$, as we show in
\autoref{thm:decision-optimality-1}. Finally, as shown in
\autoref{thm:conditional-decision-distribution-1}, the agent construction allows
one to express the agent's observation-conditional decision distribution, i.e.,
the distribution of the random variable
\begin{equation}
d_t := (\delta \circ \phi \circ \beta)(o_t)
\end{equation}
in the form 
\begin{equation}\label{eq:p-dt}
p(d_t|o_t) = \ib{o_t \le 0}^{1 - d_t}\ib{o_t > 0}^{d_t}.
\end{equation}
Using \eqnref{eq:p-dt}, we may thus equivalently express the agent model of
\eqnref{eq:A-1} in full probabilistic form as
\begin{equation}\label{eq:A-p-1} 
\mathcal{A} := \left(\mathcal{T}, O, B, V, D, p(o_t|c_t), p(d_t|o_t)\right), 
\end{equation}
which greatly simplifies the derivation of the ABM's conditional log-likelihood
function in the following.

\paragraph{Data model}
We consider the data model
\begin{equation}\label{eq:D-1}
\mathcal{D} := \left(\mathcal{T}, \mathcal{A}, p(a_t|d_t) \right),
\end{equation}
where
\begin{itemize} 
\item $\mathcal{T}$ is the task model,
\item $\mathcal{A}$ is the agent model,
\item $p(a_t|d_t)$ is the decision-conditional action probability distribution 
\begin{equation}\label{eq:p-at-giv-dt}
p(a_t|d_t) := \mathrm{Bern}(a_t;1-\tau)^{d_t}\mathrm{Bern}(a_t;\tau)^{1-d_t} \mbox{ for } \tau \in [0,0.5].
\end{equation}
\end{itemize}
\noindent From a behavioral modeling perspective, this data model embeds the
task-agent interaction in a probabilistic framework for statistical inference.
In particular, the decision-conditional action probability distribution is
constructed so that $\tau$ is the probability that an observable action $a_t$
differs from the agent's latent decision $d_t$, while $1-\tau$ is the
probability that they correspond.

\paragraph{Agentic behavioral model}

Taken together, the definitions of the task model's state, contrast difference,
and reward distributions in \eqnref{eq:T-1}, the agent model's observation and
decision distributions in \eqnref{eq:A-p-1}, and the data model's conditional
action distribution in \eqnref{eq:D-1} imply a joint probability distribution on
$s_t$, $c_t$, $o_t$, $d_t$, $a_t$, and $r_t$ across all trials $t = 1,...,T$. In
particular, the specific forms of the aforementioned distributions encode
conditional independence assumptions for these random variables that render
their joint distribution expressible in the form
\begin{multline}\label{eq:abm-1}
p(s_{1:T}, c_{1:T}, o_{1:T},d_{1:T},a_{1:T}, r_{1:T})   
\\
= 
\prod_{t=1}^T p(s_t)p(c_t|s_t) p(o_t|c_t)p(d_t|o_t)p(a_t|d_t)p(r_t|a_t,s_t).    
\end{multline}
Crucially, this joint distribution encodes the trial-by-trial interaction
between the task and the agent, as well as the observation of this process by
the experimenter. As such, this joint distribution is what we refer to as an ABM
for the Gabor contrast discrimination task.

\paragraph{Analytical conditional log-likelihood function} As shown in
\autoref{thm:conditional-log-likelihood-function-1}, the conditional
log-likelihood function of the ABM defined by \eqnref{eq:abm-1} takes the form
\begin{multline}\label{eq:llh-1} 
\ell: \mathfrak{S} \times \mathfrak{E}  \times \mathfrak{T}  \to \mathbb{R}, 
(\sigma,  \eta,  \tau) \mapsto \ell(\sigma, \eta, \tau) 
\\ 
:= \sum_{t = 1}^T \ln \left(\mathrm{Bern}(a_t;\tau)\Phi(0;c_t + \eta,\sigma^2) +\mathrm{Bern}(a_t;1-\tau)(1-\Phi(0;c_t + \eta,\sigma^2))\right).
\end{multline}
Here, $\mathfrak{S}$, $\mathfrak{E}$, and $\mathfrak{T}$ denote the parameter
space axes, which, in line with the analysis of the model's face validity (cf.
\fullref{sec-model-validation-1}), are constrained to $\mathfrak{S} := ]0,1]$,
$\mathfrak{E} := [-0.5,0.5]$, and $\mathfrak{T} := [0,0.5]$. A sampling-based
validation of the sum terms of the function is provided in \fullref{fig:S4};
exemplary profile functions of \eqnref{eq:llh-1} are visualized in
\fullref{fig:S5}.

\paragraph{Agentic conditional log-likelihood function} When simulating agent
behavior in a data generative setting, the agent's internal observation
distribution (\eqnref{eq:p_o_c}) is sampled to generate the internal agent
observation $o_t$. Depending on the outcome of the sampling, the internal
observation can thus considerably deviate from the sum of the external contrast
difference $c_t$ and the perceptual bias parameter $\eta$. When studying the
agent's face validity, such trial-by-trial internal agent variability is one of
the generators of action variability. However, when evaluating the conditional
log-likelihood function in an agentic manner, i.e., by simulating the agent's
internal processes and evaluating an observed action log probability given these
processes, reliance on a single observation of $o_t$ is necessarily error-prone
and, for fixed experimental and action variables, results in the undesirable
feature of variable conditional log-likelihood values. The analytical
conditional log-likelihood function of \eqnref{eq:llh-1} eschews this by
appropriately marginalizing out all possible values of $o_t$. Specifically,
evaluating the contrast difference-conditional action probability on a given
trial in the derivation of \eqnref{eq:llh-1} entails the analytical evaluation
of
\begin{multline}\label{eq:agentic-cllh-1}
p(a_t|c_t)
= \mathrm{Bern}(a_t;\tau)\int_{\mathbb{R}} \mathrm{Norm}(o_t;c_t+\eta,\sigma^2)\ib{o_t \le 0}\,do_t
\\ + \mathrm{Bern}(a_t;1-\tau)\int_{\mathbb{R}} \mathrm{Norm}(o_t;c_t+\eta,\sigma^2)\ib{o_t > 0}\,do_t
\end{multline}
for $t = 1,...,T$. From a computational perspective, the integral terms in
\eqnref{eq:agentic-cllh-1} can be approximated using their Monte Carlo
estimators \citep{hammersley1964, robert2004}
\begin{equation}\label{eq:monte-carlo-1}
\begin{aligned}
\int_{\mathbb{R}} \mathrm{Norm}(o_t;c_t+\eta,\sigma^2)\ib{o_t \le 0}\,do_t
\approx \frac{1}{n_s}\sum_{i=1}^{n_s}\ib{o_t^{(i)} \le 0}
=1-\frac{1}{n_s}\sum_{i=1}^{n_s}(\delta\circ\phi\circ\beta)(o_t^{(i)})
\end{aligned}
\end{equation}
and
\begin{equation}\label{eq:monte-carlo-2}
\begin{aligned}
\int_{\mathbb{R}} \mathrm{Norm}(o_t;c_t+\eta,\sigma^2)\ib{o_t > 0}\,do_t
\approx \frac{1}{n_s}\sum_{i=1}^{n_s}\ib{o_t^{(i)} > 0}
=\frac{1}{n_s}\sum_{i=1}^{n_s}(\delta\circ\phi\circ\beta)(o_t^{(i)}),
\end{aligned}
\end{equation}
where
\begin{equation}
o_t^{(1)}, ..., o_t^{(n_s)} \sim \mathrm{Norm}\left(c_t + \eta, \sigma^2\right)
\end{equation}
and $n_s$ denotes the number of samples. This means that evaluating the agentic
conditional log-likelihood function must rely on the average of a pre-specified
number of $n_s$ transformed internal agent observations
$\left(\delta \circ \phi \circ \beta \right) \left(o_t^{(i)} \right)$. In
\fullref{fig:S6}, we visualize the variability entailed by this numerical
integration approach for evaluating the conditional log-likelihood function with
a few samples. More generally, evaluating the analytical conditional log-likelihood entails
an approximately 800-fold increase in evaluation speed in the present
implementation compared to using its agentic counterpart (cf.
\fullnameref{sec:gabor-contrast-discrimination-agent-analysis} and
\textit{abm\_likelihoods.py}). Having extensively validated our analytical
results using simulations, we performed all validation and evaluation analyses
reported below using the analytical conditional log-likelihood function form of
\eqnref{eq:llh-1} as implemented in \textit{abm\_llh.py}.

\begin{table}[t]
\caption{Intuitions, parameter constraints, and conditional log-likelihood
functions of the ABM variants applied in the Gabor contrast discrimination case
study.}
\label{tab:abm-variants-1}
\centering
\small
\renewcommand{\arraystretch}{1.8}
\begin{tabularx}{\linewidth}{@{}l X@{}}
\toprule
\textbf{A0}
& Random null model ($\tau := 0.5$)
\\
&
$\ell(\cdot) 
:= \sum_{t = 1}^T \ln \mathrm{Bern}(a_t;0.5)$
\\ \midrule
\textbf{A1}
& Perceptual bias-free model ($\eta := 0$)
\\
&
$\ell(\sigma,\tau) 
:= \sum_{t = 1}^T \ln \Big(\mathrm{Bern}(a_t;\tau)\Phi(0;c_t,\sigma^2) + \mathrm{Bern}(a_t;1-\tau)\left(1-\Phi(0;c_t,\sigma^2)\right)
\Big)$
\\ \midrule
\textbf{A2}
& Full model as defined in \eqnref{eq:abm-1}
\\
&
$\ell(\sigma,\eta,\tau) 
:= \sum_{t = 1}^T \ln \Big(\mathrm{Bern}(a_t;\tau)\Phi(0;c_t + \eta,\sigma^2) + \mathrm{Bern}(a_t;1-\tau)\big(1-\Phi(0;c_t + \eta,\sigma^2)\big)
\Big)$
\\
\bottomrule
\end{tabularx}
\end{table}

\paragraph{Agent model variants} 

For evaluating the observed data in light of the agent model \eqref{eq:abm-1},
we consider three nested variants of \eqnref{eq:abm-1} and, in turn,
\eqnref{eq:llh-1} (cf. \autoref{tab:abm-variants-1}). Specifically, we consider
\eqnref{eq:abm-1} with the action noise parameter fixed at $\tau := 0.5$ as a
random null model A0. Based on \eqnref{eq:p-at-giv-dt}, this specification
entails that the agent's decisions are rendered irrelevant with respect to the
agent's actions, which are determined uniformly at random from the set
$\{0,1\}$. Second, we consider a perceptual bias-free variant of
\eqnref{eq:abm-1}, for which we set $\eta := 0$ and which we label A1. Finally,
we consider \eqref{eq:abm-1} with none of its three parameters pre-specified,
which we label A2. In A2, the sensory bias affects the observation distribution,
while the agent continues to form its beliefs under an assumed unbiased
observation model. Accordingly, the notation $p(s_t| o_t)$ in the belief-state
definition refers to the posterior under this assumed model, which only
coincides with the posterior under the actual observation model when $\eta=0$.
For nonzero sensory bias, the agent does not correct its observations for $\eta$.
For the implementation of the respective model conditional
log-likelihood functions, please refer to \textit{abm\_llh.py}.

\subsection{Parameter estimation and model comparison}\label{sec:parameter-estimation-and-model-comparison}

To obtain parameter estimates and approximations of the respective model's
marginal conditional log-likelihood, we used SciPy's minimization function,
\textit{minimize}, with the Nelder-Mead algorithm as the optimization method
\citep[]{gao2012, virtanen2020}. The parameter bounds used for fitting were
informed by the face-validity simulations and were set to
$\sigma\in[0.1,0.9]$, $\eta\in[-0.5,0.5]$, and $\tau\in[0,0.5]$,
wherever the respective parameters were free. Fitting proceeded over continuous
intervals with the same endpoints as the simulation grids. To reduce the risk of identifying
local rather than global minimizers over the respective parameter space of the
negative conditional log-likelihood function, the initial parameter estimates were
selected to be evenly spaced within the respective parameter domains, and the
minimization was repeated for a predefined number of initial parameter estimate
values. For each model with $k>0$ free parameters, we specified
$\lceil 100^{1/k} \rceil$ initial parameter estimate values per parameter
dimension. We used all combinations of these values as starting points on a
Cartesian grid, yielding 100 starting points for models with one or two free
parameters and 125 for models with three free parameters. Across all
repetitions of the parameter estimation procedure, only the parameter estimates
achieving the smallest negative conditional log-likelihood were regarded as
maximum likelihood (ML) estimates \citep{wilson2019, horvath2021, usee2025}. For
further implementation details on the parameter estimation and model comparison
procedure, please refer to \textit{abm\_estimate.py}.

To evaluate how plausible the different behavioral models are relative to one
another, given the experimental or simulated data, we first evaluated all model-
and participant-specific Bayesian information criterion (BIC) scores according
to (\citet{schwarz1978}, cf. \fullnameref{sec:bic})
\begin{equation}\label{eq:bic-main}
\text{BIC}_{ij} := \ell_{ij}(\hat{\theta}) - \frac{k_j}{2} \ln T,
\end{equation}
where $i = 1,...,n$ indexes datasets, $j = 1,...,m$ indexes models, $\ell_{ij}$
denotes the conditional log-likelihood function of model $j$ for dataset $i$,
$\hat{\theta}$ denotes the model- and dataset-specific maximum likelihood
parameter estimate, $k_j$ denotes the number of free model parameters, and $T$
denotes the total number of trials per dataset. Note that according to
\eqnref{eq:bic-main}, higher BIC values indicate higher model plausibilities. As
a first group-level indicator for the most plausible model, we summed the
model-specific BIC values across all datasets, corresponding to the approximated
log joint conditional probability of all datasets under the given model and the
assumption of pairwise dataset independence. We then subjected the BIC values of
all agents and datasets to the PEP random-effects Bayesian model selection
procedure. For each model, this method yields a PEP $\varphi_j \in [0,1]$ with
$\sum_{j=1}^m\varphi_j= 1$, combining the posterior probability that the given
model has the highest true but unknown generative probability under the
alternative hypothesis with a uniform contribution under the null hypothesis,
weighted by the posterior probabilities of the two hypotheses. As previously
mentioned, the fine-grained interpretation of a
model's PEP is documented in
\fullnameref{sec:protected-exceedance-probabilities}.

\subsection{Model validation}\label{sec-model-validation-1}

To validate the ABM framework for the Gabor contrast discrimination task, we
performed simulations that assessed its face validity, recoverability, and
parameter identifiability. To this end, we simulated datasets using the A0, A1,
and A2 variants of the ABM encoded by \eqnref{eq:abm-1} in a generative setting
that mimicked the experimental scenario regarding the number of participants,
experimental blocks per participant, and experimental trials per block. For the
parameter-free agent A0, the simulated data were generated without further
specifications. For models A1 and A2, we opted for a parameter scope that allows
the models to express the experimentally observed effects and a parameter
coverage that balances computational efficiency with theoretical insight. We
therefore decided on a resolution of nine evenly spaced parameter values for
each component of the model-specific parameter vector, which we set to
$\sigma \in \{0.1$, $0.2$, $0.3$, $0.4$, $0.5$, $0.6$, $0.7$, $0.8$, $0.9\}$,
$\eta \in \{-0.5$, $-0.375$, $-0.25$, $-0.125$, $0$, $0.125$, $0.25$, $0.375$,
$0.5\}$, and $\tau \in \{0$, $0.0625$, $0.125$, $0.1875$, $0.25$, $0.3125$,
$0.375$, $0.4375$, $0.5\}$. This resulted in 81 true but unknown parameter
values for the two-dimensional model A1 and 729 true but unknown parameter
values for the three-dimensional model A2. For each parameter setting, we then
generated $n = 59$ datasets comprising 10 blocks and 30 trials per block. Across
these simulated participant datasets, we evaluated the group psychometric
function as the primary descriptive statistic. We assessed the A0, A1, and A2 variants of the
ABM \eqnref{eq:abm-1} using the procedures described in
\fullref{sec:parameter-estimation-and-model-comparison}. For implementation
details, please refer to \textit{abm\_validate.py} and
\textit{abm\_validation.py}. \fullref{fig:4} visualizes the model validation
results.

\paragraph{Model face validity} To validate the formulation and implementation
of the ABM and to explore its behavioral repertoire, we evaluated the effect of
each of the three parameters of variant A2 on the group psychometric function
(\fullref{fig:4}A).

The first subpanel shows the effect of varying the sensory standard deviation
$\sigma$ on the functional relationship between the normalized contrast
difference $c$ and the observed proportion of taking action $a = 1$ (``right'')
for a fixed sensory bias parameter of $\eta = 0$ and in the absence of
post-decision noise ($\tau = 0$). Increasing the sensory noise at a constant
normalized contrast difference increases behavioral variability, i.e., the
proportion of the complementary action. At low sensory noise
($\sigma = 0.1$), the psychometric function assumes a more step-function-like
appearance, whereas for high sensory noise ($\sigma = 0.9$), the proportion
of objectively correct actions increases almost linearly. Between these
extremes, the normally distributed internal agent observations generate the
familiar sigmoidal form of psychometric functions. Note that in this simulation,
the proportion of ``left'' and ``right'' actions at a contrast difference of
$c = 0$ is 0.5, corresponding to the absence of a sensory bias ($\eta = 0$).

The second subpanel of \fullref{fig:4}A shows the effect of varying the sensory
bias parameter $\eta$ for a low sensory standard deviation of $\sigma = 0.2$ and
in the absence of post-decision noise ($\tau = 0$). Setting the sensory bias
parameter to positive or negative values shifts the agent psychometric function
to the left or right, as a non-zero sensory bias requires negative or positive
contrast differences to result in an equal proportion of ``left'' and ``right''
actions.

The third subpanel of \fullref{fig:4}A shows the effect of varying the
post-decision noise parameter between $\tau = 0$, corresponding to no
post-decision noise and the scenario in which the agent's actions equal the
agent's decisions, and $\tau = 0.5$, corresponding to a full post-decision noise
scenario, in which the agent's actions are independent of the agent's decisions
and are distributed uniformly at random, on the psychometric function for
$\sigma  = 0.2$ and $\eta = 0$. Crucially, the post-decision noise parameter
affects the proportion of ``left'' and ``right'' actions independently of the
contrast difference $c$, as is most evident for the extreme values $c = -1.0$
and $c = 1.0$. Taken together, the trial-by-trial action generating architecture
of ABM variant A2 is able to reproduce the familiar parameter-dependent
characteristics of psychometric functions at the aggregate level.

\paragraph{Model recovery} To assess the degree to which the three variants of
the ABM can be reliably differentiated by our model comparison framework, we
performed several model recovery analyses. For each of the three ABM variants,
we analyzed the data generated by the given ABM variant using all three ABM
variants as analytical models. For the generative parameter-free model A0, we
obtained the PEPs $\varphi_{A0} = 1.00$, $\varphi_{A1} = 0.00$, and
$\varphi_{A2} = 0.00$, indicating that variant A0 can be reliably recovered and
offers the most parsimonious explanation for a fully random action set. The
model recovery properties of ABM variants A1 and A2 depend on the true but
unknown generative parameter values of these models. We visualize these in
\fullref{fig:4}B.

The first subpanel of \fullref{fig:4}B depicts the PEPs for the three model
variants as generative model parameter-dependent pie charts. At a given
parameter location in the generative model's parameter space, the pie chart
shows the relative PEP shares of ABM variants A0, A1, and A2 as light blue,
standard blue, and dark blue, respectively. For none of the generative parameter
values of ABM variant A1 does ABM variant A2 take a major share, indicating that
the penalty incurred for the additional parameter of this variant is sufficient
to suppress its ability to score highly in the model comparison. With respect to
the sensory noise and action noise parameters, ABM variant A1 exhibits reliable
recoverability up to an action noise level of approximately $\tau \approx 0.35$,
with somewhat lower action noise levels already hampering the recoverability of
variant A1 at very high sensory noise and therefore relatively shallow
psychometric functions (cf. first subpanel of \fullref{fig:4}A). Notably, for
scenarios with high post-decision noise, the most plausible model variant is A0,
which offers a more parsimonious explanation of the respective simulated
datasets due to its parameter-freeness. In summary, model variant A1 exhibits
sensible model recoverability properties in the context of the ABM variants
under study.

The second and third subpanels of \fullref{fig:4}B depict the model
recoverability of variant A2 for the subspaces $\mathfrak{S} \times \mathfrak{E}$
and $\mathfrak{E} \times \mathfrak{T}$ of its parameter space for constant
generative parameters of $\tau = 0.0625$ and $\sigma = 0.1$, respectively. With
respect to the variation of its $\sigma$ and $\eta$ parameters, model variant A2
exhibits good recoverability in this low post-decision noise scenario. In
addition, when mimicking model variant A1 by estimating $\hat{\eta} \approx 0$,
the more parsimonious model A1, which can achieve a similar level of modeling
accuracy, takes the largest share of PEP. This is also true for small positive
or negative values of sensory bias ($\eta = 0.125$ and $\eta = -0.125$),
particularly if the sensory standard deviation parameter is large ($\sigma > 0.5$).
Finally, with respect to the variation of its parameters $\eta$ and $\tau$, the
ABM variant A2 also exhibits reasonable recoverability properties. Again, for a
generative parameter of $\eta = 0$, ABM variant A1 assumes the highest PEP
values, while for action noise levels greater than $\tau \approx 0.35$, ABM
variant A0 outperforms variant A2 in explaining the data generated by itself.

\paragraph{Parameter recovery} To gain insight into the parameter
identifiability properties of the ABM \eqnref{eq:abm-1}, we assessed the
pairwise interactions of the identifiability of a given parameter component
under the variation of a second parameter component while fixing the third
parameter component for ABM variant A2 (\fullref{fig:4}C-E). The two subpanels
of \fullref{fig:4}C visualize the recoverability of the sensory standard deviation and
sensory bias parameters as a function of the action noise parameter for a
sensory bias of $\eta = 0$ (upper subpanel) and a sensory standard deviation of
$\sigma = 0.2$ (lower subpanel). The $x$-axis in these subpanels represents the
true but unknown generative parameter, while the $y$-axis represents the
respective estimate. The colored lines depict the group average ML estimates and
their standard errors. The gray diagonal line represents exact parameter
identification. In the absence of action noise and for the current complementary
parameters, both the sensory standard deviation and the sensory bias parameters can be
reliably identified, with a slight exception for large generative sensory
standard deviations of $\sigma > 0.7$. This is in line with the descriptive
statistics of \fullref{fig:4}A, which indicate that in the absence of action
noise and for large sensory standard deviations, the ABM behavior starts to
resemble that of a high action noise scenario, rendering the respective dataset
non-indicative with respect to the sensory standard deviation parameter value. Increasing
the action noise parameter to $\tau = 0.5$ naturally hampers the identifiability
of both $\sigma$ and $\eta$.

\FloatBarrier
\clearpage
\vspace*{\fill}
\begin{center}
\includegraphics[width=\linewidth]{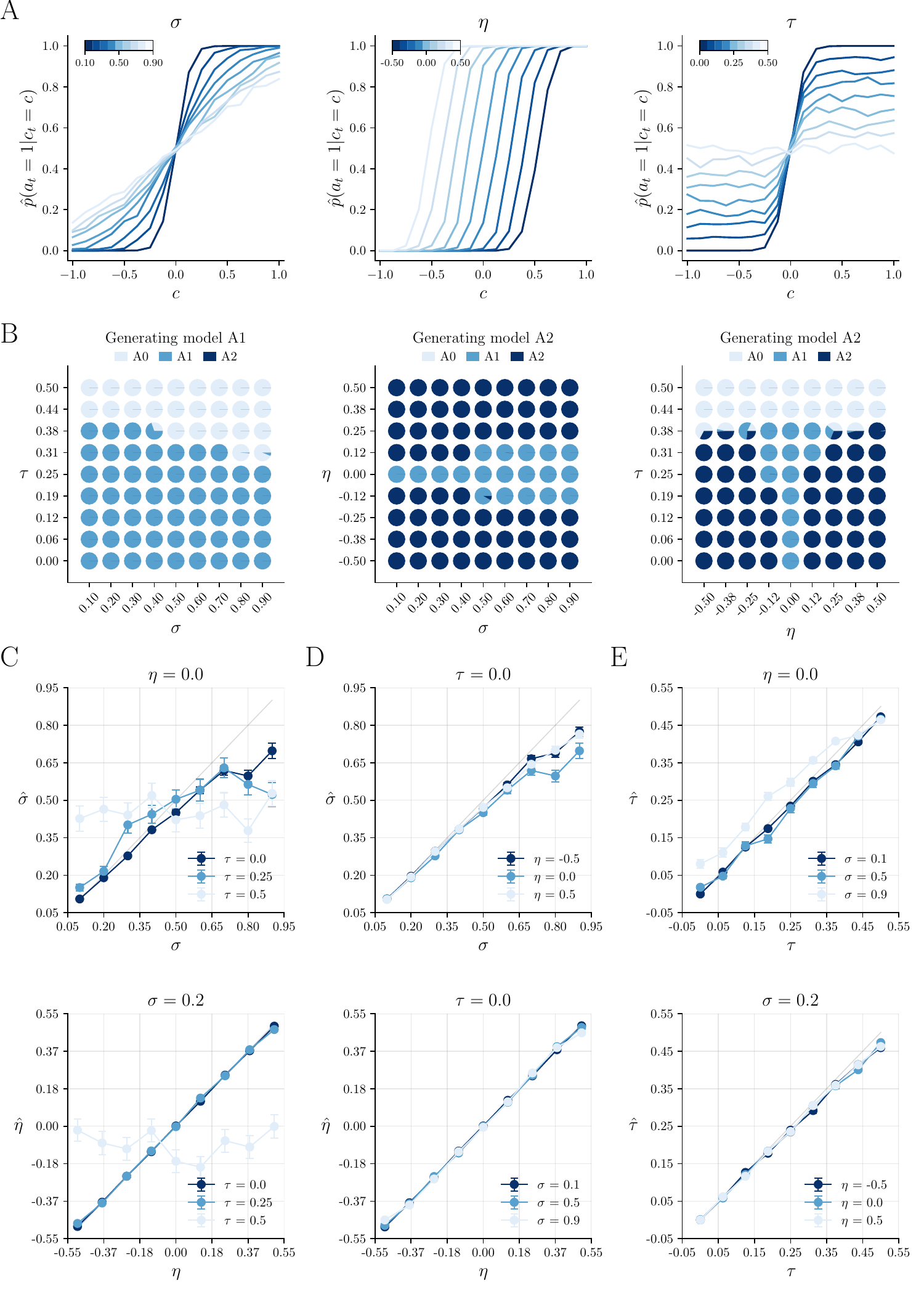}
\end{center}
\clearpage

\begingroup
\noindent\rule{\linewidth}{0.4pt}\par
\captionof{figure}{\textit{(Figure on previous page.)} Model validation for the
Gabor contrast discrimination ABM. (A) Model face validity. From left to right,
the subpanels visualize the effect of varying one of the parameters of ABM
variant A2 on the group psychometric function in isolation. For the effect of
$\sigma$ (left subpanel), the complementary parameter settings are $\eta = 0$,
$\tau = 0$, for the effect of $\eta$ (center subpanel), the complementary
parameter settings are $\sigma = 0.2$ and $\tau = 0$, and for the effect of
$\tau$ (right subpanel), the complementary parameter settings are $\sigma = 0.2$
and $\eta = 0$. (B) Model recovery. From left to right, the subpanels visualize
the PEP fractions attained by the data-analytical ABM variants A0, A1, and A2
for data generated by A1 (left subpanel) and A2 (center and right subpanels) as a
function of the generating model's parameter configuration in pie chart form.
The complementary parameter of A2 was set to $\tau = 0.0625$ for the center
subpanel and to $\sigma = 0.1$ for the right subpanel. (C) Recoverability of
$\sigma$ and $\eta$ as a function of $\tau$ for $\eta = 0$ (upper subpanel) and
$\sigma = 0.2$ (lower subpanel). (D) Recoverability of $\sigma$ as a function of
$\eta$ for $\tau = 0$ (upper subpanel) and recoverability of $\eta$ as a
function of $\sigma$ for $\tau = 0.0$ (lower subpanel). (E) Recoverability of
$\tau$ as a function of $\sigma$ for $\eta = 0.0$ (upper subpanel) and as a
function of $\eta$ for $\sigma = 0.2$ (lower subpanel).}
\label{fig:4}
\endgroup

This results in group average estimates that do not track the respective
generative parameter values and, in the case of $\eta$, a marked
increase in estimation variance and negative bias.

The two subpanels of \fullref{fig:4}D visualize the dependency of the
identifiability of the sensory standard deviation (upper subpanel) and the sensory bias
parameters (lower subpanel) on each other in the absence of action noise. Here,
increasing the sensory bias parameter in both the positive and negative
directions results in a slight underestimation of the sensory standard deviation parameter
(upper subpanel). However, varying the sensory standard deviation parameter in the absence
of action noise leaves the identifiability of the sensory bias parameter
virtually unchanged (lower subpanel).

The subpanels of \fullref{fig:4}E visualize the identifiability of the action
noise parameter under variation of the sensory standard deviation parameter for a constant
sensory bias parameter of $\eta = 0$ (upper subpanel) and under variation of the
sensory bias parameter for a constant sensory standard deviation of
$\sigma = 0.2$ (lower subpanel). In both cases, the identification of the action
noise parameter is fairly accurate. However, there is some systematic over- and
underestimation across the range of $\tau$. At constant sensory bias, low
sensory standard deviations yield a minor underestimation of the action noise
parameter, indicating that a minor proportion of the action noise-induced action
variance is actually allocated to the sensory standard deviation parameter. This is also
in line with the reverse effect observed in the upper subpanel of
\fullref{fig:4}C. At the highest level of sensory noise, however, the effect of
action noise is overestimated for low to moderate values of $\tau$, reflecting the increasing similarity between high
sensory noise and high action noise datasets, as well as the relatively small
proportion of contrast difference space on which these two parameters act
discriminatorily on the data. Finally, varying the sensory bias parameter within
its range does not affect the identifiability of the action noise parameter.
However, for the selected sensory standard deviation of $\sigma = 0.2$, a slight
and systematic underestimation of the action noise parameter can be observed.
Taken together, for the parameter scenarios discussed, the ABM variant A2
parameters are well identifiable, with some minor interactions in the context of
low action noise and a breakdown of identifiability for high action noise.

\subsection{Model evaluation}\label{sec-model-evaluation-1}

Upon model validation, we evaluated the ABM variants based on the experimental
data. \fullref{fig:5} visualizes the group-level results. The group cumulative
BIC was highest for the ABM variant A2 (\fullref{fig:5}A) with an associated PEP
of 0.84 (A1: 0.16, A0: 0.00). At the participant level, the BIC values of model
variants A2, A1, and A0 were maximal for 30, 28, and 1 of the participants,
respectively (\fullref{fig:S1}). The parameter estimates $\hat{\sigma}$ for the
sensory noise parameter for datasets best explained by ABM variants A1 and A2
varied between 0.1 and 0.6, with one outlier at 0.8 (\fullref{fig:5}C). For
datasets best explained by ABM variant A2, the parameter estimates $\hat{\eta}$
for the perceptual bias parameter formed two clusters of opposite signs between
$\pm$ 0.1 and $\pm 0.2$, with few outliers and a slight surplus of
participant-level estimates for negative over positive perceptual biases. The
parameter estimates $\hat{\tau}$ for the post-decision noise parameter varied
between 0.00 and 0.16, with a higher degree of variability for datasets better
explained by ABM variant A1 than A2. Taken together, the results indicate that,
in the evaluated model space, the ABM variant A2 has a slight overall advantage
in explaining the observed data. However, at the level of individual
participants, the additional inclusion of the perceptual bias parameter was
supported by the data for only approximately half of the participants.

\begin{figure}[!htbp]
\center
\includegraphics[width=\linewidth]{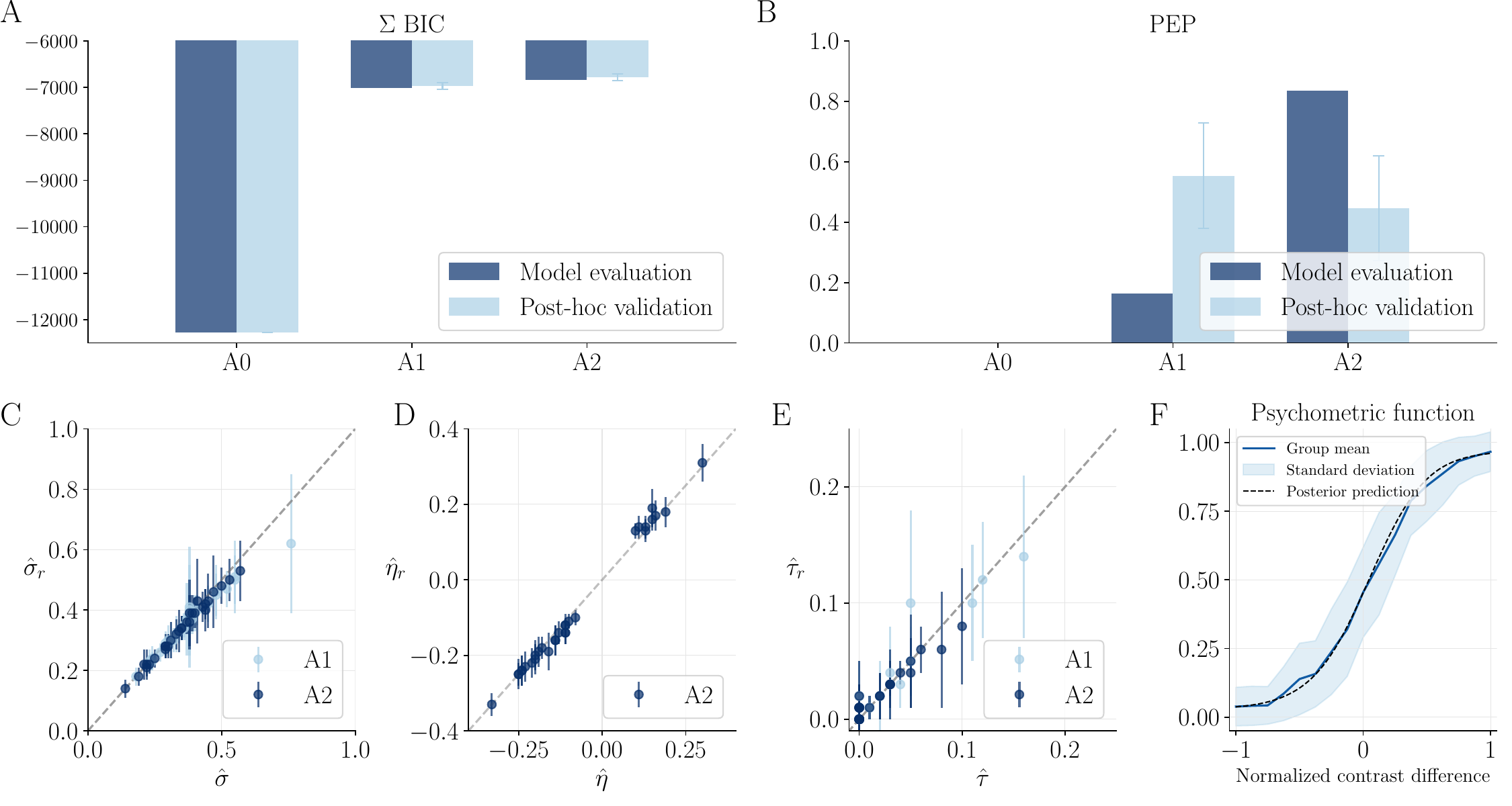}
\caption{Model evaluation and post hoc model validation of the Gabor contrast
discrimination case study. (A) Group-cumulative BIC values, corresponding to the
sums of the participant-specific BIC values for each ABM variant. Bar colors
indicate the results from the experimental data and the simulated post hoc
validation data. Error bars indicate the standard deviation across post hoc
validation simulations. (B) PEPs for each ABM variant. As in (A), bar colors
indicate the results from the experimental data and the simulated post hoc
validation data. Error bars indicate the standard deviation across post hoc
validation simulations. (C), (D), (E) Parameter estimation and post hoc recovery
results for ABM variants indicated in the respective subpanels' legends. Parameter estimates
$\hat{\sigma}, \hat{\eta}$, and $\hat{\tau}$ based on the experimental data are
displayed on the $x$-axes. The recovered average parameter estimates
$\hat{\sigma}_r$, $\hat{\eta}_r$, and $\hat{\tau}_r$ and their standard
deviations across post hoc simulations are shown on the $y$-axes. (F) Group mean
experimental data as in \autoref{fig:2}C and group mean posterior predictive
psychometric function.}\label{fig:5}
\end{figure}

\subsection{Post hoc model validation}\label{sec-posthoc-model-validation-1}

After identifying, for each participant, the most plausible ABM variant
accounting for the experimentally observed data, we performed post hoc model
validation by simulating 100 group-level datasets. These datasets were generated
on the basis of the participant-specific most plausible models and their
corresponding parameter estimates. We then subjected the generated datasets to
parameter estimation and model comparison to evaluate the average post hoc model
and parameter recoverability performance. The results of these analyses are
integrated in \fullref{fig:5}. As with the experimental data, the post hoc
simulated average group cumulative BIC was highest for ABM variant A2
(\fullref{fig:5}A). In line with the fact that, at the individual
participant level, the generative model conformed to models A1 and A2 for
approximately half of the participants, the post hoc PEPs for models A1 and A2 are
roughly identical (\fullref{fig:5}B). The post hoc parameter recovery was
successful, as indicated by the largely diagonally located experimental and
post hoc parameter estimate averages shown in \fullref{fig:5}C to E. Finally, on
a descriptive level, we used the participant-specific model comparison and
parameter estimation results to evaluate participant-specific posterior
predictive psychometric functions and, by averaging those, the group mean
posterior predictive psychometric function. Note that we here use the term
``posterior predictive'' to describe predictions conditional on the selected
models and their maximum-likelihood parameter estimates, rather than predictions
averaged over a Bayesian parameter posterior distribution.
At the participant level, the
posterior predictive psychometric functions exhibit good fits to their observed
counterparts (\fullref{fig:S1}), while at the group level, the posterior
predictive psychometric function and its observed counterpart are virtually
indistinguishable (\fullref{fig:5}F).

\subsection{Summary}\label{sec-summary-1}

In summary, this case study instantiates the ABM framework for a minimal
perceptual choice task under an agent architecture that performs optimally in
the presence of additive internal sensory noise when sensory bias and
post-decision noise are absent. We show that, from an
experimentalist’s perspective, the model can be reliably recovered from the
agent’s actions under low to moderate post-decision noise. Using empirical data,
we further demonstrate that incorporating a long-lasting perceptual bias term
improves the model’s ability to account for observed behavior. While these
results are fully consistent with standard psychometric approaches at the
descriptive and posterior predictive levels
\citep[e.g.,][]{wichmann2001, hanks2017, carandini2024}, the model additionally
specifies the algorithmic, trial-by-trial perceptual and decision-making
processes of an agent performing a Gabor contrast discrimination task.

\FloatBarrier
\section{Symmetric bandit learning}\label{sec:symmetric-bandit-learning}

\subsection{Experimental paradigm and descriptive statistics}
As a second case study, we consider a symmetric bandit task (\fullref{fig:6}).
In each trial of a block of this task, participants were presented with two
choice options, represented by the left and right cursor buttons
(\fullref{fig:6}A). In each block, one of the buttons was associated with a
probability state $s$ to obtain a reward of $+1$ and a probability state of
$1-s$ to obtain a reward of $+0$, while the other button was associated with a
probability state of $1-s$ to obtain a reward of $+1$ and a probability state
$s$ to obtain a reward of $+0$. The probability state $s$ was sampled from a
uniform distribution for each block. Each block consisted of 30 trials. The
participants' task was to press either button on each trial to maximize their
cumulative reward in a given block. Trial-wise rewards were presented for 1
second, and no time limit was applied to select an action. Each of $n = 60$
participants completed 10 blocks (300 trials in total). The better option's
probability of yielding a reward of $+1$, averaged across blocks within each
participant, had a group mean of $0.75$ (SD $0.05$). On average,
the participants selected the action with the highest associated probability to
obtain a reward of $+1$ in $81.0\%$ (SD $8.2\%$) of all the trials and obtained an
average reward of $0.69$ (SD $0.06$) (\fullref{fig:6}C). More importantly, on
average, the participants exhibited clear learning behavior, starting at chance
level for the maximizing action on the first trial and reaching an asymptotic
level of approximately $85\%$ maximizing actions on the tenth trial across
blocks (\fullref{fig:6}D). For a more in-depth coverage of the experimental
procedures and the descriptive statistics reported, please refer to
\fullnameref{sec:participants-and-procedure},
\fullnameref{sec:symmetric-bandit-learning-supplement}, and
\textit{abm\_describe.py}.

\subsection{Agentic behavioral models}

In the following, we formulate an ABM for the symmetric bandit learning task
data. After casting the task into an appropriate task model, we first consider a
Bayesian learner in \autoref{sec:agent-model-a1-bayesian-learner} and its
equivalent prediction error-correcting Rescorla–Wagner form in
\autoref{sec:equivalent-prediction-error-correcting agent}. We then consider a
generic Rescorla–Wagner learner in
\autoref{sec:agent-model-a2-rescorla-wagner-learner}.
\fullref{tab:abm-variants-2} provides an overview of the agent models considered
for the symmetric bandit learning task and their associated log-likelihood
functions.

\begin{figure}[!htbp]
\centering
\includegraphics[width=\linewidth]{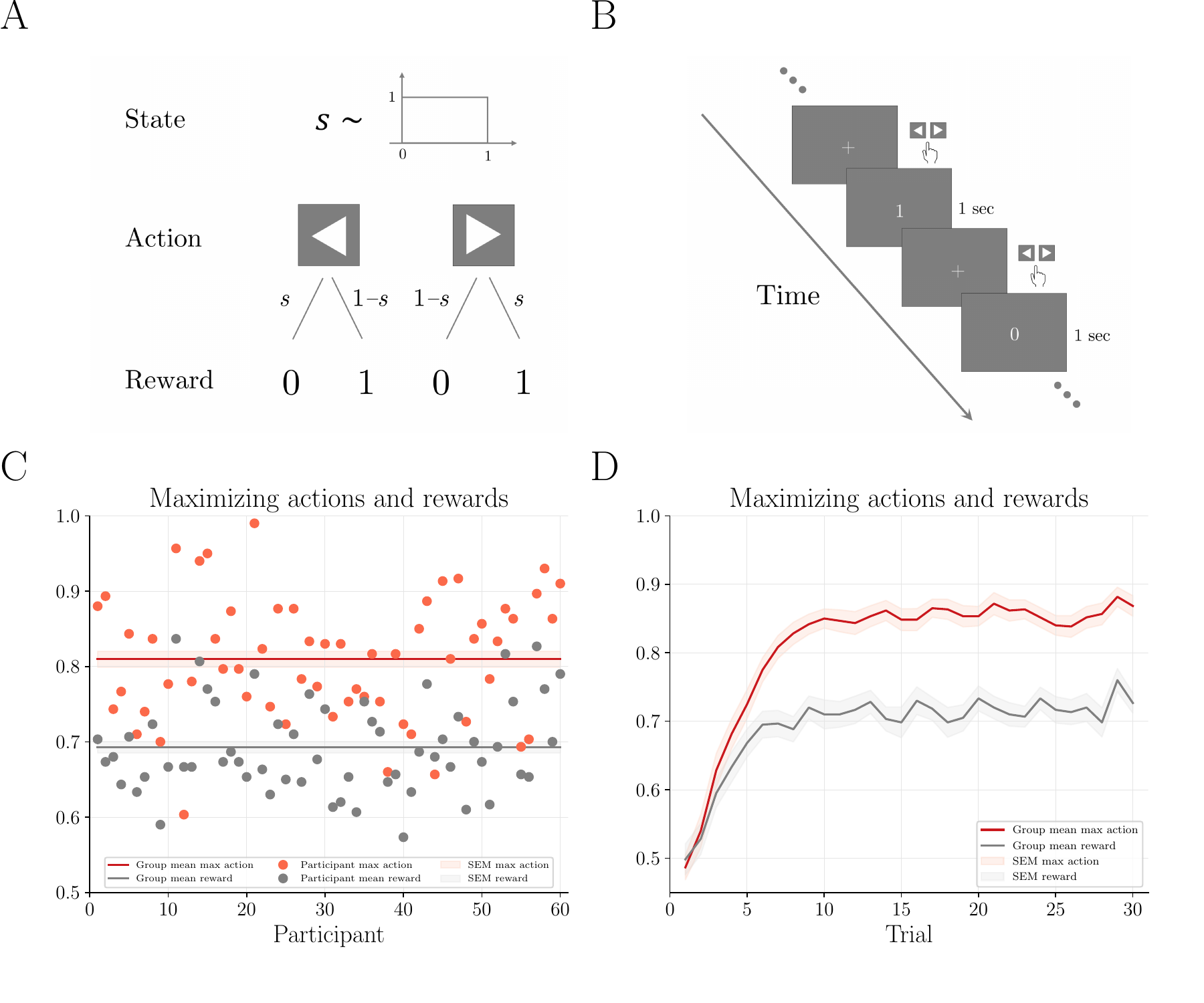}
\caption{Symmetric bandit learning task and descriptive statistics. (A)
Symmetric bandit task block mechanics. For each block, a probability state $s$
was realized from a uniform distribution on $[0,1]$ and the leftwards cursor
button endowed with a probability of $s$ to yield a reward of $+0$ and a
probability of $1-s$ to yield a reward of $+1$. Vice versa, the rightwards
cursor button was endowed with a probability of $1-s$ to yield a reward of $+0$
and a probability of $s$ to yield a reward of $+1$. Depending on whether $s$ was
realized in $[0,0.5[$ or $[0.5,1]$, either cursor button could thus be associated
with the higher probability for a reward of $+1$. (B) In each trial of each
block, participants were prompted by a fixation cross to press one of the two
cursor buttons, the respective reward distribution was sampled, and the sampled
reward was presented for 1 sec. After completing 30 trials of a given block,
participants were informed by an information screen that they could continue to
the next block by pressing the space bar. (C) Mean fraction of maximizing
actions and obtained reward across blocks for all individual participants. The
shaded areas indicate the standard error of the mean (SEM). (D) Mean fraction of
maximizing actions and obtained reward across blocks and participants as a
function of trial index. The shaded areas indicate the SEM.}\label{fig:6}
\end{figure}

\paragraph{Task model} We describe a block of the experimental task by the tuple
\begin{equation}\label{eq:T-2}
\mathcal{T} := (S,A,R, p(s_1),p(r_t|s_1,a_t)),
\end{equation}
with $t=1,\ldots,T$ indexing the trials within a block, where
\begin{itemize}
\item $S := [0,1]$ is the set of probability states, 
\item $A := \{0,1\}$ is the action set, modeling button presses ($0$ = left, $1$
= right),
\item $R := \{0,1\}$ is the set of rewards,
\item the probability state distribution is the distribution of the probability
state random variable $s_1$, taking values in $S$, and is given by
\begin{equation}
p(s_1) := \mathrm{Unif}(s_1;[0,1]),
\end{equation}
\item the probability state- and action-conditional reward distribution is the
conditional distribution of the reward random variable $r_t$, taking values
in $R$, given $s_1$ and the action random variable $a_t$, taking values in $A$,
and is given by
\begin{equation}
p(r_t|s_1,a_t) := \mathrm{Bern}(r_t;s_1)^{a_t}\mathrm{Bern}(r_t;1-s_1)^{1-a_t}.
\end{equation}
\end{itemize}
Here, the probability state of the task model is sampled once per block from the
uniform distribution $\mathrm{Unif}(s_1;[0,1])$ and governs the agent action-dependent
reward probability throughout the block according to $p(r_t|s_1,a_t)$ (cf.
\autoref{fig:7}). Notably, from the perspective of common POMDP frameworks
\citep{krishnamurthy2016, smith2022, chandramouli2024}, this corresponds to a
deterministic state transition in which the probability state remains unchanged
across trials. While such a
(pseudo)dynamic representation of the task model may make the connection of the
current task to tasks with actual state evolutions more explicit, it would come
at the cost of additional notational complexities, from which we abstain here
\citep[cf.][Chapter 5.5]{bauerle2011}. Finally, the characteristic task symmetry
is encoded in the probability state- and action-conditional reward distribution,
which renders the probability state $s_1$ the probability of obtaining a reward
of $r_t = 1$ for action $a_t = 1$, and the complementary probability of
obtaining a reward of $r_t = 1$ for action $a_t = 0$.

\subsubsection{Agent model A1 (Bayesian learner)}\label{sec:agent-model-a1-bayesian-learner}
\paragraph{Agent model} We consider the agent model
\begin{equation}\label{eq:A-2}
\mathcal{A} := \left(\mathcal{T}, B, V, D,\beta,\phi,\delta\right),
\end{equation}
where  
\begin{itemize}
\item $\mathcal{T}$ is the task model and represents the agent's knowledge of
the task properties,
\item $B$ is the set of belief states,
\begin{equation}
B:=\left\lbrace b | b: S \to \mathbb{R}_{\ge 0} \mbox{ with } \int_0^1 b(s) ds = 1\right\rbrace,
\end{equation}
i.e., the set of all probability density functions on $S$,
\item $V := [0,1]^2$ is the decision value set,
\item $D := \{0,1\}$ is the decision set, modeling the choice to execute a
specific button press (0 = left, 1 = right),
\item $\beta$ denotes the agent’s belief state update function
\begin{multline}\label{eq:beta-2}
\beta : B \times A \times R \to B, \\ 
(b_{t-1}, a_{t-1},r_{t-1}) \mapsto \beta(b_{t-1}, a_{t-1},r_{t-1}) := p(s_1|a_{1:t-1},r_{1:t-1}) =: b_t.
\end{multline} 
Here, the belief state $b_t$ takes values in $B$, with
\begin{equation}
\renewcommand{\arraystretch}{1.4}
p(s_1|a_{1:t-1},r_{1:t-1}) = \mathrm{Beta}(s_1; \alpha_t),
\end{equation}
and where the two Beta shape parameters are collected in
(cf. \autoref{thm:belief-state-2})
\begin{equation}\label{eq:beta-update}
\alpha_{t} 
:= \left(\alpha^{(1)}_{t}, \alpha^{(2)}_{t}  \right)
:= \left(1 + \sum_{k=1}^{t-1} \tilde{r}_k, t - \sum_{k=1}^{t-1} \tilde{r}_k \right) 
\end{equation}
for $t = 2,...,T$, and where the action-adjusted reward is defined by
\begin{equation}\label{eq:r-tilde}
\tilde{r}_t := \frac{1}{2}\left(1 + (-1)^{r_t + a_t}\right) 
\end{equation}
for $t = 1,...,T$. The initial belief state is defined as
\begin{equation}\label{eq:b_1}
b_1 := p(s_1) = \mathrm{Unif}(s_1;[0,1]) = \mathrm{Beta}(s_1; \alpha_1) \mbox{ with } \alpha_1 := (1,1),
\end{equation}
\item $\phi$ denotes the agent's decision value function,
\begin{equation}\label{eq:phi-2}
\phi : B \to V, b_t \mapsto \phi(b_t) := \left(v_t^d\right)_{d \in D}
=: v_t.
\end{equation}
Here, the decision value vector $v_t$ takes values in $V$, with
\begin{equation}\label{eq:bayesian-decision-values}
v_t^0 :=  1-\mathbb{E}_{p(s_1|a_{1:t-1}, r_{1:t-1})}(s_1) \mbox{ and }
v_t^1 :=   \mathbb{E}_{p(s_1|a_{1:t-1}, r_{1:t-1})}(s_1), 
\end{equation}
\item $\delta$ denotes the agent's decision function
\begin{equation}\label{eq:delta-2}
\delta : V \to D, v_t \mapsto \delta(v_t) :=
\min \left(\argmax_{d \in D} v_t^d \right) =: d_t.
\end{equation}
Here, the decision random variable $d_t$ takes values in $D$. This rule applies
for $t=1,\ldots,T$ and resolves ties in favor of decision $0$.
\end{itemize}

The agent’s conceptual framework for learning and decision making is grounded in
a rational analysis approach, which assumes optimal adaptation to its decision
environment
\citep{anderson1990, oaksford2007, friston2005a, chater2006, dayan2008, griffiths2024, ma2023}.
In essence, the agent uses a recursive conjugate Beta-Bernoulli scheme for
estimating the task probability state and employs a cumulative reward maximizing
policy for making decisions in the absence of post-decision noise, as we detail in \autoref{thm:belief-state-2} and
\autoref{thm:decision-optimality-2}, respectively. In this symmetric task, either
action provides the same distribution of evidence about the probability state
after adjusting the reward for the chosen action. Choosing the option with
lower expected reward therefore offers no additional information benefit, so
maximizing immediate expected reward also maximizes cumulative expected reward.
The formal implementation of
the agent's learning and decision architecture is thus motivated by the
fundamental goal of estimating the latent probability state $s_1$, which encodes
the probability of obtaining a reward $r_t = 1$ for action $a_t  = 1$, as well
as the complement of the probability of obtaining a reward of $r_t = 1$ for
action $a_t = 0$. To achieve this, the agent aims to assign a value $v_t^1$ to
the decision $d_t$ that estimates $s_1$ based on all reward and action
observations up to trial $t-1$, while accounting for the uncertainty in its
estimate when starting from a uniform prior distribution over $s_1$. The value
of $d = 1$ should thus be larger than that of $d = 0$ if $s_1 > 0.5$ and vice
versa. As its decision value, the agent thus uses the expectation of its
(block-specific) trial-by-trial Beta posterior distribution
$p(s_1|a_{1:t-1}, r_{1:t-1})$, which is given by
\begin{equation}
\mathbb{E}_{p(s_1|a_{1:t-1}, r_{1:t-1})}(s_1)
= \frac{\alpha^{(1)}_{t}}{\alpha^{(1)}_{t} + \alpha^{(2)}_{t}}.
\end{equation}

\begin{figure}[!t]
\includegraphics[width=\linewidth]{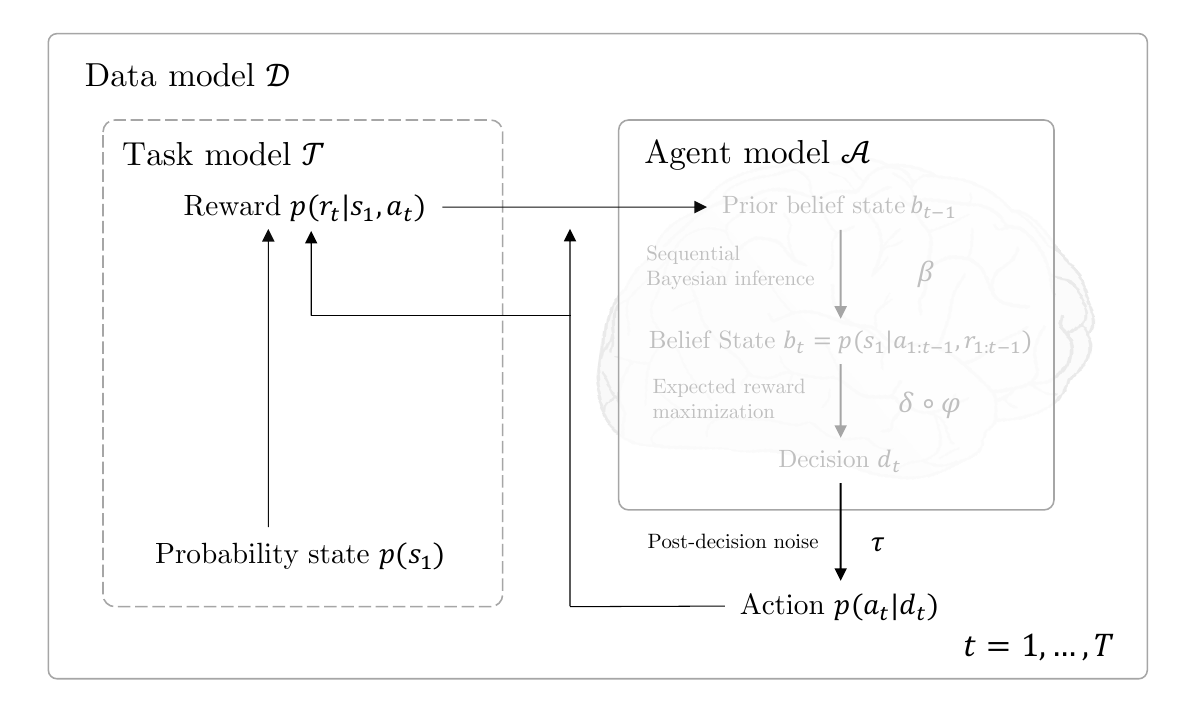}
\caption{Neurocognitive perspective of the symmetric bandit Bayesian
ABM.}\label{fig:7}
\end{figure}

This entails that $\alpha^{(1)}_{t}$ equals one plus the number of preceding
trials with adjusted reward $\tilde r=1$, while $\alpha^{(2)}_{t}$ equals one
plus the number of preceding trials with adjusted reward $\tilde r=0$, where the
added ones arise from the prior parameters $\alpha_1=(1,1)$. An adjusted reward
of one corresponds to the joint observations $(a_t = 1,r_t = 1)$ and
$(a_t = 0,r_t = 0)$. An adjusted reward of zero corresponds to the joint
observations $(a_t = 0,r_t = 1)$ and $(a_t = 1,r_t = 0)$. This counting process is implemented in the update formula
of \eqnref{eq:beta-update} for the Beta distribution parameter $\alpha_t$, where
the action-adapted reward $\tilde{r}_t$ maps the joint observations of actions
and rewards to parameter-specific observations by means of the biconditional
operation of \eqnref{eq:r-tilde}. Note that in line with the formulation of the
Gabor contrast discrimination task ABM, we consider a ``start-of-trial'' belief
state update: for $t=2,\ldots,T$, before forming the decision value and
committing to a decision,
the agent uses information from trial $t-1$, including its previous belief state
$b_{t-1}$ and its observations of its action $a_{t-1}$ and reward $r_{t-1}$, to
evaluate its belief state $b_t$ on the current trial, which is uniquely
represented by $\alpha_t$.

As shown in \autoref{thm:conditional-decision-distribution-2}, the agent
construction allows one to express the agent's observation-conditional decision
distribution, i.e., the distribution of the random variable
\begin{equation}\label{eq:dt-2}
d_t = \delta(\phi(b_t))\quad\mbox{for }t=1,\ldots,T,
\end{equation}
where $b_1=p(s_1)$ and, for $t=2,\ldots,T$,
\begin{equation}
d_t = (\delta \circ \phi \circ \beta)(b_{t-1},a_{t-1},r_{t-1}).
\end{equation}
Given the preceding action and reward history, its distribution takes the form
\begin{equation}\label{eq:p-dt-2}
p(d_t|a_{1:t-1},r_{1:t-1}) =  
\left\llbracket \alpha_t^{(1)} \le \alpha_t^{(2)} \right\rrbracket^{1-d_t}
\left\llbracket \alpha_t^{(1)}  >  \alpha_t^{(2)} \right\rrbracket^{d_t},
\end{equation}
where $\alpha_t^{(1)}$ and $\alpha_t^{(2)}$ are as in \eqnref{eq:beta-update}
for $t = 2,...,T$, $\alpha_1 =  (1,1)$, $p(d_1 = 0) = 1$, and $p(d_1 = 1) = 0$
based on \eqnref{eq:delta-2}. A sampling-based validation of \eqnref{eq:p-dt-2}
is shown in \fullref{fig:S8}A. Using \eqnref{eq:p-dt-2}, we may thus
equivalently express the agent model \eqnref{eq:A-2} in full probabilistic form
as
\begin{equation}\label{eq:A-p-2} 
\mathcal{A} := \left(\mathcal{T}, B, V, D, p(d_t|a_{1:t-1},r_{1:t-1})\right), 
\end{equation}
which greatly simplifies the derivation of the ABM's conditional log-likelihood
function in the following.

\paragraph{Data model} As in \autoref{sec:gabor-contrast-discrimination}, we
consider the data model
\begin{equation}\label{eq:D-2}
\mathcal{D} := \left(\mathcal{T}, \mathcal{A}, p(a_t|d_t) \right),
\end{equation}
where
\begin{itemize}
\item $\mathcal{T}$ is the task model,
\item $\mathcal{A}$ is the agent model,
\item $p(a_t|d_t)$ is the decision-conditional action probability distribution 
\begin{equation}\label{eq:p-at-giv-dt-2}
p(a_t|d_t) := \mathrm{Bern}(a_t;1-\tau)^{d_t}\mathrm{Bern}(a_t;\tau)^{1-d_t} \mbox{ for } \tau \in [0,0.5].
\end{equation}
\end{itemize}
Analogously to \autoref{sec:gabor-contrast-discrimination}, we also consider
\eqnref{eq:p-at-giv-dt-2} with the action noise parameter fixed at $\tau := 0.5$
as a random null model A0 (cf. \autoref{tab:abm-variants-2}).

\subsubsection{Equivalent prediction error-correcting agent}\label{sec:equivalent-prediction-error-correcting agent}

In \fullref{thm:bb-rw-equivalence}, we show that the Beta-Bernoulli symmetric
bandit agent \eqnref{eq:A-2} has the same decision values and decisions, for
identical action and reward histories, as the prediction error-correcting
adaptive learning rate Rescorla–Wagner agent \citep{rescorla1972}
\begin{equation}\label{eq:A-2-equivalent}
\mathcal{A} := \left(\mathcal{T}, S, V, D, \psi, \varphi, \delta\right),
\end{equation}
where
\begin{itemize}
\item $\mathcal{T}, V, D, \delta$ are as in \eqnref{eq:A-2},
\item $S$ is as in \eqnref{eq:T-2},
\item $\psi$ denotes the agent's probability state estimate-update function,
\begin{multline}\label{eq:psi}
\psi : S \times A \times R \to S, 
(\hat{s}_{t-1}, a_{t-1},r_{t-1}) \mapsto 
\\
\psi(\hat{s}_{t-1}, a_{t-1},r_{t-1}) := \hat{s}_{t-1} + \frac{1}{t+1}\left(\tilde{r}_{t-1} - \hat{s}_{t-1} \right) =: \hat{s}_t,
\end{multline}
such that the probability state estimate $\hat{s}_t$ takes values in $S$, with
$\tilde{r}_t$ as in \eqnref{eq:r-tilde} and $\hat{s}_1 := 0.5$,  
\item $\varphi$ denotes the agent's decision value function
\begin{equation}\label{eq:var-phi-2}
\varphi : S \to V, \hat{s}_t \mapsto \varphi(\hat{s}_t) := \left(1 - \hat{s}_t, \hat{s}_t\right) =: \left(v_t^d\right)_{d\in D}.
\end{equation}
\end{itemize}
Notably, this Rescorla–Wagner learner's state estimate equals the Bayesian 
posterior expectation for matching initializations $\hat{s}_1=0.5$ and
$\alpha_1=(1,1)$. Also note again that in contrast to the Rescorla–Wagner updates 
commonly reported in the cognitive neuroscience literature \citep[e.g.,][]{wilson2019, ma2023}, 
we consider a ``start-of-trial'' Rescorla–Wagner update that uses information from the
previous trial. This includes the previous probability state estimate
$\hat{s}_{t-1}$, its observed action $a_{t-1}$, and its observed reward
$r_{t-1}$, which are used to evaluate its probability state estimate $\hat{s}_t$
before committing to a decision on trial $t$
\citep[cf.][Chapter 14]{sutton2018}. Functionally, this is equivalent to
updating the estimate at the end of the previous trial. With estimates indexed
at trial start, an end-of-trial update uses the prediction error
$\tilde{r}_t-\hat{s}_t$ to obtain $\hat{s}_{t+1}$, whereas the corresponding
start-of-trial update uses $\tilde{r}_{t-1}-\hat{s}_{t-1}$ to obtain
$\hat{s}_t$. Because the agent model defined by \eqnref{eq:psi} and \eqnref{eq:var-phi-2}
induces exactly the same policy as the Bayesian learning agent discussed in
\autoref{sec:agent-model-a1-bayesian-learner}, we do not consider an
implementation of this agent.

\subsubsection{Agent model A2 (Rescorla–Wagner learner)}\label{sec:agent-model-a2-rescorla-wagner-learner}

Finally, we consider the Rescorla–Wagner-like agent model 
\begin{equation}\label{eq:A-3}
\mathcal{A} := \left(\mathcal{T}, S, V, D, \psi, \varphi, \delta\right),
\end{equation}
where
\begin{itemize}
\item $\mathcal{T}, V, D, \delta$ are as in \eqnref{eq:A-2}, 
\item $S$ is as in \eqnref{eq:T-2},
\item $\varphi$ is as in \eqnref{eq:var-phi-2}, 
\item $\psi$ denotes the agent's probability state estimate update function
\begin{multline}\label{eq:rw-update}
\psi: S \times A \times R \to S, 
\\ 
(\hat{s}_{t-1}, a_{t-1},r_{t-1}) \mapsto 
\psi(\hat{s}_{t-1}, a_{t-1},r_{t-1}) := \hat{s}_{t-1} + \lambda(\tilde{r}_{t-1} - \hat{s}_{t-1}) =: \hat{s}_t.
\end{multline}
Here, the probability state estimate $\hat{s}_t$ takes values in $S$, and
$\lambda \in [0,1]$ denotes the agent's learning rate, which determines the
weight assigned to the prediction error $\tilde{r}_{t-1}-\hat{s}_{t-1}$.
Unlike the trial-dependent learning rate $1/(t+1)$ of the equivalent Bayesian
learner introduced above, $\lambda$ remains constant across trials.
\end{itemize}
Rescorla–Wagner models have been popular and deemed important for understanding
associative learning for almost half a century
\citep[e.g.,][]{walkenbach1980, miller1995, siegel1996, soto2023}. From a
neurocognitive perspective, these models have repeatedly been shown to be
realizable by artificial neural networks, which may partly account for their
prominence in contemporary psychology and cognitive neuroscience
\citep[e.g.,][]{sutton1981, gluck1988, shanks1990, shanks1991, eckstein2026}.
From the Bayesian equivalence viewpoint provided in
\autoref{sec:equivalent-prediction-error-correcting agent}, this success may
also be partially explained by their capacity to act as crude expectation
parameter estimators when endowed with a fixed learning rate. Moreover, in
environments where ordinal estimation suffices for good decision making
\citep[e.g.,][]{wirth2017} or in non-stationary environments where Bayesian
relearning requires excessive evidence \citep[e.g.,][]{nassar2010, wilson2013},
Rescorla–Wagner models may indeed act as reasonably good approximative
estimators. We note that, vice versa, common criticisms with respect to Bayesian
theories of learning and decision making that lament their computational
complexity \citep[e.g.,][]{jones2011, bowers2012, kwisthout2011} may be
alleviated by noting that the Bayesian posterior expectation and induced
decisions in the present task can be computed using an adaptive-rate
Rescorla–Wagner update augmented with a trial counting mechanism.

In \fullref{thm:conditional-decision-distribution-3}, we show that the agent's
conditional decision distribution, i.e., the distribution of the random variable
\begin{equation}\label{eq:dt-3}
d_t := (\delta \circ \varphi \circ \psi)(\hat{s}_{t-1},a_{t-1},r_{t-1})
\end{equation}
given the agent's previous trial probability state estimate $\hat{s}_{t-1}$,
action $a_{t-1}$, and reward $r_{t-1}$, can be expressed in the form
\begin{equation}\label{eq:p-dt-3}
p(d_t|a_{1:t-1},r_{1:t-1}) :=  
\left\llbracket \hat{s}_t(\lambda) \le 1 - \hat{s}_t(\lambda) \right\rrbracket^{1-d_t}
\left\llbracket \hat{s}_t(\lambda)  >  1 - \hat{s}_t(\lambda) \right\rrbracket^{d_t},
\end{equation}
where the closed-form solution of the Rescorla–Wagner difference equation
\eqnref{eq:rw-update} has the form
\begin{equation}\label{eq:s_hat_t_lambda}
\hat{s}_t(\lambda) = (1 -\lambda)^{t-1}\hat{s}_1 + \lambda \sum_{j=1}^{t-1}(1-\lambda)^{t-1-j}\tilde{r}_j.
\end{equation}
Here, $\hat{s}_1=0.5$, and the expression applies for $t=1,\ldots,T$, with an
empty sum equal to zero and the convention $(1-\lambda)^0=1$, including at
$\lambda=1$. The recursive update and the decision composition above apply for
$t=2,\ldots,T$. On the first trial,
$d_1=(\delta\circ\varphi)(\hat{s}_1)=0$, following the tie-breaking convention.
Although a systematic study of the statistical properties of
$\hat{s}_t(\lambda)$ as an estimator for $s_1$ as a function of $\lambda$ is
beyond the scope of the current work, we nevertheless remark on its parametric
boundary cases. Specifically, if $\lambda = 0$, it follows that
$\hat{s}_t(\lambda) = \hat{s}_1$ for all $t = 2,...,T$. Equivalently, if
$\lambda = 1$, it follows that $\hat{s}_t(\lambda) = \tilde{r}_{t-1}$ for all
$t = 2,...,T$, as the only nonzero term in the cumulative sum in
\eqnref{eq:s_hat_t_lambda} corresponds to its last term with index $t-1$.

Using \eqnref{eq:p-dt-3}, we may thus equivalently express the agent model
\eqnref{eq:A-3} in full probabilistic form as
\begin{equation}\label{eq:A-p-3} 
\mathcal{A} := \left(\mathcal{T}, S, V, D, p(d_t|a_{1:t-1},r_{1:t-1})\right), 
\end{equation}
which greatly simplifies the derivation of the ABM's conditional log-likelihood
function.

\paragraph{Data model} For the Rescorla–Wagner
learner model, we consider the same data model as for the Bayesian learner, as
specified in \eqnref{eq:D-2}.

\paragraph{Agentic behavioral model} Taken together, the definitions of the task
model's state and state-action-conditional reward distribution in
\eqnref{eq:T-2}, the respective agent model's conditional decision distribution
in \eqnref{eq:p-dt-2} or \eqnref{eq:p-dt-3}, and the data model's conditional
action distribution in \eqnref{eq:p-at-giv-dt-2} imply a joint distribution of
the block-specific probability state $s_1$ and the decisions $d_t$, actions
$a_t$, and rewards $r_t$ across all trials $t = 1,...,T$. As in
\autoref{sec:gabor-contrast-discrimination}, the specific forms of these
distributions encode conditional independence assumptions for these random
variables that render their joint distribution expressible in the form

\begin{multline}\label{eq:abm-2}
p_\tau(s_1,d_{1:T},a_{1:T},r_{1:T}) 
\\
= p(s_1)p(d_1)p_\tau(a_1|d_1)p(r_1|s_1,a_1)\prod_{t=2}^Tp(d_t|r_{1:t-1},a_{1:t-1})p_\tau(a_t|d_t)p(r_t|s_1,a_t).
\end{multline}
Here, the parameter dependence is shown for the Bayesian learner. For the
Rescorla–Wagner learner, the joint distribution is denoted by
$p_{\lambda,\tau}$, reflecting the additional dependence of subsequent decisions
on the learning rate $\lambda$. The action distribution remains
$p_\tau(a_t|d_t)$ for both learners.
Again, it is this joint distribution that we refer to as an ABM for the
symmetric bandit task.

\subsubsection{Conditional log-likelihood functions} As shown in
\fullref{thm:conditional-log-likelihood-function-2}, the conditional
log-likelihood function of the Bayesian learner variant of \eqnref{eq:abm-2}
takes the following form
\begin{multline}\label{eq:llh-2}
\ell : \mathfrak{T} \to \mathbb{R}\cup\{-\infty\}, \tau \mapsto \ell(\tau) 
\\
:= \sum_{t=1}^T \ln \left(\left\llbracket \alpha_t^{(1)} \le \alpha_t^{(2)} \right\rrbracket \mathrm{Bern}(a_t;\tau) + \left\llbracket \alpha_t^{(1)} > \alpha_t^{(2)} \right\rrbracket \mathrm{Bern}(a_t;1-\tau) \right)
\end{multline} 
with $\alpha_t^{(1)}$ and $\alpha_t^{(2)}$ as in \eqnref{eq:beta-update}. Here,
$\mathfrak{T}:=[0,0.5]$ denotes the post-decision noise parameter space.
Note that the inclusion of $-\infty$ in the codomain of this function accounts 
for the case in which one or more observed actions have zero probability under 
the model when the post-decision noise parameter is $\tau = 0$. A sampling-based 
validation of the trialwise action probabilities inside the
logarithms in \eqnref{eq:llh-2} is provided in \fullref{fig:S8}B. 
Examples of \eqnref{eq:llh-2} for different values of the
generative, true but unknown post-decision noise parameter are visualized in
\fullref{fig:S8}C. Moreover, as shown in
\fullref{thm:conditional-log-likelihood-function-3}, the conditional
log-likelihood function of the Rescorla–Wagner learner variant of
\eqnref{eq:abm-2} takes the form
\begin{multline}\label{eq:llh-3}
\ell : \mathfrak{L} \times \mathfrak{T} \to \mathbb{R}\cup\{-\infty\}, (\lambda,\tau) \mapsto \ell(\lambda,\tau) 
\\ := 
\sum_{t=1}^T \ln \left(\left\llbracket \hat{s}_t(\lambda) \le 1-\hat{s}_t(\lambda) \right\rrbracket \mathrm{Bern}(a_t;\tau) + \left\llbracket  \hat{s}_t(\lambda) > 1- \hat{s}_t(\lambda) \right\rrbracket \mathrm{Bern}(a_t;1-\tau) \right)
\end{multline}
with $\hat{s}_t(\lambda)$ as in \eqnref{eq:s_hat_t_lambda}. Here,
$\mathfrak{L}:=[0,1]$ denotes the learning-rate parameter space, and
$\mathfrak{T}:=[0,0.5]$ denotes the post-decision noise parameter space.
Sampling-based validations of the trialwise action probabilities inside the
logarithms in \eqnref{eq:llh-3} are provided in
\fullref{fig:S10}, and exemplary profile functions of \eqnref{eq:llh-3} are
visualized in \fullref{fig:S11}.

As for the Gabor contrast discrimination ABM, evaluating the analytical
conditional log-likelihood functions of \eqnref{eq:llh-2} and
\eqnref{eq:llh-3} entails an approximately 400-fold increase in evaluation
speed for A1 and a 200-fold increase for A2 in the present implementation
compared to using their agentic counterparts. After
validating our analytical results extensively using simulations (cf.
\autoref{sec:bayesian-agent-model-analysis},
\autoref{sec:rescorla-wagner-agent-model-analysis}, and
\textit{abm\_likelihoods.py}), we used the analytical conditional log-likelihood
functions of \eqnref{eq:llh-2} and \eqnref{eq:llh-3} implemented in
\textit{abm\_llh.py} for all validation and evaluation analyses reported below.
\vspace{2mm}

\paragraph{Agent model variants, parameter estimation and model comparison}
\autoref{tab:abm-variants-2} provides an overview of the ABM variants used in
the current case study and their associated conditional log-likelihood
functions. For model validation and model evaluation, maximum-likelihood
parameter estimates, BIC approximations to the models' marginal likelihoods, and
PEPs were obtained as described in
\autoref{sec:parameter-estimation-and-model-comparison}.
For each participant, the log-likelihood was evaluated as the sum of the
block-specific log-likelihoods defined above, with the same participant-specific
model parameters used across all ten blocks. Probability states were independent
across blocks, and the initial beliefs were reset to $\alpha_1=(1,1)$ for the
Bayesian learner and $\hat{s}_1=0.5$ for the Rescorla–Wagner learner at the
beginning of each block. The within-block recursions thus covered 30 trials,
whereas the BIC sample size was the total of 300 trials per participant.
The parameter bounds used for fitting were informed by the face-validity
simulations and were set to $\lambda\in[0.1,0.9]$ and $\tau\in[0,0.5]$,
wherever the respective parameters were free. Fitting proceeded over continuous
intervals with the same endpoints as the simulation grids.

\begin{table}[t]
\caption{Intuitions, parameter constraints, and conditional log-likelihood
functions of the ABM variants applied in the symmetric bandit learning case
study.}
\label{tab:abm-variants-2}
\centering
\small
\renewcommand{\arraystretch}{1.8}
\begin{tabularx}{\linewidth}{@{}l@{\hspace{4pt}}X@{}}
\toprule
\textbf{A0}
& Random null model ($\tau := 0.5$)
\\
&
$\ell(\cdot) 
:= \sum_{t = 1}^T \ln \mathrm{Bern}(a_t;0.5)$
\\ \midrule
\textbf{A1}
& Bayesian learner 
\\
&
$
\ell(\tau) 
:= \sum_{t=1}^T \ln \left(\left\llbracket \alpha_t^{(1)} \le \alpha_t^{(2)} \right\rrbracket \mathrm{Bern}(a_t;\tau) + \left\llbracket \alpha_t^{(1)} > \alpha_t^{(2)} \right\rrbracket \mathrm{Bern}(a_t;1-\tau) \right)
$
\\ \midrule
\textbf{A2}
& Rescorla–Wagner learner 
\\
&
$
\ell(\lambda,\tau)
:= \sum_{t=1}^T \ln \left(\left\llbracket \hat{s}_t(\lambda) \le 1-\hat{s}_t(\lambda) \right\rrbracket \mathrm{Bern}(a_t;\tau) + \left\llbracket \hat{s}_t(\lambda) > 1-\hat{s}_t(\lambda) \right\rrbracket \mathrm{Bern}(a_t;1-\tau) \right)
$
\\
\bottomrule
\end{tabularx}
\end{table}

\subsection{Model validation}
To validate the ABM framework for the symmetric bandit task, we again performed
simulations that assessed its face validity, recoverability, and parameter
identifiability. Specifically, we simulated datasets using the A0, A1, and A2
variants of the ABM \eqnref{eq:abm-2} in a generative setting that mimicked the
experimental scenario regarding the number of participants, experimental blocks
per participant, and experimental trials per block. Again, for the parameter-free agent A0, the simulated data were generated without further specifications.
For models A1 and A2, we selected a parameter scope that allows the models to
express the experimentally observed effects and a parameter coverage that
balances computational efficiency with theoretical insight. As for the Gabor
contrast discrimination scenario, we thus decided on a resolution of nine evenly
spaced parameter values for each component of the model-specific parameter
vector, which we set to $\tau \in \{0$, $0.0625$, $0.125$, $0.1875$, $0.25$,
$0.3125$, $0.375$, $0.4375$, $0.5\}$, and
$\lambda \in \{0.1,0.2,0.3,0.4,0.5,0.6,0.7,0.8,0.9\}$. This resulted in nine
true but unknown parameter values for the one-dimensional parameter space of
model A1 and 81 true but unknown parameter values for the two-dimensional
parameter space of model A2. For each parameter setting, we generated $n = 60$
datasets comprising 10 blocks and 30 trials per block. Across these simulated
participant datasets, we evaluated the group maximizing action learning curve
as a descriptive statistic. We assessed the A0, A1, and A2 variants of the ABM
(\eqnref{eq:abm-2}) using the procedures described in
\fullnameref{sec:parameter-estimation-and-model-comparison}. For implementation
details, please refer to \textit{abm\_validation.py} and
\textit{abm\_validate.py}. \fullref{fig:8} visualizes the results.

\begin{figure}[!htbp]
\center
\includegraphics[width=\linewidth]{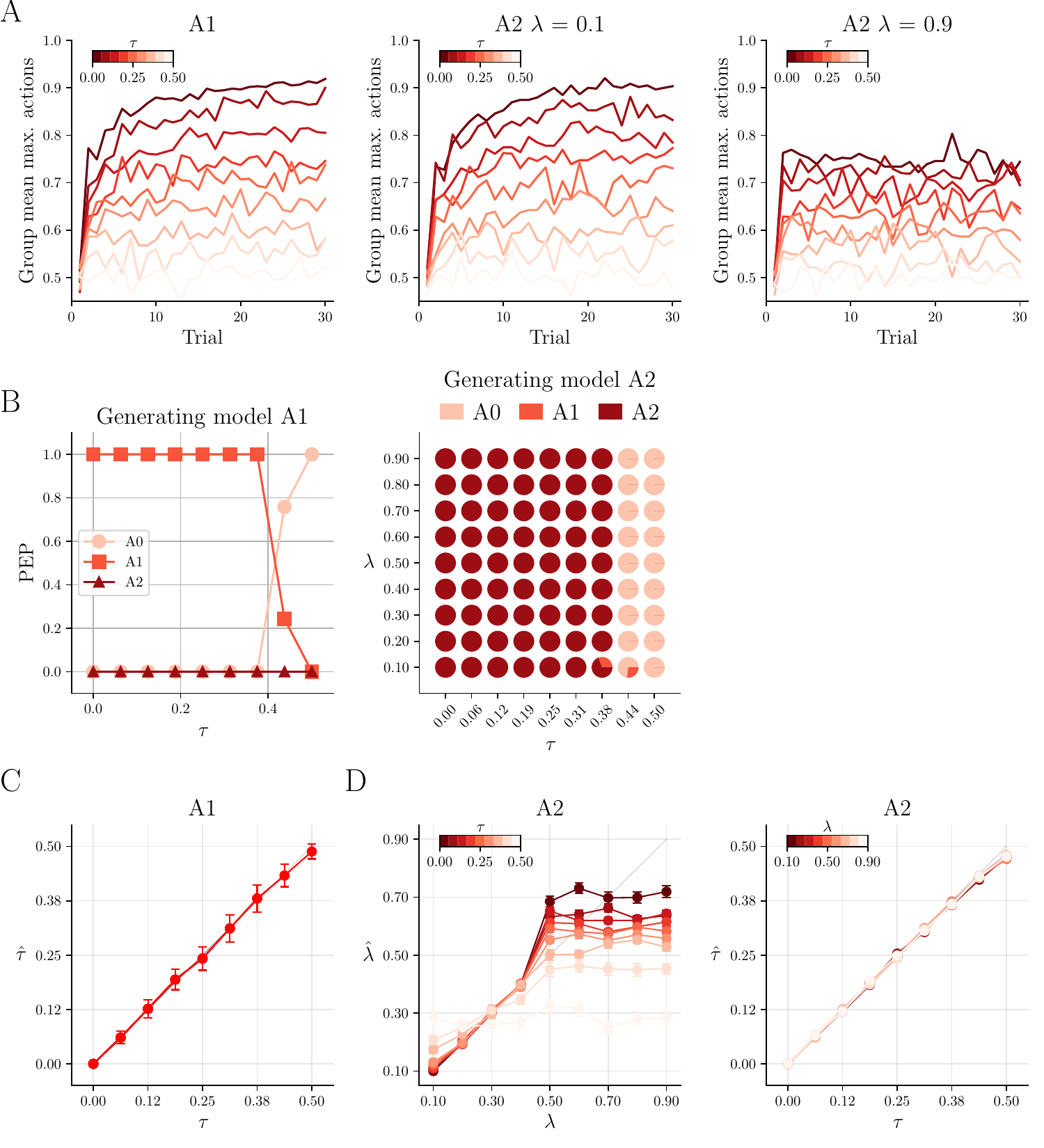}
\caption{Model validation for the symmetric bandit ABM. (A) Model face validity.
The left subpanel visualizes the effect of varying the action noise parameter
$\tau$ for ABM variant A1. The center and right subpanels visualize the effect
of varying $\tau$ for ABM variant A2 for a low ($\lambda = 0.1$) and a high
($\lambda = 0.9$) learning rate parameter. (B) Model recovery. The left subpanel
visualizes the PEP fractions attained by the data-analytical ABM variants A0,
A1, and A2 for data generated by ABM variant A1 as a function of the action
noise parameter $\tau$. The right subpanel visualizes the PEP fractions attained
by A0, A1, and A2 for data generated by A2 as a function of the generating
model's parameter configuration in pie chart form. (C) Recoverability of the
action noise parameter $\tau$ for data generated by ABM variant A1. (D)
Recoverability of the learning rate parameter $\lambda$ as a function of $\tau$
for data generated by ABM variant A2 (left subpanel) and recoverability of
$\tau$ as a function of $\lambda$ for data generated by ABM variant A2 (right
subpanel). Error bars indicate standard deviations in (C) and standard errors
of the mean in (D).}\label{fig:8}
\end{figure}

\paragraph{Model face validity} To validate the formulation and implementation
of the ABM variants and to explore their behavioral repertoires, we evaluated
the effects of the post-decision noise parameter $\tau$ on the group maximizing
action learning function of model variants A1 and A2 (\fullref{fig:8}A). The
first subpanel shows the effect of varying the post-decision noise parameter on
the group mean maximizing action learning curve of model A1. As expected,
increasing the post-decision noise parameter diminishes A1's maximizing
behavior. The second subpanel depicts the same effect for model A2 with a
learning rate parameter of $\lambda = 0.1$. Notably, the group mean maximizing
action curves of model A1 exhibit a continuing learning trend even after trial
$t = 20$, which is less pronounced for model A2. Increasing the learning rate
parameter of model A2 beyond 0.1 has detrimental effects on its maximizing
action behavior. In this regime, its state estimates become too strongly
influenced by the variability in reward observations and fail to maintain the
trial-averaging property observed for low learning rates, as shown in the third
panel of \fullref{fig:8}A for $\lambda = 0.9$.

\paragraph{Model recovery} To assess the degree to which the three ABM variants
can be reliably differentiated by our model comparison framework, we again
performed a number of model recovery analyses. For each of the three ABM
variants, we analyzed the data generated by the given ABM variant using all
three ABM variants as analytical models. For the generative model A0, we
obtained the PEPs $\varphi_{A0} = 1.00$, $\varphi_{A1} = 0.00$, and
$\varphi_{A2} = 0.00$, indicating that the null model variant A0 can be reliably
recovered as the most parsimonious explanation of a fully random action dataset.
Again, the model recovery properties of ABM variants A1 and A2 depend on the
true but unknown generative parameter values of these models, as visualized in
\fullref{fig:8}B. The left subpanel of \fullref{fig:8}B depicts the model
recovery properties of model variant A1. Here, for a post-decision noise
parameter value of $\tau < 0.4$, the model is reliably recoverable. This
indicates that the model complexity penalty assigned to the additional parameter
of model variant A2 outweighs the arguably similar data generative potentials of
model variants A1 and A2 for low learning rate parameters (cf.
\fullref{fig:8}A). For the high noise scenario, parameter values
$\tau =  0.4375$ and $\tau = 0.5$, model A0 offers a more parsimonious
explanation. As in \fullref{fig:4}, the right subpanel of \fullref{fig:8}B
depicts the relative PEP shares of ABM variants A0, A1, and A2 as pie charts in
the generative parameter space of variant A2. Across most of the parameter
space, model A2 can be reliably recovered. However, in high post-decision noise
regimes with $\tau > 0.4$, data produced by model variant A2 is best accounted
for by model variant A0, regardless of the learning parameter values.

\paragraph{Parameter recovery}
\fullref{fig:8}C visualizes the recoverability of the post-decision noise
parameter $\tau$ for model variant A1. Here, the group average estimates
$\hat{\tau}$ ($y$-axis) are virtually identical to the true but unknown values
$\tau$ ($x$-axis), while the variability of the estimates increases with $\tau$,
except for its highest value, which exhibits less variability in the estimates.
The left subpanel of \fullref{fig:8}D depicts the results of the parameter
recovery simulation for the learning rate parameter $\lambda$ of model variant
A2 as a function of the true but unknown post-decision noise parameter $\tau$.
For most values of $\lambda > 0.5$, the group averages of the learning rate parameter estimates
$\hat{\lambda}$ deviate considerably from their corresponding
true but unknown values. Appropriate recoverability is exhibited for learning
rates between $\lambda = 0.1$ and $\lambda = 0.3$ and low to medium levels of
post-decision noise of $\tau < 0.4$. For higher values of the true but unknown
post-decision noise, the learning rate parameter estimates exhibit biases.
Similarly, for true but unknown learning rate parameters $\lambda > 0.5$, the
group average parameter estimates remain approximately constant across the
actual values of $\lambda$ and instead depend predominantly on the magnitude of
the post-decision noise parameter. The right subpanel of \fullref{fig:8}D depicts the
results of the parameter recovery simulation for the post-decision noise
parameter $\tau$ of model variant A2 as a function of the true but unknown
learning rate parameter $\lambda$. Compared to parameter recovery for the
learning rate parameter, the biases exhibited by parameter estimates
$\hat{\tau}$ for the post-decision noise parameter are negligible.

\subsection{Model evaluation}\label{sec-model-evaluation-2}

Upon validating the ABM set, we evaluated it in light of the experimental data.
\fullref{fig:9} visualizes the results at the group level. The group cumulative
BIC was highest for ABM variant A1 (\fullref{fig:9}A), and its associated PEP
clearly dominated the random-effects model comparison (\fullref{fig:9}B). At the
participant level, the BIC value of model variant A1 was maximal for 44 of the
participants, while the BIC value of model variant A2 was maximal for the
remaining 16 participants (\fullref{fig:S7}). The parameter estimates
$\hat{\tau}$ for the post-decision noise parameter for datasets best explained
by ABM variant A1 ranged between 0.02 and 0.22 (\fullref{fig:9}C). The
parameter estimates $\hat{\tau}$ for the post-decision noise parameter for
datasets best explained by ABM variant A2 ranged between 0.09 and 0.25
(\fullref{fig:9}D), i.e., with a slight shift to higher values compared to those
obtained under ABM variant A1. The learning rate parameter estimates
$\hat{\lambda}$ varied between 0.10 and 0.67, with a cluster at low values and a
smaller cluster at high values (\fullref{fig:9}E). Taken together, the results
indicate that, within the evaluated model space, ABM variant A1 provides the
most plausible explanation of the observed participant actions. For a subgroup
of participants, the additional model flexibility afforded by the
Rescorla–Wagner learning parameter provides a more plausible explanation of the
data than the Bayes optimal learning dynamics of ABM variant A1.

\begin{figure}[!htbp]
\center
\includegraphics[width=\linewidth]{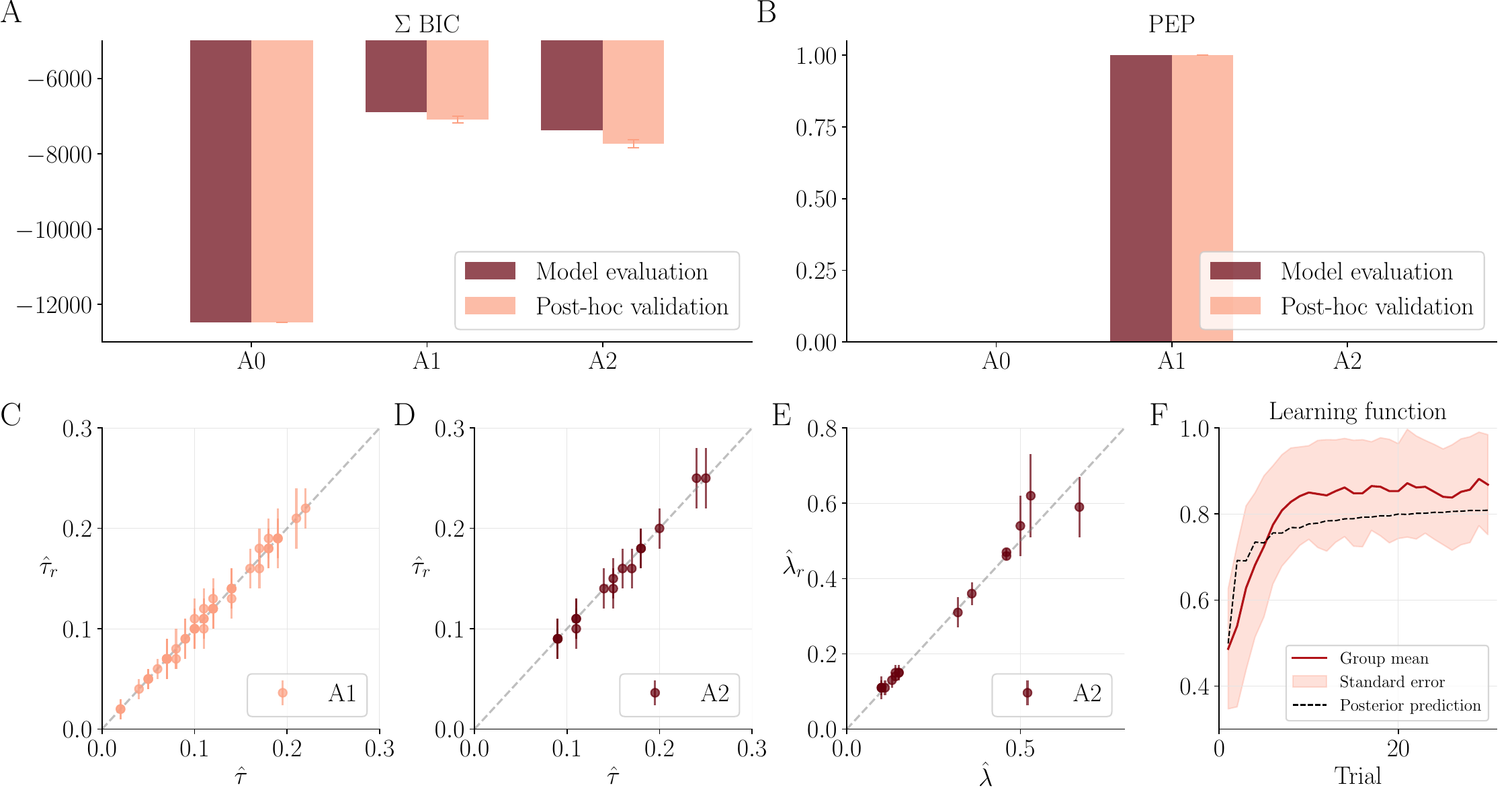}
\caption{Model evaluation and post hoc model validation of the symmetric bandit
learning case study. (A) Group-cumulative BIC values, corresponding to the sums
of the participant-specific BIC values for each ABM variant. Bar colors indicate
the results from the experimental data and the simulated post hoc validation
data. Error bars indicate the standard deviation across post hoc validation
simulations. (B) PEPs for each ABM variant. As in (A), bar colors indicate the
results from the experimental data and the simulated post hoc validation data.
Error bars indicate the standard deviation across post hoc validation
simulations. (C), (D), (E) Parameter estimation and post hoc recovery results
for ABM variants indicated in the respective subpanels' legends. Parameter estimates
$\hat{\tau}$ and $\hat{\lambda}$ based on the experimental data are displayed on
the $x$-axes. The recovered average parameter estimates $\hat{\tau}_r$ and
$\hat{\lambda}_r$ and their standard deviations across post hoc simulations are
shown on the $y$-axes. (F) Group mean experimental data as in \autoref{fig:6}D
and posterior predictive average learning function.}\label{fig:9}
\end{figure}

\subsection{Post hoc model validation}
Upon identifying the most plausible ABM variant for the experimentally observed
data of each participant, we performed post hoc model validation analyses by
generating datasets based on the identified participant-specific model
comparison and parameter estimation results. For this, we simulated 100
datasets per participant using the most plausible model for each participant
and the corresponding parameter estimates. We then subjected the generated
datasets to the parameter estimation and model comparison procedures to evaluate
the average post hoc model and parameter recoverability performance. The results
of these post hoc model validations are integrated into \fullref{fig:9}. As for
the experimental data, the post hoc simulated average group cumulative BIC was
highest for ABM variant A1 (\fullref{fig:9}A), and its associated average PEP
clearly dominated the random-effects model comparison (\fullref{fig:9}B). The
post hoc parameter recovery was largely successful for the post-decision noise
parameter and lower learning rate values, as indicated by the largely
diagonally located experimental and post hoc parameter estimates shown in
\fullref{fig:9}C to E. Notably, for higher learning rate
estimates with $\hat{\lambda} > 0.5$, the post hoc recovered parameter estimates
$\hat{\lambda}_r$ exhibited increased variability, consistent with the parameter
recovery results shown in \fullref{fig:8}D.

Finally, at the descriptive level, we evaluated the mean fraction of maximizing
actions across post hoc simulations. As in the Gabor case study, ``posterior
predictive'' here refers to predictions conditional on the selected models and
their maximum-likelihood parameter estimates.
As shown in \fullref{fig:9}F, the posterior
predictive average learning function lies well within the error bounds of the
experimental learning function and achieves a similar level of average
maximizing actions. However, on a more fine-grained level, the posterior
predictive learning function, which is dominated by the 44 datasets best
explained by ABM variant A1, exhibits a slightly different dynamic compared to
the experimental learning curve. Specifically, the conditional probability of
choosing the maximizing action is higher than that observed in the experimental
data from the first trials. To match this in overall performance and account for
additional participant-specific variability, it is lower than that observed in
the experimental data on trials following the tenth trial.

Close inspection of the single-participant posterior predictive learning curves
(\fullref{fig:S7}) indicates that a subset of about three participants (most
evidently P2, P6, P50) appears to have used a strategy that favored a stronger
degree of action variation in the first five to ten trials than predicted by the
Bayesian learning ABM variant A1. Another subset of about five participants
(most evidently P5, P8, P13, P18, P55) appears to have forgone maximizing their
cumulative reward in the initial trials by disproportionately choosing the less
rewarded option. One possible interpretation is that these participants sought
to reduce their subjective uncertainty before switching to and remaining with
the identified higher reward action. Under the specified symmetric task model,
however, choosing the less rewarded option offers no additional information
benefit, and the observed choice patterns do not uniquely identify an
uncertainty-reduction strategy. Because the dynamics of
the Bayesian learning ABM variant A1 rest on a steep rise towards the higher
reward action (cf. \fullref{fig:8}A), the parameter estimation procedure must
accordingly identify a post-decision noise parameter that accounts for these
early anti-maximizing actions, resulting in a lower average maximizing action
selection in later trials.

\subsection{Summary}\label{sec:summary-2}

In summary, this case study instantiates the ABM framework for a minimal bandit
task using both Bayesian and Rescorla–Wagner learning and decision-making
architectures, where the Bayesian architecture implements a reward-maximizing
policy in the absence of post-decision noise. We show that an adaptive-rate
Rescorla–Wagner update reproduces the Bayesian posterior expectation and induced
decisions in this task. From an
experimentalist’s perspective, we demonstrate that, in the present task, the
models can be reliably distinguished under low to moderate post-decision noise.
The learning rate parameter of the Rescorla–Wagner model is recoverable only for
relatively small values, whereas the post-decision noise parameter can be
robustly identified. In the current task setting, Bayesian agent A1 provides a
more plausible account of the empirical data while acknowledging both the
restricted model space considered and the fact that the Bayesian agent does not
fully capture certain aspects of the learning dynamics.

\clearpage 
\FloatBarrier
\section{Discussion}

Both AI research and cognitive behavioral science develop agent-centric models
of cognition and behavior. In this work, we present agentic behavioral modeling
(ABM) as a formal framework for specifying and analyzing how artificial agents
can be employed as tools for behavioral modeling within cognitive behavioral
science. As demonstrated in two case studies, ABM benefits from an explicit
differentiation of three key components: the \textit{task model}, which
formalizes the outcome-generating processes of the experimental paradigm; the
\textit{agent model}, which encodes hypotheses about the cognitive processes
underlying human task performance; and the \textit{data model}, which defines
the statistical framework for evaluating and comparing alternative agent models
on the basis of experimental data. Although virtually all contemporary
cognitive behavioral models incorporate these components, their boundaries are
often blurred. We firmly believe that clarifying these distinctions supports
accurate attribution of observed behavioral effects to task constraints (task
model), posited cognitive mechanisms (agent model), or residual variance (data
model) and promotes transparent documentation, understanding, and
interpretability of cognitive behavioral models. The principal contribution of
this work thus lies in the systematic differentiation and integration of these model
components. Within the ABM framework, all three model components are formulated
as marginal and conditional probability distributions that together define a
parameterized joint distribution. Analysis of this joint distribution, in turn,
enables detailed characterization of latent and observed variable structures,
facilitates the derivation of core agent properties such as belief states and
conditional decision distributions, and provides a foundation for probabilistic
model inference. In the remainder of this discussion, we briefly contrast ABM
with related frameworks in the literature. We note that these frameworks vary
in their degree of model agnosticism and are therefore at times only partially
comparable to ABM, which is model-agnostic by design.

\subsubsection*{Related frameworks}

\citet{daunizeau2010a} proposed a \textit{meta-Bayesian} account of behavioral
modeling that closely relates to the present agentic perspective. Their approach
is based on a principled partitioning of the behavioral model space into a
\textit{perceptual model} and a \textit{response model}, where the former
governs the agent's inference over hidden environmental states, and the latter
maps posterior beliefs to observable behavior. This joint model rests on the
assumption that agents make Bayes-optimal, utility-maximizing decisions and is
inverted using variational methods, typically under Laplace approximations. Its
empirical relevance is demonstrated in \citet{daunizeau2010} with respect to an
audio-visual associative learning paradigm. In this paradigm, reaction times are
interpreted as reflecting the dynamics of belief updating under a speed–accuracy
trade-off. From the viewpoint of ABM, the framework by \citet{daunizeau2010a} is
largely complementary: the ABM task–agent interface corresponds to the
perceptual model of \citet{daunizeau2010a}, while the agent–data interface
aligns with their response model. At the same time, \citet{daunizeau2010a} adopt
stronger agent assumptions (fully Bayesian and utility-maximizing) than ABM and
rely on a general variational inversion framework with strong roots in
continuous data analysis traditions
\citep[cf.][]{friston2007,daunizeau2009, ostwald2014}. In contrast, ABM
emphasizes explicit joint distributions and conditional likelihoods that can
remain analytically tractable in the simple and discrete choice settings
considered here. Finally, their use of the unprotected exceedance probability
framework by \citet{stephan2009} reflects their contemporary developments in
group-level model comparison. Looking forward, a promising direction is to treat
the Bayesian agents studied herein as special cases of the variational inverse
Bayesian decision-theoretic framework by \citet{daunizeau2010a}, while
simultaneously maintaining ABM's more explicit model partitioning framework, its
focus on full joint distributions and conditional likelihoods, and its generally
less constrained agent approach. In the present empirical settings, no
approximations are required; however, it is easy to see that more complex task,
agent, and data models will naturally induce intractable conditional
distributions that necessitate approximate inference in ABM.

\citet{smith2022} share with \citet{daunizeau2010a} a common emphasis on
variational inference as the central mechanism for modeling both perception and
action. They present a comprehensive tutorial on the application of
\textit{active inference} agents in \textit{partially observable Markov decision
process (POMDP)} settings, providing a unified framework for modeling
perception, learning, and decision-making. In essence, active inference agents
implement the free energy principle for both perceptual inference (i.e., belief
state formation via variational free energy minimization) and policy selection
(via expected free energy minimization of a generative model with prior
preferences over outcomes; cf. \citet{friston2012f, friston2017a}). While
\citet{smith2022} strongly emphasize the POMDP formulation, their treatment
largely abstracts from the conventional literature on solving POMDPs
\citep[cf.][]{bertsekas1995, krishnamurthy2016}, instead focusing on
implementations provided within the active inference SPM toolbox
\citep{tierney2025, parr2022}. From the perspective of ABM, several conceptual
differences emerge. In contrast to ABM’s explicit separation into task, agent,
and data models, the approach by \citet{smith2022} employs a single,
agent-centric generative model that simultaneously encodes environmental states,
agent beliefs, and action policies. Conceptually, this weakens the distinction
between the world and agent models and omits a dedicated data model layer. As a
result, comparatively less emphasis is placed on empirical data analysis; while
model fitting is discussed, conditional log likelihood formulations or related
analytical expressions are not made explicit. At the same time, both ABM and the
framework by \citet{smith2022} share similarities in their focus on discrete
state and action spaces and in their application to explore–exploit paradigms,
such as bandit tasks. This also raises the question of whether the full
machinery of variational approximations in active inference is necessary for
behavioral modeling in such scenarios. Stylistically, \citet{smith2022} adopt a
tutorial approach, revisiting foundational probabilistic concepts such as Bayes’
theorem and providing detailed guidance on implementation (e.g., function calls
within MATLAB), but comparatively less emphasis is placed on analytical
derivations. Their use of the protected exceedance probability framework follows
standard practice but is treated largely as a black-box routine. Finally, the MATLAB-based implementations discussed by
\citet{smith2022} should be distinguished from the broader active inference
software ecosystem. \citet{heins2022}, for example, introduced \textit{pymdp},
a Python library for active inference in discrete state spaces. Its current
implementation also supports Bayesian parameter estimation from behavioral
data through integration with probabilistic programming tools.
Looking forward, given the broad interest in active inference as an
agent-centric formulation of the free energy principle \citep{parr2025},
explicitly incorporating active inference agents into the ABM framework
represents a promising direction. This would involve relating their
task-specific formulations to ABM's explicit separation of task, agent, and
data models and its conditional likelihood-based inference and model
comparison.

\citet{lieder2020} extend the tradition of \textit{rational analysis} by
proposing \textit{resource-rational analysis} as an integration of rational
principles with realistic cognitive constraints, thereby building on the ACT
framework for behavioral modeling \citep{anderson1990}. Their central idea is
that cognitive mechanisms should be understood as optimizing performance under
limited computational resources, yielding resource-rational heuristics that
trade off accuracy against time, memory, and computational cost. Conceptually,
the presentation of the framework in \citet{lieder2020} is largely informal. It
is structured as a sequence of five conceptual steps: starting from a
computational-level problem formulation, positing a class of feasible algorithms
and their costs, deriving the algorithm that optimally balances resource
expenditure and approximation accuracy, evaluating its empirical predictions,
and iteratively refining assumptions. From the perspective of ABM, this approach
provides a normative account of agent models but does not explicitly address
task models or the statistical estimation and adjudication of agent models based
on empirical human data. In particular, resource-rational analysis is not
formulated in terms of a joint probability distribution over latent and observed
variables, and questions of scientific inference, likelihood-based estimation,
or group-level model inference are not alluded to in \citet{lieder2020}. While
illustrative examples such as discrete Mouselab-type decision scenarios are
discussed, these remain conceptual rather than fully specified modeling
frameworks. In contrast, ABM emphasizes explicit generative formulations,
conditional likelihoods, and concrete implementations that enable empirical
model comparison. Nevertheless, resource-rational heuristics constitute a
promising extension of the agent modeling toolbox, and incorporating such
principles into ABM represents a fruitful direction for future work. In
particular, making resource-rational analysis mathematically and
implementationally explicit would facilitate its integration with
likelihood-based inference and may provide a principled link between bounded
rationality and approximate Bayesian inference, e.g., by interpreting
sampling-based approximations such as MCMC as resource-rational strategies
\citep[cf.][]{sanborn2020}.

Finally, \citet{chandramouli2024} propose a workflow for building
computationally rational models of human behavior that is similar in spirit to
the resource-rational framework by \citet{lieder2020}, but with a stronger
emphasis on technical and practical aspects of model construction. Their
approach centers on POMDPs as a unifying formalism for modeling adaptive
behavior under cognitive and environmental constraints, and, in contrast to
\citet{smith2022}, explicitly connects to the broader literature on POMDPs and
reinforcement learning. A key contribution of \citet{chandramouli2024} is its
extensive discussion of good scientific practice in computational modeling,
including model specification, parameter estimation, and evaluation, with
particular emphasis on procedures such as prior and posterior predictive checks,
model validation, and iterative refinement based on group model inference. From
the perspective of ABM, however, the framework remains largely conceptual in
nature and does not provide a fully worked-out applied example in the sense of a
complete likelihood-based behavioral model. In addition, it does not introduce a
clear distinction between task, agent, and data models, but instead adopts a
unified POMDP-based formulation with a specific focus on (boundedly) optimal
decision-making. While this aligns with ABM in its principal use of generative
simulations and its connections to optimality principles, it differs in its
limited emphasis on explicit probabilistic formulations and conditional
likelihood-based inference. Looking forward, the framework by
\citet{chandramouli2024} offers a promising foundation for integrating
workflow-based model construction with more explicit statistical modeling
approaches. In particular, while it places less emphasis on specific inference
architectures such as variational inference, it maintains an optimality
perspective that is also central to the approaches by \citet{daunizeau2010a} and
\citet{lieder2020}. At the same time, its treatment of computational constraints
remains comparatively light relative to resource-rational analysis, suggesting
that future work could benefit from a more explicit integration of bounded
rationality assumptions within a unified, likelihood-based modeling framework.

The frameworks discussed above leave different aspects of agent modeling open.
The meta-Bayesian approach of \citet{daunizeau2010a} permits alternative
perceptual and response models while organizing them around Bayesian inference
and utility-based decision-making. The active inference formulation presented
by \citet{smith2022} likewise accommodates different generative models, but
specifies variational inference and free energy objectives as organizing
principles for perception and policy selection. Resource-rational analysis
leaves considerable freedom in the choice of cognitive mechanisms,
representations, and resource costs, while requiring these choices to be
evaluated in terms of the trade-off between performance and computational
expenditure \citep{lieder2020}. Similarly, the computational rationality
framework of \citet{chandramouli2024} accommodates different models of cognitive
constraints while explaining behavior through policies that optimize expected
utility subject to those constraints. ABM does not prescribe a particular
inference algorithm or require agents to satisfy an optimality principle. Its
model agnosticism instead concerns the choice of cognitive mechanisms within an
explicitly specified probabilistic architecture. In this sense, the related
frameworks and ABM are complementary: the former provide substantive principles
for constructing agent models, while ABM provides a common task--agent--data
formulation within which their empirical adequacy can be assessed.

\subsubsection*{Methodological challenges and future directions}

As always in model-based data analysis, core questions of ABM concern how
cognitively motivated and interesting models can be used reliably and
efficiently in the analysis of empirical data. We believe that understanding the
conditional log likelihood functions implied by these models is an important
part of this endeavor. Their analysis can clarify which assumptions connect
cognitive mechanisms to observations, which parameters the data can distinguish,
and where inference may encounter computational or statistical difficulties.
Such understanding remains valuable when analytical likelihood evaluation is
infeasible and simulation-based, often termed ``likelihood-free'', methods are
employed \citep{fenglerLikelihoodApproximationNetworks2021,radev2022}. While
these methods broaden the range of models accessible to empirical analysis, we
believe that generally dispensing with the explicit evaluation of conditional
log likelihood functions or approximations does not necessarily foster our
understanding of the inference frameworks we use to understand human cognition.

Our current choice of participant-level maximum likelihood estimation, BIC
approximations to log model evidence, and PEP-based group-level inference was
primarily pragmatic. For the illustrative case studies considered here, this
summary-statistics approach offers a comparatively simple and transparent
combination of parameter estimation and model comparison. Maximum likelihood
provides participant-specific parameter estimates without requiring prior
specification or posterior approximation, while BIC provides computationally
inexpensive evidence approximations as necessitated by the subsequent PEP
group-level analysis. Our choice here serves the methodological and didactic
aims of the present work and is not meant to imply that these procedures are
generally preferable to Bayesian parameter estimation or predictive model
comparison. We recognize that these methods have limitations and that, in
particular, the accuracy of BIC's approximation to log model evidence requires
scrutiny \citep{schwarz1978,konishi2008}. This motivates our parameter- and
model-recovery simulations, which assess how reliably the chosen estimation and
model-comparison procedures recover known generating parameters and models in
the scenarios considered here.

The PEP framework provides a useful connection between participant-level model
assessment and group-level inference \citep{stephan2009,rigoux2014}. Its
hierarchy concerns participant model assignments and population model
frequencies, allowing different models to account for different participants.
While the variational inference approach of this framework allows for
approximating the posterior over these quantities conditional on the
participant-level evidence approximations, it does not introduce a hierarchy
over participant parameters or allow participants to inform one another's
parameter estimates. Moreover, our use of independently estimated point
estimates does not propagate parameter uncertainty into subsequent analyses,
such as comparisons between participant groups. In this regard, hierarchical
empirical Bayesian approaches offer a promising extension by representing
participant parameters as draws from model-specific and, where appropriate,
group-specific distributions \citep{pirayHierarchicalBayesianInference2019}. In
the ABM setting, an attractive goal is to combine this parameter hierarchy with
inference over participant model assignments and population model frequencies.
Extending the ABM approach to diverse agent architectures nevertheless requires
model-specific work on prior specification, posterior approximation, and
computational implementation. Variational methods may facilitate this work, yet
the study of their statistical accuracy and optimization behavior remains an
active field of research
\citep{wang2019a,wu2024,margossian2025,wu2026,guilmeau2026}. Finally, these
considerations connect to current discussions of experimental design for
group-level model inference in computational behavioral modeling
\citep{piray2025}. Together with the original methodological contributions, our
treatment of protected exceedance probabilities, which explicitly account
for the null hypothesis of equal model probabilities \citep{rigoux2014}, in
\fullnameref{sec:protected-exceedance-probabilities} may provide a starting
point for further work on recovery, power, and sample size.

\subsubsection*{Conclusion} 

To conclude, we have reviewed computational behavioral modeling from an
agent-centric perspective. To this end, we partitioned the modeling process into
a task model that captures the choice environment of the experimental design, an
agent model as an algorithmic hypothesis of human cognitive processes, and a
data model as a statistical embedding that enables inference and model
adjudication. Using concrete application scenarios, we have shown how these
components jointly specify a probability model that serves as the central pivot
of the modeling process, enabling, where appropriate, the derivation of optimal
policies and providing the foundation for parameter estimation and principled
model comparison. Throughout this work, we have focused on overarching
conceptual principles, their rigorous mathematical formalization, and their
concrete realization in computational implementations. We hope that this work
provides a useful and extensible foundation for researchers engaging in
computational behavioral modeling and that it supports future methodological
advances at this exciting interface between human and artificial intelligence
research.

\clearpage
\small 
\subsubsection*{Declarations}

\paragraph{\textbf{Funding}}
This research did not receive any specific grant from funding agencies in the
public, commercial, or not-for-profit sectors.

\paragraph{\textbf{Declaration of Competing Interest}} None.

\paragraph{\textbf{Declaration of the Use of Generative AI and
AI-assisted Technologies}} The authors used generative AI tools, including ChatGPT, Gemini, Claude,
and GitHub Copilot, to assist with formal derivations and code development for 
quality assurance and refinement. During manuscript preparation,
the authors used OpenAI’s GPT (via Writefull for Overleaf) to support language
editing, including grammar, readability, and phrasing. All outputs generated by
these tools were carefully reviewed, verified, and, where necessary, revised by
the authors. The authors retain full responsibility for the content of this
work.

\paragraph{\textbf{Data and Code Availability}} The data and code supporting the
findings of this study are publicly available via PsychArchives at
\url{https://doi.org/10.23668/psycharchives.21896}. This repository includes all
datasets and analysis scripts necessary to reproduce the reported results.
Access and use are subject to a scientific use license in accordance with the
specifications of the Leibniz Institute for Psychology Information (ZPID).

\paragraph{\textbf{Author Contributions}} 
D.O.: 
Conceptualization, 
Data curation, 
Formal analysis,
Investigation,
Methodology,
Project administration,
Software,
Validation,
Visualization,
Writing - original draft,
Writing - review and editing;
R.B.: 
Conceptualization, 
Formal analysis,
Methodology,
Writing - review and editing;
F.U.: 
Conceptualization,
Writing - review and editing;
B.F.: 
Conceptualization,
Writing - review and editing;
J.S.: 
Writing - review and editing;
S.M.: 
Investigation,  
Writing - review and editing.

\clearpage    
\setlength{\bibsep}{0pt}
\bibliography{Ostwald-et-al-R1}
\clearpage

\setcounter{page}{1}
\pagenumbering{arabic}

\begin{appendices}
\renewcommand{\thesection}{S\arabic{section}}
\titleformat{\section}
{\large\bfseries}
{\thesection}
{1em}
{}
\renewcommand{\theequation}{S.\arabic{equation}}
\renewcommand{\thetable}{S.\arabic{table}}
\renewcommand{\thefigure}{S.\arabic{figure}}
\renewcommand{\thetheorem}{S.\arabic{theorem}}
\noindent \Large{\textbf{Supplementary Material}}
\small

\clubpenalty=10000
\widowpenalty=10000
\displaywidowpenalty=10000
\section{Participants and procedure}\label{sec:participants-and-procedure}
\subsection{Participants}
Participants were recruited through the Magdeburg SONA Research Participation
System and social media platforms. In total, $n = 60$ participants gave their
written informed consent for participation in the study and dissemination of
their anonymized data for research purposes. Due to a technical issue, the data
of one participant were not recorded for the Gabor contrast discrimination task.
With the exception of two individuals, all participants were undergraduate
students enrolled in the Bachelor of Science program in Psychology, holding a
general qualification for university entrance as their highest attained
educational level. The remaining two participants were enrolled in a Master of
Science program and had previously completed a bachelor's degree. Participants
were aged between 18 and 28 years (average $21.55$ years, SD $= 2.57$). Twelve
participants identified as male, and 48 participants identified as female.
Participants received course credit for their participation or participated
without compensation. The study was approved by the ethics committee of the Otto
von Guericke University Magdeburg Medical Faculty (Ethik-Kommission der
Otto-von-Guericke-Universität an der Medizinischen Fakultät und am
Universitätsklinikum Magdeburg A.ö.R.; 2025/05/12, reference number 68/25).
Participants were informed that participation in the study is voluntary and may
be terminated at any time without giving any reasons.

\subsection{Procedure}
The study consisted of a one-time visit of approximately one hour to the
Institute of Psychology at Otto von Guericke University Magdeburg. After
providing their written informed consent, participants completed three tasks:
the Gabor contrast discrimination task, the symmetric bandit learning task, and
a target shooting video game. The tasks were conducted in counterbalanced order
across participants. The behavioral data from the target shooting game are
not part of the current communication. Participants were seated in front of a
computer monitor (approximately 60 cm distance, 23-inch LED monitor with a
resolution of 1920 x 1080 pixels and a refresh rate of 60 Hz). For the Gabor
contrast discrimination task, participants rested their heads on a chin rest to
maintain a constant distance from the monitor. For each task, the participants
received written task instructions in German (cf.
\fullnameref{sec:gabor-contrast-discrimination-task-instructions} and
\fullnameref{sec:symmetric-bandit-learning-task-instructions}). Both the Gabor
contrast discrimination task and the symmetric bandit task were presented in 10
blocks of 30 trials each. After each block, participants could take a break and
initiate the next block themselves. Each task took approximately 10 minutes to
complete. To familiarize themselves with the tasks before data recording,
participants completed a few training trials and had the opportunity to ask
any questions.

The experimental tasks were implemented in Python 3.8.10 using PsychoPy 2025.1.0
in the form of a specific agent of the task's ABM. For details, please refer to
the task-specific versions of \textit{abm\_experiment.py}, \textit{abm\_abm.py},
and \textit{abm\_agent\_H.py}. The Gabor patch stimuli were custom-created
.png files (cf. \fullnameref{sec:gabor-patch-stimuli}) with an 815 x 308 pixel
resolution and were presented on a gray background. For inspecting the stimuli,
please refer to the repository folder Data/Stimuli. The symmetric bandit stimuli
were created directly using PsychoPy 2025.1.0.
\clearpage 

\section{Gabor contrast discrimination}\label{sec:gabor-contrast-discrimination-supplement}
\subsection{Gabor contrast discrimination task instructions}\label{sec:gabor-contrast-discrimination-task-instructions}

\noindent\textbf{German original} In jedem Durchgang (Trial) dieser Aufgabe
sehen Sie zwei Gabor-Patches – einen links und einen rechts von einem
Fixationskreuz. Einer der beiden Patches hat einen höheren Kontrast, d.\,h. einen
größeren Unterschied zwischen schwarzen und weißen Bildanteilen. Ihre Aufgabe
ist es, durch Drücken der linken oder rechten Pfeiltaste anzugeben, welcher
Patch den höheren Kontrast hat. Bei einer korrekten Antwort erhalten Sie +1
Punkt, bei einer falschen Antwort +0 Punkte. Ihr Ziel ist es, möglichst viele
Punkte zu sammeln, indem Sie möglichst oft korrekt antworten. Manchmal ist der
Kontrastunterschied deutlich sichtbar, manchmal ist er subtiler. Nachdem Sie
sich mit der Aufgabe vertraut gemacht haben, absolvieren Sie 10 Blöcke mit
jeweils 30 Trials.

\vspace{2mm}
\noindent\textbf{English translation} In each trial of the task, you will see
two Gabor patches to the left and right of a fixation cross. One of the patches
has a higher contrast, i.e., a larger difference between the black and white
parts of the image than the other. Your task is to indicate by pressing the left
or right arrow key which of the Gabor patches has the higher
contrast. If your answer is correct, you receive a reward of +1. If your answer
is incorrect, you receive a reward of +0. Your aim is to maximize your sum of
rewards, i.e., to be correct as often as possible. Sometimes, the difference in
contrast between the Gabor patches is obvious, and sometimes it is more subtle.
After familiarizing yourself with the task, you will perform 10 blocks of 30
trials each.

\subsection{Gabor stimuli}\label{sec:gabor-patch-stimuli}
Gabor patches were created according to the two-dimensional Gabor function 
\begingroup
\small
\begin{equation}
\begin{aligned}[c]
&g : \mathbb{R}^2 \to \mathbb{R},
(x,y)\mapsto g(x,y) :=
\\
&\gamma\exp\left(-\frac{(x \cos \phi + y \sin \phi)^2 + (-x \sin \phi + y \cos \phi)^2}{2\sigma^2}\right)
\cos\left(2\pi\frac{x \cos \phi + y \sin \phi}{\lambda} + \theta\right),
\end{aligned}
\end{equation}
\endgroup
\noindent where $\gamma$ encodes the amplitude scaling, $\phi$ encodes the
stimulus orientation, $\lambda$ encodes the wavelength of the sinusoidal factor,
$\theta$ encodes the sinusoidal phase offset, and $\sigma$ is the standard
deviation parameter of the Gaussian envelope \citep[cf.][]{daugman1985}. For
each Gabor patch, the orientation and sinusoidal phase offset were set to
$\phi := 0$ and $\theta := 0$, respectively. The wavelength of the sinusoidal
factor was set to $\lambda := 15$, and the standard deviation parameter of the
Gaussian envelope was set to $\sigma := 15$. For each stimulus comprising two
Gabor patches to the left and right of the center of the screen, let
$\gamma_l, \gamma_r \in [0,1]$ denote the amplitude scalings of the left and
right Gabor patches, respectively. The contrast difference of a stimulus was
defined as
\begin{equation}
\gamma_r  - \gamma_l =: c
\end{equation}
and the mean amplitude scaling of each stimulus was constrained to
\begin{equation}
\frac{1}{2}(\gamma_l + \gamma_r) = \mu
\end{equation}
with $\mu := \frac{1}{2}$. Stimuli with $c \in [-\kappa,\kappa]$ for
$\kappa = \frac{1}{2}$ in
steps of $0.01$ were thus created by setting
\begin{equation}
\gamma_l := \mu - \frac{1}{2}c \mbox{ and } \gamma_r := c + \gamma_l 
\end{equation}
as implemented in \textit{abm\_stimuli.py}.

\clearpage
\vspace*{\fill}
\begin{figure}[!htbp]
\centering
\begin{subfigure}{\linewidth}
\centering
\includegraphics[width=\linewidth]{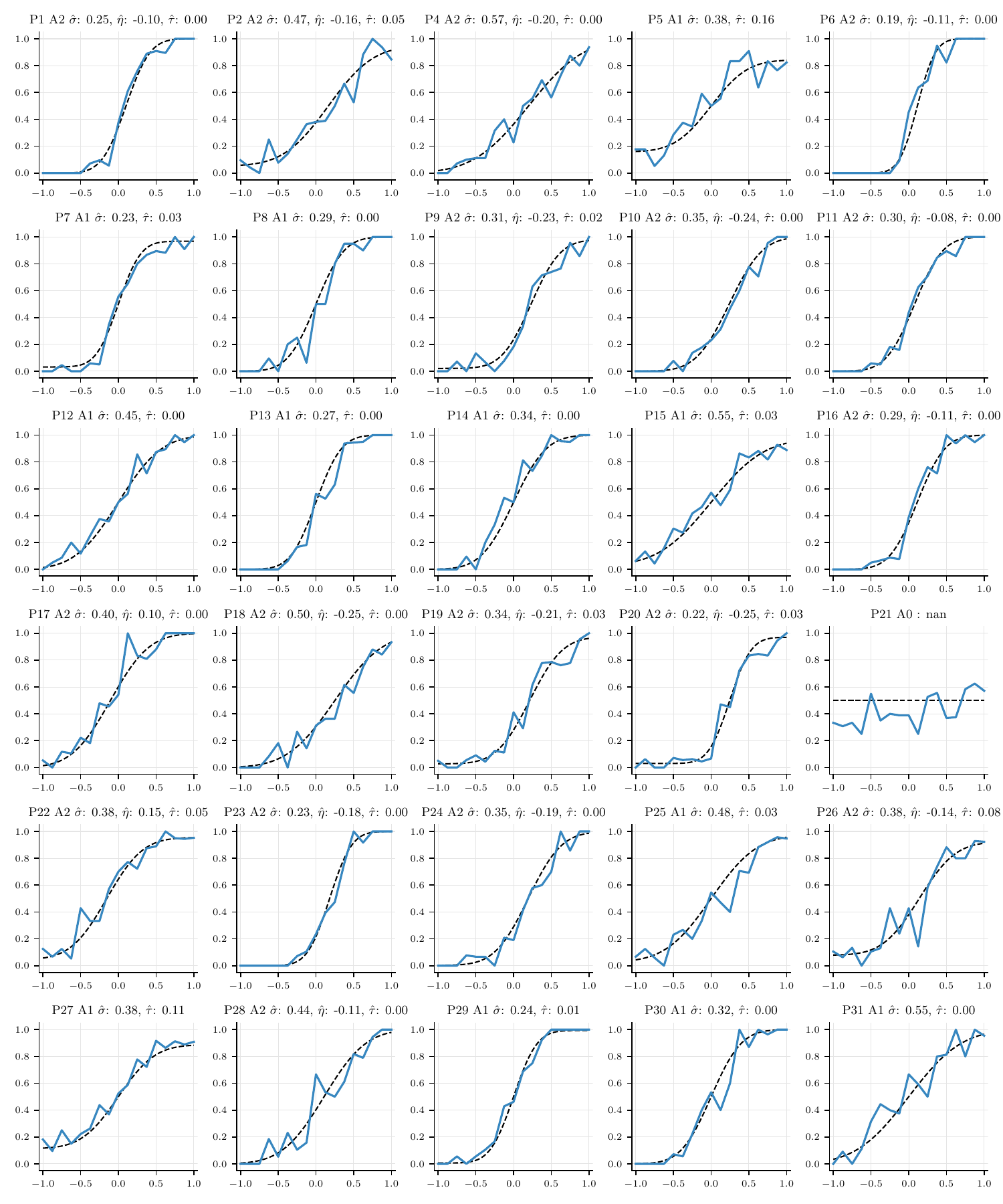}
\caption{Descriptive statistics of the Gabor contrast discrimination task and
posterior predictive psychometric functions for participants 1 to 31 ($x$-axis:
normalized contrast differences, $y$-axis: observed proportion of ``right''
actions).}
\label{fig:S1A}
\end{subfigure}
\end{figure}
\clearpage
\begin{figure}[!htbp]\ContinuedFloat
\centering
\begin{subfigure}{\linewidth}
\centering
\includegraphics[width=\linewidth]{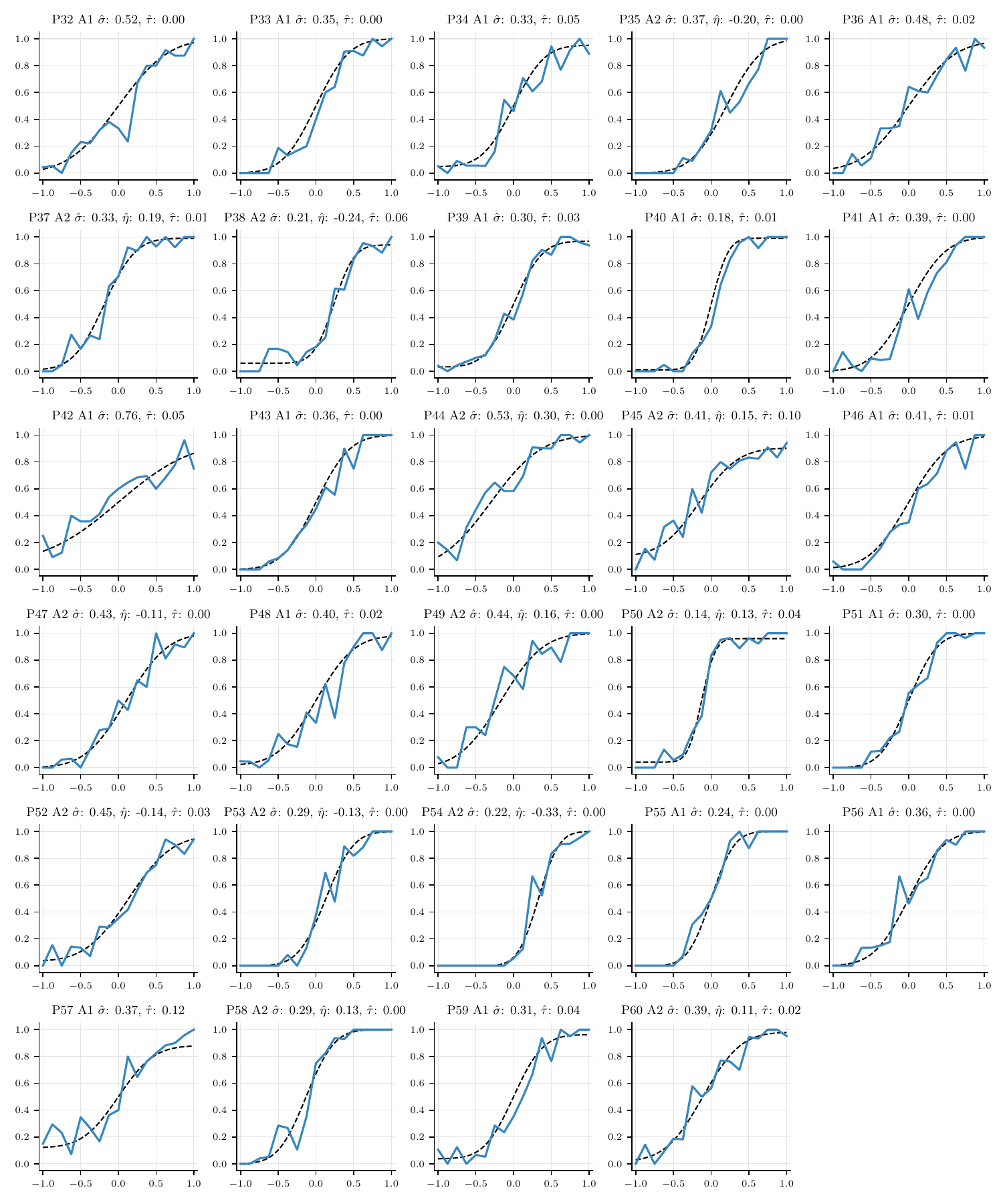}
\caption{Descriptive statistics of the Gabor contrast discrimination task and
posterior predictive psychometric functions for participants 32 to 60 ($x$-axis:
normalized contrast differences, $y$-axis: observed proportion of ``right''
actions).}
\label{fig:S1B}
\end{subfigure}
\caption{Descriptive statistics of the Gabor contrast discrimination task and
posterior predictive psychometric functions. On each panel, the blue line
depicts the proportion of ``right'' actions in 17 equally spaced contrast
difference bins with an approximate width of 0.12. The dashed black line depicts
the participant-specific posterior predictive psychometric function,
$p(a_t = 1 | c_t)$, based on the participant-specific most plausible model
and its maximum-likelihood parameter estimates.}\label{fig:S1}
\end{figure}
\clearpage

\subsection{Gabor contrast discrimination agent model analysis}\label{sec:gabor-contrast-discrimination-agent-analysis}
In the following derivations, we make frequent use of the following shorthand
notations for multiple sums and multiple integrals of multivariate real-valued
functions $f : \mathbb{R}^n \to \mathbb{R}$,
\begin{equation}
\sum_{x_{1:n}} f(x_1,...,x_n)  := \sum_{x_1} \cdots \sum_{x_n} f(x_1,...,x_n) 
\end{equation}
and
\begin{equation}
\int f(x_1,...,x_n)\, d{x_{1:n}}  := \int \cdots \int f(x_1,...,x_n) \,dx_1\cdots dx_n,
\end{equation}
respectively. In derivations, we also make frequent use of the following theorem.

\begin{theorem}[Exchangeability of integration/summation and multiplication]\label{thm:exchangeability}
For $i = 1,\ldots,n$, let $f_i: \mathbb{R} \to \mathbb{R}$ denote absolutely
integrable univariate real-valued functions. Then,
\begin{equation}\label{eq:ex-1}
\int_{\mathbb{R}^n}\left(\prod_{i=1}^{n}f_i(x_i)\right)\,dx_{1:n}
= \prod_{i=1}^{n}\left(\int_{\mathbb{R}}f_i(x_i)\,dx_i\right).
\end{equation}
Similarly, for $i = 1,\ldots,n$, let $D_i$ be a finite set and let
$f_i:D_i\to\mathbb{R}$. Then,
\begin{equation}\label{eq:ex-2}
\sum_{x_1\in D_1}\cdots\sum_{x_n\in D_n}
\left(\prod_{i=1}^{n}f_i(x_i)\right)
= \prod_{i=1}^{n}\left(\sum_{x_i\in D_i}f_i(x_i)\right).
\end{equation}
In both identities, each factor depends only on its corresponding coordinate
$x_i$, and integration or summation extends over the full Cartesian product of
the coordinate domains. The same factorization holds for mixtures of integration
and finite summation under the corresponding assumptions.
\end{theorem}

\begin{proof}
The theorem follows from the linearity properties of integrals and sums and the
distributive law. We show \eqnref{eq:ex-1} by induction. Identity
\eqref{eq:ex-2} and the corresponding identities for mixtures of integration and
finite summation follow by analogy.

\vspace{2mm}
\noindent\textit{Base case.} For $n=1$, the identity is immediate. Let $n=2$.
Then,
\begin{align}
\begin{split}
\int_{\mathbb{R}^2}\left(\prod_{i=1}^{2}f_i(x_i)\right)\,dx_{1:2}
&=\int_{\mathbb{R}}\left(\int_{\mathbb{R}}f_1(x_1)f_2(x_2)\,dx_1\right)\,dx_2\\
&=\int_{\mathbb{R}}f_2(x_2)\left(\int_{\mathbb{R}}f_1(x_1)\,dx_1\right)\,dx_2\\
&=\left(\int_{\mathbb{R}}f_1(x_1)\,dx_1\right)\left(\int_{\mathbb{R}}f_2(x_2)\,dx_2\right)\\
&=\prod_{i=1}^{2}\left(\int_{\mathbb{R}}f_i(x_i)\,dx_i\right).
\end{split}
\end{align}

\vspace{2mm}
\noindent\textit{Induction step.} Let $n\geq3$ and assume that
\eqnref{eq:ex-1} holds for $n-1$. Then,
\begin{align}
\begin{split}
&\int_{\mathbb{R}^n}\left(\prod_{i=1}^{n}f_i(x_i)\right)\,dx_{1:n}\\
&\quad=\int_{\mathbb{R}}\left(\int_{\mathbb{R}}\cdots\left(\int_{\mathbb{R}}f_1(x_1)\cdots f_{n-1}(x_{n-1})f_n(x_n)\,dx_1\right)\cdots\,dx_{n-1}\right)\,dx_n\\
&\quad=\int_{\mathbb{R}}f_n(x_n)\left(\int_{\mathbb{R}}\cdots\left(\int_{\mathbb{R}}f_1(x_1)\cdots f_{n-1}(x_{n-1})\,dx_1\right)\cdots\,dx_{n-1}\right)\,dx_n\\
&\quad=\int_{\mathbb{R}}f_n(x_n)\left(\int_{\mathbb{R}^{n-1}}\left(\prod_{i=1}^{n-1}f_i(x_i)\right)\,dx_{1:n-1}\right)\,dx_n\\
&\quad=\left(\int_{\mathbb{R}}f_n(x_n)\,dx_n\right)\left(\int_{\mathbb{R}^{n-1}}\left(\prod_{i=1}^{n-1}f_i(x_i)\right)\,dx_{1:n-1}\right)\\
&\quad=\left(\int_{\mathbb{R}}f_n(x_n)\,dx_n\right)\left(\prod_{i=1}^{n-1}\left(\int_{\mathbb{R}}f_i(x_i)\,dx_i\right)\right)\\
&\quad=\prod_{i=1}^{n}\left(\int_{\mathbb{R}}f_i(x_i)\,dx_i\right),
\end{split}
\end{align}
where we used the induction hypothesis in the penultimate equality.
\end{proof}

\begin{theorem}[Belief state]\label{thm:belief-state-1}
Let 
\begin{equation}
p(s_{1:T},c_{1:T},o_{1:T},d_{1:T},a_{1:T},r_{1:T})
= \prod_{t=1}^T p(s_t)p(c_t|s_t)p(o_t|c_t)p(d_t|o_t)p(a_t|d_t)p(r_t|s_t,a_t)
\end{equation} 
denote the ABM induced by \eqnref{eq:T-1}, \eqnref{eq:A-p-1}, and
\eqnref{eq:D-1}. Then,
\begin{equation}\label{eq:bs-1-1}
p(s_{1:T}|o_{1:T}) = \prod_{t=1}^T p(s_t|o_t)
\end{equation}
and, for $\eta = 0$,
\begin{equation}\label{eq:bs-1-2}
p(s_t|o_t)
= \frac{ \left(\Phi(0;o_t,\sigma^2) - \Phi(-\kappa;o_t,\sigma^2)\right)^{1-s_t}
         \left(\Phi(\kappa;o_t,\sigma^2)  - \Phi(0;o_t,\sigma^2)\right)^{s_t}}
        {\left(\Phi(\kappa;o_t,\sigma^2) - \Phi(-\kappa;o_t,\sigma^2)\right)}.
\end{equation}\label{eq:bs-2}
\end{theorem}

\begin{proof}
We first note that, with respect to \eqnref{eq:bs-1-1}, we have
\begin{align*}
& p(s_{1:T}|o_{1:T}) \\
& = \frac{p(s_{1:T},o_{1:T})}{p(o_{1:T})}                                                                                                                  \\
& = \frac{\int\sum_{d_{1:T}}\sum_{a_{1:T}}\sum_{r_{1:T}}p(s_{1:T},c_{1:T},o_{1:T},d_{1:T},a_{1:T},r_{1:T})\,dc_{1:T}}{p(o_{1:T})}                          \\
& = \frac{\int\sum_{d_{1:T}}\sum_{a_{1:T}}\sum_{r_{1:T}}\prod_{t=1}^T p(s_t)p(c_t|s_t)p(o_t|c_t)p(d_t|o_t)p(a_t|d_t)p(r_t|s_t,a_t)\,dc_{1:T}}{p(o_{1:T})}  \\
& = \frac{\prod_{t=1}^T\int\sum_{d_{t}}\sum_{a_{t}}\sum_{r_{t}}p(s_t)p(c_t|s_t)p(o_t|c_t)p(d_t|o_t)p(a_t|d_t)p(r_t|s_t,a_t)\,dc_{t}}{p(o_{1:T})}           \\
& = \frac{\prod_{t=1}^Tp(s_t)\int p(c_t|s_t)p(o_t|c_t)\sum_{d_{t}}p(d_t|o_t)\sum_{a_{t}}p(a_t|d_t)\sum_{r_{t}}p(r_t|s_t,a_t)\,dc_{t}}{p(o_{1:T})}           \\
& = \frac{\prod_{t=1}^Tp(s_t)\int p(c_t|s_t)p(o_t|c_t)\,dc_{t}}{p(o_{1:T})}                                                                                 \\
& = \frac{\prod_{t=1}^Tp(s_t)\int p(c_t|s_t)p(o_t|c_t)\,dc_{t}}{\sum_{s_{1:T}}p(s_{1:T},o_{1:T})}                                                           \\
& = \frac{\prod_{t=1}^Tp(s_t)\int p(c_t|s_t)p(o_t|c_t)\,dc_{t}}{\sum_{s_{1:T}}\prod_{t=1}^Tp(s_t)\int p(c_t|s_t)p(o_t|c_t)\,dc_{t}}                          \\  
& = \frac{\prod_{t=1}^Tp(s_t)\int p(c_t|s_t)p(o_t|c_t)\,dc_{t}}{\prod_{t=1}^T \sum_{s_t} p(s_t)\int p(c_t|s_t)p(o_t|c_t)\,dc_{t}}                            \\  
& = \prod_{t=1}^T \frac{p(s_t)\int p(c_t|s_t)p(o_t|c_t)\,dc_{t}}{\sum_{s_t} p(s_t)\int p(c_t|s_t)p(o_t|c_t)\,dc_{t}}                                         \\ 
& = \prod_{t=1}^T \frac{\int p(s_t)p(c_t|s_t)p(o_t|c_t)\,dc_{t}}{\sum_{s_t} \int p(s_t)p(c_t|s_t)p(o_t|c_t)\,dc_{t}}                                         \\  
& = \prod_{t=1}^T \frac{\int p(s_t,c_t)p(o_t|c_t,s_t)\,dc_{t}}{\sum_{s_t} \int p(s_t, c_t)p(o_t|c_t,s_t)\,dc_{t}}                                            \\  
& = \prod_{t=1}^T \frac{\int p(s_t,c_t,o_t)\,dc_{t}}{\sum_{s_t} \int p(s_t,c_t, o_t)\,dc_{t}}                                                                \\
& = \prod_{t=1}^T \frac{p(s_t,o_t)}{p(o_t)}                                                                                                                \\  
& = \prod_{t=1}^T p(s_t|o_t),                                                                                                                              \\  
\end{align*}
where we used \autoref{thm:exchangeability} in the fourth and ninth equalities. 
With respect to \eqnref{eq:bs-1-2}, the preceding derivation gives
\begin{equation}
p(s_t|o_t) = \frac{p(s_t,o_t)}{p(o_t)} = \frac{\int p(s_t,c_t)p(o_t|c_t,s_t)\,dc_{t}}{\sum_{s_t} \int p(s_t, c_t)p(o_t|c_t,s_t)\,dc_{t}}.
\end{equation}
For the numerator, we thus obtain with $\eta = 0$
\begin{align*} 
& p(s_t,o_t)                                                            \\
& = \int p(s_t,c_t)p(o_t|c_t,s_t)\,dc_{t}                                \\
& = \int_{-\kappa}^\kappa p(s_t)p(c_t|s_t)p(o_t|c_t)\,dc_t              \\
& = p(s_t) \int_{-\kappa}^\kappa p(c_t|s_t)p(o_t|c_t)\,dc_t             \\
& = 
\mathrm{Bern}\left(s_t;0.5\right) 
\int_{-\kappa}^\kappa 
\mathrm{Unif}\left(c_t;[-\kappa, 0]\right)^{1-s_t}
\mathrm{Unif}\left(c_t;]0,\kappa]\right)^{s_t}
\mathrm{Norm}\left(o_t;c_t,\sigma^2\right)
\,dc_t           
\\
& = 
\left(\frac{1}{2}\right)^{s_t}\left(1-\frac{1}{2}\right)^{1-s_t} 
\\
& \quad \times
\int_{-\kappa}^\kappa 
\left(\frac{1}{\kappa}\ib{c_t \in [-\kappa,0]}\right)^{1-s_t}
\left(\frac{1}{\kappa}\ib{c_t \in ]0,\kappa]}\right)^{s_t}
\left(\frac{1}{\sqrt{2\pi\sigma^2}}e^{-\frac{1}{2\sigma^2}(o_t-c_t)^2}\right)
\,dc_t  
\\
& = 
\left(\frac{1}{2}\right)^{s_t}\left(\frac{1}{2}\right)^{1-s_t} 
\\
& \quad \times
\int_{-\kappa}^\kappa 
\left(\frac{1}{\kappa}\right)^{1-s_t}
\left(\ib{c_t \in [-\kappa,0]}\right)^{1-s_t}
\left(\frac{1}{\kappa}\right)^{s_t}
\left(\ib{c_t \in ]0,\kappa]}\right)^{s_t}
\left(\frac{1}{\sqrt{2\pi\sigma^2}}e^{-\frac{1}{2\sigma^2}(c_t - o_t)^2}\right)
\,dc_t           
\\
& = 
\left(\frac{1}{2}\right)^{s_t}\left(\frac{1}{2}\right)^{1-s_t} 
\left(\frac{1}{\kappa}\right)^{1-s_t}
\left(\frac{1}{\kappa}\right)^{s_t}
\\
& \quad \times
\int_{-\kappa}^\kappa 
\left(\ib{c_t \in [-\kappa,0]}\right)^{1-s_t}
\left(\ib{c_t \in ]0,\kappa]}\right)^{s_t}
\left(\frac{1}{\sqrt{2\pi\sigma^2}}e^{-\frac{1}{2\sigma^2}(c_t - o_t)^2}\right)
\,dc_t           
\\
& = 
\frac{1}{2\kappa} 
\int_{-\kappa}^\kappa 
\left(\ib{c_t \in [-\kappa,0]}\right)^{1-s_t}
\left(\ib{c_t \in ]0,\kappa]}\right)^{s_t}
\left(\frac{1}{\sqrt{2\pi\sigma^2}}e^{-\frac{1}{2\sigma^2}(c_t-o_t)^2}\right)
\,dc_t           
\\
& = 
\frac{1}{2\kappa} 
\left(\int_{-\kappa}^0 \frac{1}{\sqrt{2\pi\sigma^2}}e^{-\frac{1}{2\sigma^2}(c_t-o_t)^2}\,dc_t\right)^{1-s_t}
\left(\int_0^\kappa    \frac{1}{\sqrt{2\pi\sigma^2}}e^{-\frac{1}{2\sigma^2}(c_t-o_t)^2}\,dc_t\right)^{s_t}
\\
& = 
\frac{1}{2\kappa} 
\left(
\int_{-\infty}^0 \frac{1}{\sqrt{2\pi\sigma^2}}e^{-\frac{1}{2\sigma^2}(c_t-o_t)^2} \,dc_t -
\int_{-\infty}^{-\kappa} \frac{1}{\sqrt{2\pi\sigma^2}}e^{-\frac{1}{2\sigma^2}(c_t-o_t)^2} \,dc_t
\right)^{1-s_t}
\\
& \quad \times
\left(\int_0^\kappa \frac{1}{\sqrt{2\pi\sigma^2}}e^{-\frac{1}{2\sigma^2}(c_t-o_t)^2}\,dc_t\right)^{s_t}
\\
& = 
\frac{1}{2\kappa} 
\left(\Phi\left(0;o_t,\sigma^2\right) - \Phi\left(-\kappa;o_t,\sigma^2\right)\right)^{1-s_t}
\\
& \quad \times
\left(
\int_{-\infty}^\kappa \frac{1}{\sqrt{2\pi\sigma^2}}e^{-\frac{1}{2\sigma^2}(c_t-o_t)^2}\,dc_t - 
\int_{-\infty}^0 \frac{1}{\sqrt{2\pi\sigma^2}}e^{-\frac{1}{2\sigma^2}(c_t-o_t)^2}\,dc_t
\right)^{s_t}
\\
& 
= 
\frac{1}{2\kappa}
\left(\Phi\left(0;o_t,\sigma^2\right)      - \Phi\left(-\kappa;o_t,\sigma^2\right)\right)^{1-s_t}
\left(\Phi\left(\kappa;o_t,\sigma^2\right) - \Phi\left(0;o_t,\sigma^2\right)\right)^{s_t},
\end{align*}
\noindent where in the eighth equality we used the fact that for both $s_t = 0$
and $s_t = 1$, the factor in
front of the integral term takes on the value $\frac{1}{2\kappa}$. For the
denominator, we obtain
\begin{align*}
p(o_t)
& = \sum_{s_t} p(s_t,o_t)                                       
\\
& = 
\frac{1}{2\kappa}
\sum_{s_t} 
\left(\Phi\left(0;o_t, \sigma^2\right)      - \Phi\left(-\kappa;o_t, \sigma^2\right)\right)^{1-s_t}
\left(\Phi\left(\kappa;o_t, \sigma^2\right) - \Phi\left(0;o_t, \sigma^2\right)\right)^{s_t}
\\
& = 
\frac{1}{2\kappa}
\left(
\Phi\left(0;o_t, \sigma^2\right)      - \Phi\left(-\kappa;o_t, \sigma^2\right) + 
\Phi\left(\kappa;o_t, \sigma^2\right) - \Phi\left(0;o_t, \sigma^2\right)
\right)
\\
& = \frac{1}{2\kappa}\left(\Phi\left(\kappa;o_t, \sigma^2\right) - \Phi\left(-\kappa;o_t, \sigma^2\right)\right).
\end{align*}
Combining the results of the numerator and the denominator, we thus obtain
\begin{equation}
p(s_t|o_t) = 
\frac{\left(\Phi\left(0;o_t, \sigma^2\right)- \Phi\left(-\kappa;o_t, \sigma^2\right)\right)^{1-s_t}
      \left(\Phi\left(\kappa;o_t, \sigma^2\right)-\Phi\left(0;o_t, \sigma^2\right)\right)^{s_t}}  
     {\Phi\left(\kappa;o_t, \sigma^2\right) - \Phi\left(-\kappa;o_t, \sigma^2\right)}. 
\qedhere
\end{equation}
\end{proof}

\begin{figure}[!htbp]
\centering
\includegraphics[width=\linewidth]{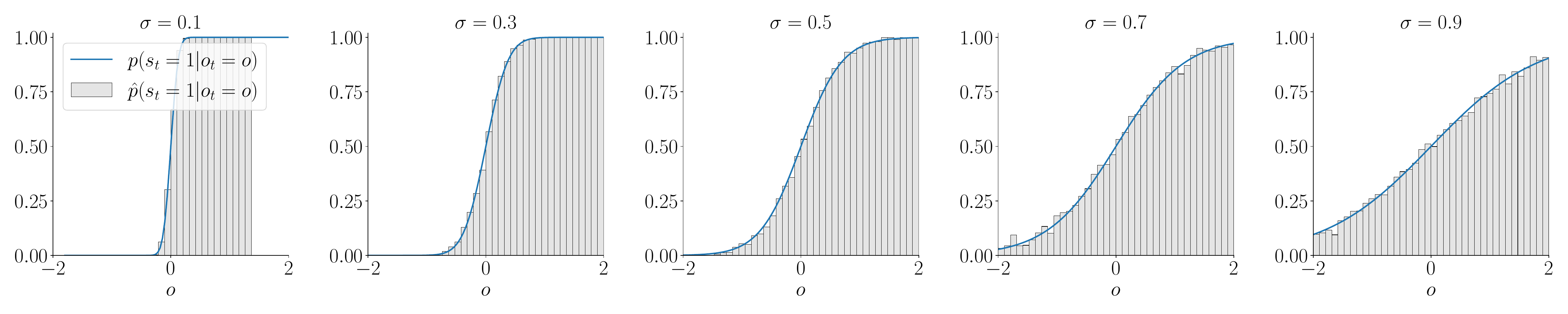}
\caption{Sampling-based validation of \fullref{thm:belief-state-1}. The figure
shows the histogram-based estimation of $p(s_t = 1|o_t = o)$ based on simulated
joint observations of $s_t$ and $o_t$ generated during model validation
($\hat{p}(s_t = 1|o_t = o)$, gray bars) and the analytical result of
\fullref{thm:belief-state-1} (blue lines). For varying levels of $\sigma$ and
$\eta := 0$, a good match is observed. For full analytical details, please refer
to \textit{abm\_figure\_S2.py}.}\label{fig:S2}
\end{figure}

\begin{theorem}[Decision optimality]\label{thm:decision-optimality-1}
Let 
\begin{equation}
p(s_{1:T}, c_{1:T},o_{1:T},d_{1:T},a_{1:T}, r_{1:T})
= \prod_{t=1}^T p(s_t)p(c_t|s_t) p(o_t|c_t)p(d_t|o_t)p(a_t|d_t)p(r_t|s_t,d_t)
\end{equation} 
denote the ABM induced by \eqnref{eq:T-1}, \eqnref{eq:A-1}, and
\eqnref{eq:D-1}
under the assumptions of no sensory bias, $\eta=0$, and no post-decision
noise, $\tau=0$, so that
$a_t=d_t$ almost surely and
\begin{equation}
p(r_t|s_t, d_t) :=\mathrm{Bern}(r_t; \ib{d_t = s_t}).
\end{equation}
Let $\varphi$ denote a \textit{Markovian policy}, i.e., a function of the form
\citep[cf.][]{krishnamurthy2016}
\begin{equation}
\varphi : O \to D, o_t \mapsto \varphi(o_t) =: d_t
\end{equation}
and let $\Pi$ denote the set of all such policies in the current scenario.
For the decision function $\delta$, ties between the two decision values are
resolved in favor of decision $0$. Then, the policy
\begin{equation}
\varphi^* := \delta \circ \phi \circ \beta
\end{equation}
maximizes the agent's expected cumulative reward, i.e.,
\begin{equation}
\varphi^* \in \argmax_{\varphi \in \Pi} \left(\mathbb{E}_{p(s_{1:T}|o_{1:T})}\left(\sum_{t = 1}^T \mathbb{E}_{p(r_t|s_t,d_t)}(r_t)\right)\right),
\end{equation}
where $d_t=\varphi(o_t)$.
\end{theorem}

\begin{proof}
We first note that
\begin{align}
\begin{split}
\mathbb{E}_{p(r_t|s_t,d_t)}(r_t)
& = 0 \cdot p(r_t = 0|s_t,d_t) + 1 \cdot p(r_t = 1|s_t,d_t)     \\
& = p(r_t = 1|s_t,d_t)                                          \\
& = \ib{d_t = s_t}^{1} \left(1 - \ib{d_t = s_t}\right)^{1-1}    \\
& = \ib{d_t = s_t}.   
\end{split}
\end{align}
We next note that
\begin{equation}
\mathbb{E}_{p(s_{1:T}|o_{1:T})}\left(\sum_{t=1}^T \ib{d_t = s_t} \right) 
= \sum_{t=1}^T \mathbb{E}_{p(s_t|o_t)} \left(\ib{d_t = s_t}\right),  \\
\end{equation}
because 
\begin{align*}
\mathbb{E}_{p(s_{1:T}|o_{1:T})}\left(\sum_{t=1}^T \ib{d_t=s_t}\right)                                
& = \sum_{s_{1:T}} p(s_{1:T}|o_{1:T})\left(\sum_{t=1}^T \ib{d_t=s_t} \right)                                            \\
& = \sum_{s_{1:T}}\sum_{t=1}^T p(s_{1:T}|o_{1:T})\ib{d_t=s_t}                                                           \\
& = \sum_{t=1}^T  \sum_{s_{1:T}}  p(s_{1:T}|o_{1:T})\ib{d_t=s_t}                                                        \\
& = \sum_{t=1}^T  \sum_{s_t}\sum_{s_{1:T}\setminus s_t} p(s_{1:T}|o_{1:T})\ib{d_t=s_t}                                  \\
& = \sum_{t=1}^T  \sum_{s_t}\ib{d_t=s_t} \sum_{s_{1:T}\setminus s_t} p(s_{1:T}|o_{1:T})                                 \\
& = \sum_{t=1}^T  \sum_{s_t}\ib{d_t=s_t} \sum_{s_{1:T}\setminus s_t} \prod_{k=1}^T p(s_k|o_k)                           \\
& = \sum_{t=1}^T  \sum_{s_t}p(s_t|o_t) \ib{d_t=s_t} \sum_{s_{1:T}\setminus s_t} \prod_{k=1, k \neq t}^T p(s_k|o_k)      \\
& = \sum_{t=1}^T  \sum_{s_t}p(s_t|o_t) \ib{d_t=s_t} \prod_{k=1, k \neq t}^T \sum_{s_k}  p(s_k|o_k)                      \\
& = \sum_{t=1}^T  \sum_{s_t}p(s_t|o_t) \ib{d_t=s_t}                                                                     \\
& = \sum_{t=1}^T  \mathbb{E}_{p(s_t|o_t)} \left(\ib{d_t=s_t}\right),                                                    \\
\end{align*}
where we used \autoref{thm:exchangeability} for the eighth equality.
Substituting $\varphi(o_t)$ for $d_t$, we thus consider the maximization problem
\begin{equation}
\max_{\varphi \in \Pi} \left(\sum_{t=1}^T \mathbb{E}_{p(s_t|o_t)} \left(\ib{\varphi(o_t)=s_t}\right) \right).
\end{equation}
Given the additive separability of the objective function, this maximization problem 
is equivalent to the $T$ maximization problems \cite[cf.][Section 7.4]{nocedal2006}
\begin{equation}
\max_{\varphi \in \Pi}  \left(\mathbb{E}_{p(s_t|o_t)} \left(\ib{\varphi(o_t)=s_t}\right)\right) \mbox{ for } t = 1,...,T.
\end{equation}
We are thus concerned with identifying a function $\varphi$ for which
$\varphi(o_t) \in \{0,1\}$ and
which maximizes the posterior expected value of $\ib{\varphi(o_t)=s_t}$. To this
end,
we next note that
\begin{equation}
\mathbb{E}_{p(s_t|o_t)} \left(\ib{\varphi(o_t)=s_t}\right) = p(s_t = \varphi(o_t)|o_t).
\end{equation}\label{eq:st-varphi}
This holds because
\begin{equation}
\mathbb{E}_{p(s_t|o_t)} \left(\ib{\varphi(o_t)=s_t}\right)
= p(s_t = 0|o_t)\ib{\varphi(o_t) = 0} + p(s_t = 1|o_t)\ib{\varphi(o_t) = 1}
\end{equation}
and thus for $\varphi(o_t) = 0$, we have
\begin{align}
\begin{split}
\mathbb{E}_{p(s_t|o_t)} \left(\ib{\varphi(o_t)=s_t}\right)   
& =  p(s_t = 0|o_t)\ib{0 = 0} + p(s_t = 1|o_t)\ib{0 = 1}    \\
& =  p(s_t = 0|o_t)\cdot 1 + p(s_t = 1|o_t)\cdot 0          \\
& =  p(s_t = 0|o_t)                                         \\
& =  p(s_t = \varphi(o_t)|o_t) 
\end{split}
\end{align}
and  for $\varphi(o_t) = 1$, we have
\begin{align}
\begin{split}
\mathbb{E}_{p(s_t|o_t)} \left(\ib{\varphi(o_t)=s_t}\right)   
& =  p(s_t = 0|o_t)\ib{1 = 0} + p(s_t = 1|o_t)\ib{1 = 1}    \\
& =  p(s_t = 0|o_t)\cdot 0 + p(s_t = 1|o_t)\cdot 1          \\
& =  p(s_t = 1|o_t)                                         \\
& =  p(s_t = \varphi(o_t)|o_t). 
\end{split}
\end{align}
We are thus led to consider
\begin{equation}
\max_{\varphi(o_t) \in \{0,1\}}  p(s_t = \varphi(o_t)|o_t)  
= \max \left\lbrace p(s_t = 0|o_t), p(s_t = 1|o_t)\right\rbrace  
\mbox{ for } t = 1,...,T.
\end{equation}
Resolving ties in favor of decision $0$, the maximizing decision is thus
$\varphi(o_t)=0$ if $p(s_t=0|o_t)\geq p(s_t=1|o_t)$ and $\varphi(o_t)=1$ if
$p(s_t=0|o_t)<p(s_t=1|o_t)$. We have thus found
that the policy
\begin{equation}
\varphi : O \to D, o_t \mapsto \varphi(o_t) :=
\begin{cases}
0 & p(s_t = 0|o_t) \geq p(s_t = 1|o_t)  \\
1 & p(s_t = 0|o_t)<p(s_t = 1|o_t) \\
\end{cases}
\end{equation}
maximizes the cumulative expected reward. But this is exactly the functional
form of $\delta \circ \phi \circ \beta$,
because on each trial $t=1,\ldots,T$, it selects a decision $d\in\{0,1\}$ that
maximizes the posterior probability $b_t^d$. Thus,
\begin{equation}
\delta \circ \phi \circ \beta = \varphi^* \in    
 \argmax_{\varphi \in \Pi} \left(\mathbb{E}_{p(s_{1:T}|o_{1:T})}\left(\sum_{t = 1}^T \mathbb{E}_{p(r_t|s_t,d_t)}(r_t)\right)\right).
\end{equation}
\end{proof}

\begin{theorem}[Conditional decision distribution]\label{thm:conditional-decision-distribution-1}
Given the agent model $\mathcal{A}$ defined in \eqnref{eq:A-1}, using the
unbiased belief-state mapping irrespective of $\eta$ and resolving
decision-value ties in favor of decision $0$, the distribution of an agent's
decision
\begin{equation}
d_t := (\delta \circ \phi \circ \beta)(o_t)
\end{equation}
given $o_t$ can be written as
\begin{equation}
p(d_t|o_t)=\ib{o_t\leq0}^{1-d_t}\ib{o_t>0}^{d_t}.
\end{equation}
\end{theorem}

\begin{proof}
We first note that because $d_t$ is a function of $o_t$, $p(d_t|o_t)$ is
degenerate, taking on the value $1$ if
$d_t=(\delta\circ\phi\circ\beta)(o_t)$ and $0$ otherwise.
Next, because $d_t$ takes values in $D=\{0,1\}$, it suffices to consider
$p(d_t=0|o_t)$, as $p(d_t=1|o_t)=1-p(d_t=0|o_t)$ then follows immediately.
We hence assume that $d_t=0$ and show that this is equivalent to $o_t\leq0$:
\begin{align}
\begin{split}
d_t&=0\\
\Longleftrightarrow\quad v_t^0&\geq v_t^1\\
\Longleftrightarrow\quad v_t^0&=1\\
\Longleftrightarrow\quad b_t^0&\in[0.5,1]\\
\Longleftrightarrow\quad p(s_t=1|o_t)&\leq p(s_t=0|o_t)\\
\Longleftrightarrow\quad
\frac{\Phi(0;o_t,\sigma^2)-\Phi(-\kappa;o_t,\sigma^2)}
{\Phi(\kappa;o_t,\sigma^2)-\Phi(-\kappa;o_t,\sigma^2)}
&\geq
\frac{\Phi(\kappa;o_t,\sigma^2)-\Phi(0;o_t,\sigma^2)}
{\Phi(\kappa;o_t,\sigma^2)-\Phi(-\kappa;o_t,\sigma^2)}\\
\Longleftrightarrow\quad
\Phi(0;o_t,\sigma^2)-\Phi(-\kappa;o_t,\sigma^2)
&\geq\Phi(\kappa;o_t,\sigma^2)-\Phi(0;o_t,\sigma^2)\\
\Longleftrightarrow\quad
\int_{-\kappa}^{0}\mathrm{Norm}(c;o_t,\sigma^2)\,dc
&\geq\int_{0}^{\kappa}\mathrm{Norm}(c;o_t,\sigma^2)\,dc\\
\Longleftrightarrow\quad
\int_{0}^{\kappa}\mathrm{Norm}(-c;o_t,\sigma^2)\,dc
&\geq\int_{0}^{\kappa}\mathrm{Norm}(c;o_t,\sigma^2)\,dc\\
\Longleftrightarrow\quad
\mathrm{Norm}(-c;o_t,\sigma^2)
&\geq\mathrm{Norm}(c;o_t,\sigma^2)
\quad\text{for all }c\in]0,\kappa[\\
\Longleftrightarrow\quad
(-c-o_t)^2&\leq(c-o_t)^2
\quad\text{for all }c\in]0,\kappa[\\
\Longleftrightarrow\quad
4co_t&\leq0\quad\text{for all }c\in]0,\kappa[\\
\Longleftrightarrow\quad o_t&\leq0.
\end{split}
\end{align}
Here, the posterior probabilities refer to the agent's unbiased belief-state
mapping. The integral-to-pointwise equivalence holds because the difference
between the two Gaussian densities has the same sign throughout $]0,\kappa[$,
as established by the subsequent steps.

Thus, $p(d_t=0|o_t)=1$ if $o_t\leq0$ and $p(d_t=0|o_t)=0$ if $o_t>0$.
Consequently, the conditional PMF of $d_t$ given $o_t$ is
\begin{equation}
p(d_t|o_t)=\ib{o_t\leq0}^{1-d_t}\ib{o_t>0}^{d_t}.
\end{equation}
\end{proof}

\begin{figure}[!htbp]
\centering
\includegraphics[width=\linewidth]{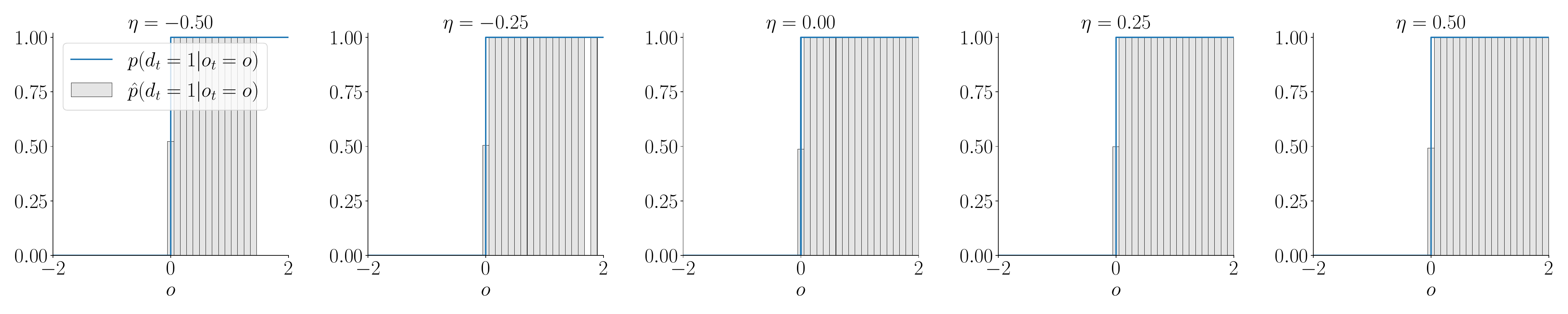}
\caption{Sampling-based validation of
\fullref{thm:conditional-decision-distribution-1}. The figure shows the
histogram-based estimation of $p(d_t = 1|o_t = o)$ based on simulated joint
observations of $d_t$ and $o_t$ as generated during model validation
($\hat{p}(d_t = 1|o_t = o)$, gray bars) and the analytical result of
\fullref{thm:conditional-decision-distribution-1} (blue lines). For varying
levels of $\eta$, a good match is observed. For full analytical details, please
refer to \textit{abm\_figure\_S3.py}.}\label{fig:S3}
\end{figure}

\begin{theorem}[Conditional log likelihood function]\label{thm:conditional-log-likelihood-function-1}
Let 
\begin{multline}
p(s_{1:T}, c_{1:T}, o_{1:T},d_{1:T},a_{1:T}, r_{1:T})   
\\ 
= \prod_{t=1}^T p(s_t)p(c_t|s_t)p(o_t|c_t)p(d_t|o_t)p(a_t|d_t)p(r_t|a_t,s_t) 
\end{multline} 
denote the ABM induced by \eqnref{eq:T-1}, \eqnref{eq:A-p-1}, and \eqnref{eq:D-1}.
Then,
\begin{equation}\label{eq:cll-1-1}
p(a_{1:T}|s_{1:T},c_{1:T})=\prod_{t=1}^T p(a_t|c_t),
\end{equation}
where
\begingroup
\small
\begin{equation}\label{eq:cll-1-2}
\fitproofline{p(a_t|c_t)=\mathrm{Bern}(a_t;\tau)\Phi(0;c_t+\eta,\sigma^2)+\mathrm{Bern}(a_t;1-\tau)\left(1-\Phi(0;c_t+\eta,\sigma^2)\right).}
\end{equation}
\endgroup
The corresponding conditional log likelihood is
\begin{equation}
\ell(\sigma,\eta,\tau)=\sum_{t=1}^T\ln p(a_t|c_t).
\end{equation}
\end{theorem}

\begin{proof}
With respect to \eqnref{eq:cll-1-1}, we have
\begingroup
\small
\begin{align*}
&p(a_{1:T}|s_{1:T},c_{1:T})\\
&=\fitproofline{\frac{p(s_{1:T},c_{1:T},a_{1:T})}{p(s_{1:T},c_{1:T})}}\\
&=\fitproofline{\frac{\int\sum_{d_{1:T}}\sum_{r_{1:T}} p(s_{1:T},c_{1:T},o_{1:T},d_{1:T},a_{1:T},r_{1:T})\,do_{1:T}}{\sum_{a_{1:T}}\int\sum_{d_{1:T}}\sum_{r_{1:T}} p(s_{1:T},c_{1:T},o_{1:T},d_{1:T},a_{1:T},r_{1:T})\,do_{1:T}}}\\
&=\fitproofline{\frac{\int\sum_{d_{1:T}}\sum_{r_{1:T}} \prod_{t=1}^T p(s_t)p(c_t|s_t)p(o_t|c_t)p(d_t|o_t)p(a_t|d_t)p(r_t|a_t,s_t)\,do_{1:T}}{\sum_{a_{1:T}}\int\sum_{d_{1:T}}\sum_{r_{1:T}} \prod_{t=1}^T p(s_t)p(c_t|s_t)p(o_t|c_t)p(d_t|o_t)p(a_t|d_t)p(r_t|a_t,s_t)\,do_{1:T}}}\\
&=\fitproofline{\frac{\prod_{t=1}^T \int\sum_{d_t}\sum_{r_t} p(s_t)p(c_t|s_t)p(o_t|c_t)p(d_t|o_t)p(a_t|d_t)p(r_t|a_t,s_t)\,do_t}{\prod_{t=1}^T \sum_{a_t}\int\sum_{d_t}\sum_{r_t} p(s_t)p(c_t|s_t)p(o_t|c_t)p(d_t|o_t)p(a_t|d_t)p(r_t|a_t,s_t)\,do_t}}\\
&=\fitproofline{\prod_{t=1}^T \frac{\int\sum_{d_t}p(s_t)p(c_t|s_t)p(o_t|c_t)p(d_t|o_t)p(a_t|d_t)\left(\sum_{r_t}p(r_t|a_t,s_t)\right)\,do_t}{\sum_{a_t}\int\sum_{d_t}p(s_t)p(c_t|s_t)p(o_t|c_t)p(d_t|o_t)p(a_t|d_t)\left(\sum_{r_t}p(r_t|a_t,s_t)\right)\,do_t}}\\
&=\fitproofline{\prod_{t=1}^T \frac{p(s_t)p(c_t|s_t)\int\sum_{d_t}p(o_t|c_t)p(d_t|o_t)p(a_t|d_t)\,do_t}{p(s_t)p(c_t|s_t)\sum_{a_t}\int\sum_{d_t}p(o_t|c_t)p(d_t|o_t)p(a_t|d_t)\,do_t}}\\
&=\fitproofline{\prod_{t=1}^T \frac{\int\sum_{d_t}p(o_t|c_t)p(d_t|o_t)p(a_t|d_t)\,do_t}{\sum_{a_t}\int\sum_{d_t}p(o_t|c_t)p(d_t|o_t)p(a_t|d_t)\,do_t}}\\
&=\fitproofline{\prod_{t=1}^T \frac{\int\sum_{d_t}p(o_t|c_t)p(a_t,d_t|o_t)\,do_t}{\sum_{a_t}\int\sum_{d_t}p(o_t|c_t)p(a_t,d_t|o_t)\,do_t}}\\
&=\fitproofline{\prod_{t=1}^T \frac{\int\sum_{d_t}p(a_t,d_t,o_t|c_t)\,do_t}{\sum_{a_t}\int\sum_{d_t}p(a_t,d_t,o_t|c_t)\,do_t}}\\
&=\fitproofline{\prod_{t=1}^T \frac{p(a_t|c_t)}{\sum_{a_t}p(a_t|c_t)}}\\
&=\fitproofline{\prod_{t=1}^T p(a_t|c_t).}
\end{align*}
\endgroup
Here, the fourth equality follows from multiple applications of
\autoref{thm:exchangeability}. The sixth equality uses
$\sum_{r_t}p(r_t|a_t,s_t)=1$. The eighth and ninth equalities follow from the
chain rule of probability and the conditional independence properties of the
model. The final equality uses $\sum_{a_t}p(a_t|c_t)=1$.
\vspace{2mm}

\noindent With respect to \eqnref{eq:cll-1-2}, we have
\begingroup
\small
\begin{align*}
& p(a_t| c_t) \\
&=\fitproofline{\frac{p(c_t,a_t)}{p(c_t)}} \\
&=\fitproofline{\frac{\sum_{s_t}\int \sum_{d_t}\sum_{r_t} p(s_t,c_t,o_t,d_t,a_t,r_t)\,do_t } {\sum_{s_t}\int \sum_{d_t}\sum_{r_t}\sum_{a_t} p(s_t,c_t,o_t,d_t,a_t,r_t)\,do_t }} \\
&=\fitproofline{\frac{\sum_{s_t}\int \sum_{d_t}\sum_{r_t} p(s_t)p(c_t|s_t)p(o_t|c_t)p(d_t|o_t)p(a_t|d_t)p(r_t|a_t,s_t)\,do_t } {\sum_{s_t}\int \sum_{d_t}\sum_{r_t}\sum_{a_t} p(s_t)p(c_t|s_t)p(o_t|c_t)p(d_t|o_t)p(a_t|d_t)p(r_t|a_t,s_t)\,do_t }} \\
&=\fitproofline{\frac{\sum_{s_t}\int \sum_{d_t} p(s_t)p(c_t|s_t)p(o_t|c_t)p(d_t|o_t)p(a_t|d_t)\sum_{r_t} p(r_t|a_t,s_t)\,do_t } {\sum_{s_t}\int \sum_{d_t}\sum_{a_t} p(s_t)p(c_t|s_t)p(o_t|c_t)p(d_t|o_t)p(a_t|d_t)\sum_{r_t} p(r_t|a_t,s_t)\,do_t }} \\
&=\fitproofline{\frac{\sum_{s_t}\int \sum_{d_t} p(s_t)p(c_t|s_t)p(o_t|c_t)p(d_t|o_t)p(a_t|d_t) \,do_t } {\sum_{s_t}\int \sum_{d_t}\sum_{a_t} p(s_t)p(c_t|s_t)p(o_t|c_t)p(d_t|o_t)p(a_t|d_t) \,do_t }} \\
&=\fitproofline{\frac{\sum_{s_t}p(s_t)p(c_t|s_t)\int \sum_{d_t} p(o_t|c_t)p(d_t|o_t)p(a_t|d_t) \,do_t } {\sum_{s_t}p(s_t)p(c_t|s_t)\int \sum_{d_t} p(o_t|c_t)p(d_t|o_t)\sum_{a_t} p(a_t|d_t) \,do_t }} \\
&=\fitproofline{\frac{\sum_{s_t}p(s_t,c_t) \int \sum_{d_t} p(o_t|c_t)p(d_t|o_t)p(a_t|d_t) \,do_t } {\sum_{s_t}p(s_t,c_t) \int \sum_{d_t}p(o_t,d_t|c_t)\,do_t }} \\
&=\fitproofline{\frac{p(c_t) \int \sum_{d_t} p(o_t|c_t)p(d_t|o_t)p(a_t|d_t) \,do_t } {p(c_t)}} \\
&=\fitproofline{\int \sum_{d_t} p(o_t|c_t)p(d_t|o_t)p(a_t|d_t)\,do_t} \\
&=\fitproofline{\int\sum_{d_t}\mathrm{Norm}(o_t;c_t+\eta,\sigma^2)\ib{o_t\le0}^{1-d_t}\ib{o_t>0}^{d_t}\mathrm{Bern}(a_t;\tau)^{1-d_t}\mathrm{Bern}(a_t;1-\tau)^{d_t}\,do_t} \\
&=\fitproofline{\int\mathrm{Norm}(o_t;c_t+\eta,\sigma^2)\ib{o_t\le0}\mathrm{Bern}(a_t;\tau)\,do_t+\int\mathrm{Norm}(o_t;c_t+\eta,\sigma^2)\ib{o_t>0}\mathrm{Bern}(a_t;1-\tau)\,do_t} \\
&=\fitproofline{\mathrm{Bern}(a_t;\tau)\int\mathrm{Norm}(o_t;c_t+\eta,\sigma^2)\ib{o_t\le0}\,do_t+\mathrm{Bern}(a_t;1-\tau)\int\mathrm{Norm}(o_t;c_t+\eta,\sigma^2)\ib{o_t>0}\,do_t} \\
&=\fitproofline{\mathrm{Bern}(a_t;\tau)\int_{-\infty}^0 \mathrm{Norm}(o_t;c_t+\eta,\sigma^2)\,do_t + \mathrm{Bern}(a_t;1-\tau) \int_{0}^\infty \mathrm{Norm}(o_t;c_t +\eta,\sigma^2)\,do_t} \\
&=\fitproofline{\mathrm{Bern}(a_t;\tau)\int_{-\infty}^0 \mathrm{Norm}(o_t;c_t+\eta,\sigma^2)\,do_t + \mathrm{Bern}(a_t;1-\tau)\left(1 - \int_{-\infty}^0 \mathrm{Norm}(o_t;c_t+\eta,\sigma^2)\,do_t\right)} \\
&=\fitproofline{\mathrm{Bern}(a_t;\tau)\Phi(0;c_t+\eta,\sigma^2) + \mathrm{Bern}(a_t;1-\tau)(1-\Phi(0;c_t+\eta,\sigma^2)).}
\end{align*}
\endgroup
The log likelihood expression follows by taking the logarithm of
\eqnref{eq:cll-1-1}.
\end{proof}

\noindent

\begin{figure}[!htbp]
\centering
\includegraphics[width=1\linewidth]{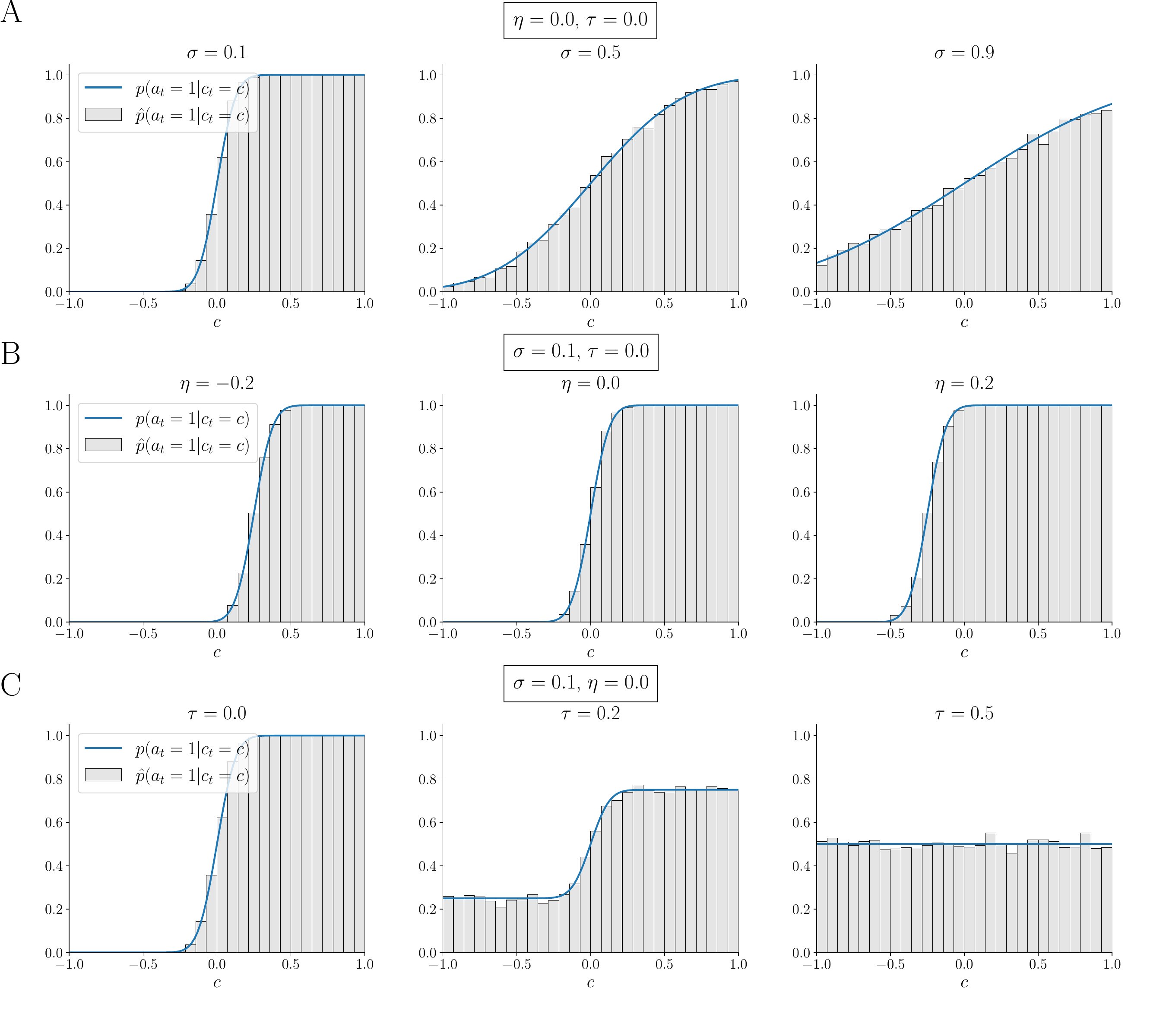}
\caption{Sampling-based validation of \eqnref{eq:cll-1-2} in
\fullref{thm:conditional-log-likelihood-function-1}. The figure depicts the
histogram-based estimation of $p(a_t = 1|c_t = c)$ based on simulated joint
observations of $a_t$ and $c_t$ as generated during model validation
($\hat{p}(a_t = 1|c_t = c)$, gray bars) and the analytical result of
\eqnref{eq:cll-1-2} of \fullref{thm:conditional-log-likelihood-function-1} (blue
lines). Each row of the figure shows the effect of varying one of the three
parameters $\sigma, \eta, \tau$ on $p(a_t = 1|c_t = c)$ while keeping the other
two constant as indicated in the heading of each row. The fixed parameters are
$\eta=0.0$ and $\tau=0.0$ in row A, $\sigma=0.1$ and $\tau=0.0$ in row B,
and $\sigma=0.1$ and $\eta=0.0$ in row C.
For full analytical details, please refer to
\textit{abm\_figure\_S4.py}.}\label{fig:S4}
\end{figure}

\clearpage
\begin{figure}[p]
\centering
\includegraphics[width=\linewidth,height=0.55\textheight,keepaspectratio]{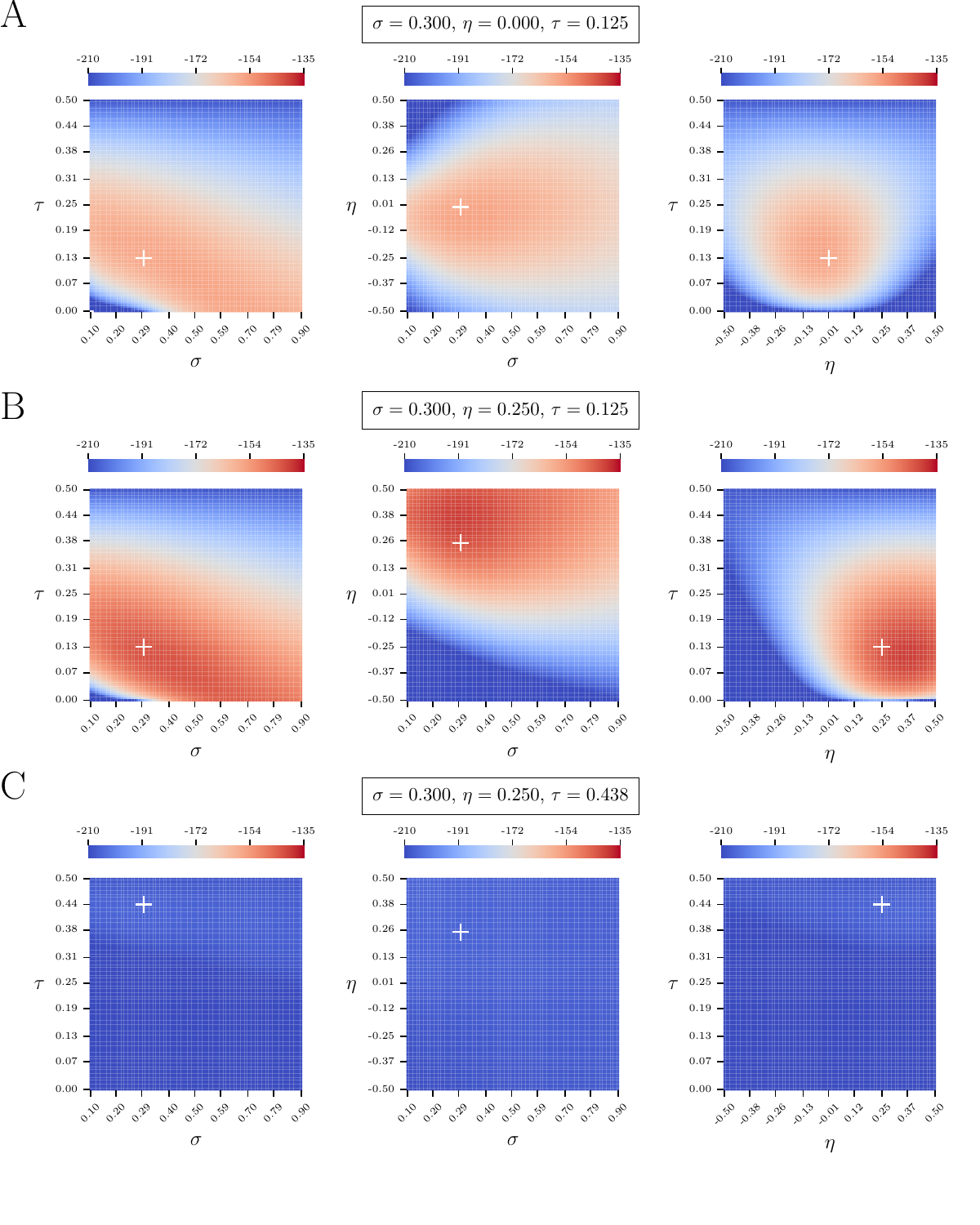}
\caption{Exemplary profile
conditional log likelihood functions of the form of \eqnref{eq:llh-1} and based
on \fullref{thm:conditional-log-likelihood-function-1}. These profile functions
are derived from the joint observations of normalized contrast differences and
agent actions of a single simulated participant dataset (300 trials), as these
functions form the basis for parameter estimation and marginal likelihood
approximation (cf. \fullnameref{sec:parameter-estimation-and-model-comparison}).
The true, but unknown, parameter values for each simulated dataset are given
above each subpanel row and depicted in each parameter subspace as white crosses
(A: $\sigma = 0.300$, $\eta = 0.000$, $\tau = 0.125$, B: $\sigma = 0.300$,
$\eta = 0.250$, $\tau = 0.125$, C: $\sigma = 0.300$, $\eta = 0.250$,
$\tau = 0.438$). The color coding is identical across all subpanels. In the low
post-decision noise scenarios of the first two subpanel rows ($\tau = 0.125$),
the profile functions exhibit adequate curvature for the numerical maximization
of the conditional log likelihood in each profile dimension and the conditional
log likelihood maxima are located in the vicinity of the true but unknown
values, indicating that the simulated dataset is indicative of its generative
parameters. In the high post-decision noise scenario of the lower subpanel row
($\tau = 0.438$), the curvature virtually disappears, essentially removing
indicative information about the generative parameters from the profile function
and hampering the successful numerical identification of the true, but unknown,
generative model parameters. For full analytical details, please refer to
\textit{abm\_figure\_S5.py}.}\label{fig:S5}
\end{figure}

\clearpage
\begin{figure}[p]
\centering
\includegraphics[width=\linewidth,height=0.55\textheight,keepaspectratio]{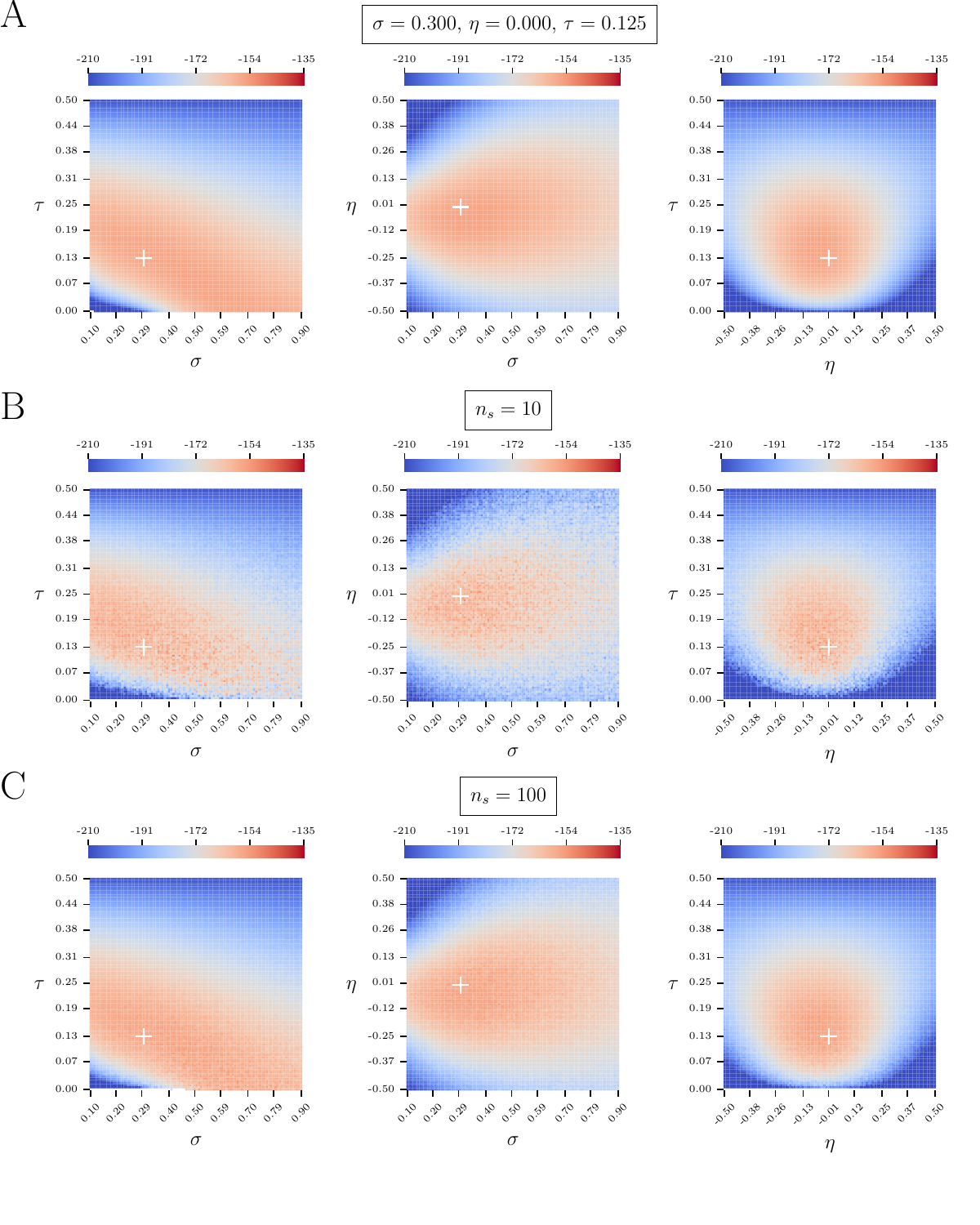}
\caption{Analytic and agentic
profile log likelihood examples. As in \fullref{fig:S5}, the profile functions
are derived from the joint observations of normalized contrast differences and
agent actions of a single simulated participant dataset (300 trials) with true,
but unknown, parameter values of $\sigma = 0.300$, $\eta = 0.000$,
and $\tau = 0.125$ as indicated by the white crosses in each parameter subspace.
Panel A visualizes analytic conditional log likelihood profiles identical to
those depicted in \fullref{fig:S5}A. Panels B and C visualize agentic
conditional log likelihood profiles for $n_s = 10$ and $n_s = 100$ samples of
the internal agent observation $o_t$, respectively, for evaluating the Monte
Carlo estimators \eqref{eq:monte-carlo-1} and \eqref{eq:monte-carlo-2}. Notably,
use of only a few samples induces non-functional variability in the conditional
log likelihood profile, which can hamper the identification of maxima. For full analytical details, please refer to
\textit{abm\_likelihoods.py} and \textit{abm\_figure\_S6.py}.}\label{fig:S6}
\end{figure}

\clearpage
\section{Symmetric bandit learning}\label{sec:symmetric-bandit-learning-supplement}
\subsection{Symmetric bandit learning task instructions}\label{sec:symmetric-bandit-learning-task-instructions}

\noindent\textbf{German original} In jedem Durchgang (Trial) dieser Aufgabe
können Sie wählen, ob Sie die linke oder die rechte Pfeiltaste drücken.
Innerhalb eines Blocks ist jeder Taste eine feste Wahrscheinlichkeit zugeordnet,
mit der das Drücken dieser Taste zu einer Belohnung von +1 oder +0 Punkten
führt.
Beispiel: In einem Block führt das Drücken der linken Pfeiltaste mit einer
Wahrscheinlichkeit von 0,8 zu +1 Punkt und mit 0,2 zu +0 Punkten. Das bedeutet:
Im Durchschnitt erhalten Sie bei 10 Durchgängen mit der linken Taste etwa 8-mal
+1 Punkt und 2-mal +0 Punkte. Die rechte Taste ist in diesem Fall mit den
gegenläufigen Wahrscheinlichkeiten verknüpft (also +1 mit 0,2 und +0 mit 0,8).
In jedem Block ist eine der beiden Tasten vorteilhafter, d.\,h. mit einer höheren
Wahrscheinlichkeit für +1 verknüpft – manchmal ist dieser Unterschied
offensichtlich, manchmal subtiler. Ihr Ziel ist es, die Taste mit der höheren
+1-Wahrscheinlichkeit so oft wie möglich zu drücken. Welche Taste das ist,
müssen Sie in jedem Block neu herausfinden, da die Wahrscheinlichkeiten mit
Beginn eines neuen Blocks neu vergeben werden. Nachdem Sie sich mit der Aufgabe
vertraut gemacht haben, absolvieren Sie 10 Blöcke mit jeweils 30 Trials.

\vspace{2mm}
\noindent\textbf{English translation} In each task trial, you can choose between
pressing the left or right arrow key. For a given block of the task, each
button is associated with fixed
probabilities that pressing it will yield a reward of +1 or +0. For example, in
a given block, pressing the left arrow key will yield a reward of +1 with
a probability of 0.8 and a reward of +0 with a probability of 0.2. When
pressing the left button, you will thus observe a +1 reward on average in 8 out of
10 trials and a reward of +0 in 2 out of 10 trials. Pressing the other button
(in this example, the right button) will yield the rewards with opposite
probabilities
assigned to them. In the example, if pressing the left button yields a reward of
+1 with
a probability of 0.8, then pressing the right button yields a reward of +1 with a
probability of 0.2,
and likewise for obtaining rewards of +0. Thus, one of the buttons is always
associated with a higher probability of obtaining a reward of +1 than the
other button. Sometimes this difference is pronounced. Sometimes it is more subtle.
Your aim is therefore to press the button with the higher associated
probability for +1 as
often as possible. In each block, you need to determine which button it is, as
the probabilities
will be reset with the start of a new block. After familiarizing yourself with
the task, you will
perform 10 blocks of 30 trials each.

\clearpage
\vspace*{\fill}
\begin{figure}[!htbp]
\centering
\begin{subfigure}{\linewidth}
\centering
\includegraphics[width=\linewidth]{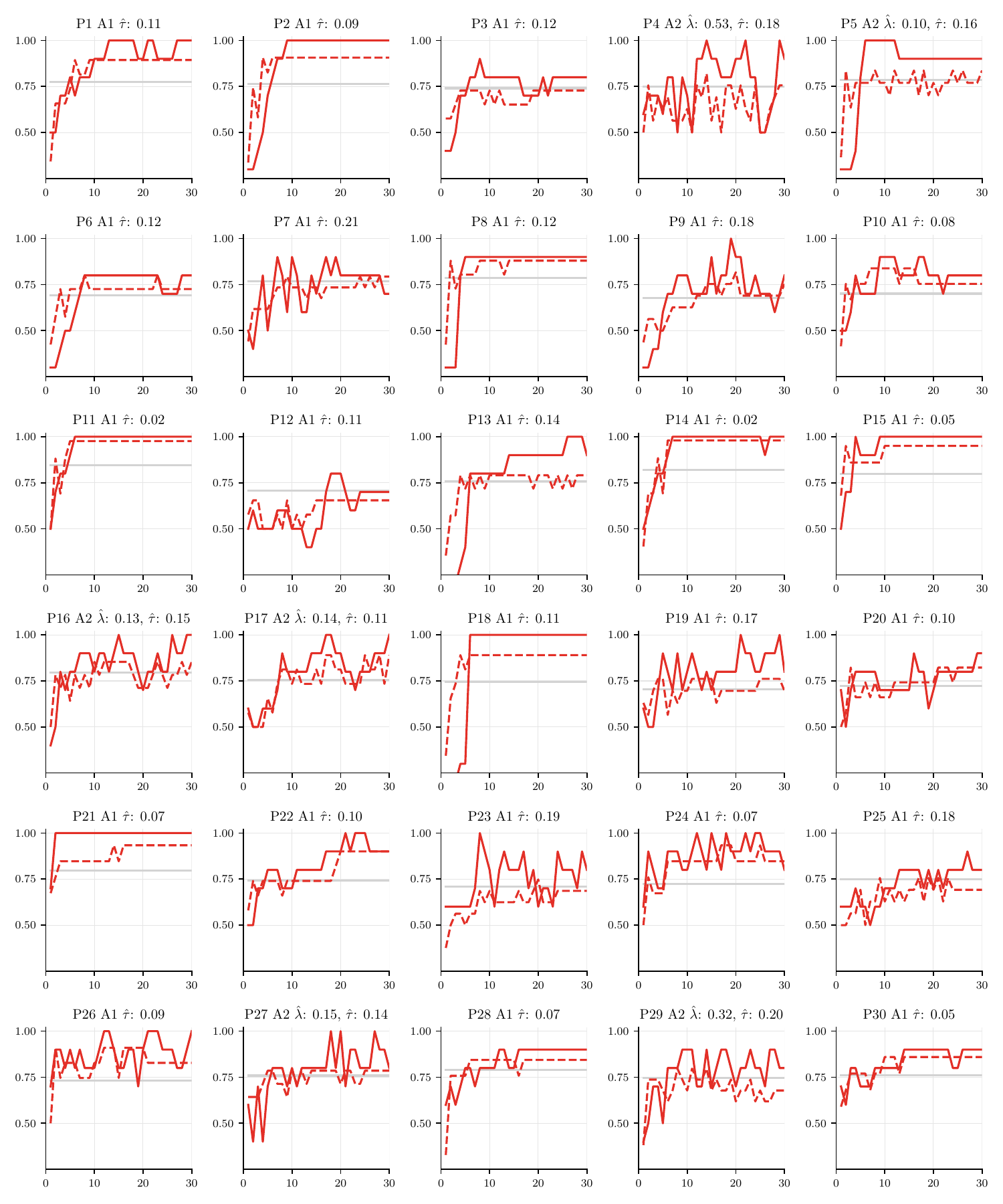}
\caption{Symmetric bandit task descriptive statistics and posterior predictive
learning functions for participants 1 to 30 ($x$-axis: trial, $y$-axis: observed
proportion of maximizing actions).}
\label{fig:S7A}
\end{subfigure}
\end{figure}
\clearpage
\begin{figure}[!htbp]\ContinuedFloat
\centering
\begin{subfigure}{\linewidth}
\centering
\includegraphics[width=\linewidth]{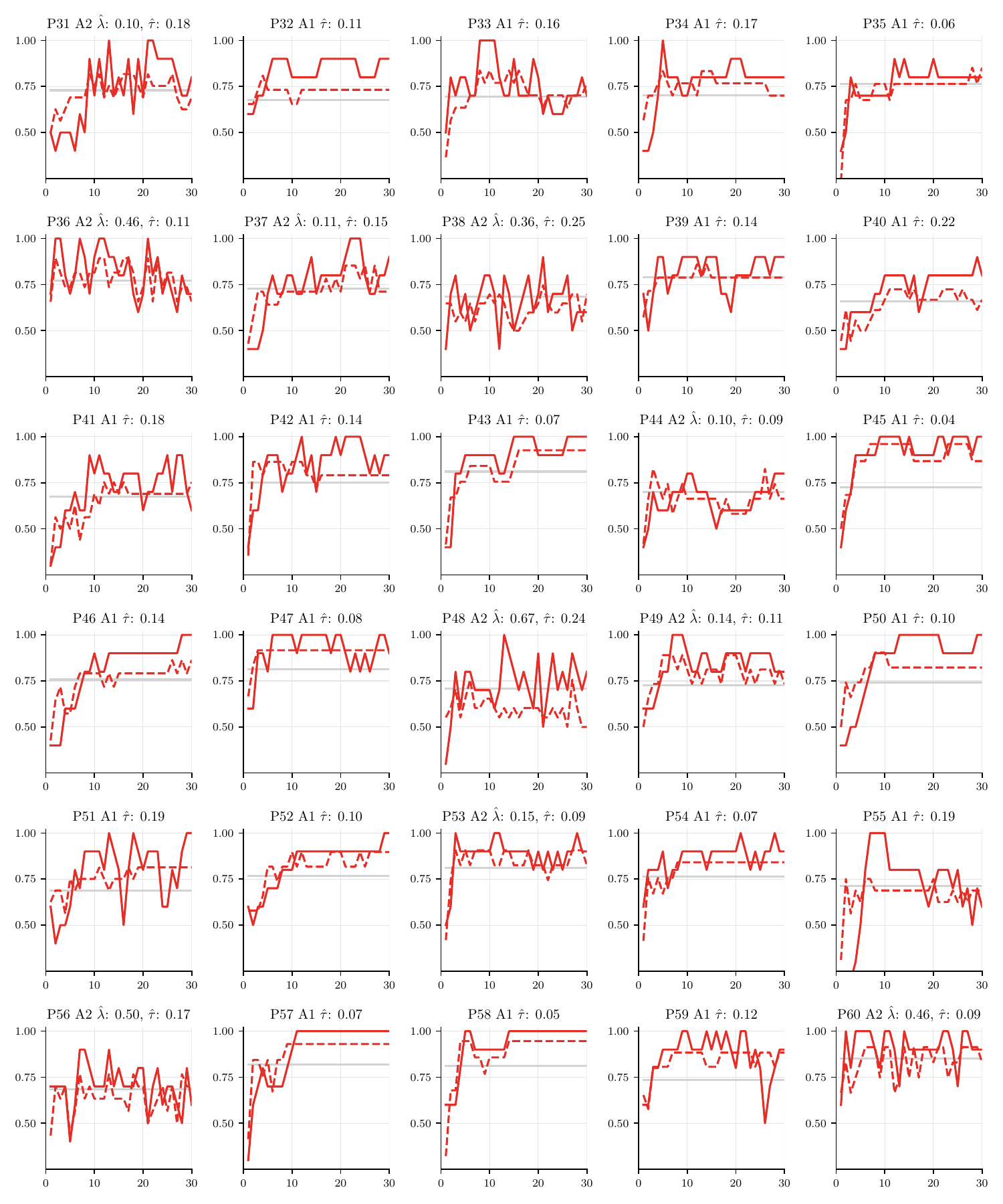}
\caption{Symmetric bandit task descriptive statistics and posterior predictive
learning functions for participants 31 to 60 ($x$-axis: trial, $y$-axis:
observed proportion of maximizing actions).}
\label{fig:S7B}
\end{subfigure}
\caption{Symmetric bandit task descriptive statistics for individual
participants. In each subpanel, the solid red line depicts the observed
proportion of maximizing actions across blocks at each trial position. The
horizontal light-gray line corresponds to the average adjusted probability
state across blocks. The dashed red line depicts the participant-specific
posterior predictive learning function, i.e., the probability of choosing the
maximizing action, averaged across blocks at each trial position, based on the
participant-specific most plausible model and its maximum-likelihood parameter
estimates.}\label{fig:S7}
\end{figure}
\clearpage

\subsection{Bayesian agent model analysis}\label{sec:bayesian-agent-model-analysis}
The blockwise constant probability state can equivalently be represented by
$s_t=s_{t-1}$ almost surely for $t=2,\ldots,T$. Formally, its transition
distribution is the Dirac probability measure $\delta_{s_{t-1}}$, which assigns
probability one to the preceding state. Because the state space is continuous,
this measure does not have an ordinary density with respect to Lebesgue measure.

\begin{theorem}[Belief state]\label{thm:belief-state-2}
Let
\begin{multline}
p(s_1,d_{1:T},a_{1:T},r_{1:T})
\\
= p(s_1)p(d_1)p(a_1|d_1)p(r_1|s_1,a_1)
\prod_{t=2}^T p(d_t|a_{1:t-1},r_{1:t-1})p(a_t|d_t)p(r_t|s_1,a_t)
\end{multline}
denote the ABM induced by \eqnref{eq:T-2}, \eqnref{eq:A-p-2} and \eqnref{eq:D-2}
and let
\begin{equation}
p(s_1|a_{1:0},r_{1:0}) := p(s_1) = \mathrm{Unif}(s_1;[0,1]) = \mathrm{Beta}(s_1; \alpha_1) \mbox{ with } \alpha_1 := (1,1)
\end{equation}
denote the agent's block-specific prior probability state distribution. Then for
$t = 2,...,T$, it holds that
\begin{equation}\label{eq:bs-theorem-2}
p(s_1|a_{1:t-1},r_{1:t-1}) = \mathrm{Beta}\left(s_1; \alpha_t \right)
\mbox{ with } \alpha_t = \left(1 + \sum_{k=1}^{t-1} \tilde{r}_k,  t -  \sum_{k=1}^{t-1} \tilde{r}_k \right).
\end{equation}
where
\begin{equation}
\tilde{r}_k := \frac{1}{2}\left(1 + (-1)^{r_k+a_k}\right).
\end{equation}
\end{theorem}

\begin{proof}
We first note that, for $t = 2,...,T$,
\begin{equation}\label{eq:bs-2-proof-1}
p(s_1|a_{1:t-1},r_{1:t-1})  = \frac{p(s_1)\prod_{k=1}^{t-1} p(r_k|s_1,a_k) }{\int_0^1 p(s_1)\prod_{k=1}^{t-1} p(r_k|s_1,a_k)\,d s_1},
\end{equation}
because
\begin{align*}
& p(s_1|a_{1:t-1},r_{1:t-1})                                                                                                        \\                                                                                          
& = \frac{p(s_1,a_{1:t-1},r_{1:t-1})}{p(a_{1:t-1},r_{1:t-1})}                                                                                                                                                                   \\
& = \frac{\sum_{d_{1:t-1}}p(s_1,d_{1:t-1},a_{1:t-1},r_{1:t-1})}{\int_0^1 \sum_{d_{1:t-1}} p(s_1,d_{1:t-1},a_{1:t-1},r_{1:t-1})\,d s_1}                                                                                                       \\
& = \frac{\sum_{d_{1:t-1}}p(s_1)p(d_1)p(a_1|d_1)p(r_1|s_1,a_1)\prod_{k=2}^{t-1} p(d_k|a_{1:k-1},r_{1:k-1})p(a_k|d_k)p(r_k|s_1,a_k)}{\int_0^1 \sum_{d_{1:t-1}} p(s_1,d_{1:t-1},a_{1:t-1},r_{1:t-1})\,d s_1}                                 \\
& = \frac{\sum_{d_1}p(s_1)p(d_1)p(a_1|d_1)p(r_1|s_1,a_1)\prod_{k=2}^{t-1} \sum_{d_k} p(d_k|a_{1:k-1},r_{1:k-1})p(a_k|d_k)p(r_k|s_1,a_k)}{\int_0^1 \sum_{d_{1:t-1}} p(s_1,d_{1:t-1},a_{1:t-1},r_{1:t-1})\,d s_1}                          \\
& = \frac{\sum_{d_1}p(s_1)p(a_1,d_1)p(r_1|s_1,a_1)\prod_{k=2}^{t-1} \sum_{d_k} p(a_k, d_k|a_{1:k-1},r_{1:k-1})p(r_k|s_1,a_k)}{\int_0^1 \sum_{d_{1:t-1}} p(s_1,d_{1:t-1},a_{1:t-1},r_{1:t-1})\,d s_1}                                     \\
& = \frac{p(s_1)p(r_1|s_1,a_1)\sum_{d_1}p(a_1,d_1)\prod_{k=2}^{t-1} p(r_k|s_1,a_k) \sum_{d_k} p(a_k, d_k|a_{1:k-1},r_{1:k-1})}{\int_0^1 \sum_{d_{1:t-1}} p(s_1,d_{1:t-1},a_{1:t-1},r_{1:t-1})\,d s_1}                                    \\
& = \frac{p(s_1)p(r_1|s_1,a_1)p(a_1)\prod_{k=2}^{t-1} p(r_k|s_1,a_k) p(a_k|a_{1:k-1},r_{1:k-1})}{\int_0^1 \sum_{d_{1:t-1}} p(s_1,d_{1:t-1},a_{1:t-1},r_{1:t-1})\,d s_1}                                                                  \\
& = \frac{p(s_1)p(r_1|s_1,a_1)p(a_1)\prod_{k=2}^{t-1} p(r_k|s_1,a_k) \prod_{k=2}^{t-1} p(a_k|a_{1:k-1},r_{1:k-1})}{\int_0^1 p(s_1)p(r_1|s_1,a_1)p(a_1)\prod_{k=2}^{t-1} p(r_k|s_1,a_k)\prod_{k=2}^{t-1} p(a_k|a_{1:k-1},r_{1:k-1})\,d s_1}   \\
& = \frac{p(a_1)\prod_{k=2}^{t-1} p(a_k|a_{1:k-1},r_{1:k-1})p(s_1)p(r_1|s_1,a_1)\prod_{k=2}^{t-1} p(r_k|s_1,a_k) }{p(a_1)\prod_{k=2}^{t-1} p(a_k|a_{1:k-1},r_{1:k-1})\int_0^1 p(s_1)p(r_1|s_1,a_1)\prod_{k=2}^{t-1} p(r_k|s_1,a_k)\,d s_1}   \\
& = \frac{p(s_1)p(r_1|s_1,a_1)\prod_{k=2}^{t-1} p(r_k|s_1,a_k) }{\int_0^1 p(s_1)p(r_1|s_1,a_1)\prod_{k=2}^{t-1} p(r_k|s_1,a_k)\,d s_1}                                                                                                \\
& = \frac{p(s_1)\prod_{k=1}^{t-1} p(r_k|s_1,a_k) }{\int_0^1 p(s_1)\prod_{k=1}^{t-1} p(r_k|s_1,a_k)\,d s_1},
\end{align*}
where we have used \autoref{thm:exchangeability} in the fourth equality.
\vspace{2mm}

\noindent We next show that \eqnref{eq:bs-theorem-2} follows from
\eqnref{eq:bs-2-proof-1} by induction with respect to $t$.

\vspace{3mm}
\noindent \textit{Base case.} For $t=2$, we have from \eqnref{eq:bs-2-proof-1}
\begin{equation} 
p(s_1|a_{1},r_{1}) 
\propto p(s_1)p(r_1|s_1,a_1) 
= \mathrm{Beta}\left(s_1;(1,1)\right)\mathrm{Bern}(r_1;s_1)^{a_1}\mathrm{Bern}(r_1;1-s_1)^{1-a_1}.
\end{equation}
Assume first that $a_1 = 0$. Then,
\begin{align}
\begin{split}
p(s_1|a_1 = 0, r_1)
& \propto \mathrm{Beta}\left(s_1;(1,1)\right)\mathrm{Bern}(r_1;s_1)^{0}\mathrm{Bern}(r_1;1-s_1)^{1-0}       \\
& = \mathrm{Beta}\left(s_1; (1,1)\right)\mathrm{Bern}(r_1;1-s_1)                                        \\    
& = (s_1)^{1-1}(1-s_1)^{1-1}(1-s_1)^{r_1}(s_1)^{1-r_1}            \\
& = (s_1)^{1 + 1 - r_1 - 1}(1-s_1)^{1 + r_1 - 1}.                 \\
\end{split}
\end{align}
Assume now that $a_1 = 1$. Then, analogously
\begin{align}
\begin{split}
p(s_1|a_1 = 1, r_1)
& \propto \mathrm{Beta}\left(s_1; (1,1)\right)\mathrm{Bern}(r_1;s_1)^{1}\mathrm{Bern}(r_1;1-s_1)^{1-1}      \\
& = \mathrm{Beta}\left(s_1; (1,1)\right)\mathrm{Bern}(r_1;s_1)                                       \\    
& = (s_1)^{1-1}(1-s_1)^{1-1}(s_1)^{r_1}(1-s_1)^{1-r_1}    \\
& = (s_1)^{1 + r_1 - 1}(1-s_1)^{1 + 1 - r_1 - 1}.          \\
\end{split}
\end{align}
We thus observe the following update rules for the parameters of the Beta
distribution:
\nopagebreak[4]
\begin{itemize}
\begin{small}
\item If $a_1 = 0$, then $\alpha_2^{(1)} =1 + 1 - r_1$ and
$\alpha_2^{(2)} = 1 + r_1$.
\item If $a_1 = 1$, then $\alpha_2^{(1)} = 1 + r_1\quad\,\,\,\,$ and
$\alpha_2^{(2)} = 1 + 1 - r_1$.
\end{small}
\end{itemize}
For $\alpha^{(1)}_2$, we thus have the update rule
\begin{equation}\label{eq:bs-2-proof-2}
\alpha_2^{(1)} = 1 + 
\begin{cases}
1 &  \mbox{ for } a_1 = 0 \mbox{ and } r_1 = 0 \\
0 &  \mbox{ for } a_1 = 0 \mbox{ and } r_1 = 1 \\
0 &  \mbox{ for } a_1 = 1 \mbox{ and } r_1 = 0 \\
1 &  \mbox{ for } a_1 = 1 \mbox{ and } r_1 = 1 \\
\end{cases}
\end{equation}
and for $\alpha^{(2)}_2$, we have the update rule
\begin{equation}\label{eq:bs-2-proof-3}
\alpha_2^{(2)} = 1  + 
\begin{cases}
0 & \mbox{ for } a_1 = 0 \mbox{ and } r_1 = 0 \\
1 & \mbox{ for } a_1 = 0 \mbox{ and } r_1 = 1 \\
1 & \mbox{ for } a_1 = 1 \mbox{ and } r_1 = 0 \\
0 & \mbox{ for } a_1 = 1 \mbox{ and } r_1 = 1 \\
\end{cases}.
\end{equation}
With
\begin{equation}\label{eq:bs-2-proof-4}
\tilde{r}_1 = 
\begin{cases}
\frac{1}{2}\left(1 + (-1)^{0+0} \right) = 1 & \mbox{ for } a_1 = 0 \mbox{ and } r_1 = 0 \\
\frac{1}{2}\left(1 + (-1)^{0+1} \right) = 0 & \mbox{ for } a_1 = 0 \mbox{ and } r_1 = 1 \\
\frac{1}{2}\left(1 + (-1)^{1+0} \right) = 0 & \mbox{ for } a_1 = 1 \mbox{ and } r_1 = 0 \\
\frac{1}{2}\left(1 + (-1)^{1+1} \right) = 1 & \mbox{ for } a_1 = 1 \mbox{ and } r_1 = 1 \\
\end{cases}
\end{equation}
and the parametric form of the Beta distribution, it thus follows that
\begin{equation}
p(s_1|a_1,r_1) = \mathrm{Beta}(s_1; \alpha_2) 
\end{equation}
with 
\begin{align}
\begin{split}
\alpha_2 
& = \left(1+\tilde{r}_1, 1 + (1 - \tilde{r}_1)\right)                           \\
& = \left(1+\tilde{r}_1, 2 - \tilde{r}_1\right)                                 \\
& = \left(1+\sum_{k=1}^{1}  \tilde{r}_k, 2-\sum_{k=1}^{1}  \tilde{r}_k \right)  \\
& = \left(1+\sum_{k=1}^{t-1}\tilde{r}_k, t-\sum_{k=1}^{t-1}\tilde{r}_k \right)\quad\mbox{for }t=2.
\end{split}
\end{align}
Thus, \eqnref{eq:bs-theorem-2} holds for $t = 2$.

\vspace{1mm}
\noindent \textit{Induction step.} Assume now that 
\begin{equation}
p(s_1|a_{1:t-2},r_{1:t-2}) = \mathrm{Beta}(s_1;\alpha_{t-1})
\end{equation}
with
\begin{equation}
\alpha_{t-1} = \left(1 + \sum_{k=1}^{t-2}\tilde{r}_k, t-1 - \sum_{k=1}^{t-2}\tilde{r}_k\right)
\end{equation}
for some $t = 3,...,T$.
Then, from \eqnref{eq:bs-2-proof-1}, we obtain
\begin{align}
\begin{split}
p(s_1|a_{1:t-1},r_{1:t-1})  
& \propto p(s_1)\prod_{k=1}^{t-1} p(r_k|s_1,a_k)                                                                \\  
& = \left(p(s_1)\prod_{k=1}^{t-2} p(r_k|s_1,a_k)\right) p(r_{t-1}|s_1,a_{t-1})                                  \\
& \propto p(s_1|a_{1:t-2},r_{1:t-2})p(r_{t-1}|s_1,a_{t-1})                                                     \\
& = \mathrm{Beta}(s_1;\alpha_{t-1})\mathrm{Bern}(r_{t-1};s_1)^{a_{t-1}}\mathrm{Bern}(r_{t-1};1-s_1)^{1-a_{t-1}}.
\end{split}
\end{align}
Assume first that $a_{t-1} = 0$. Then,
\begingroup\small
\begin{align}
\begin{split}
p(s_1|a_{1:t-2},a_{t-1} = 0,r_{1:t-1})
& \propto \mathrm{Beta}\left(s_1;\left(\alpha^{(1)}_{t-1},\alpha^{(2)}_{t-1}\right)\right)\mathrm{Bern}(r_{t-1};s_1)^{0}\mathrm{Bern}(r_{t-1};1-s_1)^{1-0}       \\
& = \mathrm{Beta}\left(s_1; \left(\alpha^{(1)}_{t-1},\alpha^{(2)}_{t-1}\right)\right)\mathrm{Bern}(r_{t-1};1-s_1)                                        \\    
& \propto (s_1)^{\alpha^{(1)}_{t-1}-1}(1-s_1)^{\alpha^{(2)}_{t-1}-1}(1-s_1)^{r_{t-1}}(s_1)^{1-r_{t-1}}            \\
& = (s_1)^{\alpha^{(1)}_{t-1} + 1 - r_{t-1} - 1}(1-s_1)^{\alpha^{(2)}_{t-1} + r_{t-1} - 1}.                 \\
\end{split}
\end{align}
\endgroup
Assume now that $a_{t-1} = 1$. Then, analogously
\begingroup\small
\begin{align}
\begin{split}
p(s_1|a_{1:t-2},a_{t-1} = 1,r_{1:t-1})
& \propto \mathrm{Beta}\left(s_1; \left(\alpha^{(1)}_{t-1},\alpha^{(2)}_{t-1}\right)\right)\mathrm{Bern}(r_{t-1};s_1)^{1}\mathrm{Bern}(r_{t-1};1-s_1)^{1-1}      \\
& = \mathrm{Beta}\left(s_1; \left(\alpha^{(1)}_{t-1},\alpha^{(2)}_{t-1}\right)\right)\mathrm{Bern}(r_{t-1};s_1)                                       \\    
& \propto (s_1)^{\alpha^{(1)}_{t-1}-1}(1-s_1)^{\alpha^{(2)}_{t-1}-1}(s_1)^{r_{t-1}}(1-s_1)^{1-r_{t-1}}    \\
& = (s_1)^{\alpha^{(1)}_{t-1} + r_{t-1} - 1}(1-s_1)^{\alpha^{(2)}_{t-1} + 1 - r_{t-1} - 1}.          \\
\end{split}
\end{align}
\endgroup
We thus observe the following update rules for the parameters of the Beta
distribution:
\nopagebreak[4]
\begin{itemize}
\begin{small}
\item If $a_{t-1} = 0$, then
$\alpha^{(1)}_{t} = \alpha^{(1)}_{t-1} + 1 - r_{t-1}$ and
$\alpha_{t}^{(2)} = \alpha^{(2)}_{t-1} + r_{t-1}$.
\item If $a_{t-1} = 1$, then
$\alpha^{(1)}_{t} = \alpha^{(1)}_{t-1} + r_{t-1}\quad\,\,\,\,$ and
$\alpha_{t}^{(2)} = \alpha^{(2)}_{t-1} + 1 - r_{t-1}$.
\end{small}
\end{itemize}
For $\alpha^{(1)}_{t}$, we thus have the update rule
\begin{equation}\label{eq:bs-2-proof-5}
\alpha^{(1)}_{t} = \alpha^{(1)}_{t-1} + 
\begin{cases}
1 &  \mbox{ for } a_{t-1} = 0 \mbox{ and } r_{t-1} = 0 \\
0 &  \mbox{ for } a_{t-1} = 0 \mbox{ and } r_{t-1} = 1 \\
0 &  \mbox{ for } a_{t-1} = 1 \mbox{ and } r_{t-1} = 0 \\
1 &  \mbox{ for } a_{t-1} = 1 \mbox{ and } r_{t-1} = 1 \\
\end{cases}
\end{equation}
and for $\alpha^{(2)}_t$, we have the update rule
\begin{equation}\label{eq:bs-2-proof-6}
\alpha^{(2)}_{t} = \alpha^{(2)}_{t-1} +
\begin{cases}
0 & \mbox{ for } a_{t-1} = 0 \mbox{ and } r_{t-1} = 0 \\
1 & \mbox{ for } a_{t-1} = 0 \mbox{ and } r_{t-1} = 1 \\
1 & \mbox{ for } a_{t-1} = 1 \mbox{ and } r_{t-1} = 0 \\
0 & \mbox{ for } a_{t-1} = 1 \mbox{ and } r_{t-1} = 1 \\
\end{cases}.
\end{equation}
With
\begin{equation}\label{eq:bs-2-proof-7}
\tilde{r}_{t-1} = 
\begin{cases}
\frac{1}{2}\left(1 + (-1)^{0+0} \right) = 1 & \mbox{ for } a_{t-1} = 0 \mbox{ and } r_{t-1} = 0 \\
\frac{1}{2}\left(1 + (-1)^{0+1} \right) = 0 & \mbox{ for } a_{t-1} = 0 \mbox{ and } r_{t-1} = 1 \\
\frac{1}{2}\left(1 + (-1)^{1+0} \right) = 0 & \mbox{ for } a_{t-1} = 1 \mbox{ and } r_{t-1} = 0 \\
\frac{1}{2}\left(1 + (-1)^{1+1} \right) = 1 & \mbox{ for } a_{t-1} = 1 \mbox{ and } r_{t-1} = 1 \\
\end{cases}
\end{equation}
and the parametric form of the Beta distribution, it thus follows that 
\begin{equation}
p(s_1|a_{1:t-1},r_{1:t-1}) = \mathrm{Beta}(s_1;\alpha_t) 
\end{equation}
with
\begin{align}
\begin{split}
\alpha_t 
& = \left(\alpha^{(1)}_{t-1} + \tilde{r}_{t-1}, \alpha^{(2)}_{t-1} + 1 - \tilde{r}_{t-1} \right) \\
& = \left(1 + \sum_{k=1}^{t-2}\tilde{r}_k  + \tilde{r}_{t-1}, t - 1  - \sum_{k=1}^{t-2}\tilde{r}_k +  1 - \tilde{r}_{t-1} \right) \\
& = \left(1 + \sum_{k=1}^{t-1}\tilde{r}_k, t - \sum_{k=1}^{t-1}\tilde{r}_k  \right).
\end{split}
\end{align}
Notably, 
\begin{align}\label{eq:alpha-update}
\begin{split}
\alpha_t 
& = \left(\alpha^{(1)}_{t-1} + \tilde{r}_{t-1}, \alpha^{(2)}_{t-1} + 1 - \tilde{r}_{t-1} \right)  \\                                               
& = \left(\alpha_{t-1}^{(1)} + \frac{1}{2}\left(1 + (-1)^{r_{t-1} + a_{t-1}}\right), 
          \alpha_{t-1}^{(2)} +  1 - \frac{1}{2}\left(1 + (-1)^{r_{t-1} + a_{t-1}}\right)\right)
\end{split}
\end{align}
implies the Markovian belief state update as formulated in the belief state
update function \eqnref{eq:beta-2}, because
$\alpha_{t-1} = \left(\alpha_{t-1}^{(1)},\alpha_{t-1}^{(2)} \right)$ corresponds
to the parametric representation of the previous belief state $b_{t-1}$. For
later reference, we also note that \eqnref{eq:alpha-update} implies, for
$\tilde{r}_{t-1} = 0$, that
\begin{align}\label{eq:post-exp-r-tilde-0}
\begin{split}
\mathbb{E}_{p(s_1|a_{1:t-1},r_{1:t-1})}(s_1) 
& = \frac{\alpha_t^{(1)}}{\alpha_t^{(1)}+\alpha_t^{(2)}}                            \\
& = \frac{\alpha_{t-1}^{(1)}+0}{\alpha_{t-1}^{(1)} + 0 +\alpha_{t-1}^{(2)} + 1 - 0} \\
& = \frac{\alpha_{t-1}^{(1)}}{\alpha_{t-1}^{(1)} + \alpha_{t-1}^{(2)} + 1}
\end{split}
\end{align}
and, for $\tilde{r}_{t-1} = 1$, that
\begin{align}\label{eq:post-exp-r-tilde-1}
\begin{split}
\mathbb{E}_{p(s_1|a_{1:t-1},r_{1:t-1})}(s_1) 
& = \frac{\alpha_t^{(1)}}{\alpha_t^{(1)}+\alpha_t^{(2)}}                            \\
& = \frac{\alpha_{t-1}^{(1)}+1}{\alpha_{t-1}^{(1)} + 1 +\alpha_{t-1}^{(2)} + 1 - 1} \\
& = \frac{\alpha_{t-1}^{(1)}+1}{\alpha_{t-1}^{(1)} + 1 + \alpha_{t-1}^{(2)}}.
\end{split}
\end{align}
\end{proof}

\begin{theorem}[Decision optimality]\label{thm:decision-optimality-2}
Let
\begin{multline}
p(s_1,d_{1:T}, a_{1:T},r_{1:T})
\\
= p(s_1)p(d_1)p(a_1|d_1)p(r_1|s_1,d_1)\prod_{t=2}^Tp(d_t|a_{1:t-1},r_{1:t-1})p(a_t|d_t)p(r_t|s_1,d_t)
\end{multline}
denote the ABM induced by \eqnref{eq:T-2}, \eqnref{eq:A-p-2} and
\eqnref{eq:D-2} under
the assumption of absence of post-decision noise, i.e., $\tau=0$. Then
$a_t=d_t$, and the decision-conditional reward distribution is
\begin{equation}
p(r_t|s_1,d_t) := \mathrm{Bern}(r_t;s_1)^{d_t}\mathrm{Bern}(r_t;1-s_1)^{1-d_t}.
\end{equation}
Let $\varphi=(\varphi_t)_{t=1}^T$ denote a history-dependent policy
\citep[cf.][]{krishnamurthy2016, puterman1994}, with
\begin{equation}
\varphi_t:A^{t-1}\times R^{t-1}\to D,\qquad
(a_{1:t-1},r_{1:t-1})\mapsto\varphi_t(a_{1:t-1},r_{1:t-1})=:d_t,
\end{equation}
and let $\Pi$ denote the set of all such policies in the current scenario. Then,
the policy $\varphi^*$ induced by the agent's recursive belief update and
decision rule,
\begin{equation}
\varphi_t^*(a_{1:t-1},r_{1:t-1}):=\delta\bigl(\phi(b_t)\bigr),
\end{equation}
where
\begin{equation}
b_1=p(s_1),\qquad b_t=\beta(b_{t-1},a_{t-1},r_{t-1})
\quad\mbox{for }t=2,\ldots,T,
\end{equation}
and
\begin{equation}
\delta(v_t):=\min\argmax_{d\in D}v_t^d,
\end{equation}
maximizes the agent's cumulative expected reward, i.e.,
\begin{equation}
\varphi^*\in\argmax_{\varphi\in\Pi}
\mathbb{E}^{\varphi}\left(\sum_{t=1}^T r_t\right),
\end{equation}
where $\mathbb{E}^{\varphi}$ denotes expectation under the task model and policy
$\varphi$.
\end{theorem}

\begin{proof}
We first note that
\begin{align}
\begin{split}
\mathbb{E}_{p(r_t|s_1,d_t)}(r_t)
& = 0\cdot \mathrm{Bern}(0;s_1)^{d_t}\mathrm{Bern}(0;1-s_1)^{1-d_t} + 1 \cdot \mathrm{Bern}(1;s_1)^{d_t}\mathrm{Bern}(1;1-s_1)^{1-d_t}  \\
& = \mathrm{Bern}(1;s_1)^{d_t}\mathrm{Bern}(1;1-s_1)^{1-d_t}                                                                  \\
& = (s_1)^{d_t}(1-s_1)^{1-d_t}.
\end{split}
\end{align}

Substituting $\varphi_t(a_{1:t-1},r_{1:t-1})$ for $d_t$ and applying the law
of iterated expectation, we thus consider the maximization problem
\begin{equation}
\fitproofline{\max_{\varphi\in\Pi}\mathbb{E}^{\varphi}\left(
\sum_{t=1}^T\mathbb{E}_{p(s_1|a_{1:t-1},r_{1:t-1})}\left(
(s_1)^{\varphi_t(a_{1:t-1},r_{1:t-1})}
(1-s_1)^{1-\varphi_t(a_{1:t-1},r_{1:t-1})}\right)\right).}
\end{equation}
To justify trialwise maximization, note that the adjusted reward satisfies
\begin{equation}
\tilde{r}_t=\ib{a_t=r_t}
\end{equation}
and hence
\begin{align}
\begin{split}
p(\tilde{r}_t=1|s_1,a_t=0)&=p(r_t=0|s_1,a_t=0)=s_1,\\
p(\tilde{r}_t=1|s_1,a_t=1)&=p(r_t=1|s_1,a_t=1)=s_1.
\end{split}
\end{align}
Thus, conditional on the current belief, the distribution of $\tilde{r}_t$, and
therefore of the updated belief
\begin{equation}
\alpha_{t+1}=\left(\alpha_t^{(1)}+\tilde{r}_t,\,
\alpha_t^{(2)}+1-\tilde{r}_t\right),
\end{equation}
is independent of the chosen action. Both actions consequently have the same
expected optimal continuation reward. Backward induction from trial $T$
therefore shows that maximizing cumulative expected reward reduces, at each
trial and history, to maximizing the immediate posterior expected reward. For 
$t=1,\ldots,T$, we are thus concerned with identifying a function
$\varphi_t$ for which
\begin{equation}
\varphi_t(a_{1:t-1},r_{1:t-1})\in\{0,1\}
\end{equation}
and which maximizes the posterior expected value of
\begin{equation}
(s_1)^{\varphi_t(a_{1:t-1},r_{1:t-1})}
(1-s_1)^{1-\varphi_t(a_{1:t-1},r_{1:t-1})}.
\end{equation}
We are thus led to consider

\begin{align}
\begin{split}
& \max_{d \in \{0,1\}}
\mathbb{E}_{p(s_1|a_{1:t-1},r_{1:t-1})}\left((s_1)^{d}(1-s_1)^{1-d} \right)
\\
& = 
\max
\left\lbrace  
\mathbb{E}_{p(s_1|a_{1:t-1},r_{1:t-1})}\left((s_1)^{0}(1-s_1)^{1-0} \right),
\mathbb{E}_{p(s_1|a_{1:t-1},r_{1:t-1})}\left((s_1)^{1}(1-s_1)^{1-1} \right)
\right\rbrace
\\
& = 
\max
\left\lbrace  
\mathbb{E}_{p(s_1|a_{1:t-1},r_{1:t-1})}\left(1-s_1\right),
\mathbb{E}_{p(s_1|a_{1:t-1},r_{1:t-1})}\left(s_1  \right)
\right\rbrace
\\
& = 
\max
\left\lbrace  
1 - \mathbb{E}_{p(s_1|a_{1:t-1},r_{1:t-1})}\left(s_1\right),
\mathbb{E}_{p(s_1|a_{1:t-1},r_{1:t-1})}\left(s_1  \right)
\right\rbrace
\\
& = 
\max
\left\lbrace  
1 - \frac{\alpha^{(1)}_t}{\alpha^{(1)}_t + {\alpha^{(2)}_t}}, 
\frac{\alpha^{(1)}_t}{\alpha^{(1)}_t + {\alpha^{(2)}_t}}
\right\rbrace.
\end{split}
\end{align}
With ties resolved in favor of decision $0$, we have
\begin{align}
\begin{split}
\min\argmax_{d \in \{0,1\}}
\mathbb{E}_{p(s_1|a_{1:t-1},r_{1:t-1})}  &  \left((s_1)^{d} (1-s_1)^{1-d} \right)  = 0 
\\ \Leftrightarrow
1 - \frac{\alpha^{(1)}_t}{\alpha^{(1)}_t + {\alpha^{(2)}_t}} & \ge  \frac{\alpha^{(1)}_t}{\alpha^{(1)}_t + {\alpha^{(2)}_t}}
\\ \Leftrightarrow
 \frac{\alpha^{(1)}_t + {\alpha^{(2)}_t}}{\alpha^{(1)}_t + {\alpha^{(2)}_t}} - \frac{\alpha^{(1)}_t}{\alpha^{(1)}_t + {\alpha^{(2)}_t}} & \ge  \frac{\alpha^{(1)}_t}{\alpha^{(1)}_t + {\alpha^{(2)}_t}}
 \\ \Leftrightarrow
 \frac{\alpha^{(2)}_t}{\alpha^{(1)}_t + {\alpha^{(2)}_t}} & \ge  \frac{\alpha^{(1)}_t}{\alpha^{(1)}_t + {\alpha^{(2)}_t}}
\\ \Leftrightarrow
\alpha^{(1)}_t & \le \alpha^{(2)}_t
\end{split}
\end{align}
and
\begin{align}
\begin{split}
\min\argmax_{d \in \{0,1\}}
\mathbb{E}_{p(s_1|a_{1:t-1},r_{1:t-1})} &  \left((s_1)^{d}(1-s_1)^{1-d} \right) = 1 
\\ \Leftrightarrow
1 - \frac{\alpha^{(1)}_t}{\alpha^{(1)}_t + {\alpha^{(2)}_t}} & <  \frac{\alpha^{(1)}_t}{\alpha^{(1)}_t + {\alpha^{(2)}_t}}
\\ \Leftrightarrow
 \frac{\alpha^{(1)}_t + {\alpha^{(2)}_t}}{\alpha^{(1)}_t + {\alpha^{(2)}_t}} - \frac{\alpha^{(1)}_t}{\alpha^{(1)}_t + {\alpha^{(2)}_t}} & <  \frac{\alpha^{(1)}_t}{\alpha^{(1)}_t + {\alpha^{(2)}_t}}
 \\ \Leftrightarrow
 \frac{\alpha^{(2)}_t}{\alpha^{(1)}_t + {\alpha^{(2)}_t}} & <  \frac{\alpha^{(1)}_t}{\alpha^{(1)}_t + {\alpha^{(2)}_t}}
\\ \Leftrightarrow
\alpha^{(1)}_t & > \alpha^{(2)}_t.
\end{split}
\end{align}
Such a policy is thus of the form
\begin{multline}
\varphi_t : A^{t-1}\times R^{t-1}\to D,\\
(a_{1:t-1},r_{1:t-1})\mapsto\varphi_t(a_{1:t-1},r_{1:t-1}):=
\begin{cases}
0 & \mbox{for }\alpha_t^{(1)}\le\alpha_t^{(2)},\\
1 & \mbox{for }\alpha_t^{(1)}>\alpha_t^{(2)}.
\end{cases}
\end{multline}
But as shown in \autoref{thm:conditional-decision-distribution-2}, this is
exactly the decision rule induced by the agent's recursive belief update,
because
\begin{align}
\begin{split}
\delta(\phi(b_t))=0
&\Longleftrightarrow p(d_t=0|a_{1:t-1},r_{1:t-1})=1
\Longleftrightarrow\alpha_t^{(1)}\le\alpha_t^{(2)},\\
\delta(\phi(b_t))=1
&\Longleftrightarrow p(d_t=1|a_{1:t-1},r_{1:t-1})=1
\Longleftrightarrow\alpha_t^{(1)}>\alpha_t^{(2)}.
\end{split}
\end{align}
Thus,
\begin{equation}
\varphi^*\in\argmax_{\varphi\in\Pi}
\mathbb{E}^{\varphi}\left(\sum_{t=1}^T r_t\right).
\end{equation}
\end{proof}

\begin{theorem}[Conditional decision distribution]\label{thm:conditional-decision-distribution-2}
Given the agent model $\mathcal{A}$ defined in \eqnref{eq:A-2}, the distribution
of an agent's decision
\begin{equation}
d_t = \delta(\phi(b_t)),\qquad t=1,\ldots,T,
\end{equation}
where $b_1=p(s_1)$ and
\begin{equation}
b_t=\beta(b_{t-1},a_{t-1},r_{t-1})
\quad\mbox{for }t=2,\ldots,T,
\end{equation}
given the preceding action and reward history can be written as
\begin{equation}
p(d_t|a_{1:t-1},r_{1:t-1}) = 
\left\llbracket \alpha^{(1)}_{t} \le \alpha^{(2)}_{t} \right\rrbracket^{1-d_t}
\left\llbracket \alpha^{(1)}_{t}  >  \alpha^{(2)}_{t}\right\rrbracket^{d_t}, 
\end{equation}
where
\begin{equation}
\alpha^{(1)}_t := 1 + \sum_{k=1}^{t-1} \tilde{r}_k \mbox{ and } \alpha^{(2)}_t := t - \sum_{k=1}^{t-1} \tilde{r}_k. 
\end{equation}
For $t=1$, the histories are empty and the sums are zero.
\end{theorem}

\begin{proof}
We first note that because $d_t$ is a function of $a_{1:t-1}$ and $r_{1:t-1}$,
$p(d_t|a_{1:t-1},r_{1:t-1})$ is degenerate, taking on the value $1$ if
$d_t=\delta(\phi(b_t))$ and $0$ otherwise.
Next, because $d_t$ takes on values in $D = \{0,1\}$, it suffices to consider
$p(d_t = 0|a_{1:t-1},r_{1:t-1})$, as
$p(d_t = 1|a_{1:t-1},r_{1:t-1}) = 1 - p(d_t = 0|a_{1:t-1},r_{1:t-1})$ follows
immediately. We hence assume that $d_t = 0$ and show that this is equivalent to
$\alpha^{(1)}_t \le \alpha^{(2)}_t$:
\begin{align}
\begin{split}
d_t 
& = 
0
\\\Leftrightarrow
v_t^{0} 
& \ge 
v_t^{1}
\\\Leftrightarrow
1 - \mathbb{E}_{p(s_1|a_{1:t-1}, r_{1:t-1})}(s_1)
& \ge 
\mathbb{E}_{p(s_1|a_{1:t-1}, r_{1:t-1})}(s_1)
\\\Leftrightarrow
1 - \frac{\alpha^{(1)}_t}{\alpha^{(1)}_t + \alpha^{(2)}_t}
& \ge 
\frac{\alpha^{(1)}_t}{\alpha^{(1)}_t + \alpha^{(2)}_t}
\\\Leftrightarrow
\frac{\alpha^{(1)}_t + \alpha^{(2)}_t}{\alpha^{(1)}_t + \alpha^{(2)}_t} - \frac{\alpha^{(1)}_t}{\alpha^{(1)}_t + \alpha^{(2)}_t}
& \ge 
\frac{\alpha^{(1)}_t}{\alpha^{(1)}_t + \alpha^{(2)}_t}
\\\Leftrightarrow
\frac{\alpha^{(2)}_t}{\alpha^{(1)}_t + \alpha^{(2)}_t}
& \ge 
\frac{\alpha^{(1)}_t}{\alpha^{(1)}_t + \alpha^{(2)}_t}
\\\Leftrightarrow
\alpha^{(1)}_t
& \le 
\alpha^{(2)}_t.
\end{split}
\end{align}
The identities 
$p(d_t = 0|a_{1:t-1}, r_{1:t-1}) = 1$ if $\alpha^{(1)}_t \le \alpha^{(2)}_t$, and
$p(d_t = 0|a_{1:t-1}, r_{1:t-1}) = 0$ if $\alpha^{(1)}_t > \alpha^{(2)}_t$,
imply that
$p(d_t = 1|a_{1:t-1}, r_{1:t-1}) = 0$ if $\alpha^{(1)}_t \le \alpha^{(2)}_t$, and 
$p(d_t = 1|a_{1:t-1}, r_{1:t-1}) = 1$ if $\alpha^{(1)}_t > \alpha^{(2)}_t$, and
thus allow
for writing the conditional distribution of $d_t$ given $a_{1:t-1}$ and
$r_{1:t-1}$ as
\begin{equation}
p(d_t|a_{1:t-1}, r_{1:t-1}) = 
\left\llbracket\alpha^{(1)}_t \le \alpha^{(2)}_t\right\rrbracket^{1-d_t}
\left\llbracket\alpha^{(1)}_t > \alpha^{(2)}_t\right\rrbracket^{d_t}.
\end{equation}
\end{proof}

\begin{theorem}[Conditional log likelihood function]\label{thm:conditional-log-likelihood-function-2}
Let
\begin{multline}
p_\tau(s_1,d_{1:T},a_{1:T},r_{1:T})
\\
= p(s_1)p(d_1)p_\tau(a_1|d_1)p(r_1|s_1,a_1)
\prod_{t=2}^T p(d_t|a_{1:t-1},r_{1:t-1})p_\tau(a_t|d_t)p(r_t|s_1,a_t)
\end{multline}
denote the ABM induced by \eqnref{eq:T-2}, \eqnref{eq:A-p-2} and
\eqnref{eq:D-2}. Then,
\begin{equation}\label{eq:cll-2-1}
p_\tau(a_{1:T},r_{1:T})=cL(\tau),
\end{equation}
where
\begin{equation}
c:=\int_0^1 p(s_1)\prod_{t=1}^T p(r_t|s_1,a_t)\,d s_1
\end{equation}
is constant with respect to $\tau$ for the observed action and reward
sequences, and
\begin{equation}
L(\tau):=\prod_{t=1}^T\sum_{d_t}p(d_t|a_{1:t-1},r_{1:t-1})p_\tau(a_t|d_t).
\end{equation}
Moreover,
\begin{equation}\label{eq:cll-2-2}
\begin{aligned}
&\sum_{d_t}p(d_t|a_{1:t-1},r_{1:t-1})p_\tau(a_t|d_t)\\
&=\ib{\alpha_t^{(1)}\le\alpha_t^{(2)}}\mathrm{Bern}(a_t;\tau)+\ib{\alpha_t^{(1)}>\alpha_t^{(2)}}\mathrm{Bern}(a_t;1-\tau),
\end{aligned}
\end{equation}
where $\alpha_t^{(1)}$ and $\alpha_t^{(2)}$ are as defined in
\autoref{thm:belief-state-2}. Consequently, up to an additive term independent
of $\tau$, the log likelihood is
\begin{equation}
\begin{aligned}
\ell(\tau)&:=\ln L(\tau)\\
&=\fitproofline{\sum_{t=1}^T\ln\left(\ib{\alpha_t^{(1)}\le\alpha_t^{(2)}}\mathrm{Bern}(a_t;\tau)+\ib{\alpha_t^{(1)}>\alpha_t^{(2)}}\mathrm{Bern}(a_t;1-\tau)\right).}
\end{aligned}
\end{equation}
\end{theorem}

\begin{proof}
With respect to the factorization of the joint action--reward distribution,
we have
\begingroup
\small
\begin{align*}
&p_\tau(a_{1:T},r_{1:T})\\
&=\sum_{d_{1:T}}\int_0^1 p_\tau(s_1,d_{1:T},a_{1:T},r_{1:T})\,d s_1\\
&=\fitproofline{\sum_{d_{1:T}}\int_0^1 p(s_1)p(d_1)p_\tau(a_1|d_1)p(r_1|s_1,a_1)\prod_{t=2}^T p(d_t|a_{1:t-1},r_{1:t-1})p_\tau(a_t|d_t)p(r_t|s_1,a_t)\,d s_1}\\
&=\fitproofline{c\sum_{d_{1:T}}p(d_1)p_\tau(a_1|d_1)\prod_{t=2}^T p(d_t|a_{1:t-1},r_{1:t-1})p_\tau(a_t|d_t)}\\
&=\fitproofline{c\left(\sum_{d_1}p(d_1)p_\tau(a_1|d_1)\right)\prod_{t=2}^T\sum_{d_t}p(d_t|a_{1:t-1},r_{1:t-1})p_\tau(a_t|d_t)}\\
&=c\prod_{t=1}^T\sum_{d_t}p(d_t|a_{1:t-1},r_{1:t-1})p_\tau(a_t|d_t)\\
&=cL(\tau),
\end{align*}
\endgroup
where we have used the definition of $c$ in the third equality,
\autoref{thm:exchangeability} in the fourth equality, and the empty-history
convention
\begin{equation}
p(d_1|a_{1:0},r_{1:0})=p(d_1)
\end{equation}
in the fifth equality. Because the task distributions $p(s_1)$ and
$p(r_t|s_1,a_t)$ do not depend on $\tau$, the factor $c$ is independent of $\tau$.

Next, with \autoref{thm:conditional-decision-distribution-2} and
\eqnref{eq:D-2}, we have
\begin{align}
\begin{split}
&\sum_{d_t}p(d_t|a_{1:t-1},r_{1:t-1})p_\tau(a_t|d_t)\\
&=\fitproofline{\sum_{d_t}\ib{\alpha_t^{(1)}\le\alpha_t^{(2)}}^{1-d_t}\ib{\alpha_t^{(1)}>\alpha_t^{(2)}}^{d_t}\mathrm{Bern}(a_t;1-\tau)^{d_t}\mathrm{Bern}(a_t;\tau)^{1-d_t}}\\
&=\ib{\alpha_t^{(1)}\le\alpha_t^{(2)}}\mathrm{Bern}(a_t;\tau)+\ib{\alpha_t^{(1)}>\alpha_t^{(2)}}\mathrm{Bern}(a_t;1-\tau).
\end{split}
\end{align}
Taking logarithms therefore gives
\begin{equation}
\ln p_\tau(a_{1:T},r_{1:T})=\ln c+\ell(\tau),
\end{equation}
with
\begin{align}
\begin{split}
\ell(\tau)&=\sum_{t=1}^T\ln\left(\sum_{d_t}p(d_t|a_{1:t-1},r_{1:t-1})p_\tau(a_t|d_t)\right)\\
&=\fitproofline{\sum_{t=1}^T\ln\left(\ib{\alpha_t^{(1)}\le\alpha_t^{(2)}}\mathrm{Bern}(a_t;\tau)+\ib{\alpha_t^{(1)}>\alpha_t^{(2)}}\mathrm{Bern}(a_t;1-\tau)\right).}
\end{split}
\end{align}
Since $\ln c$ is independent of $\tau$, maximizing $\ell(\tau)$ is equivalent
to maximizing the joint log likelihood.
\end{proof}

\begin{figure}[!htbp]
\centering
\includegraphics[width=\linewidth]{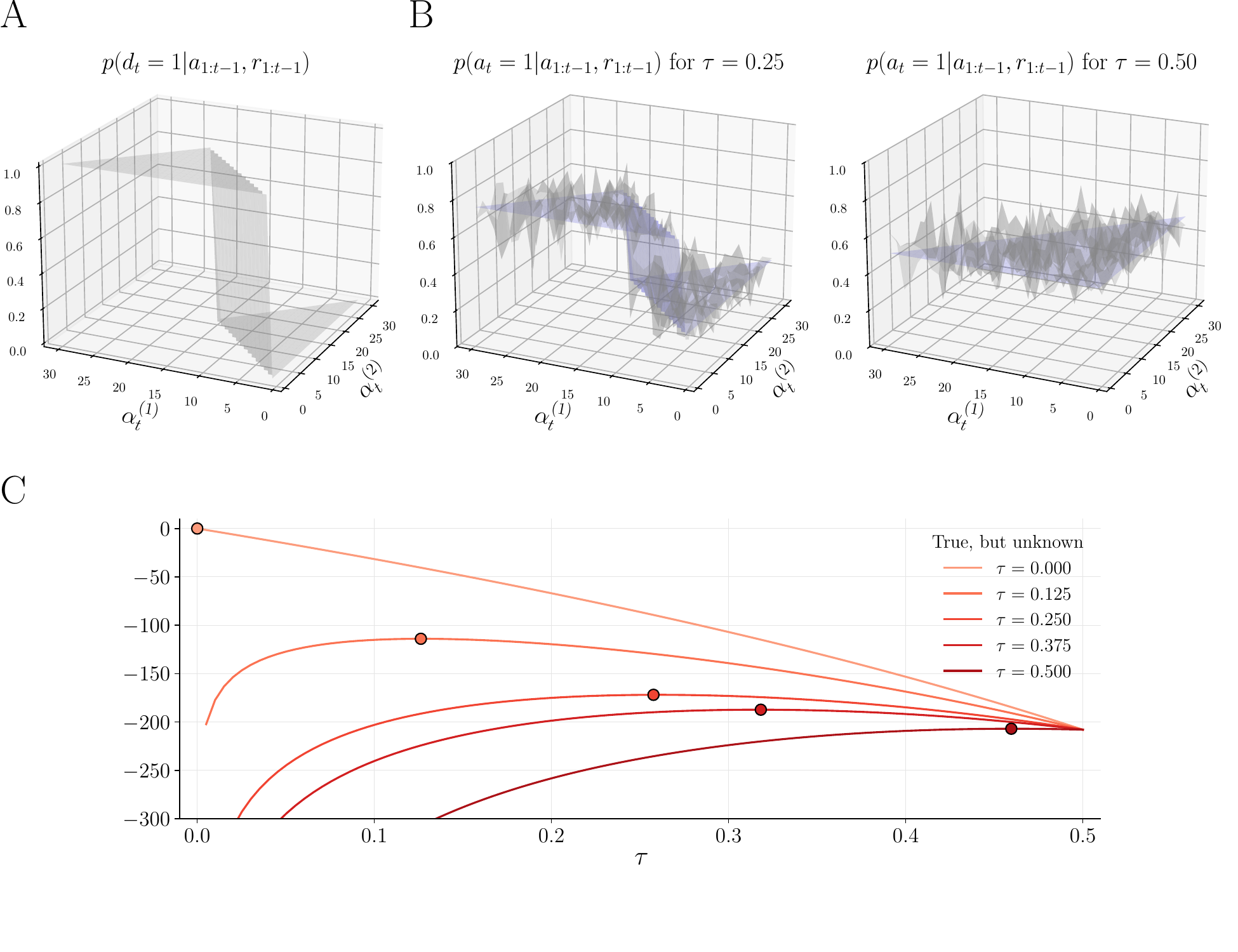}
\caption{Sampling-based validations of
\fullref{thm:conditional-decision-distribution-2} and
\fullref{thm:conditional-log-likelihood-function-2}. Panel A. Histogram-based
estimation of $p(d_t = 1|a_{1:t-1},r_{1:t-1})$, parameterized by
$\alpha_t=(\smash{\alpha_t^{(1)}},\smash{\alpha_t^{(2)}})$, as generated
during model validation and shown as a gray transparent surface. Because
$p(d_t|a_{1:t-1},r_{1:t-1})$ is a degenerate distribution, its estimate at
sampled parameter combinations is identical to the analytical result of
\fullref{thm:conditional-decision-distribution-2}. The experimental constraint
of 30 trials per block implies $\alpha_t^{(1)}+\alpha_t^{(2)}=t+1\le31$.
Parameter combinations exceeding this bound are therefore inadmissible and
are not associated with conditional decision probabilities. Panel B.
Histogram-based estimation of $p(a_t=1|a_{1:t-1},r_{1:t-1})$ based on simulated
joint observations of $a_t$ and $\alpha_t$ as generated during model validation
for two settings of the post-decision noise parameter $\tau$. Each subpanel
shows the histogram-based estimate of $p(a_t=1|a_{1:t-1},r_{1:t-1})$ as a gray
wrinkled surface, reflecting the noise in its estimation, and the analytical
result of $p(a_t=1|a_{1:t-1},r_{1:t-1})$ according to
\fullref{thm:conditional-log-likelihood-function-2} as a blue transparent
surface. As for Panel A, parameter combinations with
$\alpha_t^{(1)}+\alpha_t^{(2)}>31$ are inadmissible in the current experimental
context and thus not associated with a conditional action probability. Panel C.
Exemplary conditional log likelihood functions based on
\fullref{thm:conditional-log-likelihood-function-2}. Each curve uses a separate
simulated participant dataset comprising 10 blocks of 30 trials, generated
with the parameter value indicated by its shade of red and documented in the
legend. The dots mark the maxima over the evaluated parameter grid. The
profiles illustrate how support for candidate values of $\tau$ varies within
each dataset. Each plotted function differs from the joint log likelihood by
an additive constant that depends on the observed actions and rewards but not
on $\tau$. For full analytical details, please refer to
\textit{abm\_figure\_S8.py.}}\label{fig:S8}
\end{figure}

\begin{theorem}[Beta-Bernoulli and Rescorla–Wagner learning equivalence]\label{thm:bb-rw-equivalence}
Let $\phi$ and $\varphi$ denote the decision value functions of the
Beta-Bernoulli agent model \eqnref{eq:A-2} and the
adaptive learning rate Rescorla–Wagner agent model \eqnref{eq:A-2-equivalent}
for the symmetric bandit task. Then, for $t=1,\ldots,T$,
\begin{equation}
\phi(b_t) = \varphi(\hat{s}_t) \textnormal{ for all } a_{1:t-1} \textnormal{ and } r_{1:t-1}
\end{equation}
and thus, given identical action and reward observation histories, both make
identical decisions.
\end{theorem}

\begin{proof}
From \eqnref{eq:phi-2}, \eqnref{eq:bayesian-decision-values} and
\eqnref{eq:var-phi-2}, we first note that
\begin{equation}
\phi(b_t) = \varphi(\hat{s}_t) \textnormal{ for all } a_{1:t-1} \textnormal{ and } r_{1:t-1}
\end{equation}
is equivalent to
\begin{equation}\label{eq:posterior-expectation-equals-s-hat}
\mathbb{E}_{p(s_1|a_{1:t-1},r_{1:t-1})}(s_1) = \hat{s}_t \textnormal{ for all } a_{1:t-1} \textnormal{ and } r_{1:t-1}.
\end{equation}
It therefore suffices to show that the adaptive learning rate Rescorla–Wagner
update function \eqnref{eq:psi} produces the same posterior expected value
update as implied by the belief state update function of \eqnref{eq:beta-2}.
We show \eqnref{eq:posterior-expectation-equals-s-hat} by induction with respect
to $t$.

\vspace{1mm}
\noindent \textit{Base case.} Let $t = 1$. Then, with $\hat{s}_1 = 0.5$ by definition and from
\eqnref{eq:b_1},

\begin{equation}
\mathbb{E}_{p(s_1|a_{1:0},r_{1:0})}(s_1) = \frac{\alpha_1^{(1)}}{\alpha_1^{(1)} + \alpha_1^{(2)}} = \frac{1}{1+1} = 0.5.
\end{equation}

\vspace{1mm}

\noindent \textit{Induction step.} Assume that 
\begin{equation}
\hat{s}_{t-1} = \mathbb{E}_{p(s_1|a_{1:t-2},r_{1:t-2})}(s_1) = \frac{\alpha_{t-1}^{(1)}}{\alpha_{t-1}^{(1)} + \alpha_{t-1}^{(2)}}
\end{equation}
for some $t = 2,\ldots,T$ (the base case shows that this assumption is justified
for $t=2$). For ease of notation, define
\begin{equation}
a := \alpha_{t-1}^{(1)} \textnormal{ and } b := \alpha_{t-1}^{(2)},
\end{equation}
which gives
\begin{equation}
\hat{s}_{t-1} = \frac{a}{a+b} \textnormal{ and } t = a + b.
\end{equation}
We consider the cases
\begin{equation}
\tilde{r}_{t-1} = 0 \Leftrightarrow (a_{t-1} = 0 \land r_{t-1} = 1) \lor (a_{t-1} = 1 \land r_{t-1} = 0)
\end{equation}
and
\begin{equation}
\tilde{r}_{t-1} = 1 \Leftrightarrow (a_{t-1} = 0 \land r_{t-1} = 0) \lor (a_{t-1} = 1 \land r_{t-1} = 1)
\end{equation}
in turn.
\vspace{1mm}

\noindent (1) Let $\tilde{r}_{t-1} = 0$. Then 
\begin{align*}
\hat{s}_t 
& = \hat{s}_{t-1} + \frac{1}{t+1}\left(\tilde{r}_{t-1} - \hat{s}_{t-1}\right)               \\
& = \frac{a}{a+b} + \frac{1}{a+b+1}\left(0 - \frac{a}{a+b}\right)                           \\
& = \frac{a}{a+b} - \frac{a}{(a+b+1)(a+b)}                                                  \\
& = \frac{a(a+b+1)}{(a+b+1)(a+b)} - \frac{a}{(a+b+1)(a+b)}                                  \\
& = \frac{a(a+b+1) - a}{(a+b+1)(a+b)}                                                       \\
& = \frac{a^2+ab+a-a}{(a+b+1)(a+b)}                                                         \\
& = \frac{a^2+ab}{(a+b+1)(a+b)}                                                             \\
& = \frac{a(a+b)}{(a+b+1)(a+b)}                                                             \\
& = \frac{a}{a+b+1}                                                                         \\
& = \frac{\alpha_{t-1}^{(1)}}{\alpha_{t-1}^{(1)} + \alpha_{t-1}^{(2)} + 1 }                 \\
& = \mathbb{E}_{p(s_1|a_{1:t-1},r_{1:t-1})}(s_1),                                           \\
\end{align*}
where the last equality corresponds to \eqnref{eq:post-exp-r-tilde-0}.

\vspace{2mm}

\noindent (2) Let $\tilde{r}_{t-1} = 1$. Then 
\begin{align*}
\hat{s}_t 
& = \hat{s}_{t-1} + \frac{1}{t+1}\left(\tilde{r}_{t-1} - \hat{s}_{t-1}\right)               \\
& = \frac{a}{a+b} + \frac{1}{a+b+1}\left(1 - \frac{a}{a+b}\right)                           \\
& = \frac{a}{a+b} + \frac{1}{a+b+1} - \frac{a}{(a+b+1)(a+b)}                                \\
& = \frac{a(a+b+1)}{(a+b+1)(a+b)} + \frac{a+b}{(a+b+1)(a+b)} - \frac{a}{(a+b+1)(a+b)}       \\
& = \frac{a(a+b+1)+b}{(a+b+1)(a+b)}                                                         \\
& = \frac{a^2 + ab + a + b}{(a+b+1)(a+b)}                                                   \\
& = \frac{(a + b)(a + 1)}{(a+b+1)(a+b)}                                                     \\
& = \frac{a + 1}{a+1+b}                                                                     \\
& = \frac{\alpha_{t-1}^{(1)} + 1}{\alpha_{t-1}^{(1)} + 1  + \alpha_{t-1}^{(2)}}             \\
& = \mathbb{E}_{p(s_1|a_{1:t-1},r_{1:t-1})}(s_1),                                            \\
\end{align*}
where the last equality corresponds to \eqnref{eq:post-exp-r-tilde-1}.
Thus, \eqnref{eq:posterior-expectation-equals-s-hat} holds for all
$t=1,\ldots,T$. The agents therefore have identical decision values for every
action and reward history. Their shared decision rule consequently gives
identical decisions, including ties.
\end{proof}

We validated the equivalence of the agents' decision value functions numerically
in \textit{abm\_equivalence.py}, using an absolute tolerance of $10^{-5}$.

\FloatBarrier
\subsection{Rescorla–Wagner agent model analysis}\label{sec:rescorla-wagner-agent-model-analysis}

\begin{theorem}[Conditional decision distribution]\label{thm:conditional-decision-distribution-3}
Given the agent model $\mathcal{A}$ defined in \eqnref{eq:A-3}, the distribution
of an agent's decision
\begin{equation}
d_t=\delta(\varphi(\hat{s}_t)),\qquad t=1,\ldots,T,
\end{equation}
where $\hat{s}_1=0.5$ and
\begin{equation}
\hat{s}_t=\psi(\hat{s}_{t-1},a_{t-1},r_{t-1})
\quad\mbox{for }t=2,\ldots,T,
\end{equation}
given the preceding action and reward history can be written as
\begin{equation}
p(d_t|a_{1:t-1},r_{1:t-1}) = 
\left\llbracket \hat{s}_{t} \le 1 - \hat{s}_{t} \right\rrbracket^{1-d_t}
\left\llbracket \hat{s}_{t}  >  1 - \hat{s}_{t} \right\rrbracket^{d_t}, 
\end{equation}
where
\begin{equation}\label{eq:thm-conditional-decision-distribution-3-1}
\hat{s}_t = (1-\lambda)^{t-1}\hat{s}_1  + \lambda\sum_{j=1}^{t-1} (1-\lambda)^{t-1-j}\tilde{r}_j,\qquad\lambda\in[0,1].
\end{equation}
For $t=1$, the histories are empty and the sum is zero. We use the convention
$(1-\lambda)^0=1$, including at $\lambda=1$.
\end{theorem}

\begin{proof}
We first note that because $d_t$ is a function of $a_{1:t-1}$ and $r_{1:t-1}$,
$p(d_t|a_{1:t-1},r_{1:t-1})$ is degenerate, taking on the value $1$ if
$d_t=\delta(\varphi(\hat{s}_t))$ and $0$ otherwise. Next, because $d_t$ takes on values in $D = \{0,1\}$, it suffices to
consider $p(d_t = 0|a_{1:t-1},r_{1:t-1})$, as
$p(d_t = 1|a_{1:t-1},r_{1:t-1}) = 1 - p(d_t = 0|a_{1:t-1},r_{1:t-1})$ follows
immediately. We hence assume that $d_t = 0$ and show that this is equivalent to
$\hat{s}_{t} \le 1 - \hat{s}_{t}$:
\begin{align}
\begin{split}
d_t 
& = 
0
\\\Leftrightarrow
v_t^{0} 
& \ge 
v_t^{1}
\\\Leftrightarrow
1 - \hat{s}_{t}
& \ge 
\hat{s}_{t}
\\\Leftrightarrow
\hat{s}_{t}
& \le 
1 - \hat{s}_{t}.
\end{split}
\end{align}
The identities 
$p(d_t = 0|a_{1:t-1}, r_{1:t-1}) = 1$ if $\hat{s}_{t} \le 1 - \hat{s}_{t}$, and 
$p(d_t = 0|a_{1:t-1}, r_{1:t-1}) = 0$ if $\hat{s}_{t}  >  1 - \hat{s}_{t}$,
imply that
$p(d_t = 1|a_{1:t-1}, r_{1:t-1}) = 0$ if $\hat{s}_{t} \le 1 - \hat{s}_{t}$, and 
$p(d_t = 1|a_{1:t-1}, r_{1:t-1}) = 1$ if $\hat{s}_{t} > 1 - \hat{s}_{t}$, and
thus allow
for writing the conditional distribution of $d_t$ given $a_{1:t-1}$ and
$r_{1:t-1}$ as
\begin{equation}
p(d_t|a_{1:t-1},r_{1:t-1}) = 
\left\llbracket \hat{s}_{t} \le 1 - \hat{s}_{t} \right\rrbracket^{1-d_t}
\left\llbracket \hat{s}_{t}  >  1 - \hat{s}_{t} \right\rrbracket^{d_t}. 
\end{equation}
Finally, we validate the closed-form solution
\eqnref{eq:thm-conditional-decision-distribution-3-1}
of the probability state estimate $\hat{s}_{t}$ of the Rescorla–Wagner update
equation by means of induction
with respect to $t$.

\vspace{2mm}
\noindent \textit{Base case.} For $t=1$, the expression equals $\hat{s}_1=0.5$ by the empty-sum
and zero-exponent conventions. Let $t=2$. Then, according to
\eqnref{eq:rw-update},
\begin{align}\label{eq:rw-base-case}
\begin{split}
\hat{s}_{2} 
& = \hat{s}_1 + \lambda(\tilde{r}_1 - \hat{s}_1)                                            \\
& = \hat{s}_1 - \lambda\hat{s}_1 + \lambda\tilde{r}_1                                       \\
& = (1-\lambda)\hat{s}_1 + \lambda(1-\lambda)^0\tilde{r}_1                                  \\
& = (1-\lambda)\hat{s}_1 + \lambda(1-\lambda)^{2-1-1}\tilde{r}_1                            \\
& = (1-\lambda)\hat{s}_1 + \lambda\sum_{j=1}^1 (1-\lambda)^{2-1-j}\tilde{r}_j               \\
& = (1-\lambda)^{2-1}\hat{s}_1 + \lambda\sum_{j=1}^{2-1} (1-\lambda)^{2-1-j}\tilde{r}_j     \\
& = (1-\lambda)^{t-1}\hat{s}_1 + \lambda\sum_{j=1}^{t-1} (1-\lambda)^{t-1-j}\tilde{r}_j\quad\mbox{for }t=2.
\end{split}
\end{align}

\vspace{2mm}
\noindent \textit{Induction step.} Assume that 
\begin{equation}
\hat{s}_{t-1} = (1-\lambda)^{t-2}\hat{s}_1  + \lambda\sum_{j=1}^{t-2} (1-\lambda)^{t-2-j}\tilde{r}_j
\end{equation}
holds for some $t = 3,...,T$. The base-case calculation in
\eqnref{eq:rw-base-case} establishes this assumption for $t=3$.
Then, according to \eqnref{eq:rw-update},
\begin{align}
\begin{split}
\hat{s}_{t} 
& = \fitproofline{\hat{s}_{t-1} + \lambda(\tilde{r}_{t-1} - \hat{s}_{t-1})} \\
& = \fitproofline{(1-\lambda)\hat{s}_{t-1} + \lambda\tilde{r}_{t-1}} \\
& = \fitproofline{(1-\lambda)\left((1-\lambda)^{t-2}\hat{s}_1  + \lambda\sum_{j=1}^{t-2} (1-\lambda)^{t-2-j}\tilde{r}_j \right) + \lambda\tilde{r}_{t-1}} \\
& = \fitproofline{(1-\lambda)(1-\lambda)^{t-2}\hat{s}_1  +  (1-\lambda)\lambda\sum_{j=1}^{t-2} (1-\lambda)^{t-2-j}\tilde{r}_j + \lambda\tilde{r}_{t-1}} \\
& = \fitproofline{(1-\lambda)^{t-1}\hat{s}_1  +  \lambda\sum_{j=1}^{t-2} (1-\lambda)(1-\lambda)^{t-2-j}\tilde{r}_j + (1-\lambda)^{0}\lambda\tilde{r}_{t-1}} \\
& = \fitproofline{(1-\lambda)^{t-1}\hat{s}_1  +  \lambda\sum_{j=1}^{t-2} (1-\lambda)^{t-1-j}\tilde{r}_j + (1-\lambda)^{t-1-(t-1)}\lambda\tilde{r}_{t-1}} \\
& = \fitproofline{(1-\lambda)^{t-1}\hat{s}_1  +  \lambda\sum_{j=1}^{t-1} (1-\lambda)^{t-1-j}\tilde{r}_j.} \\
\end{split}
\end{align}
Thus, the closed-form solution holds for all $t=1,\ldots,T$.
\end{proof}

\begin{figure}[!htbp]
\center
\includegraphics[width=\linewidth]{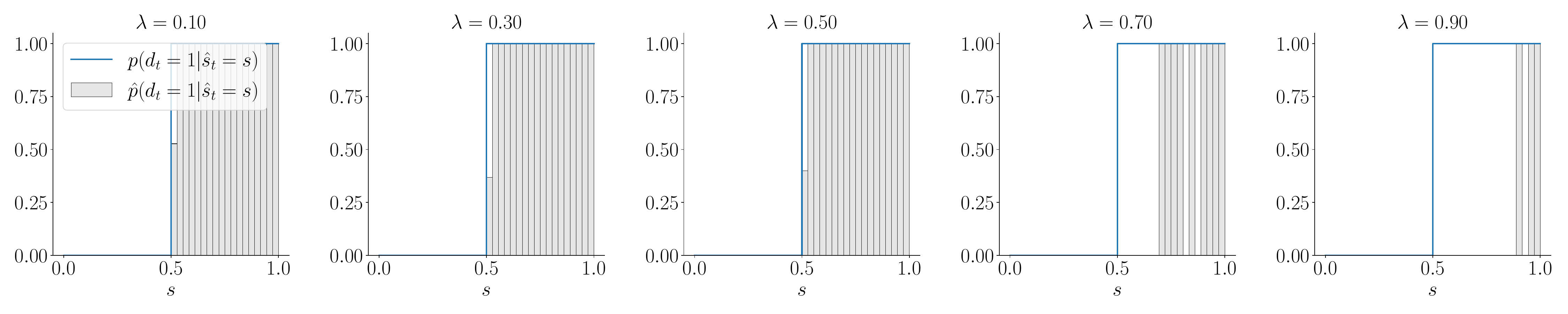}
\caption{Sampling-based validation of
\fullref{thm:conditional-decision-distribution-3}. The figure shows the
histogram-based estimation of $p(d_t=1|\hat{s}_t=s)$ based on simulated joint
observations of $d_t$ and $\hat{s}_t$ as generated during model validation
($\hat{p}(d_t=1|\hat{s}_t=s)$, gray bars) and the analytical result of
\fullref{thm:conditional-decision-distribution-3} (blue lines). For varying
levels of $\lambda$, a good match is observed. For large learning rates, the
simulated probability state estimates concentrate near $0$ or $1$.
Intermediate bins may therefore contain no observations, so no empirical
conditional decision probability can be estimated there. For full analytical
details, please refer to \textit{abm\_figure\_S9.py}.}\label{fig:S9}
\end{figure}

\newpage
\begin{theorem}[Conditional log likelihood function]\label{thm:conditional-log-likelihood-function-3}
Let
\begin{multline}
p_{\lambda,\tau}(s_1,d_{1:T},a_{1:T},r_{1:T})
\\ = p(s_1)p_\lambda(d_1)p_\tau(a_1|d_1)p(r_1|s_1,a_1)\prod_{t=2}^Tp_\lambda(d_t|a_{1:t-1},r_{1:t-1})p_\tau(a_t|d_t)p(r_t|s_1,a_t)
\end{multline}
denote the ABM induced by \eqnref{eq:T-2}, \eqnref{eq:A-p-3}, and
\eqnref{eq:D-2}, with $\lambda\in[0,1]$ and $\tau\in[0,0.5]$. Then,
\begin{equation}\label{eq:cll-3-1}
p_{\lambda,\tau}(a_{1:T},r_{1:T})=cL(\lambda,\tau),
\end{equation}
where
\begin{equation}
c:=\int_0^1p(s_1)\prod_{t=1}^Tp(r_t|s_1,a_t)\,d s_1
\end{equation}
is independent of $(\lambda,\tau)$ for fixed observed actions and rewards, and
\begin{equation}
L(\lambda,\tau):=\prod_{t=1}^T\sum_{d_t}p_\lambda(d_t|a_{1:t-1},r_{1:t-1})p_\tau(a_t|d_t).
\end{equation}
Moreover,
\begin{multline}\label{eq:cll-3-2}
\sum_{d_t}p_\lambda(d_t|a_{1:t-1},r_{1:t-1})p_\tau(a_t|d_t)
\\ = \ib{\hat{s}_t(\lambda)\le 1-\hat{s}_t(\lambda)}\mathrm{Bern}(a_t;\tau)+\ib{\hat{s}_t(\lambda)>1-\hat{s}_t(\lambda)}\mathrm{Bern}(a_t;1-\tau),
\end{multline}
with $\hat{s}_t(\lambda)$ as in
\eqnref{eq:thm-conditional-decision-distribution-3-1} and $\hat{s}_1=0.5$.
The conditional log likelihood function is thus
\begin{align}
\begin{split}
\ell(\lambda,\tau)&:=\ln L(\lambda,\tau)\\
&=\fitproofline{\sum_{t=1}^T\ln\left(\ib{\hat{s}_t(\lambda)\le 1-\hat{s}_t(\lambda)}\mathrm{Bern}(a_t;\tau)+\ib{\hat{s}_t(\lambda)>1-\hat{s}_t(\lambda)}\mathrm{Bern}(a_t;1-\tau)\right).}
\end{split}
\end{align}
It equals the joint log likelihood up to an additive constant independent of
$(\lambda,\tau)$.
\end{theorem}

\begin{proof}
Given the same factorization structure of the ABMs in
\fullref{thm:conditional-log-likelihood-function-2} and the current theorem,
the joint action--reward distribution factorizes analogously as
\begin{align}
\begin{split}
p_{\lambda,\tau}(a_{1:T},r_{1:T})
&=c\prod_{t=1}^T\sum_{d_t}p_\lambda(d_t|a_{1:t-1},r_{1:t-1})p_\tau(a_t|d_t)\\
&=cL(\lambda,\tau).
\end{split}
\end{align}
Integrating out the task state produces the same factor $c$, because the task
distributions $p(s_1)$ and $p(r_t|s_1,a_t)$ do not depend on $\lambda$ or
$\tau$. With respect to \eqnref{eq:cll-3-2}, we have with
\autoref{thm:conditional-decision-distribution-3} and \eqnref{eq:p-at-giv-dt-2}
\begin{align}
\begin{split}
&\sum_{d_t}p_\lambda(d_t|a_{1:t-1},r_{1:t-1})p_\tau(a_t|d_t)\\
&=\fitproofline{\sum_{d_t}\ib{\hat{s}_t(\lambda)\le 1-\hat{s}_t(\lambda)}^{1-d_t}\ib{\hat{s}_t(\lambda)>1-\hat{s}_t(\lambda)}^{d_t}\mathrm{Bern}(a_t;1-\tau)^{d_t}\mathrm{Bern}(a_t;\tau)^{1-d_t}}\\
&=\fitproofline{\ib{\hat{s}_t(\lambda)\le 1-\hat{s}_t(\lambda)}\mathrm{Bern}(a_t;\tau)+\ib{\hat{s}_t(\lambda)>1-\hat{s}_t(\lambda)}\mathrm{Bern}(a_t;1-\tau).}
\end{split}
\end{align}
Taking logarithms gives
\begin{equation}
\ln p_{\lambda,\tau}(a_{1:T},r_{1:T})=\ln c+\ell(\lambda,\tau),
\end{equation}
where
\begin{align}
\begin{split}
\ell(\lambda,\tau)
&=\sum_{t=1}^T\ln\left(\sum_{d_t}p_\lambda(d_t|a_{1:t-1},r_{1:t-1})p_\tau(a_t|d_t)\right)\\
&=\fitproofline{\sum_{t=1}^T\ln\left(\ib{\hat{s}_t(\lambda)\le 1-\hat{s}_t(\lambda)}\mathrm{Bern}(a_t;\tau)+\ib{\hat{s}_t(\lambda)>1-\hat{s}_t(\lambda)}\mathrm{Bern}(a_t;1-\tau)\right).}
\end{split}
\end{align}
Since $\ln c$ is independent of $(\lambda,\tau)$, maximizing
$\ell(\lambda,\tau)$ is equivalent to maximizing the joint log likelihood.
\end{proof}

\begin{figure}[!htbp]
\centering
\includegraphics[width=\linewidth]{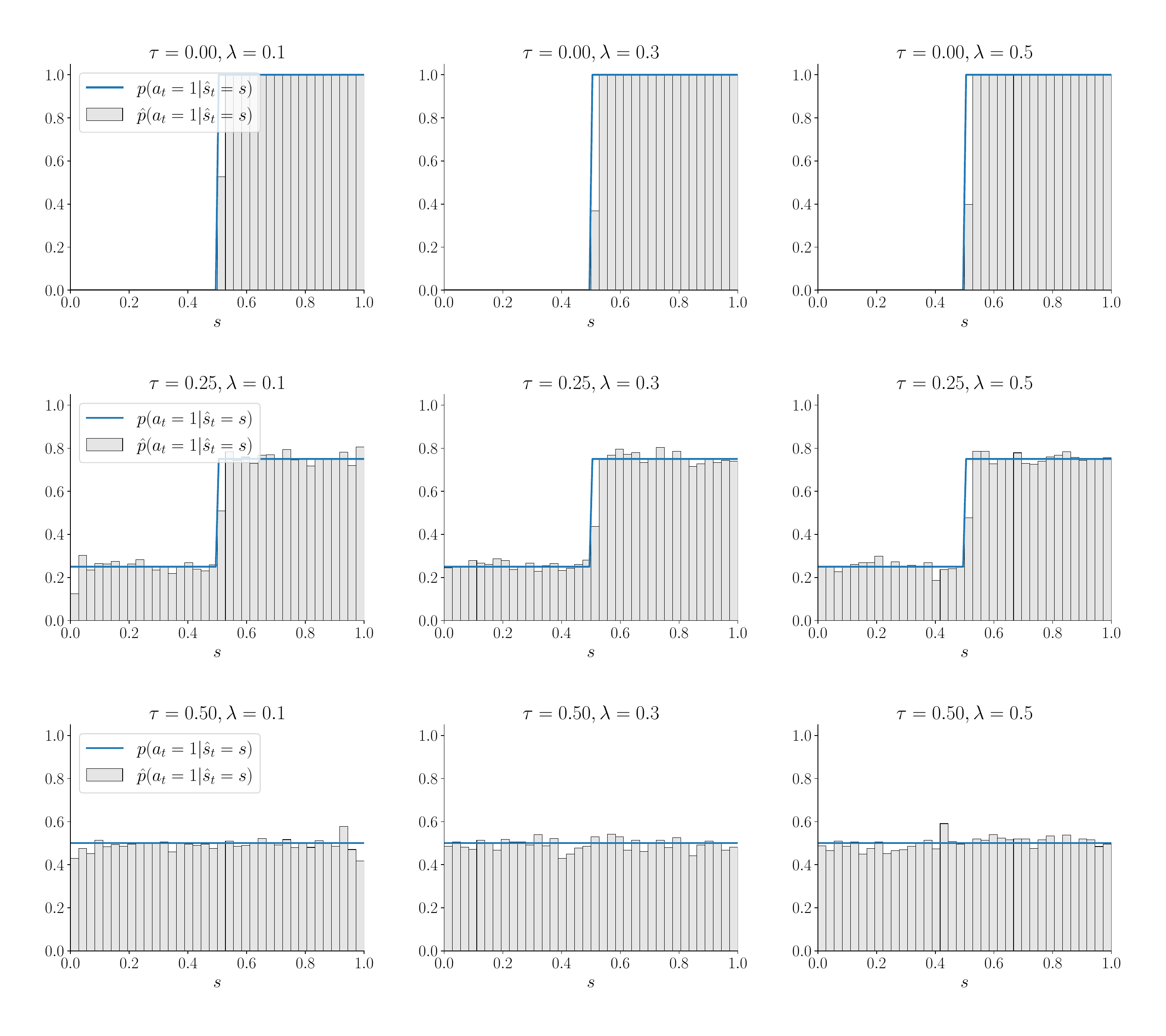}
\caption{Sampling-based validation of the trialwise action probabilities in
\fullref{thm:conditional-log-likelihood-function-3}. The figure shows
histogram-based estimates of $p(a_t=1|\hat{s}_t=s)$ based on simulated joint
observations of $a_t$ and $\hat{s}_t$ generated during model validation
($\hat{p}(a_t=1|\hat{s}_t=s)$, gray bars), together with the analytical
probabilities given in \fullref{thm:conditional-log-likelihood-function-3}
(blue lines). The learning rate parameter $\lambda$ varies between columns.
Relatively low learning rates are used to obtain observations across a broad
range of expected reward estimates. The post-decision noise parameter $\tau$
varies between rows. Bins without observations provide no empirical
conditional probability estimate. For full analytical details, please refer
to \textit{abm\_figure\_S10.py}.}\label{fig:S10}
\end{figure}

\clearpage
\vspace*{\fill}
\begin{center}
\includegraphics[width=\linewidth]{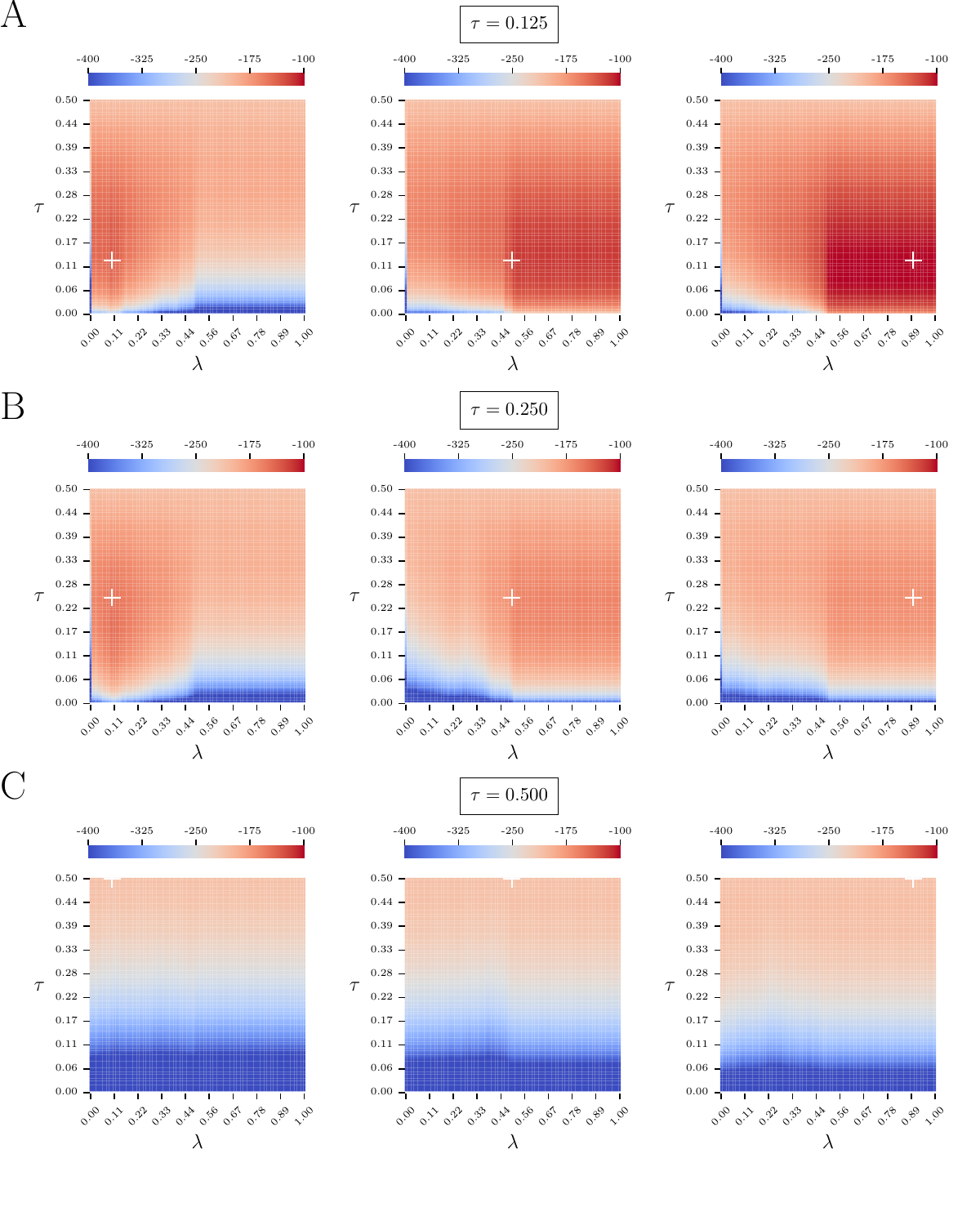}
\end{center}
\clearpage
    
\clearpage
\noindent\rule{\linewidth}{0.4pt}  
\captionof{figure}{\textit{(Figure on previous page.)} Exemplary conditional
log likelihood functions of the form of \eqnref{eq:llh-3}, based on
\fullref{thm:conditional-log-likelihood-function-3}. Each surface uses the
observed actions and rewards from a separate simulated participant dataset
comprising 10 blocks of 30 trials. These functions form the basis for
parameter estimation and marginal likelihood approximation (see
\autoref{sec:parameter-estimation-and-model-comparison}). Each plotted
function equals the joint log likelihood up to an additive constant that
depends on the observed actions and rewards but not on $(\lambda,\tau)$.
The true but unknown post-decision noise parameter value is given above each
row. White crosses indicate the grid points nearest the true but unknown
parameter combinations on the respective log likelihood surfaces. The color
coding of the log likelihood values is identical across all plots.
Panel A depicts a low post-decision noise scenario with $\tau=0.125$. For a
low true but unknown learning rate of $\lambda=0.1$, the conditional
log likelihood surface provides information about the learning rate. For
candidate learning rates $\lambda>0.5$, however, the preceding trial's
adjusted reward determines the decision independently of the precise learning
rate. The conditional log likelihood is therefore constant in $\lambda$ over
this range, explaining the non-identifiability observed in the parameter
recovery results (cf. \fullref{fig:8}D, left subpanel).
For the higher post-decision noise scenario in Panel B ($\tau=0.25$), the
surfaces show less pronounced variation with the learning rate. Finally,
Panel C shows datasets generated with $\tau=0.5$, for which actions are
independent of the agent's decisions and therefore contain no information
about the generative learning rate. At the candidate value $\tau=0.5$, the
conditional log likelihood is exactly $-300\ln 2$ for every $\lambda$. At other
candidate noise values, finite-sample variation can still produce differences
across learning rates. This lack of information about the generative learning
rate hampers its recovery (cf. \fullref{fig:8}D, left subpanel). For full
analytical details, please refer to \textit{abm\_figure\_S11.py}.}\label{fig:S11}

\clearpage
\section{Protected exceedance probabilities}\label{sec:protected-exceedance-probabilities}
\suppressfloats[t]
\subsection{Group-level model inference}

The Bayesian model selection approach for group studies from which protected
exceedance probabilities (PEPs) are derived was originally proposed by
\citet{stephan2009} and refined by \citet{rigoux2014}. We refer to this
framework as \textit{group-level model inference} (GMI), and in what follows, we
present a unified documentation of the framework along with its derivation. To
this end, we first formulate the GMI probabilistic model. Then, we discuss how
PEPs are defined under the GMI framework. Finally, we address how PEPs are
estimated based on a group of behavioral
datasets using a free-form variational inference approach. We briefly review the
theoretical background of this free-form variational inference approach in
\fullref{sec:free-form-mean-field-variational-inference}. Crucially, for
group-level model inference, the GMI framework only requires log model evidences
from individual participant datasets. We discuss our chosen approach to
estimating log model evidences in terms of the Bayesian information criterion
(BIC) in \fullref{sec:bic}. Throughout, we defer proofs to \fullref{sec:proofs}.
\vspace{2mm}

In terms of notation, we denote probabilities by $\mathbb{P}$ and probability
mass/density functions (PMFs/PDFs) by standard lowercase letters such as $p$ or
$q$. We denote the PDF of the Dirichlet distribution of an $m$-dimensional random
vector $\pi$ taking on values in the $(m-1)$-dimensional probability simplex and
with parameter vector $\alpha \in \mathbb{R}^{m}_{>0}$ by
$\mbox{Dir}(\pi;\alpha)$, and we denote the PMF of a random $m$-dimensional
vector $\mu$ taking on values in the set of $m$-dimensional canonical unit
vectors $E_m := \{e_1,...,e_m\}$ and with parameter vector $\pi$ by
$\mbox{Mult}(\mu;\pi)$ (cf. \autoref{tab:pep-distributions}).
In the multinoulli PMF, we use the convention $0^0 := 1$.

\begin{table}[t]
\caption{Parametric distributions used in group-level model inference.
Note that $\delta_{x_0}$ denotes a probability measure
concentrated at $x_0$. For integrable functions $f$, integration with respect
to this measure satisfies $\int f(x)\,\delta_{x_0}(dx)=f(x_0)$.}
\label{tab:pep-distributions}
\centering
\small
\renewcommand{\arraystretch}{1.8}
\begin{tabularx}{\linewidth}{@{}l X@{}}
\toprule
Multinoulli PMF
& $\mbox{Mult}(x;\pi) := \prod_{i=1}^m \pi_i^{x_i}$  for  $x \in \{e_1,...,e_m\}$, $\pi \in [0,1]^m, \sum_{i=1}^m \pi_i = 1$ 
\\ \midrule
Dirichlet PDF 
& $\mbox{Dir}(x;\alpha) := \frac{\Gamma\left(\sum_{i=1}^m \alpha_i\right)}{\prod_{i=1}^m\Gamma\left(\alpha_i\right)}\prod_{i=1}^m x_i^{\alpha_i - 1}$
for 
$x_i > 0, \sum_{i=1}^m x_i = 1, \alpha \in \mathbb{R}^{m}_{>0}$
\\ \midrule
Dirac point mass
& $\delta_{x_0}(A):=\ib{x_0\in A}$ for measurable sets $A$.
\\ \bottomrule
\end{tabularx}

\end{table}

\subsubsection*{Model formulation}\label{sec:model-formulation}

To formalize the GMI model, let $i = 1,...,n$ index participants and let
$j = 1,...,m$ index behavioral models of the model space, such that 
$n$ denotes the total number of participants and $m$ denotes the total 
number of models in the model space. The GMI model is composed of the following
random entities:
\begin{itemize}

\item A \textit{hypothesis} $\eta$ taking on values in $\{0,1\}$, for which
$\eta = 0$
encodes the \textit{null hypothesis} of all equal model probabilities, and
$\eta = 1$ encodes the \textit{alternative hypothesis} of not all equal model
probabilities.
From an epistemic viewpoint, it is assumed that there exists a true but unknown
hypothesis
value in the world, and the probabilities assigned to $\eta$ encode 
our uncertainty about this value.

\item A vector of \textit{model probabilities} $\pi$ taking on values in
$[0,1]^{m}$ with $\sum_{j=1}^m \pi_j = 1$.
$\pi = (\pi_1,...,\pi_m)^T$ is thus an $m$-dimensional vector with components 
$\pi_j \in [0,1]$ for $j = 1,...,m$ summing to one, such as
$\pi = (0.2,0.1,0.3,0.4)^T$
for $m = 4$. The $j$th entry $\pi_j$ encodes the probability of model $j$.
Again, from an epistemic
viewpoint, it is assumed that there exist true but unknown model probability 
values in the world, and the probabilities assigned to $\pi$ encode our 
uncertainty about these values.

\item A set of \textit{unobserved participant model indicators} $\mu_i$ taking
on values in $\{0,1\}^{m}$ with $\sum_{j=1}^m \mu_{i_j} = 1$ for $i = 1,...,n$.
$\mu_i$ is thus an $m$-dimensional canonical unit vector (sometimes
referred to as a \textit{one-hot vector} or a \textit{1-of-K vector}) comprising
all $0$s except
for a single $1$, the index of which indicates the behavioral model index that 
generates the behavioral data of participant $i$. We 
denote the set of $m$-dimensional canonical unit vectors, i.e., the outcome set 
of the $\mu_i$'s, by $E_m := \{e_1,...,e_m\}$. For example, for $m = 4$,
$\mu_{16} = e_3 = (0,0,1,0)^T$ encodes that the behavioral data of participant 
$i = 16$ are generated by model $j = 3$. We denote the set of all unobserved
participant
model indicators by $\mu_{1:n} := \{\mu_1,...,\mu_n\}$.

\item The \textit{observed participant behavioral datasets} $y_i$ for
$i = 1,...,n$. No particular
assumptions are made about the form of the $y_i$ and their outcome sets. As an
example, $y_i$ could
be the vector of all actions taken by a participant in a sequential learning
task or a neuroimaging data time series. The probabilities assigned to the $y_i$
encode our uncertainty about the model-based
data-generating process, which is sometimes referred to as \textit{observation
error}.
We denote the set of all observed participant behavioral datasets by
$y_{1:n} := \{y_1,...,y_n\}$.
\end{itemize}
\noindent The GMI model is given by the joint distribution of
$\eta,\pi,\mu_{1:n}$ and $y_{1:n}$,
which is assumed to factorize as follows:
\begin{equation}\label{eq:gmi}
p(\eta,\pi,\mu_{1:n},y_{1:n}) = p(\eta)p(\pi|\eta)\prod_{i=1}^{n}p(\mu_i|\pi)p(y_i|\mu_i).
\end{equation}
We unpack the components of \eqnref{eq:gmi} below:

\begin{itemize}
\item The marginal PMF of the hypothesis is given by 
\begin{equation}\label{eq:p-eta}
p(\eta) := \frac{1}{2} \mbox{ for } \eta = 0 \mbox{ and } \eta = 1.
\end{equation}
From a Bayesian perspective, \eqnref{eq:p-eta} encodes our prior probability that
the null hypothesis of all equal model probabilities holds. This probability
(and its complementary probability, the probability that the alternative hypothesis 
of not all equal model probabilities holds) is assumed to take on the value 
$1/2$, encoding maximal prior uncertainty regarding the hypothesis state.

\item The conditional distribution of the model probabilities under the null
hypothesis is given by
\begin{equation}\label{eq:p-eta-0}
\pi|\eta=0\sim\delta_{\pi_0},
\qquad
\pi_0:=\left(\frac{1}{m},\ldots,\frac{1}{m}\right)^T,
\end{equation}
where $\delta_{\pi_0}$ denotes the Dirac point mass at $\pi_0$. Under the
alternative hypothesis, the conditional PDF is given by
\begin{equation}\label{eq:p-eta-1}
p(\pi|\eta = 1)
:= \mbox{Dir}(\pi; \alpha)
:= \frac{\Gamma\left(\sum_{j = 1}^{m} \alpha_j \right)}{\prod_{j = 1}^{m} \Gamma(\alpha_j)}\prod_{j=1}^{m} \pi_j^{\alpha_j - 1}
\mbox{ for }
\alpha \in \mathbb{R}^m_{>0}.
\end{equation}
For $\eta=0$, the conditional distribution assigns probability $1$ to
$\pi=\pi_0$ and encodes the absence of uncertainty about the model
probabilities $\pi$. For $\eta=1$, the conditional PDF is of Dirichlet form
and encodes our uncertainty about the behavioral model probabilities $\pi$
as a function of the prior parameter vector $\alpha$. For the common choice
of $\alpha:=1_m$, where $1_m$ denotes the $m$-dimensional vector of ones,
this uncertainty is maximal in the sense of a continuous uniform distribution
on the $(m-1)$-dimensional probability simplex.

\item The conditional PMF of the unobserved participant model indicators is
given by the Multinoulli distribution
\begin{equation}\label{eq:p-mu-i-pi}
p(\mu_i|\pi) := \mbox{Mult}(\mu_i;\pi) := \prod_{j=1}^{m} \pi_j^{\mu_{i_j}}.
\end{equation}
$p(\mu_i|\pi)$ thus assigns the probability $\pi_j$ to the $j$th canonical unit
vector.
For example, for $m = 4$ and $\mu_i = (0,1,0,0)^T$, we have
\begin{equation}\label{eq:p-mu-i-pi-example}
p\left((0,1,0,0)^T|\pi\right) = \pi_1^0 \cdot \pi_2^1 \cdot \pi_3^0 \cdot \pi_4^0 = 1 \cdot \pi_2 \cdot 1 \cdot 1 = \pi_2. 
\end{equation}

\item Like the precise form of the $y_i$, $i = 1,...,n$, the conditional PMF/PDF of 
the observed participant data $y_i$ is left unspecified but is assumed to exist.
\end{itemize}

\noindent The factorization properties of the joint distribution of
$\eta,\pi,\mu_{1:n}, y_{1:n}$
given by \eqnref{eq:gmi} encode two important conditional independence
properties. First,
given the model probabilities $\pi$, the unobserved model indicators
$\mu_1,...,\mu_n$
are identically and independently distributed. Conditional on $\pi$, the model indicator value of the $i$th participant is
thus assumed to have no influence on the model indicator of another participant.
Likewise, conditional on the participant model indicators $\mu_{1:n}$, the
observed behavioral datasets $y_1,...,y_n$ are independently distributed, with
each $y_i$ distributed according to $p(y_i|\mu_i)$, again encoding the absence
of influences between participants.

\subsubsection*{Protected exceedance probabilities}

The core quantity of interest in the GMI framework is the \textit{protected
exceedance probability of model $j$}. Following \citet{rigoux2014}, this
combines the exceedance probability conditional on the alternative hypothesis
with a uniform contribution under the null hypothesis:
\begin{align}\label{eq:varphi-j}
\begin{split}
\varphi_j
:=\frac{1}{m}\mathbb{P}(\eta=0|y_{1:n})
+\mathbb{P}\left(\pi_j>\max_{k\ne j}\pi_k|\eta=1,y_{1:n}\right)\left(1-\mathbb{P}(\eta=0|y_{1:n})\right).
\end{split}
\end{align}
The null contribution $1/m$ follows the limiting convention of
\citet[Eq.~(7) and footnote~2]{rigoux2014}. Under a symmetric Dirichlet prior,
the posterior exceedance probabilities tend to $1/m$ as the prior
concentration tends to infinity for fixed observed data. This convention
distinguishes the limiting exceedance probabilities from the probability of
attaining the maximum under the exact Dirac point mass, where all model
probabilities are equal.

Under the alternative hypothesis, ties have probability zero because the
posterior distribution of $\pi$ is absolutely continuous on the probability
simplex. Consequently, as shown in \fullref{sec:proofs}, PEPs sum to one over
all models in the model space:
\begin{equation}\label{eq:varphi-normalization}
\sum_{j=1}^m\varphi_j=1.
\end{equation}
The PEP of model $j$ thus combines its
exceedance probability under the alternative hypothesis with the uniform
null contribution, weighted by the posterior probabilities of the two
hypotheses. Under the equal hypothesis priors specified above, and assuming
finite, strictly positive marginal likelihoods under both hypotheses, the
posterior probability of the null hypothesis can be written as
\begin{equation}\label{eq:p-eta-giv-y}
\mathbb{P}(\eta=0|y_{1:n})
=\frac{1}{1+\exp\left(\ln p(y_{1:n}|\eta=1)-\ln p(y_{1:n}|\eta=0)\right)},
\end{equation}
as shown in \fullref{sec:proofs}.
Together, \eqnref{eq:varphi-j} and \eqnref{eq:p-eta-giv-y} set the stage for
inference in the GMI model. First, to evaluate the exceedance probability of
model $j$ conditional on the alternative hypothesis and the observed data,
it is necessary to evaluate the conditional PDF $p(\pi|\eta=1,y_{1:n})$.
Under the null hypothesis, the contribution is fixed at $1/m$ by the limiting
convention described above. Second, to evaluate the posterior probability of
the null hypothesis, it is necessary to evaluate the log marginal likelihoods
$\ln p(y_{1:n}|\eta=0)$ and $\ln p(y_{1:n}|\eta=1)$. Together, these two tasks
correspond to the fundamental aims of probabilistic inference reviewed in
\fullref{sec:free-form-mean-field-variational-inference}.

\subsubsection*{Inference}\label{sec:inference}

To infer the necessary quantities for evaluating the $j$th model's PEP given
by \eqnref{eq:varphi-j}, we follow
\citet{rigoux2014} and use a two-step inference approach. In the first step,
variational updates yield the exact participant-indicator posterior
distributions and log marginal likelihood conditional on the null hypothesis
$\eta=0$, with $\pi=\pi_0$ fixed. In the second step, variational inference is
used to approximate the posterior distributions and log marginal likelihood
conditional on the alternative hypothesis $\eta=1$. To simplify the
documentation of these steps, we define the reduced \textit{null
hypothesis-conditional GMI model}
\begin{equation}\label{eq:GMI-0}
p_0(\mu_{1:n},y_{1:n})
:=p(\mu_{1:n},y_{1:n}|\eta=0)
=\prod_{i=1}^n p(\mu_i|\pi_0)p(y_i|\mu_i),
\end{equation}
where $\pi_0=(1/m,...,1/m)^T$, and the
\textit{alternative hypothesis-conditional GMI model}
\begin{equation}\label{eq:GMI-1}
p_1(\pi,\mu_{1:n},y_{1:n}) := p(\pi,\mu_{1:n},y_{1:n}|\eta = 1)
= p_1(\pi)\prod_{i=1}^n p(\mu_i|\pi)p(y_i|\mu_i),
\end{equation}
where
\begin{equation}
p_1(\pi) := p(\pi|\eta = 1).
\end{equation}
The introduction of the null hypothesis-conditional GMI model is the core
contribution of \citet{rigoux2014}, while the inference approach for the
alternative hypothesis-conditional GMI was already established by
\citet{stephan2009}. In the following two subsections, we discuss the inference
approaches for \eqnref{eq:GMI-0} and \eqnref{eq:GMI-1} in turn.

\subsubsection*{Inference conditional on $\eta = 0$}

\citet{rigoux2014} derive variational updates for the null
hypothesis-conditional GMI model. Under this hypothesis, the model
probabilities are fixed at $\pi_0=(1/m,...,1/m)^T$, so inference concerns the
participant model indicators only. The resulting model is
\begin{equation}\label{eq:p0-inf-1}
p_0(\mu_{1:n},y_{1:n})=\prod_{i=1}^n p(\mu_i|\pi_0)p(y_i|\mu_i).
\end{equation}
Throughout this subsection and its associated proofs, we assume that all
participant model evidences are finite and strictly positive, i.e.,
$0<p(y_i|\mu_i=e_j)<\infty$ for $i=1,...,n$ and $j=1,...,m$.
We consider the mean-field factorization
\begin{equation}\label{eq:p0-inf-2}
q_0(\mu_{1:n}):=\prod_{i=1}^n q_0(\mu_i).
\end{equation}
Under the null hypothesis, this factorization is exact. A single update
yields the posterior distributions of the participant model indicators:
\begin{equation}
p_0(\mu_{1:n}|y_{1:n})=\prod_{i=1}^n q_0^{(1)}(\mu_i).
\end{equation}
The evidence lower bound (ELBO, see
\fullref{sec:free-form-mean-field-variational-inference}) evaluated at these
distributions therefore equals the log marginal likelihood:
\begin{equation}\label{eq:p0-inf-32}
\ln p_0(y_{1:n})=\mathrm{ELBO}\left(q_0^{(1)}(\mu_{1:n})\right).
\end{equation}
These equalities hold for the supplied participant model evidences.
Approximating those evidences using BIC constitutes a separate approximation.
Two things are noteworthy here.
\begin{itemize}
\item Under the null hypothesis, the posterior distribution of $\pi$ remains
the Dirac point mass at $\pi_0$:
\begin{equation}
\pi|\eta=0,y_{1:n}\sim\delta_{\pi_0}.
\end{equation}
The null contribution to the PEP is $1/m$,
following the symmetric Dirichlet limiting convention of
\citet[Eq.~(7) and footnote~2]{rigoux2014} described above.
\item As no variational update of $\pi$ is required, an algorithm alternating
between updates of the distributions on $\mu_{1:n}$ and $\pi$ is unnecessary.
The participant-indicator updates do not depend on one another and are
completed in a single step. Likewise, the ELBO is a function of the
$q_0(\mu_i)$ only.
\end{itemize}

\paragraph{Variational distribution updates for $q_0(\mu_i)$}

As shown in \fullref{sec:proofs}, the single variational update for each
$\mu_i$, $i=1,...,n$, takes the form

\begin{equation}
q^{(1)}_0(\mu_i) = \mbox{Mult}\left(\mu_i; \pi^{(1)}_i\right),
\end{equation}
where
\begin{equation}\label{eq:q0-mu-i}
\pi_i^{(1)} := \left(\pi_{ij}^{(1)}\right)_{j = 1,...,m}
\mbox{ with }
\pi_{ij}^{(1)} := 
\frac{\exp\left(\ln p(y_i|\mu_i = e_j) \right)}
        {\sum_{\mu_i\in E_m}\exp\left(\ln p(y_i|\mu_i) \right)}.
\end{equation}
The updated variational distribution of $\mu_i$ is thus given by a Multinoulli
distribution with a parameter vector depending on the participant-specific log
model evidence conditional on the $j$th model, $\ln p(y_i|\mu_i = e_j)$.
Participant-specific log model evidences can be approximated using BIC scores
(cf. \autoref{sec:bic}). This provides the link between participant- and
model-specific maximum-likelihood estimation and the evaluation of PEPs at
the group level.

\paragraph{Log marginal likelihood}
Further, as shown in \fullref{sec:proofs}, the ELBO in the current scenario
evaluates to
\begin{equation}\label{eq:elbo-p0}
\mbox{ELBO}\left(q_0^{(c)}(\mu_{1:n})\right)
= \sum_{i=1}^n \sum_{j=1}^m \pi_{ij}^{(c)} \left(\ln p(y_i|\mu_i = e_j)- \ln m- \ln \pi_{ij}^{(c)}\right),
\end{equation}
where $c=1$ denotes completion of the single update. We thus have
\begin{equation}\label{eq:p-0-y}
\ln p_0(y_{1:n}) = 
\mbox{ELBO}\left(q_0^{(c)}(\mu_{1:n})\right).
\end{equation}

\subsubsection*{Inference conditional on $\eta = 1$}

Throughout this subsection and its associated proofs, we assume that all
participant model evidences are finite and strictly positive, i.e.,
$0<p(y_i|\mu_i=e_j)<\infty$ for $i=1,...,n$ and $j=1,...,m$.

\citet{stephan2009} proposed a free-form variational inference algorithm for the
alternative
hypothesis-conditional GMI
\begin{equation}\label{eq:p1-inf-1}
p_1(\pi,\mu_{1:n},y_{1:n}) = p_1(\pi)\prod_{i=1}^n p(\mu_i|\pi)p(y_i|\mu_i)  
\end{equation}
based on the mean-field approximation
\begin{equation}\label{eq:p1-inf-2}
q_1(\pi,\mu_{1:n}) := q_1(\pi)\prod_{i=1}^n q_1(\mu_i)
\end{equation}
such that, upon the convergence of the algorithm,
\begin{equation}\label{eq:p1-inf-3}
p_1(\pi|y_{1:n}) \approx  q_1(\pi), 
\quad
p_1(\mu_{1:n}|y_{1:n}) \approx \prod_{i=1}^n q_1(\mu_i), 
\end{equation}
and
\begin{equation}\label{eq:p1-inf-4}
\ln p_1(y_{1:n}) \approx \mbox{ELBO}\left(q_1(\pi,\mu_{1:n})\right).
\end{equation}
In the following, we first state the resulting variational distribution update
equations that constitute a coordinate-ascent variational inference (CAVI)
algorithm (cf. \fullref{sec:free-form-mean-field-variational-inference}).
We then consider a Monte Carlo approach to estimating model exceedance
probabilities under the alternative hypothesis. Finally, we document the ELBO
approximation of the log marginal likelihood under this hypothesis.

\paragraph{Variational distribution updates for $q_1(\mu_i)$}

As shown in \fullref{sec:proofs}, for algorithm iterations $k = 0,1,2,...$, the 
variational distribution updates for $\mu_i$ with $i = 1,...,n$ take the form
\begin{equation}\label{eq:q1-mu-i}
q^{(k+1)}_1(\mu_i) = \mbox{Mult}\left(\mu_i;\pi^{(k+1)}_i\right),
\mbox{where }
\pi_i^{(k+1)} := \left(\pi_{ij}^{(k+1)}\right)_{j = 1,...,m}
\end{equation}
with
\begin{equation}
\pi_{ij}^{(k+1)} := 
\frac{\exp\left(\ln p(y_i|\mu_i = e_j) + \psi\left(\alpha_j^{(k)}\right) - \psi\left(\alpha_s^{(k)}\right)\right)}
        {\sum_{r=1}^m\exp\left(\ln p(y_i|\mu_i = e_r) + \psi\left(\alpha_r^{(k)}\right) - \psi\left(\alpha_s^{(k)}\right)\right)},
\end{equation}
where $\psi$ denotes the digamma function and
$\alpha_s^{(k)} := \sum_{j=1}^m \alpha_j^{(k)}$.
The $(k+1)$th variational distribution of $\mu_i$ is thus given by a Multinoulli
distribution with a parameter vector depending on the participant-specific log
model evidence conditional on the $j$th model, $\ln p(y_i|\mu_i = e_j)$ (again
providing
the link between  participant- and model-specific maximum-likelihood 
estimation and the evaluation of PEPs at the group level), and
the parameters of the variational distribution of $\pi$ on the $k$th iteration of the
algorithm. 

\paragraph{Variational distribution updates for $q_1(\pi)$}
As shown in \fullref{sec:proofs}, for algorithm iterations $k = 0,1,2,...$, the
variational
distribution updates for $\pi$ take the form
\begin{equation}\label{eq:q1-pi}
q^{(k+1)}_1(\pi) := \mbox{Dir}\left(\pi; \alpha^{(k+1)}\right)
\mbox{, where }
\alpha^{(k+1)}_j := \alpha_j + \beta_j^{(k+1)} 
\end{equation}
with
\begin{equation} 
\beta_j^{(k+1)} := \sum_{i=1}^n \pi_{ij}^{(k+1)}
\mbox{ for } j = 1,...,m. 
\end{equation}
The $(k+1)$th variational distribution of $\pi$ is thus a Dirichlet distribution
with a parameter vector depending on the prior parameters $\alpha_j$
(cf. \eqnref{eq:p-eta-1}) and the newly updated variational probabilities
$\pi_{ij}^{(k+1)}$ of the participant model indicators
(cf. \eqnref{eq:q1-mu-i}).

\vspace{2mm}
\noindent The variational updates of \eqnref{eq:q1-mu-i} and \eqnref{eq:q1-pi}
are iterated alternately
until the algorithm converges. We denote the iteration of convergence by $c$,
such that upon the convergence of the algorithm, we have
\begin{equation}
p_1(\pi|y_{1:n}) \approx  q^{(c)}_1(\pi) = \mbox{Dir}\left(\pi; \alpha^{(c)}\right)
\end{equation}
and
\begin{equation}
p_1(\mu_{1:n}|y_{1:n}) \approx \prod_{i=1}^n q_1^{(c)}(\mu_i) = \prod_{i=1}^n \mbox{Mult}\left(\mu_i;\pi^{(c)}_i\right).
\end{equation}

\subsubsection*{Exceedance probability}

To evaluate the exceedance probability of model $j$ conditional on $\eta = 1$,
\citet{stephan2009} and \citet{rigoux2014} use a Monte Carlo integration
approach. To illustrate
this, let
\begin{equation}
\pi^{(1)}, ...,\pi^{(s)} \overset{\mathrm{iid}}{\sim} \mbox{Dir}\left(\alpha^{(c)}\right)
\end{equation}
denote a set of $s$ independent samples from the converged variational distribution
on $\pi$. Here, $\mbox{Dir}(\alpha^{(c)})$ denotes the Dirichlet distribution
with parameter vector $\alpha^{(c)}$. For $j = 1,...,m$, let
\begin{equation}
\mathbb{I}_j : [0,1]^m \to \{0,1\},
\pi \mapsto \mathbb{I}_j\left(\pi\right)
:= 
\begin{cases}
1 & \mbox{ if } \pi_j =    \max\lbrace \pi_{1},...,\pi_{m} \rbrace \\
0 & \mbox{ if } \pi_j \neq \max\lbrace \pi_{1},...,\pi_{m} \rbrace \\
\end{cases}
\end{equation}
denote the indicator function of the event that $\pi_j$ is the largest of 
the model probabilities $\pi_1,...,\pi_m$. Then, \citet{stephan2009} and
\citet{rigoux2014} use
\begin{equation}
\hat{\mathbb{P}}(\pi_j = \max_{r=1,\ldots,m}\pi_r|\eta = 1, y_{1:n}) := \frac{1}{s} \sum_{k=1}^s \mathbb{I}_j\left(\pi^{(k)}\right)
\end{equation}
as an estimate of the exceedance probability of the $j$th model conditional on
$\eta = 1$.
In our implementation of the approach, we set $s := 10^6$ and use the \texttt{scipy.stats}
module to sample from the Dirichlet distribution with parameter $\alpha^{(c)}$.

\subsubsection*{Log marginal likelihood}
Finally, as shown in \fullref{sec:proofs}, the ELBO approximation of the log
marginal likelihood
in the current scenario evaluates to 
\begin{align}\label{eq:elbo-p1}
\begin{split}
& \mbox{ELBO}\left(q_1^{(c)}(\pi,\mu_{1:n})\right) = \\  
&  \ln \Gamma\left(\sum_{j = 1}^{m} \alpha_j \right)
- \sum_{j=1}^m \ln  \Gamma(\alpha_j)
+ \sum_{j=1}^m (\alpha_j-1) \left(\psi\left(\alpha_j^{(c)}\right) - \psi\left(\alpha_s^{(c)}\right)\right)
\\
&
\quad + \sum_{i=1}^n \sum_{j=1}^m \pi_{ij}^{(c)} \left(\psi\left(\alpha_j^{(c)}\right) - \psi\left(\alpha_s^{(c)}\right)\right)
\\
&
\quad + \sum_{i=1}^n \sum_{j=1}^m \pi_{ij}^{(c)} \ln p(y_i|\mu_i = e_j)
\\
&
\quad - \ln \Gamma\left(\sum_{j = 1}^{m} \alpha_j^{(c)} \right)
      + \sum_{j=1}^m \ln \Gamma\left(\alpha_j^{(c)}\right)
      - \sum_{j=1}^m \left(\alpha_j^{(c)}-1\right) \left(\psi\left(\alpha_j^{(c)}\right) - \psi\left(\alpha_s^{(c)}\right)\right)
\\
&
\quad - \sum_{i=1}^n \sum_{j=1}^m \pi_{ij}^{(c)} \ln \pi_{ij}^{(c)}, 
\end{split}
\end{align}
where $\Gamma$ denotes the gamma function. The vector $\alpha$ contains the
parameters of the prior distribution of $\pi$ (cf. \eqnref{eq:p-eta-1}),
while $\alpha^{(c)}$ contains the parameters of the converged variational
distribution of $\pi$. Upon convergence of the algorithm, we thus have
\begin{equation}
\ln p_1(y_{1:n}) \approx 
\mbox{ELBO}\left(q_1^{(c)}(\pi,\mu_{1:n})\right).
\end{equation}

\subsubsection*{Summary}

In summary, the PEP $\varphi_j$ of model $j$ can be 
evaluated based on \eqnref{eq:varphi-j} and \eqnref{eq:p-eta-giv-y} by
approximating
the participant-specific log model evidence conditional on the $j$th model using
the corresponding BIC score
\begin{equation}\label{eq:lme-bic}
\ln p(y_i|\mu_i = e_j) \approx \mbox{BIC}_{ij}.
\end{equation}
Under the null hypothesis, we use the uniform contribution $1/m$ in
\eqnref{eq:varphi-j}, following the limiting convention described above.
Under the alternative hypothesis, we approximate the exceedance probability by
\begin{equation}
\mathbb{P}(\pi_j = \max_{r=1,\ldots,m}\pi_r|\eta = 1, y_{1:n})
\approx
\hat{\mathbb{P}}(\pi_j = \max_{r=1,\ldots,m}\pi_r|\eta = 1, y_{1:n}).
\end{equation}
For the log marginal likelihoods, we use
\begin{equation}
\ln p(y_{1:n}|\eta = 0) = \mbox{ELBO}\left(q_0^{(1)}(\mu_{1:n})\right)
\end{equation}
and
\begin{equation}
\ln p(y_{1:n}|\eta = 1) \approx \mbox{ELBO}\left(q_1^{(c)}(\pi,\mu_{1:n})\right).
\end{equation}

\clearpage

\subsection{Free-form mean-field variational inference}\label{sec:free-form-mean-field-variational-inference}
In this subsection, we review the theoretical background of free-form mean-field 
variational inference in condensed form, with proofs deferred to
\fullref{sec:proofs}. For
extended introductions to variational inference, see, for example, \citet{fox2012}, 
\citet{ostwald2014}, and \citet{blei2017}. We commence by introducing the
fundamental problem of probabilistic inference.
Throughout this subsection, integrals denote sums for discrete variables and
the corresponding combinations of sums and integrals for mixed discrete and
continuous variables.

\begin{definition}[Fundamental problem of probabilistic inference]
Let $y := (y_1,...,y_n)$ denote a set of observable random 
entities modeling observed data, and let $z := (z_1,...,z_m)$ denote a set of
latent random
entities. Assume that $y$ and $z$ form a \textit{probabilistic model} with PDF/PMF
\begin{equation}
p(y,z) = p(y|z)p(z).
\end{equation}
Then, the \textit{fundamental problem of probabilistic inference} is to
approximate (1) the conditional distribution $p(z|y)$, and
(2) the \textit{log model evidence} $\ln p(y) = \ln \int p(y,z)\,dz$.

\end{definition}

In Bayesian parameter estimation problems, the conditional distribution $p(z|y)$
is often referred to as the \textit{posterior distribution}, and the log model
evidence
is also referred to as the \textit{log marginal likelihood}. Note that for any fixed 
dataset, the log model evidence is a fixed scalar quantity. In our current
application
to the hypothesis-conditional GMIs, the set of observable random entities $y$ is
given by
the set of observed participant behavioral datasets $y_1,...,y_n$, while the latent variables $z$ comprise the model probabilities $\pi$
and the unobserved participant model indicators $\mu_1,...,\mu_n$. Under the
null hypothesis, the model probabilities are fixed at $\pi_0$, and inference
concerns the participant model indicators only. The following
theorem is central to the variational inference approach.

\begin{theorem}[Log model evidence decomposition]\label{thm:log-model-evidence-decomposition}
Let $p(y,z)$ denote a probabilistic model, and let $q(z)$ denote the PDF/PMF of a 
distribution over the latent random variables, referred to as \textit{the
variational distribution}.
Assume that $0<p(y)<\infty$, that $q$ is absolutely continuous with respect
to $p(\,\cdot| y)$, and that the expectations defining the ELBO and
Kullback--Leibler divergence are finite.
Then, the following log model evidence decomposition holds:
\begin{equation}
\ln p(y) = \mathrm{ELBO}(q(z)) + \mathrm{KL}(q(z)||p(z|y)),
\end{equation}
where
\begin{equation}
\mathrm{ELBO}(q(z)) := \int q(z) \ln \left(\frac{p(y,z)}{q(z)}\right)\,dz
\end{equation}
is referred to as the \textit{evidence lower bound}, and
\begin{equation}
\mathrm{KL}(q(z)||p(z|y)) := \int q(z) \ln \left(\frac{q(z)}{p(z|y)}\right)\,dz
\end{equation}
is referred to as the \textit{Kullback--Leibler divergence} between $q(z)$ and
$p(z|y)$.
\end{theorem}

The variational distribution $q(z)$ serves as an approximation to the conditional
distribution $p(z|y)$, and, as alluded to below, the ELBO serves as an
approximation to
the log model evidence. Commonly, $q(z)$ is of parameterized functional form, and
we refer to its parameters as \textit{variational parameters}. Owing to its origin in
statistical physics, the ELBO is often referred to as \textit{variational free
energy}.
The Kullback--Leibler divergence serves as a measure of discrepancy between the variational
distribution $q(z)$ and the conditional distribution $p(z|y)$. The following 
properties of the Kullback--Leibler divergence enable variational inference. The
first of these is a direct consequence of Jensen's inequality \citep{jensen1906}.

\begin{theorem}[Kullback--Leibler divergence]\label{thm:kullback-leibler-divergence}
Let $\mathrm{KL}\left(q(z)||p(z)\right)$ denote the Kullback--Leibler divergence 
of two PMFs/PDFs $q(z)$ and $p(z)$. Then,
\begin{equation}
\mathrm{KL}\left(q(z)||p(z)\right)\geq 0,
\end{equation}
with equality if and only if $q(z)=p(z)$ almost everywhere.
\end{theorem}
The non-negativity of the Kullback--Leibler divergence endows the ELBO with its
characterizing property:
\begin{theorem}[Evidence lower bound]\label{thm:evidence-lower-bound}
The ELBO is a lower bound of the log model evidence,
\begin{equation}
\mathrm{ELBO}(q(z)) \le \ln p(y).
\end{equation}
\end{theorem}

The fixed scalar nature of $\ln p(y)$ for a given dataset and the fact that
the ELBO is always smaller than or equal to $\ln p(y)$ suggest two equivalent options
for solving the fundamental problem of probabilistic inference using a variational
approach: either (1) minimizing $\mathrm{KL}(q(z)||p(z|y))$ with respect to $q(z)$,
such that $q(z)\to p(z|y)$ and $\mathrm{ELBO}(q(z)) \to \ln p(y)$, or (2) maximizing
 $\mathrm{ELBO}(q(z))$ with respect to $q(z)$, such that likewise 
$\mathrm{ELBO}(q(z)) \to \ln p(y)$ and $q(z) \to p(z|y)$. The free-form 
variational inference approach discussed below follows the latter approach. Its
final building block is the notion of a \textit{mean-field variational
approximation},
which merely refers to an independence assumption on $q(z)$.

\begin{definition}[Mean-field variational approximation and inference]\label{def:mean-field-variational-approximation-and-inference}
Let $z=(z_1,...,z_m)$ denote the latent random variables of a probabilistic
model and let $\mathcal{S}:=\{s_1,...,s_S\}$ denote a partition of
$\{1,...,m\}$ into $S$ nonempty, mutually exclusive subsets, such that
$\bigcup_{i=1}^S s_i=\{1,...,m\}$ and $s_i\cap s_j=\emptyset$ for $i\neq j$.
Let $z_{s_i}:=(z_j)_{j\in s_i}$ denote the corresponding blocks of latent
variables. Then the assumption that the blocks $z_{s_1},...,z_{s_S}$ are
mutually independent under $q$, i.e.,
\begin{equation}
q(z) := \prod_{i = 1}^S q(z_{s_i}),
\end{equation}
where $q(z_{s_i})$ denotes the PDF/PMF of the distribution of the elements of 
$z_{s_i}$, $i = 1,\ldots,S$, is referred to as a \textit{mean-field approximation}. 
\end{definition}

Clearly, variational mean-field approximations neglect potential non-independencies 
in the conditional distribution $p(z|y)$. Commonly employed variational mean-field
approximations are a full factorization $q(z) = \prod_{i = 1}^m q(z_i)$ or a 
binary factorization $q(z) = q(z_s)q(z_{\setminus s})$, where $s$ is a nonempty proper subset of $\{1,...,m\}$,
$z_s:=(z_j)_{j\in s}$, and
$z_{\setminus s}:=(z_j)_{j\in\{1,...,m\}\setminus s}$. In our current
application, a factorization across the blocks $\pi,\mu_1,...,\mu_n$ is
employed. We are now in a position to formulate the general free-form mean-field
variational inference approach that forms the foundation for inference in the 
current application and gives rise to the variational distribution update equations
of the main text.

\begin{theorem}[Free-form mean-field variational inference]\label{thm:free-form-mean-field-variational-inference}
Let $p(y,z)$ denote a probabilistic model comprising observable random variables 
$y = (y_1,...,y_n)$ and latent random variables $z = (z_1,...,z_m)$. Let
$q(z)$
denote a variational distribution and assume that $q(z)$ factorizes in a binary 
manner, i.e., $q(z) = q(z_s)q(z_{\setminus s})$ with $s$ a nonempty proper subset of $\{1,...,m\}$,
$z_s:=(z_j)_{j\in s}$, and
$z_{\setminus s}:=(z_j)_{j\in\{1,...,m\}\setminus s}$.
For fixed $q(z_{\setminus s})$, assume that the normalization constant
$\gamma_s$ is finite and strictly positive. Then, setting
\begin{equation}
q(z_s) := \frac{1}{\gamma_s} \exp \left(\int q(z_{\setminus s}) \ln p(y,z) \,dz_{\setminus s}\right),
\end{equation}
where $\gamma_s$ denotes a normalization constant independent of $z_s$, maximizes 
the evidence lower bound with respect to $q(z_s)$.
\end{theorem}

The functional form of the maximizing variational distributions depends on the 
right-hand side of the update equation. Hence the term ``free-form''. Note that
exchanging the roles of $q(z_s)$ and $q(z_{\setminus s})$ maximizes the ELBO 
with respect to $q(z_{\setminus s})$. This implies the following coordinate-wise 
ELBO maximization approach, referred to as the coordinate-ascent variational
inference
(CAVI) algorithm.

\begin{definition}[Coordinate-ascent variational inference (CAVI)]\label{def:coordinate-ascent-variational-inference}
Let $p(y,z)$ denote a probabilistic model comprising observable random variables
$y = (y_1,...,y_n)$ and latent random variables $z = (z_1,...,z_m)$. 
Let $q(z)$ denote the PDF/PMF of a variational distribution and assume that $q(z)$ 
factorizes in a binary manner. Then, the following algorithm, referred to as 
\textit{coordinate-ascent variational inference}, performs successive
coordinate-wise maximizations of the evidence lower bound:

\vspace{2mm}

\noindent \textit{Initialization}
\vspace{2mm}

\noindent (0) Initialize $q^{(0)}(z_s)$ and $q^{(0)}(z_{\setminus s})$
appropriately,
e.g., by setting  $q^{(0)}(z_s) := \smallint p(z)\,dz_{\setminus s}$ and 
$q^{(0)}(z_{\setminus s}) :=  \smallint p(z)\,dz_s$, evaluate 
$\mathrm{ELBO}\left(q^{(0)}(z_s)q^{(0)}(z_{\setminus s})\right)$, and select a 
convergence criterion $\delta > 0$.

\vspace{2mm}
\noindent \textit{Iterations}
\vspace{2mm}

\noindent For $k = 0,1,2,...$
\vspace{1mm}

(1) Set 
\begin{equation}
q^{(k+1)}(z_s) := \frac{1}{\gamma_s} \exp \left(\smallint q^{(k)}(z_{\setminus s}) \ln p(y,z) \,dz_{\setminus s} \right)  
\end{equation}\label{eq:q-update-1}

(2) Set 
\begin{equation}
q^{(k+1)}(z_{\setminus s}) := \frac{1}{\gamma_{\setminus s}} \exp \left(\smallint q^{(k+1)}(z_s) \ln p(y,z) \,dz_s \right)  
\end{equation}\label{eq:q-update-2}

(3) Evaluate 
\begin{equation}
\mathrm{ELBO}\left(q^{(k+1)}(z_s)q^{(k+1)}(z_{\setminus s})\right) 
\end{equation}

\vspace{1mm}

\noindent until

\begin{equation}\label{eq:elbo-update}
\big\vert\mathrm{ELBO}\left(q^{(k+1)}(z_s)q^{(k+1)}(z_{\setminus s})\right) -
\mathrm{ELBO}\left(q^{(k)}(z_s)q^{(k)}(z_{\setminus s})\right)\big\vert < \delta
\end{equation}

\vspace{2mm}

\noindent \textit{Return}

\begin{equation}
p(z|y) \approx q^{(k+1)}(z_s)q^{(k+1)}(z_{\setminus s}) 
\mbox{ and } 
\ln p(y) \approx \mathrm{ELBO}\left(q^{(k+1)}(z_s)q^{(k+1)}(z_{\setminus s})\right) 
\end{equation} 
\end{definition}

\clearpage
\subsection{Bayesian information criterion}\label{sec:bic}

Let 
\begin{equation}
p(y,\theta) = p(y|\theta)p(\theta)
\end{equation}
denote a probabilistic model, i.e., a joint distribution of data $y$ and
parameters $\theta$. The Bayesian information criterion (BIC) was proposed by
\citet{schwarz1978} as an approximation of the log marginal likelihood
\begin{equation}
\ln p(y) = \ln \int p(y,\theta)\,d\theta.
\end{equation}
The BIC is given by
\begin{equation}\label{eq:bic}
\mbox{BIC} := \ell(\hat{\theta}) - \frac{k}{2}\ln n,
\end{equation} 
 where $\ell(\hat{\theta})$ denotes a maximum value of the log likelihood
 $\ln p(y|\theta)$ of the model, $k$ denotes the dimension of $\theta$ (i.e.,
 the number of
 scalar parameters of the probabilistic model), and $n$ denotes the
 number of data points. The BIC formula rests on Laplace's integral
 approximation method, which we first briefly review for the case of a
 multivariate real-valued function. To this end, let
 $f : \mathbb{R}^k \to \mathbb{R}$ be a twice continuously differentiable
 convex function with a unique minimum at $\tilde{x}\in\mathbb{R}^k$ and a
 positive-definite Hessian $H^f(\tilde{x})$. Assume that $I_n$ is finite for
 sufficiently large $n$. Then, \textit{Laplace's integral approximation method}
 corresponds to approximating the integral
\begin{equation}
I_n := \int_{\mathbb{R}^k} \exp(-nf(x))\,dx
\end{equation}
for large $n \in \mathbb{N}$ as
\begin{equation}\label{eq:laplace_int_approx}
I_n \approx 
\left(\frac{2\pi}{n} \right)^{\frac{k}{2}}
|H^f(\tilde{x})|^{-\frac{1}{2}}
\exp(-nf(\tilde{x})),
\end{equation}
where $|H^f(\tilde{x})|$ denotes the determinant of the Hessian of $f$ evaluated
at $\tilde{x}$. For the motivation and derivation of Laplace's integral
approximation method, see, e.g., \citet{lauritzen2009}. Using Laplace's integral
approximation method, it can readily be shown that an approximation of the log
marginal likelihood $\ln p(y)$ is furnished by
\begin{equation}\label{eq:lml_approx}
\ln p(y) \approx
\frac{k}{2}\ln 2\pi - \frac{k}{2} \ln n - \frac{1}{2} \ln |H^f(\tilde{\theta})| + \ln p(y|\tilde{\theta}) + \ln p(\tilde{\theta}),
\end{equation} 
where $\tilde{\theta}$ denotes the posterior mode. We retain
$\hat{\theta}$ for a maximizer of the log likelihood function 
\begin{equation}
\ell : \mathbb{R}^k \to \mathbb{R}, \theta \mapsto \ell(\theta) := \ln p(y|\theta).
\end{equation}
This may be seen as follows: Define
\begin{equation}
f : \mathbb{R}^k \to \mathbb{R}, 
\theta \mapsto f(\theta) := -\frac{1}{n}\left(\ln p(y|\theta) + \ln p(\theta)\right),
\end{equation}
and let $\tilde{\theta}$ denote a minimum location of $f$. Then, we have with
the Laplace integral approximation of eq. \eqref{eq:laplace_int_approx}
\begin{align*}
\ln p(y) 
& = \ln \int p(y|\theta)p(\theta)\,d\theta																						\\
& = \ln \int \exp\left(\ln\left(p(y|\theta)p(\theta)\right)\right)\,d\theta														\\
& = \ln \int \exp\left(\ln p(y|\theta) + \ln p(\theta)\right)\,d\theta															\\
& = \ln \int \exp\left(-n\left(-\frac{1}{n}\left(\ln p(y|\theta) + \ln p(\theta)\right)\right)\right)\,d\theta					\\
& = \ln \int \exp\left(-n f(\theta)\right)\,d\theta																				\\
& \approx 
    \ln \left(\left(\frac{2\pi}{n} \right)^{\frac{k}{2}}|H^f(\tilde{\theta})|^{-\frac{1}{2}}\exp(-nf(\tilde{\theta}))\right)	\\
& = \frac{k}{2}\ln 2\pi - \frac{k}{2} \ln n - \frac{1}{2} \ln |H^f(\tilde{\theta})| + \ln p(y|\tilde{\theta}) + \ln p(\tilde{\theta}).
\end{align*}
Under standard regularity conditions, the BIC approximation follows by
retaining the leading terms of the Laplace expansion as the sample size
increases \citep{schwarz1978} \citep[see also][Section 9.1.3,
pp.~215--217]{konishi2008}. For fixed parameter dimension $k$, the term
$\frac{k}{2}\ln(2\pi)$ is constant. If the prior density is continuous and
strictly positive near the limiting parameter value and the normalized
Hessian $H^f(\tilde{\theta})$ converges to a finite, positive-definite matrix,
the terms $\ln p(\tilde{\theta})$ and
$-\frac12\ln|H^f(\tilde{\theta})|$ likewise remain bounded in probability.
Together, these contributions are of order $O_p(1)$, meaning that they do not
grow with the sample size, although they need not vanish. Replacing the
posterior mode $\tilde{\theta}$ by the maximum-likelihood estimate
$\hat{\theta}$ does not affect the retained leading-order terms, giving
\begin{equation}
\ln p(y)=\ell(\hat{\theta})-\frac{k}{2}\ln n+O_p(1).
\end{equation}
Dropping this bounded-order remainder while retaining the growing penalty
$\frac{k}{2}\ln n$ yields
\begin{equation}
\ln p(y) \approx \ell(\hat{\theta}) - \frac{k}{2}\ln n = \mbox{BIC}.
\end{equation}

\clearpage
\subsection{Proofs}\label{sec:proofs}

\textit{Proof of \eqnref{eq:varphi-normalization}}
\vspace{2mm}

Under the alternative hypothesis, the posterior distribution of $\pi$ is
absolutely continuous on the probability simplex. Hence, ties between model
probabilities have probability zero, and exactly one model has the largest
probability almost surely. It follows that
\begin{equation}
\sum_{j=1}^m\mathbb{P}\left(\pi_j>\max_{k\ne j}\pi_k|\eta=1,y_{1:n}\right)=1.
\end{equation}
Using the definition of the PEPs in
\eqnref{eq:varphi-j}, we therefore obtain
\begin{align}
\sum_{j=1}^m\varphi_j
&=\sum_{j=1}^m\frac{1}{m}\mathbb{P}(\eta=0|y_{1:n})
+\fitproofline{\left(1-\mathbb{P}(\eta=0|y_{1:n})\right)\sum_{j=1}^m\mathbb{P}\left(\pi_j>\max_{k\ne j}\pi_k|\eta=1,y_{1:n}\right)}\notag\\
&=\mathbb{P}(\eta=0|y_{1:n})+\left(1-\mathbb{P}(\eta=0|y_{1:n})\right)\notag\\
&=1.
\end{align}
$\hfill\Box$
\vspace{2mm}

\textit{Proof of \eqnref{eq:p-eta-giv-y}}
\vspace{2mm}

Under the equal hypothesis priors specified in \eqnref{eq:p-eta}, and
assuming finite, strictly positive marginal likelihoods under both
hypotheses, we have
\begin{align}
\mathbb{P}(\eta = 0|y_{1:n}) 
& =  \frac{p(\eta = 0, y_{1:n})}{p(y_{1:n})} \notag\\
& =  \frac{\mathbb{P}(\eta = 0)p(y_{1:n}|\eta = 0)}{\sum_{\eta \in \{0,1\}}p(\eta, y_{1:n})} \notag\\
& =  \frac{\mathbb{P}(\eta = 0)p(y_{1:n}|\eta = 0)}{\sum_{\eta \in \{0,1\}}\mathbb{P}(\eta)p(y_{1:n}|\eta)} \notag\\
& =  \frac{\frac{1}{2}p(y_{1:n}|\eta = 0)}{\sum_{\eta \in \{0,1\}}\frac{1}{2}p(y_{1:n}|\eta)} \notag\\
& =  \frac{p(y_{1:n}|\eta = 0)}{\sum_{\eta \in \{0,1\}}p(y_{1:n}|\eta)} \notag\\
& =  \frac{p(y_{1:n}|\eta = 0)}{p(y_{1:n}|\eta = 0) + p(y_{1:n}|\eta = 1)} \notag\\
& =  \frac{\frac{p(y_{1:n}|\eta = 0)}{p(y_{1:n}|\eta = 0)}}{\frac{p(y_{1:n}|\eta = 0) + p(y_{1:n}|\eta = 1)}{p(y_{1:n}|\eta = 0)}} \notag\\
& =  \frac{1}{\frac{p(y_{1:n}|\eta = 0)}{p(y_{1:n}|\eta = 0)}+\frac{p(y_{1:n}|\eta = 1)}{p(y_{1:n}|\eta = 0)}} \notag\\
& =  \frac{1}{1+\frac{p(y_{1:n}|\eta = 1)}{p(y_{1:n}|\eta = 0)}} \notag\\
& =  \frac{1}{1+\exp\left(\ln\left(\frac{p(y_{1:n}|\eta = 1)}{p(y_{1:n}|\eta = 0)}\right)\right)} \notag\\
& =  \frac{1}{1+\exp\left(\ln p(y_{1:n}|\eta = 1) - \ln p(y_{1:n}|\eta = 0)\right)}.
\end{align}
$\hfill\Box$
\vspace{2mm}

\textit{Proof of \eqnref{eq:q0-mu-i}}
\vspace{2mm}

Under the null hypothesis, $\mathbb{P}(\pi=\pi_0|\eta=0)=1$.
It therefore suffices to consider the reduced model
\begin{equation}
p_0(\mu_{1:n},y_{1:n})=\prod_{i=1}^n p(\mu_i|\pi_0)p(y_i|\mu_i).
\end{equation}
In this proof and the following ELBO proof, integrals over discrete model
indicators denote sums over their possible values in $E_m$. An integral over
several indicators denotes the corresponding multiple sum.
In the following derivation, we omit the hypothesis subscript and the
single-update superscript and write $\pi=\pi_0$ for ease of notation. From
\fullref{sec:free-form-mean-field-variational-inference} and the definitions
in \eqnref{eq:p0-inf-1} and \eqnref{eq:p0-inf-2}, the CAVI algorithm update
for the variational distribution on $\mu_i$ takes the form
\begin{equation}
q(\mu_i) 
:= \frac{1}{\gamma_{\mu_i}}
\exp(I_{\mu_i})
\mbox{ with }
I_{\mu_i} :=\int \left(\prod_{\substack{k = 1 \\ k \neq i}}^n q(\mu_k)\right) \ln p(\mu_{1:n},y_{1:n}) \,d\mu_{1:n\backslash i}
\mbox{ for }
i = 1,...,n.
\end{equation}
To show the validity of \eqnref{eq:q0-mu-i}, we rewrite $I_{\mu_i}$ in terms of
the log model evidence conditional on the participant model indicator,
$\ln p(y_i|\mu_i)$. We have

\begin{align*}
I_{\mu_i}
&= \int \left(\prod_{\substack{k=1\\k\ne i}}^n q(\mu_k)\right) \ln p(\mu_{1:n},y_{1:n})\,d\mu_{1:n\backslash i}
\\
&= \int \left(\prod_{\substack{k=1\\k\ne i}}^n q(\mu_k)\right) \ln\left(\prod_{k'=1}^n p(\mu_{k'}|\pi)p(y_{k'}|\mu_{k'})\right)\,d\mu_{1:n\backslash i}
\\
&= \int \left(\prod_{\substack{k=1\\k\ne i}}^n q(\mu_k)\right) \left(\sum_{k'=1}^n\ln p(\mu_{k'}|\pi)+\sum_{k'=1}^n\ln p(y_{k'}|\mu_{k'})\right)\,d\mu_{1:n\backslash i}
\\
&= \sum_{k'=1}^n\int \left(\prod_{\substack{k=1\\k\ne i}}^n q(\mu_k)\right) \ln p(\mu_{k'}|\pi)\,d\mu_{1:n\backslash i}
+\sum_{k'=1}^n\int \left(\prod_{\substack{k=1\\k\ne i}}^n q(\mu_k)\right) \ln p(y_{k'}|\mu_{k'})\,d\mu_{1:n\backslash i}
\\
&= \int \left(\prod_{\substack{k=1\\k\ne i}}^n q(\mu_k)\right) \ln p(\mu_i|\pi)\,d\mu_{1:n\backslash i}
\\
&\quad+\sum_{\substack{k'=1\\k'\ne i}}^n\int \left(\prod_{\substack{k=1\\k\ne i}}^n q(\mu_k)\right) \ln p(\mu_{k'}|\pi)\,d\mu_{1:n\backslash i}
\\
&\quad+\int \left(\prod_{\substack{k=1\\k\ne i}}^n q(\mu_k)\right) \ln p(y_i|\mu_i)\,d\mu_{1:n\backslash i}
\\
&\quad+\sum_{\substack{k'=1\\k'\ne i}}^n\int \left(\prod_{\substack{k=1\\k\ne i}}^n q(\mu_k)\right) \ln p(y_{k'}|\mu_{k'})\,d\mu_{1:n\backslash i}
\\
&= \int \left(\prod_{\substack{k=1\\k\ne i}}^n q(\mu_k)\right) \ln p(\mu_i|\pi)\,d\mu_{1:n\backslash i}+\int \left(\prod_{\substack{k=1\\k\ne i}}^n q(\mu_k)\right) \ln p(y_i|\mu_i)\,d\mu_{1:n\backslash i}+c_1
\\
&= \ln p(\mu_i|\pi)\int\prod_{\substack{k=1\\k\ne i}}^n q(\mu_k)\,d\mu_{1:n\backslash i}+\ln p(y_i|\mu_i)\int\prod_{\substack{k=1\\k\ne i}}^n q(\mu_k)\,d\mu_{1:n\backslash i}+c_1
\\
&= \ln p(\mu_i|\pi)+\ln p(y_i|\mu_i)+c_1
\\
&= \ln\left(\prod_{j=1}^m\left(\frac{1}{m}\right)^{\mu_{i_j}}\right)+\ln p(y_i|\mu_i)+c_1
\\
&= -\sum_{j=1}^m\mu_{i_j}\ln m+\ln p(y_i|\mu_i)+c_1
\\
&= \ln p(y_i|\mu_i)-\ln m+c_1
\\
&= \ln p(y_i|\mu_i)+c_2,
\end{align*}
where we defined the constant terms
\begin{align}
c_1
&:= \sum_{\substack{k'=1\\k'\ne i}}^n\int \left(\prod_{\substack{k=1\\k\ne i}}^n q(\mu_k)\right) \ln p(\mu_{k'}|\pi)\,d\mu_{1:n\backslash i}
+ \sum_{\substack{k'=1\\k'\ne i}}^n\int \left(\prod_{\substack{k=1\\k\ne i}}^n q(\mu_k)\right) \ln p(y_{k'}|\mu_{k'})\,d\mu_{1:n\backslash i}
\end{align}
and
\begin{equation}
c_2 := c_1 - \ln m,
\end{equation}
which do not vary with $\mu_i$. Here, we used the normalization of the
variational factors and the fact that $\sum_{j=1}^m\mu_{i_j}=1$ for each model
indicator vector $\mu_i$. For the variational update of $q(\mu_i)$,
we thus obtain
\begin{align}
q(\mu_i) 
& = 
\frac{1}{\gamma_{\mu_i}} \exp\left(I_{\mu_i}\right)
\notag\\
\Leftrightarrow
q(\mu_i) 
& =
\frac{1}{\gamma_{\mu_i}}
\exp\left(\ln p(y_i|\mu_i)  + c_2 \right)
\notag\\
\Leftrightarrow
q(\mu_i) 
& =
\frac{1}{\gamma_{\mu_i}}
\exp\left(\ln p(y_i|\mu_i) \right)
\exp(c_2)
\notag\\
\Leftrightarrow
q(\mu_i) 
& =
\frac{\exp(c_2)\exp\left(\ln p(y_i|\mu_i)\right)}{\sum_{\mu_i \in E_m} \exp(c_2)\exp\left(\ln p(y_i|\mu_i)\right)}
\notag\\
\Leftrightarrow
q(\mu_i) 
& =
\frac{\exp\left(\ln p(y_i|\mu_i)\right)}{\sum_{\mu_i \in E_m} \exp\left(\ln p(y_i|\mu_i)\right)}.
\end{align}   
Finally, with $\mu_i \in E_m$, we may re-express the above in Multinoulli
distribution
form by defining
\begin{equation}
\pi_{ij}^q 
:= \frac{\exp\left(\ln p(y_i|\mu_i = e_j)\right)}
        {\sum_{\mu_i\in E_m}\exp\left(\ln p(y_i|\mu_i)\right)},
\end{equation}
such that 
\begin{equation}
q(\mu_i) = \prod_{j=1}^m \left(\pi_{ij}^q\right)^{\mu_{i_j}}.
\end{equation}
Restoring the hypothesis subscript and replacing the variational superscript
$q$ with $(1)$ to denote completion of the single update, we have
\begin{equation}
q^{(1)}_0(\mu_i) 
= \mbox{Mult}\left(\mu_i; \pi_i^{(1)}\right),
\end{equation}
where 
\begin{equation}
\pi_i^{(1)} := \left(\pi_{ij}^{(1)}\right)_{j = 1,...,m}
\mbox{ with }
\pi_{ij}^{(1)} := 
\frac{\exp\left(\ln p(y_i|\mu_i = e_j)\right)}
        {\sum_{\mu_i\in E_m}\exp\left(\ln p(y_i|\mu_i)\right)}.
\end{equation}
$\hfill\Box$
\vspace{2mm}

\textit{Proof of \eqnref{eq:elbo-p0}}
\vspace{2mm}

We work with the reduced null model, with $\pi=\pi_0$ fixed, and omit the
hypothesis subscript for ease of notation. The ELBO in the current scenario
is given by
\begin{equation}\label{eq:elbo-p0-1}
\mbox{ELBO}(q(\mu_{1:n})) 
:= \int q(\mu_{1:n})\ln \left(\frac{p(\mu_{1:n},y_{1:n})}{q(\mu_{1:n})}\right)\,d\mu_{1:n}.                           
\end{equation}
In the following, we first rewrite the right-hand side of \eqnref{eq:elbo-p0-1}
as a sum
of three integral terms. We then consider these integral terms and finally
assemble the form
provided in \eqnref{eq:elbo-p0}. Thus, we first note that

\begin{align}
& \mbox{ELBO}(q(\mu_{1:n})) 
\notag\\
& 
:= \int q(\mu_{1:n})\ln \left(\frac{p(\mu_{1:n},y_{1:n})}{q(\mu_{1:n})}\right)\,d\mu_{1:n}                           
\notag\\
&  
= \int q(\mu_{1:n})\left( \ln p(\mu_{1:n},y_{1:n}) - \ln q(\mu_{1:n})\right)\,d\mu_{1:n}                            
\notag\\
&  
= \fitproofline{\int q(\mu_{1:n})\ln p(\mu_{1:n},y_{1:n})\,d\mu_{1:n}  
- \int q(\mu_{1:n})\ln q(\mu_{1:n})\,d\mu_{1:n}}
\notag\\
&  
= \fitproofline{\int\left(\prod_{i'=1}^n q(\mu_{i'})\right)\ln \left(\prod_{i=1}^{n}p(\mu_i|\pi)p(y_i|\mu_i)\right)\,d\mu_{1:n}   
- \int\left(\prod_{i'=1}^n q(\mu_{i'})\right)\ln \left(\prod_{i=1}^n q(\mu_i)\right)\,d\mu_{1:n}}
\notag\\
&  
= \fitproofline{\int\left(\prod_{i'=1}^n q(\mu_{i'})\right)\left(\sum_{i=1}^{n} \ln p(\mu_i|\pi) + \sum_{i=1}^n \ln p(y_i|\mu_i) \right)\,d\mu_{1:n}   
- \int\left(\prod_{i'=1}^n q(\mu_{i'})\right)\left(\sum_{i=1}^n \ln q(\mu_i)\right)\,d\mu_{1:n}}
\notag\\
&  
= \fitproofline{\int\left(\prod_{i'=1}^n q(\mu_{i'})\right)\sum_{i=1}^{n} \ln p(\mu_i|\pi)\,d\mu_{1:n} 
+ \int\left(\prod_{i'=1}^n q(\mu_{i'})\right)\sum_{i=1}^n \ln p(y_i|\mu_i)\,d\mu_{1:n}   
- \int\left(\prod_{i'=1}^n q(\mu_{i'})\right)\sum_{i=1}^n \ln q(\mu_i)\,d\mu_{1:n}.}
\end{align}
We consider the remaining integral terms in turn. For a fixed $i$, all factors
$q(\mu_{i'})$ with $i'\ne i$ integrate to one by normalization. This justifies
reducing the joint integrals to single-indicator integrals below.
For the first term, we have
\begin{align}
T_1 
&
:= 
\int\left(\prod_{i'=1}^n q(\mu_{i'})\right)\sum_{i=1}^{n} \ln p(\mu_i|\pi)\,d\mu_{1:n} 
\notag\\
& 
=  \sum_{i=1}^n \int q(\mu_i)\ln p(\mu_i|\pi)\,d\mu_i  
\notag\\
& 
=  \sum_{i=1}^n \int q(\mu_i)\ln \left(\prod_{j=1}^m \pi_j^{\mu_{i_j}}\right)\,d\mu_i  
\notag\\
& 
=  \sum_{i=1}^n \int q(\mu_i) \left(\sum_{j=1}^m \mu_{i_j} \ln \pi_j \right)\,d\mu_i  
\notag\\
& 
=  \sum_{i=1}^n \sum_{\mu_i \in E_m} q(\mu_i)\sum_{j=1}^m \mu_{i_j} \ln \pi_j  
\notag\\
& 
=  \sum_{i=1}^n \sum_{j=1}^m \pi_{ij}^q \ln \pi_j,  
\end{align}
where we defined
\begin{equation}
\pi_{ij}^q := q(\mu_i=e_j).
\end{equation}
For the second term, we have
\begin{align}
T_2 
& 
:= 
\int\left(\prod_{i'=1}^n q(\mu_{i'})\right)\sum_{i=1}^n \ln p(y_i|\mu_i)\,d\mu_{1:n}   
\notag\\
& 
= 
\sum_{i=1}^n \int q(\mu_i)\ln p(y_i|\mu_i)\,d\mu_i 
\notag\\
& 
= 
\sum_{i=1}^n \sum_{\mu_i \in E_m} q(\mu_i) \ln p(y_i|\mu_i)
\notag\\
& 
= 
\sum_{i=1}^n \sum_{j=1}^m \pi_{ij}^q \ln p(y_i|\mu_i = e_j).
\end{align}
Finally, for the third term, we have
\begin{align}
T_3  
& = \int\left(\prod_{i'=1}^n q(\mu_{i'})\right)\sum_{i=1}^n \ln q(\mu_i)\,d\mu_{1:n}  
\notag\\
& = \sum_{i=1}^n \int q(\mu_i)\ln q(\mu_i)\,d\mu_{i}   
\notag\\
& = \sum_{i=1}^n \sum_{\mu_i \in E_m} q(\mu_i)\ln q(\mu_i) 
\notag\\
& = \sum_{i=1}^n \sum_{j=1}^m \pi_{ij}^q \ln \pi_{ij}^q.
\end{align}
In summary, we thus obtain  
\begin{align}
\mbox{ELBO}(q(\mu_{1:n})) 
&
= \sum_{i=1}^n \sum_{j=1}^m \pi_{ij}^q \ln \pi_j  
+ \sum_{i=1}^n \sum_{j=1}^m \pi_{ij}^q \ln p(y_i|\mu_i = e_j) 
- \sum_{i=1}^n \sum_{j=1}^m \pi_{ij}^q \ln \pi_{ij}^q  
\notag\\
&
= \sum_{i=1}^n \sum_{j=1}^m \pi_{ij}^q \ln\left(\frac{1}{m}\right)  
+ \sum_{i=1}^n \sum_{j=1}^m \pi_{ij}^q \ln p(y_i|\mu_i = e_j) 
- \sum_{i=1}^n \sum_{j=1}^m \pi_{ij}^q \ln \pi_{ij}^q  
\notag\\
&
= \sum_{i=1}^n \sum_{j=1}^m \pi_{ij}^q \ln p(y_i|\mu_i = e_j) 
- \sum_{i=1}^n \sum_{j=1}^m \pi_{ij}^q \ln m   
- \sum_{i=1}^n \sum_{j=1}^m \pi_{ij}^q \ln \pi_{ij}^q  
\notag\\
&
= \sum_{i=1}^n \sum_{j=1}^m \pi_{ij}^q \left(\ln p(y_i|\mu_i = e_j)-\ln m-\ln \pi_{ij}^q\right).  
\end{align} 
Replacing the variational superscript $q$ on $\pi_{ij}$ with $(c)$, where
$c=1$ denotes completion of the single update, yields the ELBO in
\eqnref{eq:elbo-p0}.
Substituting the updated model probabilities from \eqnref{eq:q0-mu-i}
yields the first equality below. The second follows by marginalizing the
reduced null model over the participant indicators:
\begin{align}
\mathrm{ELBO}\left(q_0^{(1)}(\mu_{1:n})\right)
&=\sum_{i=1}^n\ln\left(\frac{1}{m}\sum_{j=1}^m p(y_i|\mu_i=e_j)\right)\notag\\
&=\ln p_0(y_{1:n}).
\end{align}
$\hfill\Box$
\vspace{2mm}

\textit{Proof of \eqnref{eq:q1-mu-i}}
\vspace{2mm}

In the following three proofs, integrals over discrete model indicators
denote sums over their possible values in $E_m$. Integrals over several
indicators denote the corresponding multiple sums. Integration over $\pi$
is over the probability simplex. Throughout these three derivations, we omit
the iteration superscripts and hypothesis subscripts for ease of notation
and restore them when stating the final results.

From \fullref{sec:free-form-mean-field-variational-inference} and the definitions of
\eqnref{eq:p1-inf-1} and \eqnref{eq:p1-inf-2} in the current inference scenario,
we see that the
CAVI algorithm update for the variational distribution on $\mu_i$ takes the form
\begin{equation}\label{eq:q1-mu-1}
q(\mu_i) := \frac{1}{\gamma_{\mu_i}}
\exp\left(I_{\mu_i}\right)
\mbox{ with }
I_{\mu_i} := 
\int q(\pi)\prod_{\substack{k=1 \\ k\neq i}}^{n} q(\mu_k)\ln p(\pi, \mu_{1:n},y_{1:n}) d\mu_{1:n\backslash i} d\pi.
\end{equation}
To show the validity of \eqnref{eq:q1-mu-i}, we rewrite
$I_{\mu_i}$ in
terms of the log model evidence conditional on the participant model indicator,
$\ln p(y_i|\mu_i)$, and the variational parameters of $q(\pi)$.

\vspace{2mm}
\noindent We first note that for $i = 1,...,n$
\begin{align*}
& I_{\mu_i}
\\ 
& = 
\int
q(\pi) 
\prod_{\substack{k=1 \\ k\neq i}}^{n} q(\mu_k) 
\ln p(\pi,\mu_{1:n},y_{1:n}) 
d\mu_{1:n\backslash i}
d\pi 
\\
& = 
\int q(\pi) 
\int
\prod_{\substack{k=1 \\ k\neq i}}^{n} q(\mu_k) 
\ln p(\pi,\mu_{1:n},y_{1:n}) 
d\mu_{1:n\backslash i}
d\pi 
\\
& = 
\int q(\pi) 
\int 
\prod_{\substack{k=1 \\ k\neq i}}^{n} q(\mu_k) 
\ln\left(p(\pi)\prod_{k'=1}^{n}p(\mu_{k'}|\pi)p(y_{k'}|\mu_{k'})\right)
d\mu_{1:n\backslash i}
d\pi  
\\
& = 
\int q(\pi) 
\int\prod_{\substack{k=1 \\ k\neq i}}^{n} q(\mu_k) 
\left(\ln p(\pi)+\sum_{k'=1}^{n} \ln p(\mu_{k'}|\pi)+\sum_{k'=1}^{n} \ln p(y_{k'}|\mu_{k'})\right) 
d\mu_{1:n\backslash i}
d\pi  
\\
& = 
\resizebox{0.94\linewidth}{!}{$\displaystyle \int q(\pi) \int\left( \prod_{\substack{k=1 \\ k\neq i}}^{n} q(\mu_k) \ln p(\pi) + \prod_{\substack{k=1 \\ k\neq i}}^{n} q(\mu_k) \sum_{k'=1}^{n} \ln p(\mu_{k'}|\pi) + \prod_{\substack{k=1 \\ k\neq i}}^{n} q(\mu_k) \sum_{k'=1}^{n} \ln p(y_{k'}|\mu_{k'}) \right) \,d\mu_{1:n\backslash i} \,d\pi$}  
\\
& = 
\int q(\pi) \biggl(
\int \prod_{\substack{k=1 \\ k\neq i}}^{n} q(\mu_k)
\ln p(\pi) 
d\mu_{1:n\backslash i} 
\\
& \quad 
+
\int \prod_{\substack{k=1 \\ k\neq i}}^{n} q(\mu_k) 
\sum_{k'=1}^{n} \ln p(\mu_{k'}|\pi) 
d\mu_{1:n\backslash i} 
\\
& \quad 
+
\int \prod_{\substack{k=1 \\ k\neq i}}^{n} q(\mu_k) 
\sum_{k'=1}^{n} \ln p(y_{k'}|\mu_{k'}) 
d\mu_{1:n\backslash i}
\biggr) d\pi  
\\
& = 
\int q(\pi) 
\ln p(\pi) 
\int 
\prod_{\substack{k=1 \\ k\neq i}}^{n} q(\mu_k) 
d\mu_{1:n\backslash i}
d\pi
\\
& \quad 
+
\int q(\pi) 
\int\biggl( \prod_{\substack{k=1 \\ k\neq i}}^{n} q(\mu_k) \sum_{\substack{k'=1 \\ k' \neq i}}^{n} \ln p(\mu_{k'}|\pi) 
+ \prod_{\substack{k=1 \\ k\neq i}}^{n} q(\mu_k) \ln p(\mu_i|\pi) 
\biggr) d\mu_{1:n\backslash i}
d\pi
\\
& \quad 
+
\int q(\pi) 
\int\biggl(\prod_{\substack{k=1 \\ k\neq i}}^{n} q(\mu_k) 
  \sum_{\substack{k'=1 \\ k' \neq i}}^{n} \ln p(y_{k'}|\mu_{k'})
+ \prod_{\substack{k=1 \\ k\neq i}}^{n} q(\mu_k) \ln p(y_i|\mu_i) 
\biggr) d\mu_{1:n\backslash i}
d\pi
\\
& =
\int q(\pi) \ln p(\pi) d\pi  
\\
& \quad 
+
\int q(\pi)
\biggl(
\int \prod_{\substack{k=1 \\ k\neq i}}^{n} q(\mu_k) \sum_{\substack{k'=1 \\ k' \neq i}}^{n} \ln p(\mu_{k'}|\pi) 
d\mu_{1:n\backslash i}
+
\int
\prod_{\substack{k=1 \\ k\neq i}}^{n} q(\mu_k) \ln p(\mu_i|\pi) 
d\mu_{1:n\backslash i}
\biggr)
d\pi
\\
& \quad 
+ \int q(\pi)
\biggl(
\int 
\prod_{\substack{k=1 \\ k\neq i}}^{n} q(\mu_k) \sum_{\substack{k'=1 \\ k' \neq i}}^{n} \ln p(y_{k'}|\mu_{k'}) 
d\mu_{1:n\backslash i}
+ 
\int
\prod_{\substack{k=1 \\ k\neq i}}^{n} q(\mu_k) \ln p(y_i|\mu_i) 
d\mu_{1:n\backslash i}
\biggr)
d\pi
\\
& =
\int q(\pi) \ln p(\pi) d\pi  
\\
& \quad 
+
\int q(\pi)
\int \prod_{\substack{k=1 \\ k\neq i}}^{n} q(\mu_k) \sum_{\substack{k'=1 \\ k' \neq i}}^{n} 
\ln p(\mu_{k'}|\pi) 
d\mu_{1:n\backslash i}
d\pi
+ 
\int q(\pi)
\int\prod_{\substack{k=1 \\ k\neq i}}^{n} q(\mu_k) 
\ln p(\mu_i|\pi) 
d\mu_{1:n\backslash i}
d\pi
\\
& \quad 
+
\int q(\pi)
\int
\prod_{\substack{k=1 \\ k\neq i}}^{n} q(\mu_k) \sum_{\substack{k'=1 \\ k' \neq i}}^{n} 
\ln p(y_{k'}|\mu_{k'}) 
d\mu_{1:n\backslash i}
d\pi
+ 
\int q(\pi)
\int\prod_{\substack{k=1 \\ k\neq i}}^{n} q(\mu_k) 
\ln p(y_i|\mu_i) 
d\mu_{1:n\backslash i}
d\pi
\\
& =
\int q(\pi)
\int\prod_{\substack{k=1 \\ k\neq i}}^{n} q(\mu_k) 
\ln p(\mu_i|\pi) 
d\mu_{1:n\backslash i}
d\pi
+ 
\int q(\pi)
\int\prod_{\substack{k=1 \\ k\neq i}}^{n} q(\mu_k) 
\ln p(y_i|\mu_i) 
d\mu_{1:n\backslash i}
d\pi
+
c 
\\
& =
\int q(\pi) \ln p(\mu_i|\pi) 
\int\prod_{\substack{k=1 \\ k\neq i}}^{n} q(\mu_k) 
d\mu_{1:n\backslash i}
d\pi
+ 
\int q(\pi) \ln p(y_i|\mu_i) 
\int\prod_{\substack{k=1 \\ k\neq i}}^{n} q(\mu_k) 
d\mu_{1:n\backslash i}
d\pi
+
c 
\\
& =
\int q(\pi) \ln p(\mu_i|\pi) \,d\pi
+ 
\int q(\pi) \ln p(y_i|\mu_i) \,d\pi
+
c 
\\
& =
\ln p(y_i|\mu_i) 
+ 
\int q(\pi) \ln p(\mu_i|\pi) \,d\pi
+
c, 
\end{align*}
{\normalsize\noindent where we have defined the constant term\par}
\begin{align*}
c 
& := \fitproofline{\int q(\pi) \ln p(\pi) d\pi 
+
\int q(\pi)
\int\prod_{\substack{k=1 \\ k\neq i}}^{n} q(\mu_k) 
\sum_{\substack{k'=1 \\ k' \neq i}}^{n} \ln p(\mu_{k'}|\pi) 
d\mu_{1:n\backslash i}
d\pi
+
\int q(\pi)
\int\prod_{\substack{k=1 \\ k\neq i}}^{n} q(\mu_k) \sum_{\substack{k'=1 \\ k' \neq i}}^{n} 
\ln p(y_{k'}|\mu_{k'}) 
d\mu_{1:n\backslash i}
d\pi,}
\end{align*}
{\normalsize\noindent which does not vary with $\mu_i$.\par} 

\vspace{2mm}

\noindent For the remaining integral term, we then obtain using  
$q(\pi)=\mbox{Dir}(\pi;\alpha^q)$ and $\alpha_s^q:=\sum_{r=1}^m\alpha_r^q$
\begin{align}
\int q(\pi) \ln p(\mu_i|\pi) \,d\pi
& = \int q(\pi) \ln \left(\prod_{j=1}^m \pi_j^{\mu_{i_j}}\right)\,d\pi \notag\\
& = \int q(\pi) \sum_{j=1}^m  \mu_{i_j} \ln \pi_j \,d\pi\notag\\
& =  \sum_{j=1}^m  \mu_{i_j} \int   q(\pi) \ln \pi_j \,d\pi\notag\\
& =  \sum_{j=1}^m  \mu_{i_j} \left(\psi\left(\alpha_j^q\right) - \psi\left(\sum_{j'=1}^m \alpha_{j'}^q\right) \right)\notag\\
& =  \sum_{j=1}^m \mu_{i_j} \left(\psi\left(\alpha_j^q\right) - \psi\left(\alpha_s^q\right) \right),
\end{align}
where $\psi$ denotes the digamma function. We used the Dirichlet log-moment
identity $\mathbb{E}_q(\ln\pi_j)=\psi(\alpha_j^q)-\psi(\alpha_s^q)$. Taken together, we thus
obtain, for $\mu_i=e_j$ with $j=1,...,m$,
\begin{align}
I_{\mu_i}
&= \ln p(y_i|\mu_i)
 + \sum_{r=1}^m \mu_{i_r}
 \left(\psi\left(\alpha_r^q\right)-\psi\left(\alpha_s^q\right)\right)+c\notag\\
&= \ln p(y_i|\mu_i)
 + \psi\left(\alpha_j^q\right)-\psi\left(\alpha_s^q\right)+c.
\end{align}
Here, $\mu_i=e_j$ means that $\mu_{i_r}=1$ for $r=j$ and
$\mu_{i_r}=0$ for $r\ne j$. Because $c$ does not depend on $\mu_i$,
the factor $\exp(c)$ cancels between numerator and denominator when
normalizing $q(\mu_i)$. For $\mu_i=e_j$, the variational update is thus
\begin{align}
q(\mu_i) 
& = 
\frac{1}{\gamma_{\mu_i}} \exp\left(I_{\mu_i}\right)
\notag\\
\Leftrightarrow
q(\mu_i) 
& =
\frac{1}{\gamma_{\mu_i}}
\exp\left(\ln p(y_i|\mu_i) + \psi\left(\alpha_j^q\right) -\psi\left(\alpha_s^q\right) + c\right)
\notag\\
\Leftrightarrow
q(\mu_i) 
& =
\frac{\exp\left(\ln p(y_i|\mu_i) + \psi\left(\alpha_j^q\right) - \psi\left(\alpha_s^q\right)\right)}
     {\sum_{r=1}^m\exp\left(\ln p(y_i|\mu_i = e_r) + \psi\left(\alpha_r^q\right) - \psi\left(\alpha_s^q\right)\right)}.
\end{align}   
Finally, with $\mu_i \in E_m$, we may re-express the above in Multinoulli
distribution
form by defining
\begin{equation}
\pi_{ij}^q 
:= \frac{\exp\left(\ln p(y_i|\mu_i = e_j) + \psi\left(\alpha_j^q\right) - \psi\left(\alpha_s^q\right)\right)}
        {\sum_{r=1}^m\exp\left(\ln p(y_i|\mu_i = e_r) + \psi\left(\alpha_r^q\right) - \psi\left(\alpha_s^q\right)\right)},
\end{equation}
such that 
\begin{equation}
q(\mu_i) = \prod_{j=1}^m \left(\pi_{ij}^q\right)^{\mu_{i_j}}.
\end{equation}
In summary, by reintroducing the algorithm iteration superscripts in place of
the variational superscript $q$, as well as the hypothesis subscript, we have
\begin{equation}
q^{(k+1)}_1(\mu_i) 
= \mbox{Mult}\left(\mu_i; \pi_i^{(k+1)}\right),
\mbox{ where }
\pi_i^{(k+1)} := \left(\pi_{ij}^{(k+1)}\right)_{j = 1,...,m}
\end{equation}
with
\begin{samepage}
\begin{equation}
\pi_{ij}^{(k+1)} := 
\frac{\exp\left(\ln p(y_i|\mu_i = e_j) + \psi\left(\alpha_j^{(k)}\right) - \psi\left(\alpha_s^{(k)}\right)\right)}
        {\sum_{r=1}^m\exp\left(\ln p(y_i|\mu_i = e_r) + \psi\left(\alpha_r^{(k)}\right) - \psi\left(\alpha_s^{(k)}\right)\right)}.
\end{equation}
$\hfill\Box$
\end{samepage}
\vspace{2mm}

\textit{Proof of \eqnref{eq:q1-pi}}
\vspace{2mm}

From \fullref{sec:free-form-mean-field-variational-inference} and the definitions of
\eqnref{eq:p1-inf-1} and \eqnref{eq:p1-inf-2} in the current inference scenario,
we see that the
CAVI algorithm update for the variational distribution on $\pi$ takes the form
\begin{equation}\label{eq:q1-pi-1}
q(\pi) := \frac{1}{\gamma_\pi}
\exp\left(I_\pi\right)
\mbox{ with }
I_\pi := \int \prod_{i'=1}^n q(\mu_{i'}) \ln p(\pi, \mu_{1:n},y_{1:n}) \,d\mu_{1:n}.
\end{equation}
To show the validity of \eqnref{eq:q1-pi}, we
proceed in two steps.
In Step (1), we first rewrite $I_\pi$ in terms of the variational probabilities 
\begin{equation}\label{eq:q1-pi-2}
\pi_{ij}^q := q(\mu_i=e_j),
\end{equation}
i.e., the variational probability of the $i$th participant model indicator
corresponding
to the $j$th $m$-dimensional canonical unit vector. In Step (2), we then
substitute the resulting form of $I_\pi$ in \eqnref{eq:q1-pi-1} and derive the
Dirichlet
form for the variational distribution on $\pi$.

\vspace{2mm}
\noindent \textit{Step (1)}
\vspace{2mm}

\noindent We first note that

\begin{align}
& I_\pi
\notag\\
&  := 
\int \prod_{i'=1}^n q(\mu_{i'}) \ln p(\pi, \mu_{1:n},y_{1:n}) \,d\mu_{1:n} 
\notag\\
& 
= \int \prod_{i'=1}^n q(\mu_{i'}) \ln \left(p(\pi)\prod_{i=1}^{n}p(\mu_i|\pi)p(y_i|\mu_i) \right) \,d\mu_{1:n}
\notag\\
& 
= \int \prod_{i'=1}^n q(\mu_{i'}) \left(\ln p(\pi) + \sum_{i=1}^{n} \ln p(\mu_i|\pi) + \sum_{i=1}^n \ln p(y_i|\mu_i) \right) \,d\mu_{1:n}
\notag\\
& 
= \fitproofline{\int \prod_{i'=1}^n q(\mu_{i'})\ln p(\pi)\,d\mu_{1:n} 
+  \int \prod_{i'=1}^n q(\mu_{i'})\sum_{i=1}^{n} \ln p(\mu_i|\pi)\,d\mu_{1:n}  
+  \int \prod_{i'=1}^n q(\mu_{i'})\sum_{i=1}^n \ln p(y_i|\mu_i)\,d\mu_{1:n}}
\notag\\
& 
= \fitproofline{\ln p(\pi) \int \prod_{i'=1}^n q(\mu_{i'})\,d\mu_{1:n} 
+  \sum_{i=1}^{n} \int q(\mu_i) \ln p(\mu_i|\pi)\,d\mu_i
+  c_1}
\notag\\
&
=  \ln p(\pi) 
+  \sum_{i=1}^{n} \int q(\mu_i) \ln p(\mu_i|\pi)\,d\mu_i
+  c_1,
\end{align}
where we have defined the constant term
\begin{equation}
c_1 := \int \prod_{i'=1}^n q(\mu_{i'})\sum_{i=1}^n \ln p(y_i|\mu_i)\,d\mu_{1:n},
\end{equation}
which does not depend on $\pi$. 
\vspace{2mm}

\noindent Substituting the functional forms of $p(\pi)$ and $p(\mu_i|\pi)$ yields 
\begin{align}
I_\pi
& 
=  \ln \left(\frac{\Gamma\left(\sum_{j = 1}^{m} \alpha_j \right)}{\prod_{j = 1}^{m} \Gamma(\alpha_j)}
\prod_{j=1}^{m} \pi_j^{\alpha_j - 1}\right) 
+  \sum_{i=1}^{n} \int q(\mu_i) \ln \left(\prod_{j=1}^{m} \pi_j^{\mu_{i_j}}\right)\,d\mu_i
+  c_1
\notag\\
& =  \ln \left(\prod_{j=1}^{m} \pi_j^{\alpha_j - 1}\right) 
+  \sum_{i=1}^{n} \int q(\mu_i) \left(\sum_{j=1}^m \mu_{i_j} \ln \pi_j \right) \,d\mu_i
+  c_1 + c_2
\notag\\
& = \sum_{j=1}^m (\alpha_j - 1)\ln \pi_j 
+ \sum_{i=1}^{n} \int q(\mu_i) \left(\sum_{j=1}^m \mu_{i_j} \ln \pi_j \right)\,d\mu_i 
+ c_3,
\end{align}
where we have defined the additional constant terms
\begin{equation}
c_2 :=  \ln \left(\frac{\Gamma\left(\sum_{j = 1}^{m} \alpha_j \right)}{\prod_{j = 1}^{m} \Gamma(\alpha_j)}\right)
\mbox{ and }
c_3 := c_1 + c_2,
\end{equation}
which do not vary with $\pi$. We next consider the remaining integral term 
in further detail. Because $\mu_i$ is by definition an element of the set of 
$m$-dimensional canonical unit vectors $E_m$, we have
\begin{equation}
\int q(\mu_i) \left(\sum_{j=1}^m \mu_{i_j} \ln \pi_j \right) \,d\mu_i
= \sum_{\mu_i \in E_m} q(\mu_i) \sum_{j=1}^m \mu_{i_j} \ln \pi_j  
= \sum_{r=1}^m q(\mu_i=e_r) \ln \pi_r,  
\end{equation}
where the last equality follows because, for $\mu_i=e_r$, component $r$
equals one and all other components equal zero. Next, let
\begin{equation}
\pi_{ij}^q := q(\mu_i=e_j),
\end{equation}
such that $\pi_{ij}^q$ denotes the probability of $\mu_i$ being the $j$th
$m$-dimensional
canonical unit vector. Then,
\begin{equation}
\int q(\mu_i) \left(\sum_{j=1}^m \mu_{i_j} \ln \pi_j \right) \,d\mu_i
= \sum_{j = 1}^m \pi_{ij}^q \ln \pi_j.  
\end{equation}
We thus obtain
\begin{align}
I_\pi 
& = \sum_{j=1}^m (\alpha_j - 1)\ln \pi_j  + \sum_{i=1}^{n}\sum_{j = 1}^m \pi_{ij}^q \ln \pi_j   
+ c_3
\notag\\
& = \sum_{j=1}^m (\alpha_j - 1)\ln \pi_j  + \sum_{j=1}^m\sum_{i = 1}^n \pi_{ij}^q \ln \pi_j  
+ c_3 
\notag\\
& = \sum_{j=1}^m\left((\alpha_j - 1)\ln \pi_j  + \sum_{i = 1}^n \pi_{ij}^q \ln \pi_j\right)
+ c_3
\notag\\
& = \sum_{j=1}^m \left(\alpha_j  + \sum_{i = 1}^n \pi_{ij}^q -1 \right) \ln \pi_j 
+ c_3
\notag\\
& = \sum_{j=1}^m \left(\alpha_j  + \beta_j -1 \right) \ln \pi_j 
+ c_3,
\end{align}
where we have defined
\begin{equation}
\beta_j :=  \sum_{i = 1}^n \pi_{ij}^q.
\end{equation}

\vspace{2mm}
\noindent \textit{Step (2)}
\vspace{2mm}

\noindent Substitution of $I_\pi$ in \eqnref{eq:q1-pi-1} yields

\begin{align}
q(\pi) & = \frac{1}{\gamma_\pi} \exp(I_\pi) \notag\\
\Leftrightarrow
\ln q(\pi) & = \ln(\exp(I_\pi)) - \ln \gamma_\pi \notag\\
\Leftrightarrow
\ln q(\pi) & = I_\pi - \ln \gamma_\pi \notag\\
\Leftrightarrow
\ln q(\pi) & = \sum_{j=1}^m \left(\alpha_j  + \beta_j -1 \right) \ln \pi_j + c_3 - \ln \gamma_\pi \notag\\
\Leftrightarrow
\ln q(\pi) & = \sum_{j=1}^m \left(\alpha_j  + \beta_j -1 \right) \ln \pi_j + c_4, 
\end{align}
where we have defined 
\begin{equation}
c_4 := c_3 -\ln \gamma_\pi.
\end{equation}
We thus see that the logarithm of the variational distribution $q(\pi)$ is of
the form
of the logarithm of the functional kernel of a Dirichlet distribution with
parameters
\begin{equation}
\alpha^q_j := \alpha_j + \beta_j \mbox{ for } j = 1,...,m.
\end{equation}
As the normalization constant for this kernel is well known as the normalization
constant
of the Dirichlet distribution, we thus have  
\begin{equation}
q(\pi) = \frac{\Gamma\left(\sum_{j = 1}^{m} \alpha_j^q \right)}{\prod_{j = 1}^{m} \Gamma(\alpha_j^q)}\prod_{j=1}^{m} \pi_j^{\alpha_j^q - 1}
= \mbox{Dir}(\pi;\alpha^q). 
\end{equation}
In summary, by reintroducing the algorithm iteration superscripts in place of
the variational superscript $q$, as well as the hypothesis subscript, we have
\begin{equation}
q^{(k+1)}_1(\pi) := \mbox{Dir}\left(\pi; \alpha^{(k+1)}\right),
\end{equation}
where
\begin{samepage}
\begin{equation}
\alpha^{(k+1)}_j := \alpha_j + \beta_j^{(k+1)} 
\mbox{ with } 
\beta_j^{(k+1)} := \sum_{i=1}^n \pi_{ij}^{(k+1)}.
\end{equation}
$\hfill\Box$
\end{samepage}
\vspace{2mm}

\textit{Proof of \eqnref{eq:elbo-p1}}
\vspace{2mm}

The ELBO in the current scenario is given by
\begin{equation}\label{eq:elbo-p1-a}
\mbox{ELBO}(q(\pi,\mu_{1:n})) 
:= \int q(\pi,\mu_{1:n})\ln \left(\frac{p(\pi,\mu_{1:n},y_{1:n})}{q(\pi,\mu_{1:n})}\right)\,d\mu_{1:n}d\pi.   
\end{equation}
In the following, we first rewrite the right-hand side of \eqnref{eq:elbo-p1-a} as a
sum of five integral terms. We then consider these integral terms and finally
assemble the form provided in Equation \eqref{eq:elbo-p1}. Thus, we first note that
\begin{align*}
& \mbox{ELBO}(q(\pi,\mu_{1:n})) 
\\
& 
:= \int q(\pi,\mu_{1:n})\ln \left(\frac{p(\pi,\mu_{1:n},y_{1:n})}{q(\pi,\mu_{1:n})}\right)\,d\mu_{1:n}d\pi                           
\\
&  
= \int q(\pi,\mu_{1:n})\left( \ln p(\pi,\mu_{1:n},y_{1:n}) - \ln q(\pi,\mu_{1:n})\right)\,d\mu_{1:n}d\pi                            
\\
&  
= \int q(\pi,\mu_{1:n})\ln p(\pi,\mu_{1:n},y_{1:n})\,d\mu_{1:n}d\pi  
\\
& \quad 
-\int q(\pi,\mu_{1:n})\ln q(\pi,\mu_{1:n})\,d\mu_{1:n}d\pi                                                                      
\\
&  
= \int q(\pi)\int\prod_{i'=1}^n q(\mu_{i'})\ln \left(p(\pi)\prod_{i=1}^{n}p(\mu_i|\pi)p(y_i|\mu_i) \right)\,d\mu_{1:n}d\pi  
\\
& \quad 
- \int q(\pi)\int\prod_{i'=1}^n q(\mu_{i'})\ln \left(q(\pi)\prod_{i=1}^n q(\mu_i)\right)\,d\mu_{1:n}d\pi                             
\\
&  
= \int q(\pi)\int\prod_{i'=1}^n q(\mu_{i'})\left(\ln p(\pi) 
+ \sum_{i=1}^{n} \ln p(\mu_i|\pi) + \sum_{i=1}^n \ln p(y_i|\mu_i) \right)\,d\mu_{1:n}d\pi   
\\
& \quad 
- \int q(\pi)\int\prod_{i'=1}^n q(\mu_{i'})\left(\ln q(\pi) + \sum_{i=1}^n \ln q(\mu_i)\right)\,d\mu_{1:n}d\pi    
\\
&  
= \int q(\pi)\int\prod_{i'=1}^n q(\mu_{i'})\ln p(\pi) \,d\mu_{1:n}d\pi    
\\
& \quad
+ \int q(\pi)\int\prod_{i'=1}^n q(\mu_{i'})\sum_{i=1}^{n} \ln p(\mu_i|\pi)\,d\mu_{1:n}d\pi  
\\
& \quad
+ \int q(\pi)\int\prod_{i'=1}^n q(\mu_{i'})\sum_{i=1}^n \ln p(y_i|\mu_i)\,d\mu_{1:n}d\pi  
\\
& \quad 
- \int q(\pi)\int\prod_{i'=1}^n q(\mu_{i'})\ln q(\pi)\,d\mu_{1:n}d\pi    
\\
& \quad
- \int q(\pi)\int\prod_{i'=1}^n q(\mu_{i'})\sum_{i=1}^n \ln q(\mu_i)\,d\mu_{1:n}d\pi
\\
&  
= \int q(\pi) \ln p(\pi) \int\prod_{i'=1}^n q(\mu_{i'}) \,d\mu_{1:n}d\pi 
\\
& \quad
+ \int q(\pi)\int\prod_{i'=1}^n q(\mu_{i'})\sum_{i=1}^{n} \ln p(\mu_i|\pi)\,d\mu_{1:n}d\pi 
\\
& \quad
+ \left(\int\prod_{i'=1}^n q(\mu_{i'})\sum_{i=1}^n \ln p(y_i|\mu_i)\,d\mu_{1:n}\right)\left(\int q(\pi)\,d\pi\right)
\\
& \quad 
- \int q(\pi)\ln q(\pi)\int\prod_{i'=1}^n q(\mu_{i'})\,d\mu_{1:n}d\pi   
\\
& \quad
- \left(\int\prod_{i'=1}^n q(\mu_{i'})\sum_{i=1}^n \ln q(\mu_i)\,d\mu_{1:n}\right)\left(\int q(\pi)\,d\pi\right)
\\
& \quad 
= \int q(\pi) \ln p(\pi) d\pi
\\
& \quad    
+ \int q(\pi)\int\prod_{i'=1}^n q(\mu_{i'})\sum_{i=1}^{n} \ln p(\mu_i|\pi)\,d\mu_{1:n}d\pi  
\\
& \quad
+ \int\prod_{i'=1}^n q(\mu_{i'})\sum_{i=1}^n \ln p(y_i|\mu_i)\,d\mu_{1:n} 
\\
& \quad
- \int q(\pi)\ln q(\pi) d\pi    
\\
& \quad 
- \int\prod_{i'=1}^n q(\mu_{i'})\sum_{i=1}^n \ln q(\mu_i)\,d\mu_{1:n}.
\end{align*}
Here, the integrals of $q(\pi)$ and of the variational distributions over
indicators that do not occur in the respective logarithmic term equal one
because these distributions are normalized.
We next consider the remaining integral terms in turn, using the variational
probabilities
\begin{equation}
\pi_{ij}^q := q(\mu_i=e_j).
\end{equation}
For the first term, we have
\begin{align}
T_1 
&
:= 
\int q(\pi) \ln p(\pi) d\pi
\notag\\
&
= \int q(\pi) \ln\left(\frac{\Gamma\left(\sum_{j = 1}^{m} \alpha_j \right)}{\prod_{j = 1}^{m} \Gamma(\alpha_j)}\prod_{j=1}^{m} \pi_j^{\alpha_j - 1} \right)\,d\pi
\notag\\
&
= 
\int q(\pi) 
\left(
  \ln \Gamma\left(\sum_{j = 1}^{m} \alpha_j \right)
- \sum_{j=1}^m \ln  \Gamma(\alpha_j)
+ \sum_{j=1}^m (\alpha_j - 1) \ln \pi_j
\right)
\,d\pi
\notag\\
&
= 
  \int q(\pi) \ln \Gamma\left(\sum_{j = 1}^{m} \alpha_j \right) \,d\pi
- \int q(\pi) \sum_{j=1}^m \ln  \Gamma(\alpha_j) \,d\pi
+ \int q(\pi) \sum_{j=1}^m (\alpha_j - 1) \ln \pi_j\,d\pi
\notag\\
&
= 
\ln \Gamma\left(\sum_{j = 1}^{m} \alpha_j \right)
- \sum_{j=1}^m \ln  \Gamma(\alpha_j)
+ \sum_{j=1}^m (\alpha_j - 1) \int q(\pi)\ln \pi_j\,d\pi
\notag\\
&
= 
\ln \Gamma\left(\sum_{j = 1}^{m} \alpha_j \right)
- \sum_{j=1}^m \ln  \Gamma(\alpha_j)
+ \sum_{j=1}^m (\alpha_j - 1) \left(\psi(\alpha_j^q) - \psi(\alpha_s^q)\right).
\end{align}
For the second term, we have
\begin{align}
T_2 
&
:= 
\int q(\pi)\int\prod_{i'=1}^n q(\mu_{i'})\sum_{i=1}^{n} \ln p(\mu_i|\pi)\,d\mu_{1:n}d\pi 
\notag\\
& 
= \int q(\pi) \sum_{i=1}^n \int q(\mu_i)\ln p(\mu_i|\pi)\,d\mu_i d\pi 
\notag\\
& 
= \int q(\pi) \sum_{i=1}^n \int q(\mu_i)\ln \left(\prod_{j=1}^m \pi_j^{\mu_{i_j}}\right)\,d\mu_i d\pi 
\notag\\
& 
= \int q(\pi) \sum_{i=1}^n \int q(\mu_i) \left(\sum_{j=1}^m \mu_{i_j} \ln \pi_j \right)\,d\mu_i d\pi 
\notag\\
& 
= \int q(\pi) \sum_{i=1}^n \sum_{\mu_i \in E_m} q(\mu_i)\sum_{j=1}^m \mu_{i_j} \ln \pi_j d\pi 
\notag\\
& 
= \int q(\pi) \sum_{i=1}^n \sum_{j=1}^m \pi_{ij}^q \ln \pi_j\,d\pi 
\notag\\
& 
= \sum_{i=1}^n \sum_{j=1}^m \pi_{ij}^q \int q(\pi) \ln \pi_j\,d\pi 
\notag\\
& 
= \sum_{i=1}^n \sum_{j=1}^m \pi_{ij}^q \left(\psi(\alpha_j^q) - \psi(\alpha_s^q)\right).
\end{align}
For the third term, we have
\begin{align}
T_3 
& 
:= 
\int\prod_{i'=1}^n q(\mu_{i'})\sum_{i=1}^n \ln p(y_i|\mu_i)\,d\mu_{1:n}   
\notag\\
& 
= 
\sum_{i=1}^n \int q(\mu_i)\ln p(y_i|\mu_i)\,d\mu_i 
\notag\\
& 
= 
\sum_{i=1}^n \sum_{\mu_i \in E_m} q(\mu_i) \ln p(y_i|\mu_i)
\notag\\
& 
= 
\sum_{i=1}^n \sum_{j=1}^m \pi_{ij}^q l_{ij},
\end{align}
where we have defined the log model evidence terms
\begin{equation}
l_{ij} := \ln p(y_i|\mu_i=e_j).
\end{equation}
For the fourth term, we have
\begin{align*}
T_4
&
= 
\int q(\pi)\ln q(\pi) d\pi   
\\
&
= \int q(\pi) \ln\left(\frac{\Gamma\left(\sum_{j=1}^m\alpha_j^q \right)}{\prod_{j=1}^{m}\Gamma(\alpha_j^q)}\prod_{j=1}^{m} \pi_j^{\alpha_j^q-1} \right)\,d\pi
\\
&
= 
\int q(\pi) 
\left(
  \ln \Gamma\left(\sum_{j = 1}^{m}\alpha_j^q\right)
- \sum_{j=1}^m\ln\Gamma(\alpha_j^q)
+ \sum_{j=1}^m(\alpha_j^q-1)\ln \pi_j
\right)
\,d\pi
\\
&
= 
  \int q(\pi)\ln \Gamma\left(\sum_{j=1}^m\alpha_j^q\right) \,d\pi
- \int q(\pi)\sum_{j=1}^m \ln\Gamma(\alpha_j^q) \,d\pi
+ \int q(\pi)\sum_{j=1}^m(\alpha_j^q- 1)\ln \pi_j\,d\pi
\\
&
= 
\ln \Gamma\left(\sum_{j = 1}^{m} \alpha_j^q \right)
- \sum_{j=1}^m \ln  \Gamma(\alpha_j^q)
+ \sum_{j=1}^m (\alpha_j^q - 1) \int q(\pi)\ln \pi_j\,d\pi
\\
&
= 
\ln \Gamma\left(\sum_{j = 1}^{m} \alpha_j^q \right)
- \sum_{j=1}^m \ln \Gamma(\alpha_j^q)
+ \sum_{j=1}^m (\alpha_j^q - 1) \left(\psi(\alpha_j^q) - \psi(\alpha_s^q)\right).
\end{align*}
Finally, for the fifth term, we have
\begin{align*}
T_5  
& = \int\prod_{i'=1}^n q(\mu_{i'})\sum_{i=1}^n \ln q(\mu_i)\,d\mu_{1:n}  
\\  
& = \sum_{i=1}^n \int q(\mu_i)\ln q(\mu_i)\,d\mu_{i}   
\\  
& = \sum_{i=1}^n \sum_{\mu_i \in E_m} q(\mu_i)\ln q(\mu_i) 
\\  
& = \sum_{i=1}^n \sum_{j=1}^m \pi_{ij}^q \ln \pi_{ij}^q.    
\end{align*}
In summary, we thus obtain
\begin{align*}
\mbox{ELBO}(q(\pi,\mu_{1:n})) 
&
= 
  \ln \Gamma\left(\sum_{j = 1}^{m} \alpha_j \right)
- \sum_{j=1}^m \ln  \Gamma(\alpha_j)
+ \sum_{j=1}^m (\alpha_j - 1) \left(\psi(\alpha_j^q) - \psi(\alpha_s^q)\right)
\\
&
\quad + \sum_{i=1}^n \sum_{j=1}^m \pi_{ij}^q \left(\psi(\alpha_j^q) - \psi(\alpha_s^q)\right)
\\
&
\quad + \sum_{i=1}^n \sum_{j=1}^m \pi_{ij}^q l_{ij}
\\
&
\quad - \ln \Gamma\left(\sum_{j = 1}^{m} \alpha_j^q \right)
      + \sum_{j=1}^m \ln \Gamma(\alpha_j^q)
      - \sum_{j=1}^m (\alpha_j^q - 1) \left(\psi(\alpha_j^q) - \psi(\alpha_s^q)\right)
\\
&
\quad - \sum_{i=1}^n \sum_{j=1}^m \pi_{ij}^q \ln \pi_{ij}^q. 
\end{align*}
Equation \eqref{eq:elbo-p1} then follows immediately by replacing the
variational superscripts with the iteration of convergence superscript $(c)$.

\hfill$\Box$
\end{appendices}
\end{document}